\documentclass[12pt]{article}
\usepackage{amsmath}
\usepackage{graphicx}
\usepackage{enumerate}
\usepackage{natbib}
\usepackage{url} 
\usepackage{amsthm,amsfonts,amssymb}
\usepackage{cases}
\usepackage{algorithm} 
\usepackage{algorithmic}  
\usepackage{diagbox}
\usepackage{enumerate}
\usepackage{bm}
\usepackage{multirow}
\usepackage{subcaption}
\usepackage{colortbl}
\usepackage{booktabs}
\usepackage{bbm}
\usepackage[colorlinks,citecolor=blue,urlcolor=blue]{hyperref}
\usepackage{makecell}
\usepackage{setspace}

\newcommand{\blind}{1}

\theoremstyle{plain}
\newtheorem{axiom}{Axiom}
\newtheorem{claim}[axiom]{Claim}
\newtheorem{theorem}{Theorem}[section]
\newtheorem{lemma}[theorem]{Lemma}
\newtheorem{proposition}{Proposition}
\newtheorem{corollary}{Corollary}
\theoremstyle{remark}

\newtheorem{remark}{Remark}
\graphicspath{ {fig/}}

\begin{document}
\def\spacingset#1{\renewcommand{\baselinestretch}%
{#1}\small\normalsize} \spacingset{1}


\if1\blind
{
  \title{\bf Expected Shortfall Regression via Optimization}
  \author{Yuanzhi Li$^\dagger$, Shushu Zhang$^\dagger$, Xuming He$^\ddag$\\ 
  University of Michigan$^{\dagger}$,  Washington University in St. Louis$^\ddag$} 
  \maketitle
  \begingroup
\renewcommand\thefootnote{}\footnotetext{
\textit{The first two authors contributed equally to this work.}
}
\endgroup
} \fi

\if0\blind
{
 \bigskip
  \bigskip
  \bigskip
  \begin{center}
    {\Large\bf Expected Shortfall Regression via Optimization}
\end{center}
  \medskip
} \fi

\bigskip

\begin{abstract}
To provide a comprehensive summary of the tail distribution, the expected shortfall is defined as the average over the tail above (or below) a certain quantile of the distribution. The expected shortfall regression captures the heterogeneous covariate-response relationship and describes the covariate effects on the tail of the response distribution. Based on a critical observation that the superquantile regression from the operations research literature does not coincide with the expected shortfall regression, we propose and validate a novel optimization-based approach for the linear expected shortfall regression, without additional assumptions on the conditional quantile models. While the proposed loss function is implicitly defined, we provide a prototype implementation of the proposed approach with some initial expected shortfall estimators based on binning techniques. With practically feasible initial estimators, we establish the consistency and the asymptotic normality of the proposed estimator. The proposed approach achieves heterogeneity-adaptive weights and therefore often offers efficiency gain over existing linear expected shortfall regression approaches in the literature, as demonstrated through simulation studies. 
\end{abstract}

\noindent%
\vfill

\newpage
\spacingset{1.9} 


\newtheorem{innercustomcond}{Condition}
\newenvironment{customcond}[1]
  {\renewcommand\theinnercustomcond{#1}\innercustomcond}
  {\endinnercustomcond}

\def\E{E}
\def\var{\text{var}}
\def\Pr{\text{Pr}}
\def\P{\text{P}}
\def\convergeas{\stackrel{a.s.}{\longrightarrow}}
\def\convergeip{\stackrel{P^*}{\longrightarrow}}
\def\converged{\stackrel{d}{\longrightarrow}}
\def\d{\text{d}}
\def\Op{O_{\P}}
\def\op{o_{\P}}
\def\sumn{\sum_{i=1}^n}
\def\En{\mathbb{E}_n}
\def\Gn{\mathbb{G}_n}
\def\argmin{\mathop{\mathrm{argmin}}}

\def\ddfrac#1#2{\displaystyle\frac{\displaystyle #1}{\displaystyle #2}}
\newcommand{\ubar}[1]{\mkern1.5mu\underline{\mkern-1.5mu #1\mkern-1.5mu}\mkern1.5mu}
\def\hgamma{\hat{\gamma}}
\def\bh{\bar{h}}
\def\ubh{\ubar{h}}
\newcommand{\itemEq}[1]{%
         \begingroup%
         \setlength{\abovedisplayskip}{0pt}%
         \setlength{\belowdisplayskip}{0pt}%
         \parbox[c]{\linewidth}{\begin{flalign*}#1&&\end{flalign*}}%
         \endgroup}
\def\sumM{\sum_{m=1}^M}
\def\sumgammaM{\sum_{m=1}^M\hgamma_m}
\def\hxi{\hat{\xi}}
\def\supM{\sup_{m=1,\ldots,M}}
\def\sumd{\sum_{m=1}^M\sum_{j=1}^{n/M}}
\def\bxm{\bar{x}_m}
\def\txm{\tilde{x}_m}
\def\tX{\tilde{X}}
\def\supMS{\sup_{\substack{m=1,\ldots,M\\|s-\tau|\leq B\cdot r_n}}}
\def\CXk{X_i - \tilde{x}_m}

\section{Introduction}
The expected shortfall (ES), also known as the conditional value-at-risk (CVaR), measures the conditional mean of an outcome above (or below) a certain quantile. 
Formally, the  $\tau$-th (upper tail) ES of a random variable $Y$ with $\E(|Y|) < \infty$ is defined here as
\begin{equation}
\label{eq::def-SQ}
v_{[Y]}(\tau) = \E[Y\mid Y \geq q_{[Y]}(\tau)], \quad 0 \leq \tau < 1,
\end{equation}
where $q_{[Y]}(\tau)$ is the $\tau$-th quantile of $Y$. We consider the distribution of $Y$ to be continuous throughout the paper. Then, $v_{[Y]}(\tau)$ can be equivalently expressed as
\begin{equation}
\label{eq::def-SQ2}
v_{[Y]}(\tau) = \frac{1}{1-\tau}\int_\tau^1 q_{[Y]}(\alpha) \,\mathrm{d}\alpha.
\end{equation}

 The ES provides a useful summary of the tail distribution and has become a powerful metric in a wide range of applications. 
In finance, the ES is a popular risk measure
as it is coherent \citep{artzner1999coherent, acerbi2002coherence} and quantifies the expected loss of a portfolio under adverse scenarios \citep{jorion2003financial,Topaloglou2022,kang2019data,Barendse2021backtesting}. 
The ES has replaced the quantile as the official risk metric for determining the regulatory minimum capital requirements \citep{B2016}.
Beyond financial applications, the ES is useful in investigating heterogeneous treatment effects when the two potential outcomes differ mainly in the tail \citep{he2010detection,zhang2023high,Chen2024estimation}. As an example, we investigate the health disparities in low birth weight among different ethnic groups, where babies with low birth weight tend to have a higher risk of mortality and long-term health problems \citep{hughes20172500}. 
The ES is also widely used in
supply chain management \citep{soleimani2014reverse,sawik2013selection},
robust machine learning \citep{laguel2021superquantile,laguel2021superquantiles}, and reliability engineering \citep{rockafellar2010buffered}.

In the presence of covariates, we are often interested in studying the relationship between a set of covariates $X$ and the tail of a response variable $Y$.
In this paper, we model the conditional ES of $Y$ given $X$ in a linear framework. Specifically, we assume that the cumulative distribution function of $Y \mid X$ is continuous. At a fixed quantile level $\tau$ ($0 < \tau < 1$), we assume that 
\begin{equation}
v_{[Y\mid X]}(\tau,x) = \E\left[Y\mid Y  \geq  q_{[Y\mid X]}(\tau,x),X=x\right] = x^T\beta,
\label{eq::SQ-model-linear}
\end{equation}
where $q_{[Y\mid X]}(\tau,x)$ is the $\tau$-th quantile of $Y$ given $X = x$, and $\beta$ is the $\tau$-th ES regression coefficient of interest. We suppress the dependence on $\tau$ in $\beta$ here for notational simplicity.

While the estimation and inference for the ES have been well understood in the one-sample case \citep{chen2007nonparametric,zwingmann2016asymptotics,sun2018bootstrapping}, the literature on the linear ES regression model in~\eqref{eq::SQ-model-linear} has gained more attention in recent years. 
Due to the \textit{unelicitability} of the ES \citep{gneiting2011making}, i.e., there does not exist a loss function such that the ES is the minimizer of the expected loss, 
existing approaches for the linear ES regression often rely on quantile estimations or cumulative distribution function estimations that typically require additional modeling assumptions beyond~\eqref{eq::SQ-model-linear}. 
We classify common linear ES regression methods into two categories and provide a brief review of each category.

The first category of approaches uses the equivalent definition of the ES in~\eqref{eq::def-SQ2} and aggregates the conditional quantile or distribution functions in the tail. 
For example, \cite{peracchi2008estimating} and \cite{leorato2012asymptotically} propose to use a weighted average of the quantile regression estimators over a grid of quantile levels. 
More recently, \cite{chetverikov2022weighted} proposes to use an integration of non-parametric estimators of the conditional cumulative distribution function, which does not require parametric models for the quantile regression by leveraging the technique of debiased machine learning \citep{chernozhukov2018double}.

The second category of approaches adopts a joint model for the $\tau$-th quantile and the $\tau$-th ES.  \cite{fissler2016higher} recognizes that the quantile and ES are jointly \textit{elicitable} as a pair. With this fact in mind, \cite{dimitriadis2019joint} and \cite{patton2019dynamic} develop a joint regression framework that estimates the quantile and ES regression simultaneously. However, 
the associated optimization problem is neither smooth nor convex, which poses significant computational challenges.
Under the same joint regression models for the quantile and the ES, 
two-step procedures have been proposed that estimate the $\tau$-th quantile and ES regression sequentially, which are computationally more convenient \citep{barendse2020efficiently,peng2022advances,HTZ2022,zhang2023high,barendse2023expected}.


Additionally, non-parametric estimation of the ES regression is available in the literature. In addition to the work of \cite{peracchi2008estimating} and \cite{leorato2012asymptotically},  
\cite{cai2008nonparametric} and \cite{kato2012weighted} consider kernel-based approaches.
\cite{xiao2014right} uses the connection between the ES and the check-loss function in the quantile regression, and \cite{martins2018nonparametric} uses the extreme value theory. 
More recently, \cite{olma2021nonparametric} considers local-linear estimation based on Neyman-orthogonalized score functions.
Nonetheless, linear or parametric ES regression methods are often preferred for interpretability and statistical inference. In fact, we will show later in the paper that our proposed approach works in the spirit of interpretable machine learning and aims to improve statistical efficiency and interpretability by projecting a nonparametric ES function into a linear model.  

In this paper, we propose a novel optimization formulation for the linear ES regression model in~\eqref{eq::SQ-model-linear} without additional modeling assumptions on the quantile functions. 
This approach is inspired by \cite{rockafellar2013fundamental} and \cite{rockafellar2014random}, which estimate the ES by the superquantile defined as the minimizer of an implicitly defined loss function. 
While the superquantile 
coincides with the ES in the one-sample case in~\eqref{eq::def-SQ}, we make it clear in Section~\ref{sec::2} that the superquantile regression is not the same as the ES regression in general.
Our main contributions can be summarized as follows.
\begin{enumerate}
	\item 
 We propose and validate a new optimization framework for the ES regression that is motivated by but distinctive from \cite{rockafellar2014superquantile}. 
 We show that the ES regression coefficient is the unique minimizer of a convex loss function that depends on some unknown distributional quantities.
	\item 
  Based on this new formulation, we propose a practically feasible ES regression method, 
 and show the consistency and asymptotic normality of the proposed ES estimator.
 \item 
 Under appropriate conditions, we show that the proposed approach is 
 automatically adaptive to data heterogeneity, which explains its superior statistical efficiency over existing ES regression approaches under a wide range of models. 
\end{enumerate}

The remainder of the paper is organized as follows. We provide a new characterization for the ES regression in Section~\ref{sec::2}, which lays the foundation for our proposed estimation method called i-Rock. Then we describe and explain how i-Rock works in problems with discrete covariates in Section~\ref{sec::discrete}, where the operation and the underlying statistical theory are cleaner and easier to understand. We conduct further investigations of the i-Rock method for continuous covariates in Section~\ref{sec::continuous}, where some nonparametric initial estimates of the ES regression functions are used. The asymptotic analysis of the i-Rock estimator shows how it leans automatically towards good weighting schemes in the presence of data heterogeneity.  We demonstrate the practicality and numerical performance of the i-Rock method through numerical investigations in Section~\ref{sec::numerical}   and data applications on the health disparity of low birth weight in Section~\ref{sec::data_application}, followed by some concluding remarks in Section~\ref{sec::conclusion}.

\section{An optimization-based approach for ES regression}\label{sec::2}

Let $Y$ be a random variable, with $\tau \in (0,1)$ as the quantile level of interest, and $X = (1,\tX^T)^T \in \mathbb{R}^{p+1}$ as the covariate vector that includes an intercept term. 
In the one-sample case, \cite{rockafellar2013fundamental} and \cite{rockafellar2014random} propose an optimization-based formulation for the $\tau$-th superquantile of $Y$ that coincides with the $\tau$-th ES of $Y$ in~\eqref{eq::def-SQ2}, i.e., 
\begin{eqnarray}
v_{[Y]}(\tau) &=& \argmin_C\; C + \frac{1}{1-\tau}\int_0^1\max\{0, v_{[Y]}(\alpha) - C\}\,\mathrm{d}\alpha .
\label{eq::RRM-optim-univariate}
\end{eqnarray} 
However,  we find that a direct generalization to superquantile regression  does not generally agree with the ES regression coefficients $\beta$ in \eqref{eq::SQ-model-linear}, even at the population level, i.e., 
\begin{equation}
\label{eq::SQR}
\beta \neq \argmin_{{\theta}}\left[ \E(X^T{\theta}) + \frac{1}{1-\tau}\int_{0}^1\max\{0,v_{[Y-X^T{\theta}]}(\alpha)\}\;\mathrm{d}\alpha\right];
\end{equation}
see Appendix A for a review of the optimization formulation proposed in \cite{rockafellar2013fundamental} and \cite{rockafellar2014random}, and a counterexample that proves~\eqref{eq::SQR}. 

Given that superquantile regression 
is not the solution to expected shortfall, 
we develop a new optimization formulation targeted for the ES regression under linear model~\eqref{eq::SQ-model-linear}. 
We propose a population-level loss function that can identify the $\tau$-th conditional ES of $Y$ given $X$, namely, 
\begin{eqnarray}
L(\theta) 
&=& \E_X \,\left[ \int_{0}^1 \,\rho_\tau\left(v_{[Y\mid X]}(\alpha,X) - X^T\theta\right) \,\mathrm{d}\alpha\right]\label{eq::RRM-correct},
\end{eqnarray}
where $\rho_\tau(u) = \{ \tau - \mathbbm{1} (u<0) \} u$ is the quantile loss function \citep{koenker_2005}, 
and $\E_{X}$ denotes the expectation with respect to the random variable $X$.
We refer to $L(\theta)$ as the improved Rockafellar (i-Rock) loss function, and refer to the process of minimizing the i-Rock loss function as the i-Rock approach. In the case where $|L(\theta)|$ is infinite, our analysis for the ES regression hereafter still  holds by considering the loss function 
$\E_{X}\,\left[  \int_{0}^1 \,\rho_\tau\left(v_{[Y\mid X]}(\alpha,X) - X^T\theta\right) - \rho_\tau\left(v_{[Y\mid X]}(\alpha,X) - X^T\beta^*\right)\,\mathrm{d}\alpha \right]$ for any given $\beta^*$, which is guaranteed to take finite values under the conditions of Theorem \ref{thm::RRM-correct} below.
The i-Rock loss function in~\eqref{eq::RRM-correct} is different from the loss function in \eqref{eq::SQR}, and we will show in Theorem~\ref{thm::RRM-correct} that the i-Rock approach yields the correct coefficients for the ES regression. 
The extension from the one-sample ES loss function in \eqref{eq::RRM-optim-univariate} to the new ES regression loss function in (\ref{eq::RRM-correct}) stands as a substantial improvement that ensures the validity of the optimization-based approach for the ES regression. 

\begin{theorem}\label{thm::RRM-correct}
{\color{black} Suppose the cumulative distribution function of $Y \mid X=x$ is continuous and strictly increasing in a neighborhood of $q_{Y|X} (\tau, x)$}, and
the matrix 
\begin{equation}\label{eq::D1}
    D_1 = \E_X\left[\frac{XX^T}{v_{[Y \mid X]}(\tau,X) - q_{[Y \mid X]}(\tau,X)}\right]
\end{equation}
is positive definite. 
Then, under linear ES regression model \eqref{eq::SQ-model-linear},
\begin{equation}\label{eq::iRock}
    \beta = \argmin_{\theta}  L(\theta),
\end{equation}
where $\beta$ is the true ES regression coefficient and the minimizer is uniquely identified. 
\end{theorem}

Theorem \ref{thm::RRM-correct} shows that the ES regression coefficients are correctly and uniquely identified by minimizing the i-Rock loss function under weak conditions.
{\color{black} A sufficient condition for $D_1$ to be positive definite is that the components of $X$ are linearly independent with probability one. }
The i-Rock formulation allows us to interpret the ES regression problem of $Y$ on $X$ as a quantile regression problem of $Z$ on $X$, where $Z$ is an auxiliary response variable distributed as $Z \mid X \sim  v_{[Y\mid X]}(\xi,X)$ where $\xi\sim U(0,1)$ and $\xi$ is independent of $X$. 
This interpretation coincides with the fact that the $\tau$-th quantile of the ES process at all levels from 0 to 1 is exactly the $\tau$-th ES, due to the monotonicity of $v_{[Y\mid X]}(\alpha, X)$ over $\alpha \in (0,1)$. This serves as an intuitive validation of the proposed i-Rock formulation for the ES regression.


Similar to the formulation in~\eqref{eq::RRM-optim-univariate}, the i-Rock formulation in~\eqref{eq::iRock} is not directly feasible for empirical estimation since the i-Rock loss function \eqref{eq::RRM-correct} involves unknown conditional ES process at all quantile levels from 0 to 1. 
In the remainder of the paper, we formally propose and study a new estimation approach for the ES regression based on the i-Rock formulation, which, in essence, substitutes
$v_{[Y\mid X]}(\alpha, X)$ in \eqref{eq::RRM-correct} with appropriate initial estimators.
Typically, estimation of the ES would have high variability at extreme quantile levels due to small effective sample sizes. As a corollary from Theorem \ref{thm::RRM-correct}, we show that the i-Rock formulation does not necessarily involve the conditional ES at extreme tails.
\begin{corollary}
For any $0<\delta\leq 1$, define
\[
L^{(\delta)}(\theta) = \E_X \left[\int_{\tau- \delta\tau }^{\tau + \delta(1-\tau)}\,\rho_\tau\left(v_{[Y\mid X]}(\alpha,X) - X^T\theta\right)\,\mathrm{d}\alpha\right].
\]
Under the same conditions of Theorem \ref{thm::RRM-correct}, we have
$\beta = \argmin_{\theta}  L^{(\delta)}(\theta)$,
and the minimizer is uniquely identified.
\label{coro::truncation}
\end{corollary}
Taking a small positive value $\delta$, Corollary~\ref{coro::truncation} shows that the domain of integration in the i-Rock loss function can be much shortened without jeopardizing identification, validating a class of truncated i-Rock loss functions for the ES regression. Thus, only the conditional ES at quantile levels near $\tau$ are relevant for the i-Rock estimation of the $\tau$-th ES regression. 

In the rest of the paper, all  derivations and expositions consider the case of $\delta =1$, but no essential changes are needed to cover the case of $\delta \in (0,1)$. In our numerical studies, we use $\delta =0.5$ throughout the paper.
\section{The i-Rock approach for discrete covariates} \label{sec::discrete} 
To demonstrate the i-Rock approach for ES regression, we start with the simple setting with discrete covariates and study its theoretical properties. 
We consider covariates that take only $M$ distinct values, namely, $\{x_1,\ldots,x_M\}$, where $x_m \in \mathbb{R}^{p+1},m=1,\ldots,M,$ including an intercept term, and $M$ is a fixed number that does not depend on the sample size. To simplify notations, in this section, we assume an equal number of $i.i.d.$ observations at each covariate value with a total sample size of $n$, i.e., 
\begin{equation}\label{eq::fix_design}
    \{(x_m,Y_{mj}):\; m = 1,\ldots,M;\, j = 1\ldots,n/M\},
\end{equation}
and we write $v_{[Y \mid X]}(s,x_m)$ as $v_m(s)$ and $q_{[Y \mid X]}(s,x_m)$ as $q_m(s)$. A simple initial ES estimator is the empirical ES at each distinct value of the covariates, namely, 
\begin{equation}
\hat{v}_m(s) = \ddfrac{\sum_{j=1}^{n/M} Y_{mj}\bm{1}\{Y_{mj} \geq \hat{q}_m(s)\}}{(1-s)n/M},  m=1,\ldots,M,
\label{eq::empirical-SQ}
\end{equation}
where $\hat{q}_m(s)$ is the empirical $s$-th quantile of the 
responses at $x_m$.
Under Model~\eqref{eq::SQ-model-linear}, the i-Rock estimator with discrete covariates is obtained through
\begin{eqnarray}
\widehat{\beta}
&=& \argmin_{\theta} \; 
\sum_{m=1}^M\sum_{t=1}^T \; \rho_\tau\left(\hat{v}_m(\alpha_t) - x_m^T{\theta}\right)\label{eq::mRock-est-discrete},
\end{eqnarray}
where $\alpha_1,\ldots,\alpha_T$ is an equally-spaced grid over the interval $(0,1)$ for a sufficiently large $T$, with a pre-specified $\delta \in (0,1)$.  Computationally, $\widehat{\beta}$ in~\eqref{eq::mRock-est-discrete} can be obtained by the $\tau$-th quantile regression of $\{\hat{v}_m(\alpha_t); t=1,\ldots,T, m=1,\ldots,M\}$ on $\{x_m;t=1,\ldots,T,m=1,\ldots,M\}$, for which efficient numerical algorithms exist \citep[Section 6]{koenker_2005}.

\subsection{Asymptotic results}
To study the statistical properties of the i-Rock estimator in~\eqref{eq::mRock-est-discrete}, we start with several technical conditions on the data-generating mechanism. 

\begin{customcond}{R-X}
The Gram matrix $D_0 = \sumM x_mx_m^T/M$ is positive definite. 
\label{cond::R-X}
\end{customcond}

\begin{customcond}{R-Y1}
At each $x_m$, the cumulative distribution function of $Y_{mj}$ is continuous. Furthermore, the density function of $Y_{mj}$, denoted as $f_m(y)$, is bounded away from zero and
continuous on the interval $[q_m(\tau)-\varepsilon, q_m(\tau) + \varepsilon]$ for some $\varepsilon>0$.
\label{cond::R-Y-density}
\end{customcond}

\begin{customcond}{R-Y2}
 At each $x_m$, 
 $\E[(Y_{mj}^{+})^2] < +\infty$, where $Y_{mj}^{+} = \max\{Y_{mj},0\}$.
 \label{cond::R-Y-moment}
\end{customcond}
Condition \ref{cond::R-X} ensures the $M$ distinct covariate values are non-degenerate. 
Conditions \ref{cond::R-Y-density} and \ref{cond::R-Y-moment} are relatively weak for the response distribution. 
Under Conditions \ref{cond::R-Y-density} and \ref{cond::R-Y-moment}, $v_m(\alpha)$ is increasing and continuously differentiable with respect to $\alpha$ for $\alpha \in [\tau - \epsilon,\tau + \epsilon]$ for some $\epsilon>0$. 
With the inverse empirical ES estimator denoted as $\hat{h}_m(z) = \inf\{s\in[0,1]:\hat{v}_m(s) \geq z\}$ where $\hat{v}_m(s)$ is the empirical ES estimator in~\eqref{eq::empirical-SQ}, we present the main result for the i-Rock ES estimator in Theorem~\ref{thm::RRM-modified}.

\begin{theorem}
\label{thm::RRM-modified}
Under a fixed discrete design in~\eqref{eq::fix_design} and Conditions \ref{cond::R-X}, \ref{cond::R-Y-density}, and \ref{cond::R-Y-moment}, and suppose
$
D_1 = M^{-1}\sumM \{v_m(\tau) - q_m(\tau)\}^{-1} x_mx_m^T
$
is positive definite, we have
\[
(1-\tau)D_1\left(\widehat{\beta} - \beta\right) = M^{-1}\sumM x_m\left[\tau - \hat{h}_m\{v_m(\tau)\}\right] + \op\left(\frac{1}{\sqrt{n}}\right).
\]
In particular, the i-Rock estimator $\widehat{\beta}$ is consistent for $\beta$ in~\eqref{eq::SQ-model-linear}, and 
\[
\sqrt{n}\left(\widehat{\beta} - \beta\right)\converged \mathrm{N}\left(0,\;\;D_1^{-1}\Omega_1D_1^{-1}
\right),
\]
where 
$
\Omega_1 = M^{-1}\sumM\left[ \sigma_m^2(\tau)\{v_m(\tau) - q_m(\tau)\}^{-2}x_mx_m^T\right],
$
and $(1-\tau)\sigma_m^2(\tau) = \text{var}[Y_m\mid Y_m\geq q_m(\tau)] + \tau[v_m(\tau) - q_m(\tau)]^2$.
\end{theorem}


Theorem \ref{thm::RRM-modified} uncovers the main statistical properties of the i-Rock estimator, namely, consistency and asymptotic normality. 
It also gives an explicit connection between $\widehat{\beta}$ and $\hat{h}_m\{v_m(\tau)\}$, the inverse function of $\hat{v}_m$, via a Bahadur-type representation. 
While the implementation of the i-Rock approach only depends on $\hat{v}_m$ as in~\eqref{eq::mRock-est-discrete}, the first-order asymptotic property of $\widehat{\beta}$ depends directly and only on $\hat{h}_m\{v_m(\tau)\}$. 

Note that our consistency result in Theorem~\ref{thm::RRM-modified} is different from Theorem 3 in \cite{rockafellar2014superquantile} in the sense that the convergence limit for the Rockefellar's estimator is not the ES regression coefficient, as demonstrated in Section~\ref{sec::2}. This underscores 
the improvements achieved through our proposed i-Rock approach for the ES regression.

\subsection{Comparisons with other ES regression methods}\label{subsec::efficiency}
To further understand the performance of the proposed i-Rock approach, we compare its asymptotic efficiency with the following four existing ES regression approaches. Note that our proposed i-Rock approach only assumes the linear ES regression model in~\eqref{eq::SQ-model-linear}, 
while all except for the linearization approach below
assume linearity on both quantile and ES, i.e., $ q_{[Y\mid X]}(\tau) = X^T \eta$, and~\eqref{eq::SQ-model-linear}. 
\begin{enumerate}
    \item The linearization approach (LN) is to linearize the initial ES estimators, i.e., 
    $\hat \beta = \argmin_{u} \sum_{m=1}^M (\hat v_m(\tau) - x_m^T u)^2$, 
where $\hat v_m(\tau)$ is defined in~\eqref{eq::empirical-SQ}. 
    \item 
    The joint approaches (J1 and J2), proposed in \cite{fissler2016higher}, consider modeling the quantile and the ES jointly by minimizing a joint loss function. Here, we compare two different specifications, J1 and J2, used by \cite{dimitriadis2019joint} and \cite{patton2019dynamic}, respectively.
    \item The two-step approach (TS), proposed in \cite{barendse2020efficiently}, is formulated as
    \begin{gather*}
\hat{\eta}  =  \argmin_\eta \sumd \rho_\tau(Y_{mj} - x_m^T\eta),\quad
\hat{\beta}  =  \argmin_\beta \sumd \left[Z_{mj}(\hat{\eta}) - x_m^T\beta\right]^2,
\end{gather*}
where $Z_{mj}(\eta) =  (1-\tau)^{-1}(Y_{mj} - x_m^T \eta)\mathbbm{1}(Y_{mj}\geq x_m^T \eta) + x_m^T \eta$. 
    \item The two-step least squares approach (TSLS) is formulated as 
    \begin{equation*}
        \hat{\eta}  =  \argmin_\eta \sumd \rho_\tau(Y_{mj} - x_m^T\eta),\quad
\widehat{\beta} = \argmin_\theta \sumd \left[(Y_{mj} - x_m^T\theta)^2\cdot \bm{1}\{Y_{mj}\geq x_m^T\hat{\eta}\}\right].
\end{equation*}
\end{enumerate}
We relegate the detailed formulations for these approaches and their asymptotic variance to Appendix C.2
of the Supplementary Material. We shall examine the asymptotic efficiencies of these methods in the following two examples.

\paragraph{Case 3.1} We consider the linear homoscedastic model, namely, $Y_{mj} = x_m^T\eta + \varepsilon_{mj},  (m=1,\ldots,M;\,j=1,\ldots,n_0)$,
where 
$\sum_{m=1}^M x_m x_m^T/M$ is positive definite, and $\varepsilon_{mj}$'s are $i.i.d.$ with continuous density and finite second moment. Then, we have
$\text{AVar}_{iRock} = \text{AVar}_{TS} =\text{AVar}_{LN} =\text{AVar}_{TSLS} \leq (\text{AVar}_{J1}\wedge\text{AVar}_{J2})$,
where $\text{AVar}_{[\cdot]}$ denotes the asymptotic variance for each of the methods. 
Here, the i-Rock, the two-step, the linearization, and the two-step least squares estimators achieve the same efficiency. All of these approaches are at least as efficient as the two joint approaches, and they achieve equal efficiency if and only if $x_m^T\eta$ are constant over $m$. 
In fact, TS, J1, J2, and i-Rock estimators are asymptotically equivalent to the weighted least squares estimator
\begin{equation}\label{eq::WLS-linearize}
\argmin_{{u}}\sum_{m=1}^M w_m\left(\hat{v}_m(\tau) - x_m^T{u}\right)^2,
\end{equation}
with weights $w_m$ proportional to $1$, $\{v_m(\tau)\}^{-2}$,  $\{v_m(\tau)\}^{-3/2}$, 
and $\{v_m(\tau) - q_m(\tau)\}^{-1}$, respectively.
In this case, the joint approaches lose efficiency by incorporating non-constant weights in homoscedastic models, while the i-Rock approach is efficient since the weight $v_m(\tau) - q_m(\tau)$ remains constant over $m$ in homoscedastic models. 

\paragraph{Case 3.2} 
\begin{figure}[tb]
\begin{center}
\includegraphics[scale=0.4]{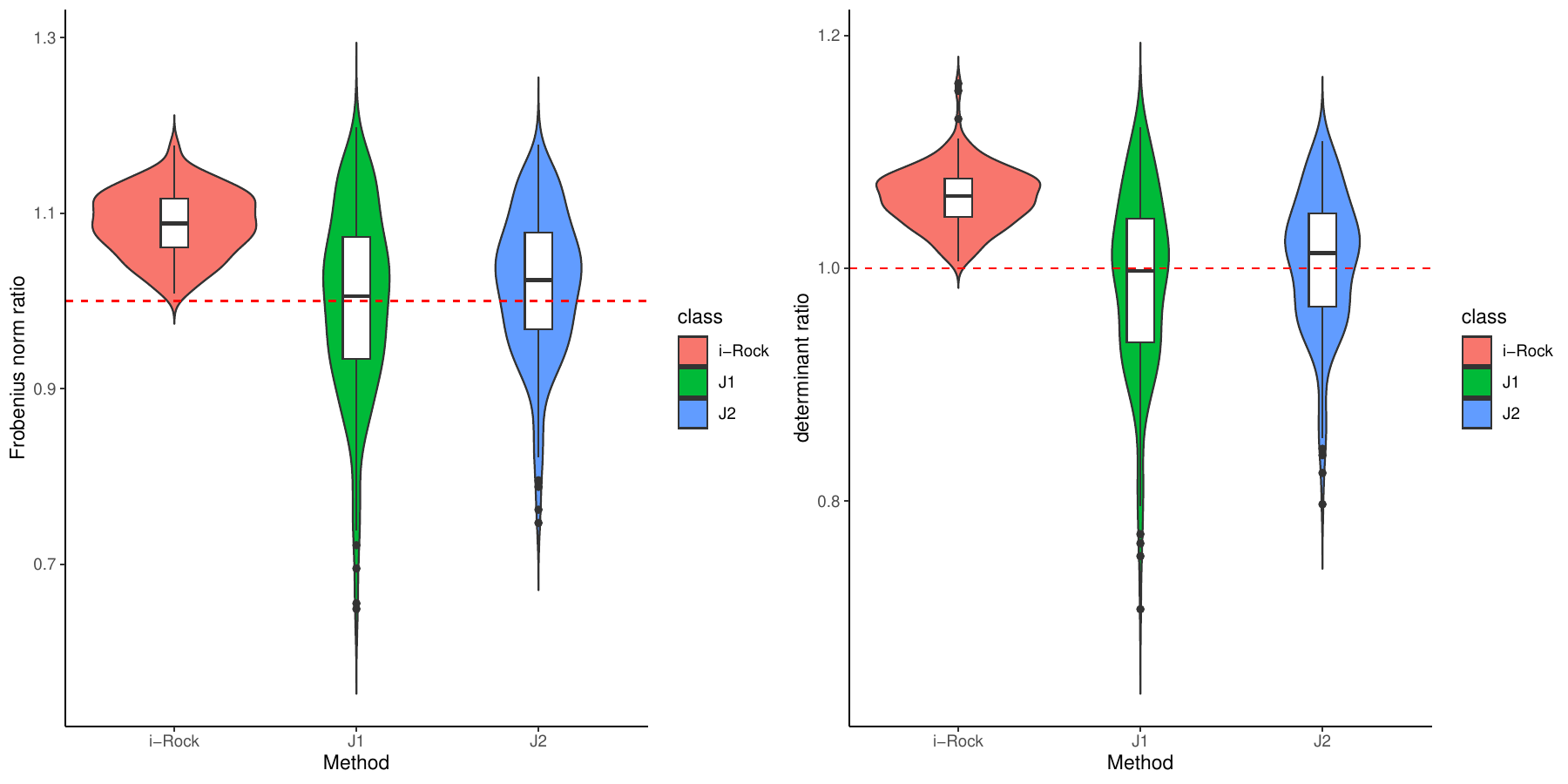}
\end{center}
\caption{\setstretch{1} The violin plot of ARE of the i-Rock and the joint approaches relative to the TS approach, 
at $\tau = 0.9$ under Model~\eqref{eq::comparison-simple} with $p = 3$. Here,
the ARE of one estimator $\widehat{\beta}_1$ relative to the other estimator $\widehat{\beta}_2$ is defined as $\lVert \text{AVar}(\widehat{\beta}_2) \rVert /\lVert \text{AVar}(\widehat{\beta}_1) \rVert$, where $\lVert\cdot \rVert$ can be Frobenius norm (on the left) and determinant (on the right). 
Each element of the covariates $\tX$ takes values independently from $\{i/10;~i=0,1,\ldots,10\}$ with equal probability. 
For $\gamma_1 = (\gamma_{10},\gamma_{11}^T)^T$ and $\gamma_2 = (\gamma_{20},\gamma_{21}^T)^T$,
we fix $\gamma_{10} = \gamma_{20} = 3$ and randomly sample $200$ different values of $\gamma_{11}$ and $\gamma_{21}$ independently and uniformly in the cube $[-1,3]^{3}$. 
}
\label{fig::asymp_effi_nid1}
\end{figure}
Next, we consider the heteroscedastic location-scale shift model, namely,
\begin{equation}
Y_m = x_m^T \gamma_1 + (x_m^T \gamma_2)\cdot \varepsilon_{mj}\quad (m=1,\ldots,M;\,j=1,\ldots,n_0),
\label{eq::comparison-simple}
\end{equation}
where 
$x_m^T \gamma_2>0$, $\varepsilon_{mj}$'s are $i.i.d.$ and follow a (scaled) normal distribution with $v_{[\varepsilon_{mj}]}(\tau) = 0$ and $\var\{\varepsilon_{mj} \mid \varepsilon_{mj} > q_{[\varepsilon_{mj}]}(\tau)\} = 1$. 
Under this more complicated model, we present empirical results of the asymptotic relative efficiency (ARE) 
under $200$ randomly sampled model parameter values in Figure~\ref{fig::asymp_effi_nid1}. 
The i-Rock approach is more efficient than the TS approach in almost all cases, even though the latter requires an additional modeling assumption on the quantile function.
The key to this improvement is the implicit weighting induced by the i-Rock loss function, as will be discussed in Section~\ref{subsec::asy_linear}. 
Under Model~\eqref{eq::comparison-simple}, the i-Rock approach is asymptotically equivalent to~\eqref{eq::WLS-linearize}
with $w_m \propto (x_m^T \gamma_2)^{-1}$.
While $w_m$ are not optimal, they are beneficial for efficiency since $w_m$ reflects the dispersion of $Y_m$ on the right tail and down-weights data with higher conditional tail variance. 
On the other hand, the performance of the joint approaches varies heavily with $\gamma_1$ and $\gamma_2$ in~\eqref{eq::comparison-simple},
resulting in up to a $20\%$ increase or decrease in efficiency compared to the TS approach.

In summary, none of the approaches considered here are universally the most efficient since none can achieve optimal asymptotic weights. However, the i-Rock approach is more efficient in most cases since its asymptotic weight $w_m \propto \{v_m(\tau) - q_m(\tau)\}^{-1}$ characterizes the data heterogeneity and often aligns closer to the optimal weights.

\section{The i-Rock approach for continuous covariates}\label{sec::continuous}


To provide general guidance and further theoretical investigations of the proposed i-Rock approach, we now consider the i-Rock approach with continuous covariates, applicable also to a mix of continuous and discrete covariates. While the same principles given in Section~\ref{sec::discrete} apply, 
the theoretical analysis in this section is much more involved. The key challenge arises from using initial estimators for the conditional ES under minimal assumptions. 
We start with a general analysis for a broad class of non-parametric initial ES estimators and show that the i-Rock approach is asymptotically equivalent to a weighted linearization of those initial estimators. To be specific, we also consider an example of the initial estimator, namely, the Neyman-orthogonalized locally linear ES estimator. 
The proof of the three theorems in this section can be found in Appendix D
of the Supplementary Material.

We formulate our i-Rock approach for continuous covariates with the idea of \textit{binning} to simplify the subsequent theoretical analysis. 
Suppose $\{(X_i,Y_i):i=1,\ldots,n\}$ is a random sample from the distribution $(X,Y) \sim \Pr$.
In the remainder of the paper, we shall write 
the conditional ES and quantile of $Y\mid X=x$ as 
$v(\tau,x)$ and $q(\tau,x)$, respectively,
omitting the subscript $[Y\mid X]$.
Let $\mathcal{X}\subset\mathbb{R}^{p+1}$ be the sample space of the covariate, and we partition $
\mathcal{X} = \bigcup_{m=1}^M A_m$,
where $A_1,\ldots,A_M$ are non-stochastic, disjoint bins, and the number of bins $M$ may grow with sample size $n$. Within each bin $A_m$, 
let $\bar{x}_m\in\mathbb{R}^{p+1}$ be the geometric center of $A_m$, i.e.,
$\bar{x}_m = |A_m|^{-1}\int x\bm{1}\{x\in A_m\}\,\d x$ with $|A_m|$ being the volume of the bin, and let $\hat{v}(\alpha,\bar{x}_m)$ be an initial estimator of $v(\alpha,\bar{x}_m)$, for $\alpha\in(0,1)$. 
Under Model~\eqref{eq::SQ-model-linear}, the i-Rock estimator is formulated as
\begin{equation}
\widehat{\beta} = \argmin_{{u}\in\mathbb{R}^{p+1}} \;\; \sumM \hgamma_m \int_{0}^1\; \rho_\tau\left(\hat{v}(\alpha,\bar{x}_m) - \bar{x}_m^T{u}\right)\,\d\alpha\label{eq::estimator-rock-binning},
\end{equation}
where $\hgamma_m$ is a user-specified weight for each bin $A_m$.
While we allow many choices of $\hgamma_m$, we would normally set $\hgamma_m$ to be 
the bin size of $A_m$, to adjust for the differences in sample sizes across bins. 
In practice, we may approximate the integration in (\ref{eq::estimator-rock-binning}) by a fine grid over $\alpha\in(0,1)$, which turns~\eqref{eq::estimator-rock-binning}
into a quantile regression problem in terms of computation.

\subsection{Asymptotic linearization of the i-Rock estimator}\label{subsec::asy_linear}
Let $f_{Y\mid X}(y;x)$ be the conditional density function of $Y\mid X = x$.
We further impose the following regularity conditions. 

\begin{customcond}{G-X}\label{cond::bin-covar}
The covariates have bounded support $\mathcal{X}\subset\mathbb{R}^{p+1}$, and have a density function $f_X(x)$ that is uniformly bounded away from $0$ and $+\infty$ on $\mathcal{X}$. Furthermore, the matrix $D_1$ defined in~\eqref{eq::D1} is positive definite.
\end{customcond}

\begin{customcond}{G-Y1}\label{cond::R-bin-Y-density}
At each $x$, $f_{Y\mid X}(y; x)$ is continuous over $y$. There exist constants $\underline{f}$, $\overline{f}$, and $\varepsilon_0 > 0$, such that
$
0 < \;\underline{f} \; \leq \inf_{\substack{(x,y):x\in\mathcal{X}\\|y-q(\tau,x)|\leq \varepsilon_0}}f_{Y\mid X}(y; x) \leq \sup_{\substack{(x,y):x\in\mathcal{X}\\|y-q(\tau,x)|\leq \varepsilon_0}}f_{Y\mid X}(y; x) \leq \;\overline{f}$.
\end{customcond}

\begin{customcond}{G-Y2}\label{cond::bin-SQ}
For each $x$, $q(s,x)$ and $v(s,x)$ are strictly increasing and continuous over $s\in(0,1)$. Furthermore, $q(\tau,x)$ and $v(\tau,x)$ are Lipschitz continuous as a function of $x$. 
\end{customcond}

\begin{customcond}{G-V1}\label{cond::vhat1}
For any $m = 1,\ldots,M$, the initial ES estimator $\hat{v}(s,\bar{x}_m)$ is left-continuous and non-decreasing in $s\in(0,1)$. Furthermore, for any constant $B > 0$ and some sequence $r_n = o(n^{-1/4})$, the initial estimator satisfies
\begin{enumerate}[(i)]
    \item \itemEq{
    \sup_{\substack{m=1,\ldots,M\\s:|s-\tau|\leq B\cdot (r_n + n^{-1/2}  )}} \left|\hat{v}(s,\bar{x}_m) - v(s,\bar{x}_m)\right| = \Op(r_n),
    }
    \item \itemEq{
    \sup_{\substack{m=1,\ldots,M\\s:|s-\tau|\leq B\cdot (r_n + n^{-1/2})}} \left|[\hat{v}(s ,\bar{x}_m) - v(s ,\bar{x}_m)] - [\hat{v}(\tau,\bar{x}_m) - v(\tau,\bar{x}_m)]\right| = \op\left(n^{-1/2}\right).
    }
\end{enumerate} 
\end{customcond}

\begin{customcond}{G-V2}\label{cond::bin-linear}
The weighted aggregation of the initial ES estimators satisfies
\[
\sumM \left[\frac{\hgamma_m \bxm }{v(\tau,\bar{x}_m)-q(\tau,\bar{x}_m)}\{\hat{v}(\tau,\bar{x}_m) - v(\tau,\bar{x}_m)\}\right] = \Op\left(n^{-1/2}\right).
\]
\end{customcond} 
Conditions \ref{cond::bin-covar}, \ref{cond::R-bin-Y-density} and \ref{cond::bin-SQ} are standard in the literature of quantile and ES regression; see, e.g., \cite[Section 4]{koenker_2005}, and \cite{dimitriadis2019joint}.
{\color{black} The bounded support assumption in Condition~\ref{cond::bin-covar} is mathematically necessary if we allow data heterogeneity and are interested in linear ES functions at two different quantile levels. In practical problems, we usually work with covariates with bounded ranges. }
Conditions~\ref{cond::bin-linear} and~\ref{cond::vhat1} impose generally weak assumptions on the initial ES estimators that allow for a broad class of non-parametric ES estimations. 
Condition~\ref{cond::vhat1} requires uniform consistency with a convergence rate slightly better than $n^{1/4}$ and asymptotic equicontinuity of the estimated ES process. 
Condition~\ref{cond::bin-linear} requires that a weighted aggregation of the initial ES estimators over bins enjoys a $n^{1/2}$-rate of convergence.
As discussed in Corollary~\ref{coro::truncation}, our analysis would only concern $\hat{v}(s,x)$ for those $s$ in a local neighborhood of $\tau$, but not at extreme quantile levels, as long as $\hat{v}(s,x)$ is monotonic in $s$. We relegate further discussion of the initial ES estimator to Appendix F of the Supplementary Material. 


Recall from~\eqref{eq::estimator-rock-binning} that $\hgamma_m$ is a weight for each bin in the i-Rock approach, and that the number of bins $M$ may depend on the sample size.
Let diam$(\cdot)$ be the diameter of a set in $\mathbb{R}^{p+1}$, and let $\hat{\pi}_m = n^{-1}\sumn \bm{1}[X_i\in A_m]$ be the proportion of data that fall into the bin $A_m$. 
In the following theorem, we establish the asymptotic relationship between the proposed i-Rock estimator and its initial ES estimators. 

\begin{theorem}
\label{thm::asymptotic-RRM-bin}
Suppose Conditions \ref{cond::bin-covar}, \ref{cond::R-bin-Y-density} and \ref{cond::bin-SQ} hold. In addition, suppose that the binning mechanism and the chosen weights satisfy
$\sup_{m=1,\ldots,M} \text{diam}(A_m) = o(1)$ and $\sup_{m=1,\ldots,M}\left| \frac{\hgamma_m}{\hat{\pi}_m} - 1\right| = \op(1)$.
Given any initial ES estimator that satisfies Conditions \ref{cond::vhat1} and \ref{cond::bin-linear}, the i-Rock estimator in~\eqref{eq::estimator-rock-binning} satisfies
\[
\left(\widehat{\beta} - \beta\right) = D_1^{-1}\sumM \left[\frac{\hgamma_m \bxm }{v(\tau,\bar{x}_m)-q(\tau,\bar{x}_m)}\{\hat{v}(\tau,\bar{x}_m) - v(\tau,\bar{x}_m)\}\right]  + \op(n^{-1/2}),
\]
where $D_1$ is defined in \eqref{eq::D1}. 
In particular, $\widehat{\beta}$ is $\sqrt{n}$-consistent for $\beta$.
\end{theorem}
Theorem~\ref{thm::asymptotic-RRM-bin} implies that $\widehat{\beta}$ is asymptotically equivalent to a weighted linearization of $\hat{v}(\tau,\bar{x}_m)$ over each bin. The i-Rock approach turns a set of non-parametric initial estimators into a parametric estimator via the i-Rock loss function in~\eqref{eq::estimator-rock-binning}. 
Specifically, consider the following weighted least squares (WLS) of the initial $\tau$-th ES estimator on the covariates
\begin{equation}
\widetilde{\beta} = \min_{u\in\mathbb{R}^{p+1}} \sumM \hgamma_mw_m\left(\hat{v}(\tau,\bar{x}_m) - \bxm ^Tu\right)^2,
\label{eq::bin-weighted-linearization}
\end{equation}
where $w_m$ is a set of known weights. With simple linear algebra,
$\widetilde{\beta}$ satisfies
\begin{equation}
    \widetilde{\beta} - \beta
    = \left(\sumM \hgamma_mw_m\bxm \bxm ^T\right)^{-1}\sumM \left[\hgamma_mw_m\bxm\{\hat{v}(\tau,\bar{x}_m) - v(\tau,\bar{x}_m)\} \right]\label{eq::WLS}.
\end{equation}
Combining \eqref{eq::WLS} and Theorem~\ref{thm::asymptotic-RRM-bin}, we note that the i-Rock estimator $\widehat{\beta}$ in~\eqref{eq::estimator-rock-binning} is asymptotically equivalent to $\widetilde{\beta}$ in \eqref{eq::bin-weighted-linearization} with $w_m = [v(\tau,\bar{x}_m) - q(\tau,\bar{x}_m)]^{-1}$.
One key feature of the i-Rock approach is that the weights are implicit and automatic. Direct calculation of the WLS would require estimating the unknown weights $w_m$. Plugging in estimated weights $\hat{w}_m$ may lead to unstable WLS estimates where there is not a sufficient amount of data. 

Another feature of the i-Rock approach is that the weights are adaptive to heterogeneity in a good way. 
Among the class of WLS estimators in~\eqref{eq::bin-weighted-linearization}, the optimal weight $w^*_m\propto \{\var[\hat{v}(\tau,\bar{x}_m)]\}^{-1}$ reflects the heterogeneity of the initial estimators (see, e.g., Section 7 of \cite{wooldridge2010econometric}). 
Although the effective weights for the i-Rock approach may not be optimal, they tend to be well correlated with $w_m^*$. 
\cite{olma2021nonparametric} shows that non-parametric conditional ES estimators often have asymptotic variance in the form of
\begin{equation}
\begin{aligned}
a_n\var[\hat{v}(\tau,\bar{x}_m) \mid X=\bar{x}_m] = &\rho_1\var[Y\mid X=\bar{x}_m, Y\geq q(\tau,\bar{x}_m)] \\
&\;\;+ \rho_2[v(\tau,\bar{x}_m)-q(\tau,\bar{x}_m)]^2 + \op(1),
\end{aligned}
\label{eq::variance-components}
\end{equation}
where $a_n$ is a scaling factor, and $\rho_1$, $\rho_2$ are two constants depending on the construction of $\hat{v}(\tau,x)$.
The weight $w_m = [v(\tau,\bar{x}_m)-q(\tau,\bar{x}_m)]^{-1}$ captures part of the variance in \eqref{eq::variance-components}, and in many cases, the two additive components in \eqref{eq::variance-components} are often well correlated across $x_m$ because they both capture the spread of the conditional distribution on the right tail. 
Therefore, the i-Rock approach can often be more efficient than a simple linearization approach that does not adapt well to the heterogeneity in the data. 

\subsection{Asymptotic normality}\label{subsec::examples}
In this section, we provide a concrete example of the initial ES estimator that satisfies the technical conditions of Theorem~\ref{thm::asymptotic-RRM-bin}, and establish the asymptotic normality of the i-Rock estimator with such an initial ES estimator. 
In particular, 
we fit a bin-wise linear ES regression with a Neyman-orthogonalized score function in \cite{barendse2020efficiently}, namely, 
\begin{equation}\label{eq::ll-SQ}
\begin{aligned}
(\hat{c}_0,\hat{c}_1) = \argmin_{\substack{c_0\in\mathbb{R}\\c_1\in\mathbb{R}^{p}}} \sum_{\substack{X_i\in A_m}} \left[\hat{Z}_i(s) - c_0 - c_1^T(\tilde{X}_i - \tilde{x}_m)\right]^2,\quad
\hat{v}(s,\bxm) = \hat{c}_0,
\end{aligned}
\end{equation}
where $\hat{Z}_i(s) = (1-s)^{-1} \{Y_i - \hat{q}(s,X_i)\}\bm{1}\{Y_i\geq \hat{q}(s,X_i)\} + \hat{q}(s,X_i)$, $X_i = (1,\tX_i^T)^T$, and $\bar{x}_m = (1,\tilde{x}_m^T)^T$.
With this initial ES estimator, we obtain the i-Rock estimator in \eqref{eq::estimator-rock-binning} with
\begin{equation}
\label{eq::ll-gweights}
\hgamma_m = S_{0m} - S_{1m}^T \bm{S}_{2m}^{-1} S_{1m},
\quad 
n^{-1} \bm{X}_m^T \bm{W}_m \bm{X}_m =
\begin{bmatrix}
S_{0m} & S_{1m}^T\\
S_{1m} & \bm{S}_{2m}
\end{bmatrix}.
\end{equation}
This ES estimation approach is similar to that in \cite{olma2021nonparametric} with a key difference that we do not require any specific form of the quantile estimator $\hat{q}(s,x)$, as long as it satisfies Condition G-Q in Appendix D.1 of the Supplementary Material. For models with linear quantile functions, this condition is readily satisfied. 

\begin{theorem}
\label{thm::example-RRM-SQ}
Under Conditions \ref{cond::bin-covar}, \ref{cond::bin-SQ}, and G-Y1', G-A1, G-A2, G-Q in Appendix D.1, 
the initial ES estimator constructed in~\eqref{eq::ll-SQ} satisfies Conditions \ref{cond::vhat1} and \ref{cond::bin-linear}, thus the conclusion of Theorem~\ref{thm::asymptotic-RRM-bin} holds. 
In particular, the i-Rock estimator \eqref{eq::estimator-rock-binning} with initial estimator in~\eqref{eq::ll-SQ} and $\hat \gamma_m$ in~\eqref{eq::ll-gweights} satisfies
\begin{equation}\label{eq::asy_normal}
    \sqrt{n}\left(\widehat{\beta} - \beta\right) \converged \mathrm{N}(0,\; D_{1}^{-1}\Omega_1D_1^{-1}),
\end{equation}
where $\Omega_1 = \E\left[\sigma^2_\tau(X)[v(\tau,X) - q(\tau,X)]^{-2}XX^T\right]$, 
and $(1-\tau)\sigma^2_\tau(x) = \var(Y\mid X = x , \,Y\geq q(\tau,x)) + \tau[v(\tau,x) - q(\tau,x)]^2$.
\end{theorem}

Theorem \ref{thm::example-RRM-SQ} shows that our 
locally linear ES estimator in~\eqref{eq::ll-SQ} satisfies Conditions \ref{cond::vhat1} and \ref{cond::bin-linear} required by Theorem \ref{thm::asymptotic-RRM-bin}, and therefore can be used as an initial estimator for the i-Rock approach. Furthermore, the resulting asymptotic variance-covariance matrix in Theorem \ref{thm::example-RRM-SQ} has the same form as that in Theorem \ref{thm::RRM-modified} for discrete covariates.

Compared to the case with discrete covariates, the main technical challenges behind Theorem \ref{thm::example-RRM-SQ} can be summarized as follows. 
First, the number of bins $M = M_n$ increases with the sample size. Therefore, the uniform convergence rate of the initial estimators over the bins needs to be carefully investigated.
Second, we need to establish the process convergence of $\hat{v}(s,\bxm)$ for a continuum of $s$. Standard empirical process tools do not apply directly to the binned data. Moreover, we also need to explicitly analyze the bias in $\hat{v}(s,\bxm)$ attributed to binning and quantile estimation.

\subsection{Optimal weighting and efficiency comparison}
To achieve optimal weights under the weighted least-squares framework in~\eqref{eq::WLS}, we take it further with additional weights $\omega_m$ on the i-Rock loss function, namely, 
\begin{equation}
\widehat{\beta} = \argmin_{{u}\in\mathbb{R}^{p+1}} \;\; \sumM \hgamma_m \omega_m \int_{0}^1\; \rho_\tau\left(\hat{v}(\alpha,\bar{x}_m) - \bar{x}_m^T{u}\right)\,\d\alpha\label{eq::weighted-iRock}.
\end{equation}
We show that 
  \begin{equation}\label{eq::weight-iRock}
    \omega_m = \frac{v(\tau,x_m) - q(\tau,x_m)}{\sigma^2_\tau(x_m)},
\end{equation}
where $\sigma^2_\tau(\cdot)$ is defined in Theorem~\ref{thm::example-RRM-SQ}, is the theoretically optimal weights in this class. 
\begin{theorem}
    Under the same assumptions of Theorem~\ref{thm::example-RRM-SQ}, the weighted i-Rock estimator \eqref{eq::weighted-iRock}--\eqref{eq::weight-iRock} achieves the minimum asymptotic variance attainable under the weighted least-squares framework~\eqref{eq::bin-weighted-linearization}, with the asymptotic variance
 $    \left[E\left\{\frac{XX^T}{\sigma^2_\tau(X)}\right\}\right]^{-1}.$
\end{theorem}
We also note that the optimally weighted i-Rock estimator is asymptotically equivalent to the optimal joint M-estimator in~\cite{dimitriadis2022efficiency} as well as the optimally weighted two-step approach in~\cite{barendse2020efficiently}. We defer the detailed comparison of the three optimally weighted approaches to Appendix D.7.1. 
In practice, the optimal weights of all these approaches need to be estimated from data. An interesting feature of the (unweighted) i-Rock approach is that it corresponds to implicit weights in the framework of~\eqref{eq::bin-weighted-linearization} that adapt well to data heterogeneity automatically and are usually highly correlated with the optimal weights as discussed at the end of Section \ref{subsec::asy_linear}. On the other hand, for general model classes, the optimally weighted ES regression estimators given here may still have an efficiency gap with the semiparametric efficiency given in~\cite{dimitriadis2022efficiency}.

The comparisons above assume that we have linear (or parametric) quantile function specifications, which is not needed for the
the i-Rock approach. In other words, the i-Rock estimators tend to be more robust against deviations from linear quantile models.

In our empirical work, we find that the (unweighted) i-Rock estimators can be comparable to the optimally weighted two-step estimators when the weights are estimated from the data; see Appendix D.7.2. In those cases, the estimation of the weights leads to some loss of finite-sample efficiency, further arguing for the value of the automatic i-Rock approach without weight estimation for data of moderate sizes.





\section{Numerical investigations}\label{sec::numerical}
In this section, we demonstrate the practical applicability of the i-Rock estimator and
investigate its numerical performance through  simulation studies. 
For discrete covariates, we implement the i-Rock approach in Algorithm 1 from Appendix G.1 of the Supplementary Material based on Section~\ref{sec::discrete} and Corollary~\ref{coro::truncation}.
For continuous or mixed covariates, we adopt a variant of~\eqref{eq::estimator-rock-binning} summarized in Algorithm 2 of the Supplementary Material, which uses the bin-wise local-linear initial ES estimator introduced in Section~\ref{subsec::examples}. 
In particular, we partition the covariate space by binning each covariate. Discrete covariates are naturally partitioned according to their distinct values, while continuous covariates are divided using breakpoints at equally spaced quantiles. For subsequent experiments, we set the number of bins for each continuous covariate as $k = \lceil 1.6 \sqrt{p} \times \{\sqrt{n}/\log(n)\}^{1/p}\rceil$, where $p$ refers to the number of continuous covariates with a slight abuse of notation.

We report simulation studies to compare the performance of the i-Rock approach with that of the two-step approach discussed in Section~\ref{subsec::efficiency}, and to check the approximate normality of the i-Rock estimator in finite-samples. 
We defer the comparison with the ``quantile average method," which involves averaging over a series of linear quantile estimators from quantile levels $\tau$ to $1$, to Appendix G.3 of the Supplementary Material.
In our studies with continuous covariates, 
we consider two types of quantile function estimators in~\eqref{eq::ll-SQ},
namely, the (global) linear quantile and the B-spline quantile estimators. 
The linear quantile estimation uses all the available data to fit a linear quantile regression.
The B-spline quantile estimation uses additive piecewise linear quantile function where two internal knots are placed at the $1/3$ and $2/3$ quantiles of each observed covariate.     
We compare our proposed i-Rock implementation with the two-step 
estimator proposed by \cite{barendse2020efficiently} in terms of the relative bias and the Root Mean Squared Error (RMSE). More specifically, we use 500 replications to
calculate  (1) the relative bias
$(\bar \beta_i - \beta_i)/SD( \hat \beta_i^{(j)} )$, 
where $\beta_i$ is the $i$-th component of $\beta$, $\hat \beta_i^{(j)}$ is the estimator for $\beta_i$ in the $j$-th replicate, and $\bar \beta_i = \frac{1}{500}\sum_{j=1}^{500} \hat \beta_i^{(j)}$; and (2) the Root Mean Squared Error (RMSE) ratio of the two-step approach over the i-Rock approach, with a ratio greater than 1 indicating better efficiency for the i-Rock estimator.  For the remainder of this section, we denote $X_{i,j}$ as the $j$-th component of $X_i$, and consider bounded covariates as in our theory. We add that the proposed estimator still performs well if we are only interested in one quantile level and the linear model holds over an unbounded covariate, or correlated covariates, as demonstrated in Appendix G.4 of the Supplementary Material. 

\paragraph{Case 5.1}
We generate data as a random sample from a linear heteroscedastic model with two-dimensional covariates, namely,
\begin{equation}\label{eq::2d_cont_disc}
    Y_i = \{1 + U \} + (2 + 2 U) X_{i,1} + \{3 + 3 U\} X_{i,2}, \quad i=1,\ldots,n,
\end{equation}
where 
$U$ follows $U(0,1)$, $X_{i,1}$ follows $U(0,4)$,  
and $X_{i,2}$ are independently distributed from $\{0,1\}$ with equal probability.
Figure~\ref{fig::2d_cont_disc} shows the relative bias and the RMSE ratio comparisons for $\tau \in \{0.8,0.9\}$ and $n\in\{5000,10000\}$. The i-Rock estimators with both the linear and the B-spline quantile regression estimation outperform the two-step estimator in bias and RMSE in all settings. This is consistent with the theoretical finding in Section~\ref{subsec::efficiency} that the i-Rock approach offers a better adaptation to heterogeneity. 
At these relatively large sample sizes, we verify that the sampling distribution of the i-Rock estimator $(\hat \beta_0,\hat \beta_1,\hat \beta_2)$ follows a normal distribution with the theoretical asymptotic variances in~\eqref{eq::asy_normal} very closely by Kolmogorov–Smirnov test; see Appendix G.2 of the Supplementary Material. 
\begin{figure}[tb]
\centering 
\begin{subfigure}[b]{0.48\textwidth}
         \centering
         \includegraphics[width = \textwidth]{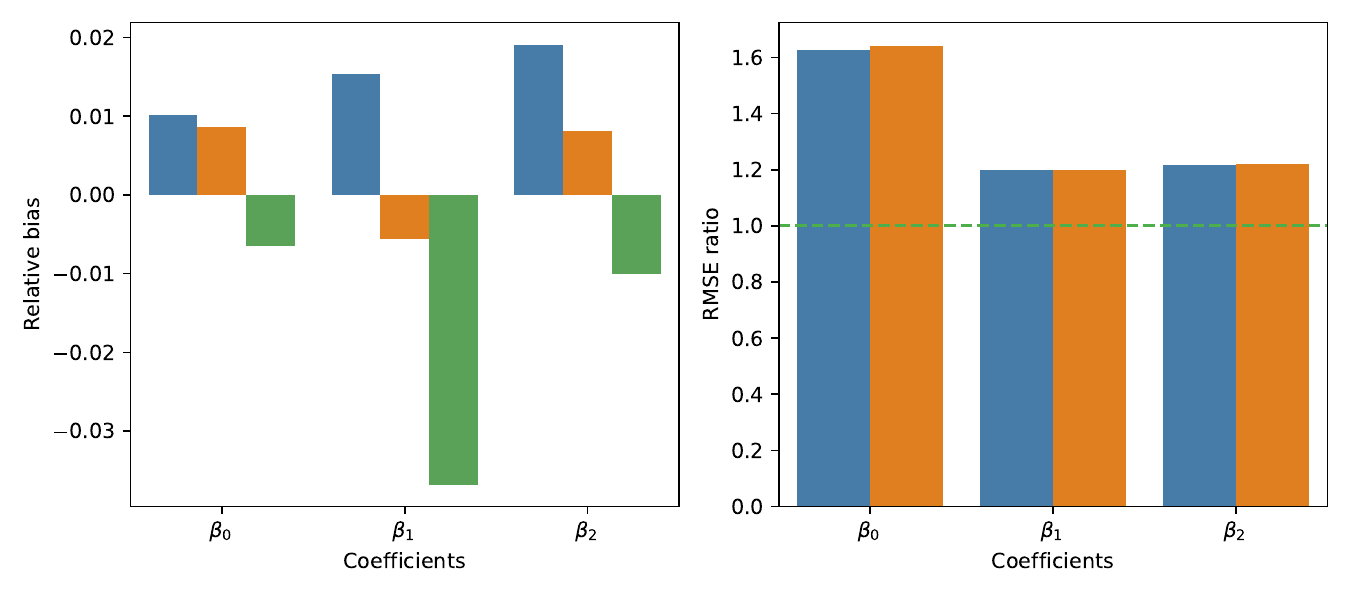}
         \caption{$n = 5000, \tau = 0.8$}
     \end{subfigure}
     \begin{subfigure}[b]{0.48\textwidth}
         \centering
         \includegraphics[width = \textwidth]{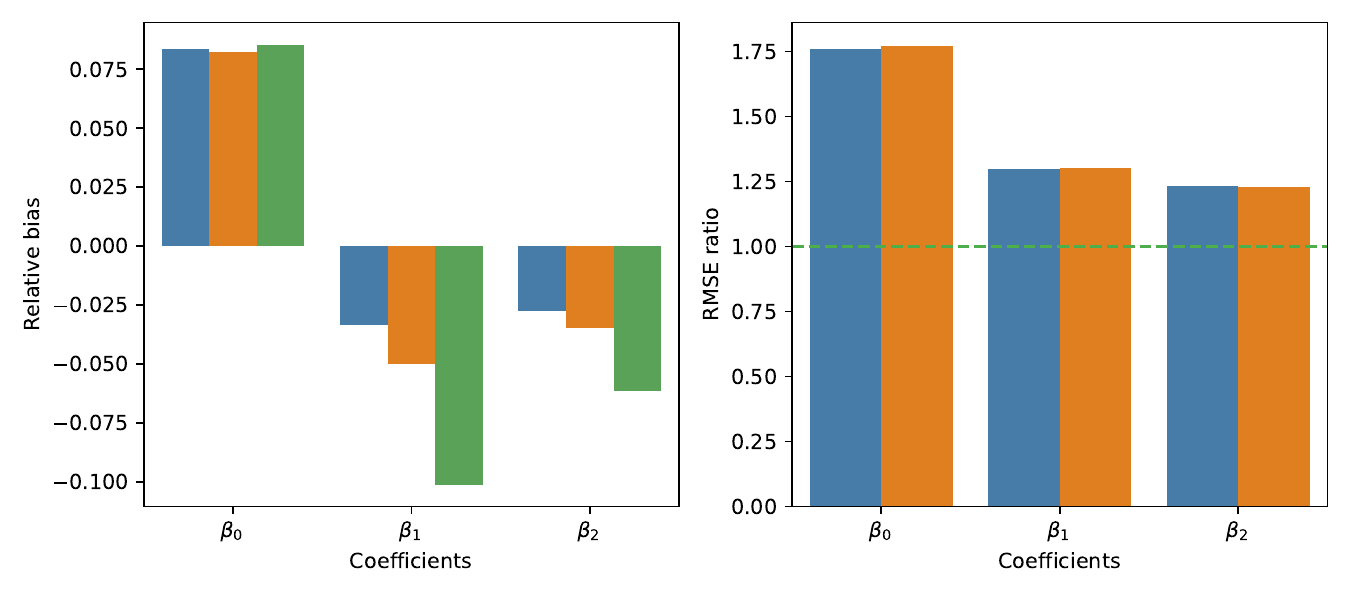}
         \caption{$n = 10000, \tau = 0.8$}
     \end{subfigure}
     \begin{subfigure}[b]{0.48\textwidth}
         \centering
         \includegraphics[width = \textwidth]{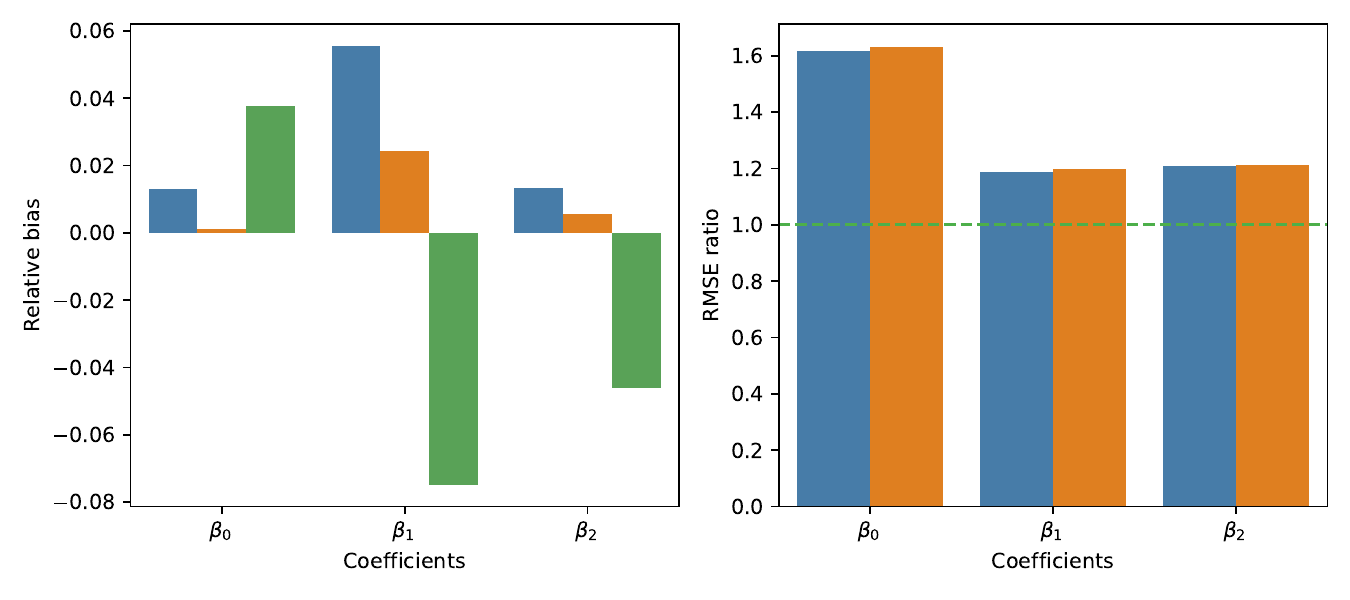}
         \caption{$n = 5000, \tau = 0.9$}
     \end{subfigure}
     \begin{subfigure}[b]{0.48\textwidth}
         \centering
         \includegraphics[width = \textwidth]{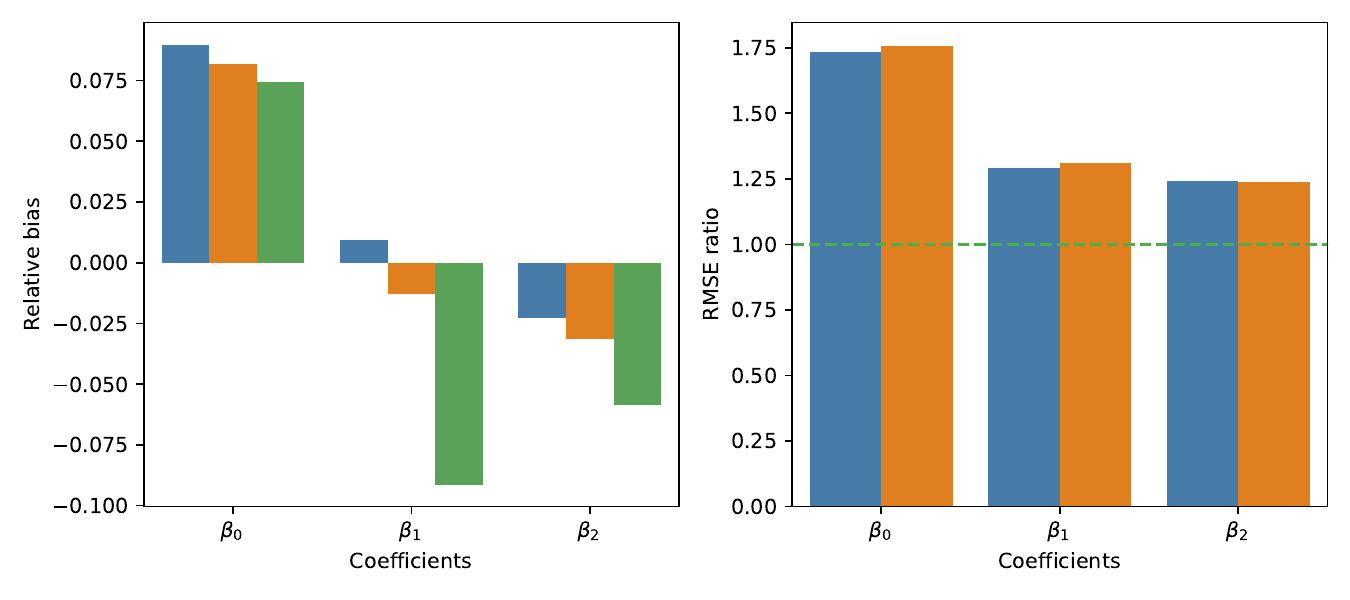}
         \caption{$n = 10000, \tau = 0.9$}
     \end{subfigure}
     \begin{subfigure}[b]{\textwidth}
         \centering
         \includegraphics[width = \textwidth]{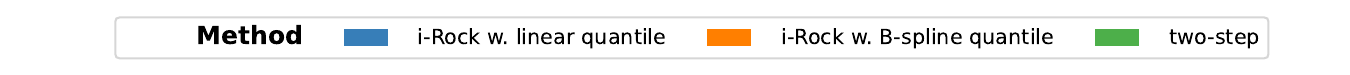}
     \end{subfigure}
\caption{Numerical comparisons of the i-Rock approach (with linear or B-spline quantile function estimation) and two-step estimator under linear heteroscedastic model~\eqref{eq::2d_cont_disc} at various quantile levels and sample sizes.}
\label{fig::2d_cont_disc}
\end{figure}
\paragraph{Case 5.2}
\begin{figure}[tb]
    \centering
     \includegraphics[scale=0.5]{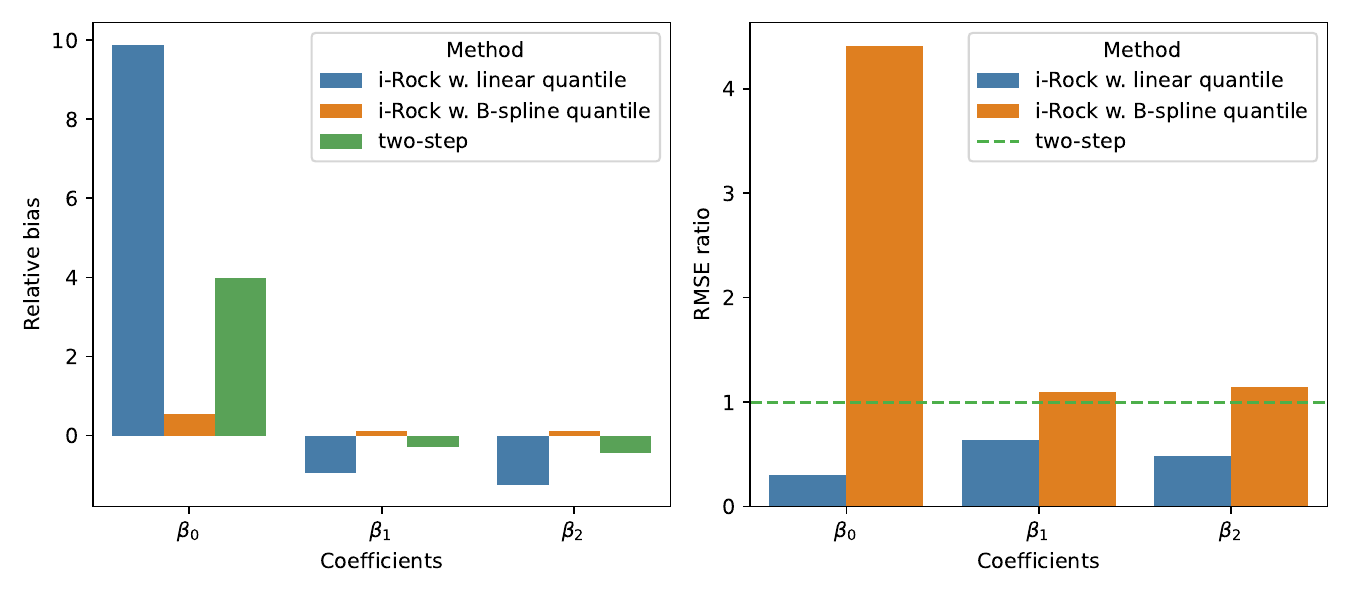}
    \caption{Numerical comparisons of the i-Rock approach (with linear or   B-spline quantile function estimation) and two-step approach under Model~\eqref{eq::sim_2d_nonlinear} at $\tau=0.9$, $n = 10000$.
    }
\label{fig:2d_nl}
\end{figure}
To further illustrate the flexibility of the i-Rock approach under non-linear quantile models, we generate data as a random sample from the following model, where the conditional expected shortfall at level 0.9, but not other levels, is linear in the covariates, namely, 
\begin{equation}\label{eq::sim_2d_nonlinear}
Y_i = -1 + 2 X_{i,1} -3 X_{i,2} + \left(24 X_{i,1}^2 + 12 X_{i,2}^2 + 5\right) (\epsilon_i-\nu_0), \quad i=1,\ldots,n,
\end{equation}
where $(X_{i,1}, X_{i,2})$ is uniformly distributed in a two-dimensional square $[-1,2]^2$,
$\epsilon_i$ follows the skewed-$t_5$ distribution with skewness 2 \citep{theodossiou1998financial} that is independent of the covariates, and $\nu_0$ is the $0.9$-th ES of the distribution of $\epsilon_i$. 
Figure~\ref{fig:2d_nl} shows the results for $\tau=0.9$, $n = 10000$. 
In this case, both the i-Rock with linear quantile function estimation and the two-step estimator suffer from the misspecification of the quantile function, while the i-Rock with the B-spline quantile function estimation is significantly less biased and more efficient than the two-step estimator (with the RMSE ratios as high as 5.2). The asymptotic normality for the i-Rock estimators with B-spline quantile function is numerically confirmed with the Q-Q plots in Figure 4 of the Supplementary Material. 

\paragraph{Case 5.3}
To further evaluate the superior adaptation to data heterogeneity of the i-Rock approach, we generate data from a two-dimensional model where the heterogeneity of the conditional distribution is more evident than in the earlier cases, namely,
\begin{equation}\label{eq:2d}
    Y_i = \{1 - \log(1 - U)\} + (2 + 2 U) X_{i,1} + \{3 - 30 \log(1-U)\} X_{i,2}, \quad i=1,\ldots,n,
\end{equation}
where $U$ is uniformly distributed on $(0,1)$, and $(X_{i,1},X_{i,2})$ are independently distributed from $\text{binomial}(2, 0.5)$. 
The RMSE ratios are summarized in Table~\ref{tab::2dim_RMSE} at several sample sizes.
In this case, the i-Rock estimator 
is significantly more efficient than the two-step estimator, due to the automatic effective weighting schemes of the former. 

\begin{table}[b]
    \centering
    \begin{tabular}{c|ccc}
n &  $\beta_0$ &  $\beta_1$ &  $\beta_2$ \\
\hline
1000 &   7.19 &   7.18 &   1.61 \\
2000 &   9.23 &   7.69 &   1.50 \\
5000 &  10.41 &   8.64 &   1.63 \\
    \end{tabular}
    \caption{RMSE ratio of the two-step estimator over the i-Rock estimator under model~\eqref{eq:2d}.}
    \label{tab::2dim_RMSE}
\end{table}
\section{Data Application}\label{sec::data_application}
Low birth weight is long-known to be associated with increased infant mortality risk and long-term health issues; See \citet{hughes20172500} for a recent review. 
The health disparity for infants with low birth weights among different ethnic groups has drawn increasing attention from policy makers, such as the National Institutes of Health (NIH) and the Centers for Disease Control and Prevention (CDC) \citep{osterman2024births}, among other researchers \citep{burris2017birth,su2021racial,pollock2021trends}. 
In this example, we use ES regression to investigate the possible contributing factors of racial disparities for low birth weight. 
We naturally focus on the lower (left-tail) ES of the birth weight distribution conditional on other factors, and measure the racial disparity as the differences in terms of the lower ES at a given quantile level (e.g., $\tau=0.05$) between two ethnic groups. 
Let $Y$ denote the birth weight, and $R=1$ and $0$ represent one of the disadvantaged group (e.g., Black, Asian, or Hispanic) and the majority group (White in this example), respectively, with $X$ as other factors under consideration. 
For each disadvantaged group, we define the health disparity function at a given covariate $X=x$ and a pre-specified quantile level $\tau$ as
\begin{equation}\label{eq::health_disparity}
\begin{aligned}
    d(\tau,x) 
    = v_{[-Y\mid (R,X)]}\{{\color{blue}1-\tau},(1,x)\} - v_{[-Y\mid (R,X)]}\{{\color{blue}1-\tau},(0,x)\},\\
\end{aligned}
\end{equation}
where the upper tail ES of $-Y$ is used here to replicate the lower tail of $Y$ and be consistent with our theory in this paper.
 
  



We use the 2022 U.S. birth-weight dataset, which is available online at the National Center for Health Statistics (\url{https://www.cdc.gov/nchs/data_access/vitalstatsonline.htm}).
In this example, we focus on male singleton births only, and include the following potential risk factors: mother's race, age, education level, presence of gestational diabetes or hypertension, smoking during the third trimester, number of prenatal visits, receipt of a Special Supplemental Nutrition Program for Women, Infants, and Children (WIC), and parents' marital status. In particular, the WIC is a program designed to help low-income pregnant women, infants, and children up to age 5 receive proper nutrition by providing vouchers for food, nutrition counseling, health care screenings and referrals (\url{https://www.fns.usda.gov/wic/about-wic-glance}).
After combining both the U.S. data and U.S. Territories Data, 
and removing entries with missing values for the variables of interest, we have a total of $n=1,534,031$ observations, including $ 826,367$ as White, $226,327$ as Black, $774,24$ as Asian, and $ 359,118$ as Hispanic individuals. 
  Due to the potential non-linear relationship between birth weight and mother's age and the number of prenatal visits, we discretize age (in years) into three categories: $<20$, $[20,34]$ (as baseline) and $>34$; and we discretize the number of prenatal visits into three categories: $\leq 5$ (as baseline), $[6,10]$, and $> 10$. 

\begin{table}[tb]
\centering
\caption{{ 
Estimated coefficients for the lower ES regression of birth weight at the quantile level $\tau = 0.05$ using the i-Rock approach. 
The numbers in the parenthesis show the standard errors.}}
\label{tab::BW-mRock}
\renewcommand{\arraystretch}{0.9} 
\resizebox{\textwidth}{!}{
\begin{tabular}{lc|lc}
  \toprule
  Covariates & Coefficients &Covariates & Coefficients\\
  \hline
\multirow{2}{*}{\shortstack[l]{\textbf{ (Intercept) }}} &
\multirow{2}{*}{\shortstack{1627.34\\{\tiny( 5.37 )}}}  & \multirow{2}{*}{\shortstack[l]{\textbf{Gestational diabetes}\\{\tiny(baseline: no gestational diabetes)}} } & \multirow{2}{*}{\shortstack{-34.02\\{\tiny( 6.20 )}}}\\
\\ \hline 
\multirow{2}{*}{\shortstack[l]{\textbf{Race $=$ black}\\{\tiny(baseline: white)}} } & \multirow{2}{*}{\shortstack{-249.97\\{\tiny( 7.98 )}}} & \multirow{2}{*}{\shortstack[l]{\textbf{Gestational hypertension}\\{\tiny(baseline: no gestational hypertension)}} } & \multirow{2}{*}{\shortstack{-447.88\\{\tiny( 4.57 )}}}
\\
\\ \hline 
\multirow{2}{*}{\shortstack[l]{\textbf{Race $=$ asian}\\{\tiny(baseline: white)}}} & \multirow{2}{*}{\shortstack{-193.55\\{\tiny( 6.27 )}}} & \multirow{2}{*}{\shortstack[l]{\textbf{Cigarettes at 3rd trimester}\\{\tiny(baseline: no cigarette use at 3rd trimester)}} } & \multirow{2}{*}{\shortstack{-174.59\\{\tiny( 5.68 )}}} \\
\\ \hline 
\multirow{2}{*}{\shortstack[l]{\textbf{Race $=$ hispanic}\\{\tiny(baseline: white)}} }  & \multirow{2}{*}{\shortstack{-61.77\\{\tiny( 1.12 )}}} & \multirow{2}{*}{\shortstack[l]{\textbf{ Mother's age $<20$ }\\{\tiny(baseline: age $[20,34]$)}} } & \multirow{2}{*}{\shortstack{-19.98\\{\tiny( 14.25 )}}} \\
\\ \hline
\multirow{2}{*}{\shortstack[l]{\textbf{Prenatal visits $\in [6,10]$ }\\{\tiny(baseline: $[0,5]$)}} } & \multirow{2}{*}{\shortstack{437.05\\{\tiny( 6.24 )}}}
& \multirow{2}{*}{\shortstack[l]{\textbf{ Mother's age $>34$ }\\{\tiny(baseline: age $[20,34]$)}} } & \multirow{2}{*}{\shortstack{-94.45\\{\tiny( 4.31 )}}} \\
\\ \hline
\multirow{2}{*}{\shortstack[l]{\textbf{Prenatal visits $> 10$ }\\{\tiny(baseline: $[0,5]$)}} } & \multirow{2}{*}{\shortstack{832.55\\{\tiny( 5.85 )}}} & \multirow{2}{*}{\shortstack[l]{\textbf{Receipt of WIC}\\{\tiny(baseline: not receipt of WIC)}}} & \multirow{2}{*}{\shortstack{8.25\\{\tiny( 3.11 )}}}
\\
\\ \hline
\multirow{2}{*}{\shortstack[l]{\textbf{Education: $\geq$ college}\\{\tiny(baseline: high school and below)}}}
&
\multirow{2}{*}{\shortstack{63.21\\{\tiny( 4.36 )}}} 
& \multirow{2}{*}{\shortstack[l]{\textbf{Unmarried}\\{\tiny(baseline: married)}} } & \multirow{2}{*}{\shortstack{-77.56\\{\tiny( 5.99 )}}} \\
\\ \hline
\multirow{2}{*}{\shortstack[l]{\textbf{Education: some college}\\{\tiny(baseline: high school and below)}}}
& \multirow{2}{*}{\shortstack{-4.25\\{\tiny( 4.72 )}}}  \\ \\
\bottomrule
\end{tabular}
}
\end{table}
To estimate the health disparity function \eqref{eq::health_disparity} at $\tau=0.05$, we start with a linear model for the lower ES of the birth weight against the dummy variables for all four races and the other potential risk factors. 
The regression coefficients and their standard errors calculated by the bootstrap
are summarized in Table~\ref{tab::BW-mRock}. With the additive model, the
health disparity in~\eqref{eq::health_disparity} for each race is constant across the other factors $X$ and equals the corresponding ES regression coefficient. For example, the results show that the average lowest 5\% of the birth weight of the Black population is lower (than that of the White population) by around 250 grams. The difference is around 193 grams for Asians and 45 grams for Hispanics.
 


\begin{figure}[tb]
\centering
\includegraphics[scale=0.3]{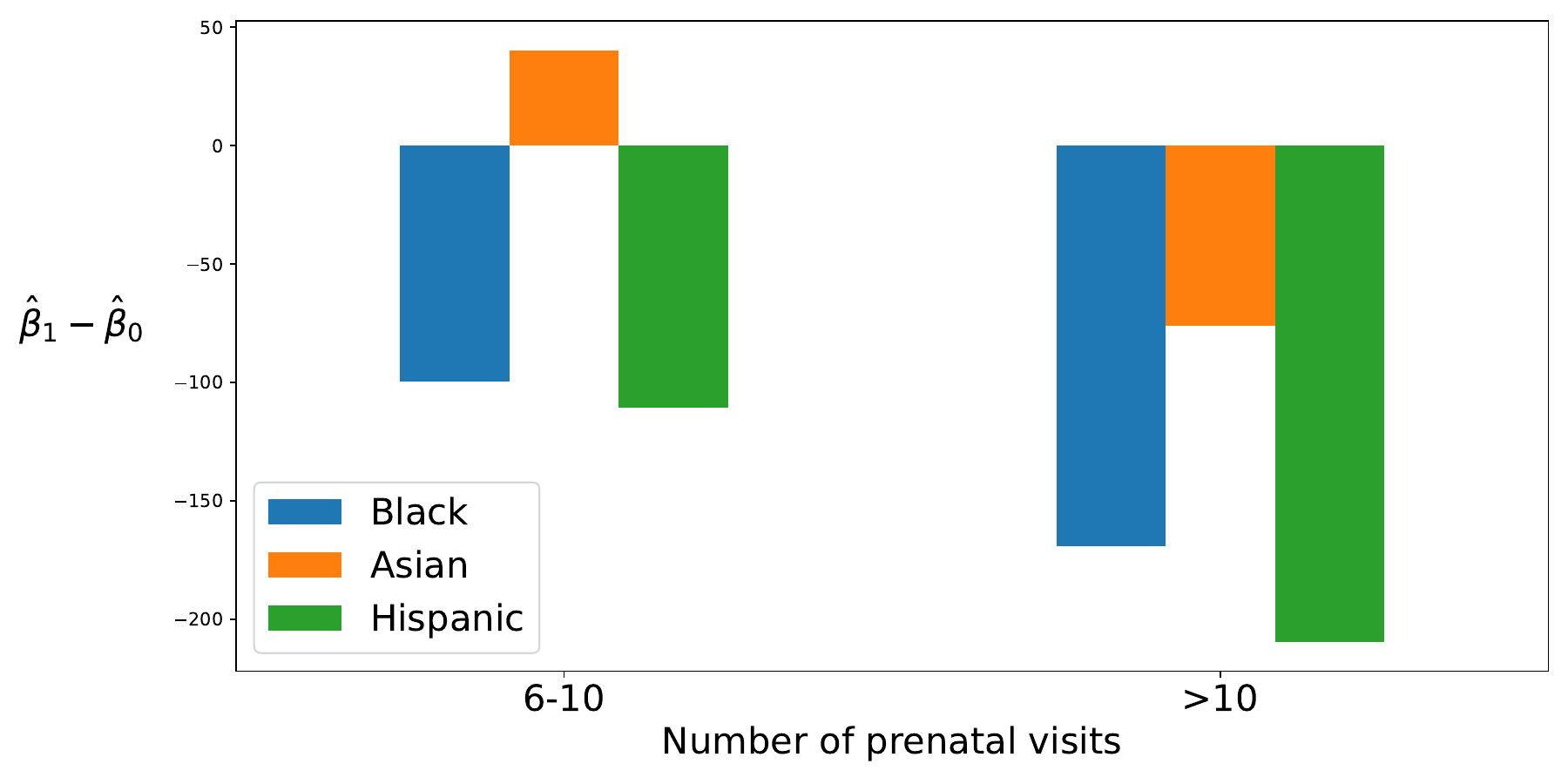}
    \caption{  The quantities shown from part of $\beta_1 - \beta_0$ associated with Equation~\eqref{eq::disparity_driver} are the birth weight disparities  (the lower $0.05$ ES) of the disadvantaged groups for subgroups defined by the number  of prenatal visits: $[6, 10]$ and $>10$, with the subgroup of  $ \leq 5$ prenatal  visits serving as the baseline.  
    } 
\label{fig::prenatal_visit_difference_in_difference}
\end{figure}

Since the additive linear model across all races is quite tentative, we now fit the linear ES regression separately for each race, which amounts to allowing two-way interactions between race and other factors $X$. Under this model,
the health disparity function \eqref{eq::health_disparity} can be written as
\begin{equation}\label{eq::disparity_driver}
    d(\tau,x) = x^T (\beta_1 - \beta_0), 
\end{equation}
where $\beta_1$ and $\beta_0$ represent the ES regression coefficient at level $\tau$ for a disadvantaged racial group and the majority group (i.e., White), respectively. 
  Conditioning on all the other factors, the prenatal visits and gestational hypertension are two noteworthy factors to the health disparity across all disadvantaged races, according to the estimates for $\beta_1-\beta_0$ presented in Figure 13 in Appendix G.5 of the online Supplementary Materials. 
In particular, Figure~\ref{fig::prenatal_visit_difference_in_difference} shows the estimated components of $\beta_1-\beta_0$ corresponding to prenatal visits, representing the differences in the disparity between each disadvantaged racial groups and White individuals when the number of prenatal visits is in $[6,10]$ or $>10$, as compared to the subgroup of no more than 5 prenatal visits.  
\begin{table}[tb]
\centering
\caption{
The percentages of participants receiving WIC and unmarried, respectively, for different ethnic groups, stratified by  the number of prenatal visits. The numbers in parenthesis show the percentages of individuals in the categories of prenatal visits within each race group.  
}
\label{tab::BW-precentage}
\resizebox{\textwidth}{!}{  
\begin{tabular}{c|ccc|ccc|ccc}
  \toprule
    \textbf{Race}&\multicolumn{3}{c|}{\textbf{White}} & \multicolumn{3}{c|}{\textbf{Black}}& \multicolumn{3}{c}{\textbf{Hispanic}}\\
      \textbf{ \# of prenatal visits } & $[0,5]$ & $[6,10]$ & $>10$&$[0,5]$ & $[6,10]$ & $>10$&$[0,5]$ & $[6,10]$ & $>10$\\
       & {\tiny (5.96\%)} & {\tiny(31.19\%)}& {\tiny(62.85\%)} & {\tiny(13.78\%)} & {\tiny(38.74\%)} & {\tiny(47.48\%)}
       &{\tiny(11.87\%)}&{\tiny(39.63\%)}&{\tiny(48.50\%)}\\
       \hline
       Receipt of WIC(\%) & 24.9 & 20.4 & 17.5& 40.2 & 45.8 & 45.5
       &44.71&46.79&45.86\\
      Unmarried(\%) & 44.5 & 28.9 &24.9 &77.7 & 70.0 & 66.6 
&64.52&54.81&49.92\\
\bottomrule
\end{tabular}
}
\end{table}

It is interesting to note that the disparities of Black and Hispanic relative to White are most evident when the number of prenatal visits is high ($>10$ times). To elucidate possible explanations for this phenomenon, we perform a logistic regression for the binary indicator for prenatal visits ($>10$ times versus $\leq5$ times) 
on the other covariates separately for each race and compare the coefficients across races. We find that how WIC enrollment and marital status associate with the frequency of prenatal visits differs across racial groups. In Table~\ref{tab::BW-precentage}, we present the percentages of WIC receipt and unmarried women, stratified by race and number of prenatal visits. The results show that the percentage of the receipts of the WIC {\it decreases} as the number of prenatal visits increases for White women, whereas the percentage actually increases among Black and Hispanic women. 
Despite potentially benefiting from the WIC program interventions, 
Black and Hispanic women with more frequent prenatal visits
may have lower socioeconomic status and wealth compared to their White counterparts. This socioeconomic gap may exacerbate disparities for the subgroup with more than 10 prenatal visits. Marital status may also play a role, as the unmarried percentage decreases drastically with more prenatal visits for White women, but not as evident for Black and Hispanic women. 
In this application, combining ES regression with mean and quantile regression provides additional insights, as detailed in Appendix G.5.
\section{Conclusions and discussions}\label{sec::conclusion}
In this paper, we propose a new optimization-based approach to the ES regression estimation. In contrast to the superquantile regression estimator of \cite{rockafellar2014superquantile},
the proposed estimator is consistent for the ES regression coefficient under heterogeneity. 

 Relative to other methods of ES regression estimation, the i-Rock approach has interesting and desirable properties and has unique challenges of its own. Compared with the computationally simpler two-step approach \citep{barendse2020efficiently}, the i-Rock estimator has several advantages: it automatically incorporates heterogeneity-adaptive weights and does not require a linear assumption in the quantile model. If optimal weights are estimated from the data and then used in these methods, both estimators achieve the same asymptotic efficiency. The joint estimation of quantile and ES regressions offers additional bandwidth for the construction of estimators, and data-adaptively chosen loss functions can lead to the same asymptotic efficiency as the optimally weighted i-Rock method. However, such joint estimation methods have to find solutions that optimize non-convex and non-smooth loss functions, while the i-Rock estimator uses convex optimization. 

The major challenge in the i-Rock approach is the need to rely on initial ES regression estimates at a grid of quantile levels near the target quantile level $\tau$.  This increases the computational burden as well as the theoretical complexity in analyzing the resulting ES regression estimators. On the other hand, if we operate under the framework with linear quantile functions, as is assumed for the two-step or the joint estimation methods, the computational and the theoretical complexities of the i-Rock estimators are greatly reduced.

Intuitively, the proposed i-Rock approach enhances the interpretability of non-linear machine learning models by projecting a nonparametric model to a linear function in the spirit of interpretable machine learning.
We hope that the i-Rock approach proposed in the paper opens a new window of opportunities for the expected shortfall regression modeling and interpretable machine learning associated with expected shortfall analysis both in theory and in practice.

\section{Acknowledgment}
The research was supported in part by the National Science Foundation Awards DMS-2345035 and DMS-1951980. The authors report there are no competing interests to declare. They also extend their gratitude to the anonymous
Associate Editor and referees for their constructive comments, which contributed to
improvements in the paper.
\clearpage
\appendix
\renewcommand{\thesection}{\Alph{section}}
\numberwithin{equation}{section}

\begin{center}
    {\LARGE \bf Supplement to ``Expected Shortfall Regression via Optimization"}
\end{center}

\medskip

\begin{quote}
In this supplement, we provide the proofs, additional numerical results and discussions, and the code to the manuscript ``Expected Shortfall Regression via Optimization". We continue to use the notations and numbered equations from the main manuscript.  
\end{quote}

\tableofcontents
\newpage

{\color{black}
\section{Rockafellar's superquantile regression}
In this section, we start with a brief review of the optimization formulation for the superquantile estimation proposed in \cite{rockafellar2013fundamental} and \cite{rockafellar2014random}. In contrast to the   one-sample case, we demonstrate through a counterexample that the ES regression coefficient in (\ref{eq::SQ-model-linear}) is not a solution to the superquantile regression in \citet{rockafellar2014superquantile}. 



\subsection{The Rockafellar formulation revisited}
Let $Y$ be a random variable and $\tau \in (0,1)$ be the quantile level of interest. 
The $\tau$-th superquantile of $Y$, denoted as $\tilde v_{[Y]}(\tau)$, is defined exactly the same as the $\tau$-th ES of $Y$ in~\eqref{eq::def-SQ2}. 
By Theorem 1 of \cite{rockafellar2014superquantile}, the $\tau$-th superquantile of $Y$ is the minimizer of a loss function with unknown population quantities $\tilde v_{[Y]}(\alpha)$ for $\alpha \in (0,1)$, i.e., 
\begin{eqnarray}
\tilde v_{[Y]}(\tau) &=& \argmin_C\; C + \frac{1}{1-\tau}\int_0^1\max\{0,\tilde v_{[Y]}(\alpha) - C\}\,\mathrm{d}\alpha
\label{eq::RRM-optim-univariate}
\end{eqnarray} 
Given a finite sample, one can substitute the function $\tilde v_{[Y]}(\alpha)$ in \eqref{eq::RRM-optim-univariate} by an empirical estimator to obtain a feasible convex optimization formulation for the superquantile. \cite{rockafellar2018superquantile} proposes efficient numerical algorithms for (\ref{eq::RRM-optim-univariate}) via a dual method that does not require an explicit estimation of $\tilde v_{[Y]}(\alpha)$ in advance. 
Since $\tilde v_{[Y]}(\tau) = v_{[Y]}(\tau)$ in the one-sample case, 
this optimization-based formulation provides a conceptually valuable alternative to the ES estimation.

\subsection{A counterexample on the regression formulation}\label{subsec::SQ_not_ES}
Let $X = (1,\tX^T)^T \in \mathbb{R}^{p+1}$ be the covariate vector that includes an intercept term. 
By Section 3.1 of~\cite{rockafellar2014superquantile}, the superquantile regression coefficient 
is defined as the minimizer of a function directly extended from~\eqref{eq::RRM-optim-univariate}, namely, 
\begin{equation}
\label{eq::SQR}
\tilde \beta = \argmin_{{\theta}}\left[ \E(X^T{\theta}) + \frac{1}{1-\tau}\int_{0}^1\max\{0,\tilde{v}_{[Y-X^T{\theta}]}(\alpha)\}\;\mathrm{d}\alpha\right].
\end{equation}
However, with a simple counterexample, we illustrate that the superquantile regression coefficient $\tilde \beta$ in~\eqref{eq::SQR} does not coincide with the ES regression coefficients $\beta$ from \eqref{eq::SQ-model-linear}, even at the population level.
Consider the following model with $p=1$
\begin{equation}
Y = 1 + \tX\,\varepsilon,\label{eq::toy-RRM}
\end{equation}
where $\tX \sim \Gamma(2,1)$ with $\E(\tX) = 2$, and $\varepsilon\sim U(-1,1)$ that is independent of $\tX$. 
We aim to estimate the $0.5$ ES regression, where the true ES regression coefficients are $\beta_0 = 1$, $\beta_1 = 0.5$.

To find the minimizer to the population level loss function in~\eqref{eq::SQR} at $\tau = 0.5$, we note from Proposition 3 of \citet{rockafellar2014superquantile} that solving (\ref{eq::SQR}) is equivalent to the following two-step procedure:
\begin{eqnarray}
\theta_1^* &\leftarrow& \arg\min_{\theta_1}\underbrace{\left\{\E(\tX^T\theta_1) + \frac{1}{1-\tau}\int_{\tau}^1v_{[Y-\tX^T\theta_1]}(\alpha)\;\mathrm{d}\alpha\right\}}_{ L_1(\theta_1)}, \label{eq::SRQ-1}\\
\theta_0^* &\leftarrow& v_{[Y-\tX^T\theta_1^*]}(\tau)\label{eq::SRQ-2},
\end{eqnarray}
where $\theta_0^*$ and $\theta_1^*$ are the population-level minimizers, and the covariates are $X^T = (1,\tX^T)$. 
We compute the analytical expression for the marginal superquantile of $Z(\theta_1)=Y-\tX^T\theta_1$ for any $\theta_1$ and any quantile level. 
For any $\theta_1\in(-1,1)$, $Z(\theta_1)$ follows tilted double exponential distribution with density function,
\begin{equation*}
f_{Z(\theta_1)}(z;\theta_1) = \begin{cases}
\frac{1}{2}\exp\left\{\frac{z}{1+\theta_1}\right\},\quad &z<0, \\
\\
\frac{1}{2}\exp\left\{\frac{-z}{1-\theta_1}\right\},\quad &z\geq 0.\\
\end{cases}
\end{equation*}
Straightforward probabilistic calculation shows that the marginal superquantile of $Z(\theta_1)=Y-\theta_1 \tX$ is
\begin{equation}
v_{[Z(\theta_1)]}(\alpha) = \begin{cases}
1 + \frac{1}{1-\alpha}\left(\alpha(1+\theta_1)(1-\log\left[\frac{2\alpha}{1+\theta_1}\right]) - 2\theta_1\right), \quad & 0 \leq \alpha \leq \frac{1 + \theta_1}{2},\\
\\
1 + (\theta_1-1)\left(\log\left[ \frac{2(1-\alpha)}{1-\theta_1} \right] - 1\right), \quad & \frac{1 + \theta_1}{2} < \alpha < 1.
\end{cases}
\label{eq::SQ_alpha}
\end{equation}
Substituting Equation (\ref{eq::SQ_alpha}) into the loss function in~\eqref{eq::SQR}, we can obtain an analytical expression for the loss function $L_1(\theta_1)$. 
Moreover, we can compute the first-order derivative to the loss function $L_1(\theta_1)$ as:
\begin{eqnarray*}
\frac{\partial L_1(\theta_1)}{\partial \theta_1} &=& E(\tX) + \frac{1}{1-1/2}\int_{1/2}^1\frac{\partial v_{[Y-\tX^T\theta_1]}(\alpha)}{\partial\theta_1}\;d\alpha\\
&=&\begin{cases}
 - 1 - \log(1-\theta_1),&\quad -1\leq \theta_1 \leq 0, \\
 \\
 2\left\{-\frac{1}{2} - \mathrm{Li}_2(\frac{1}{2}) + \mathrm{Li}_2(\frac{\theta_1+1}{2}) + (\frac{1}{2} - \log2)\log(1+\theta_1)\right\},&\quad 0<\theta_1\leq1,
\end{cases}
\end{eqnarray*}
where $\mathrm{Li}_2(x) = -\int_0^x \log(1-z)/z\,dz$. 

Figure \ref{fig::analytic-loss} below shows the loss function in (\ref{eq::SQR}) and its derivative under Model (\ref{eq::toy-RRM}). 
Since the loss function in (\ref{eq::SQR}) is convex and differentiable \citep{rockafellar2008risk,rockafellar2013fundamental}, we can use a first-order method, e.g., the Newton-Raphson method, to solve the minimization problem (\ref{eq::SQR}). We use the convex optimization toolbox in MATLAB and obtain the population-level minimizer to \eqref{eq::SQR} as $\tilde \beta_1 = 0.7041$, marked by the red line in Figure \ref{fig::analytic-loss}, while the true ES regression coefficient is $\beta_1 = 0.5$, marked by the blue line. Note we focus on the population level loss function, therefore the clear discrepancy between $\tilde \beta_1$ and $\beta_1$ shows the superquantile regression approach proposed by \cite{rockafellar2014superquantile} fails to give the correct coefficients for the ES regression. 
\footnote{
Although the foregoing derivation for $L_1(\theta_1)$ is only valid for $\theta\in(0,1)$, it does not affect our conclusion. This is due to the global convexity of the loss function in \eqref{eq::SQR} \citep{rockafellar2008risk}: A local minimizer within $[-1,1]$ must also be the global minimizer. }

\begin{figure}[ht]
\centering
\begin{minipage}{0.48\textwidth}
\includegraphics[width =\textwidth]{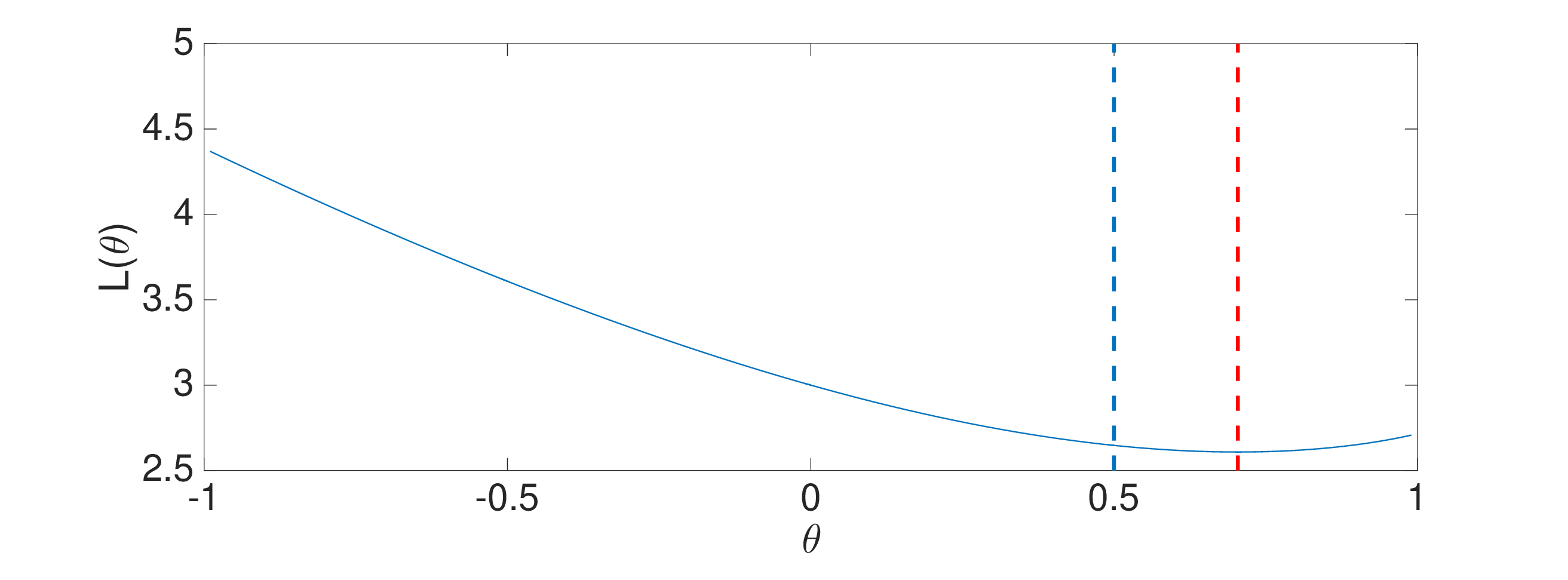}
\end{minipage}
\hfill
\begin{minipage}{0.48\textwidth}
\includegraphics[width = \textwidth]{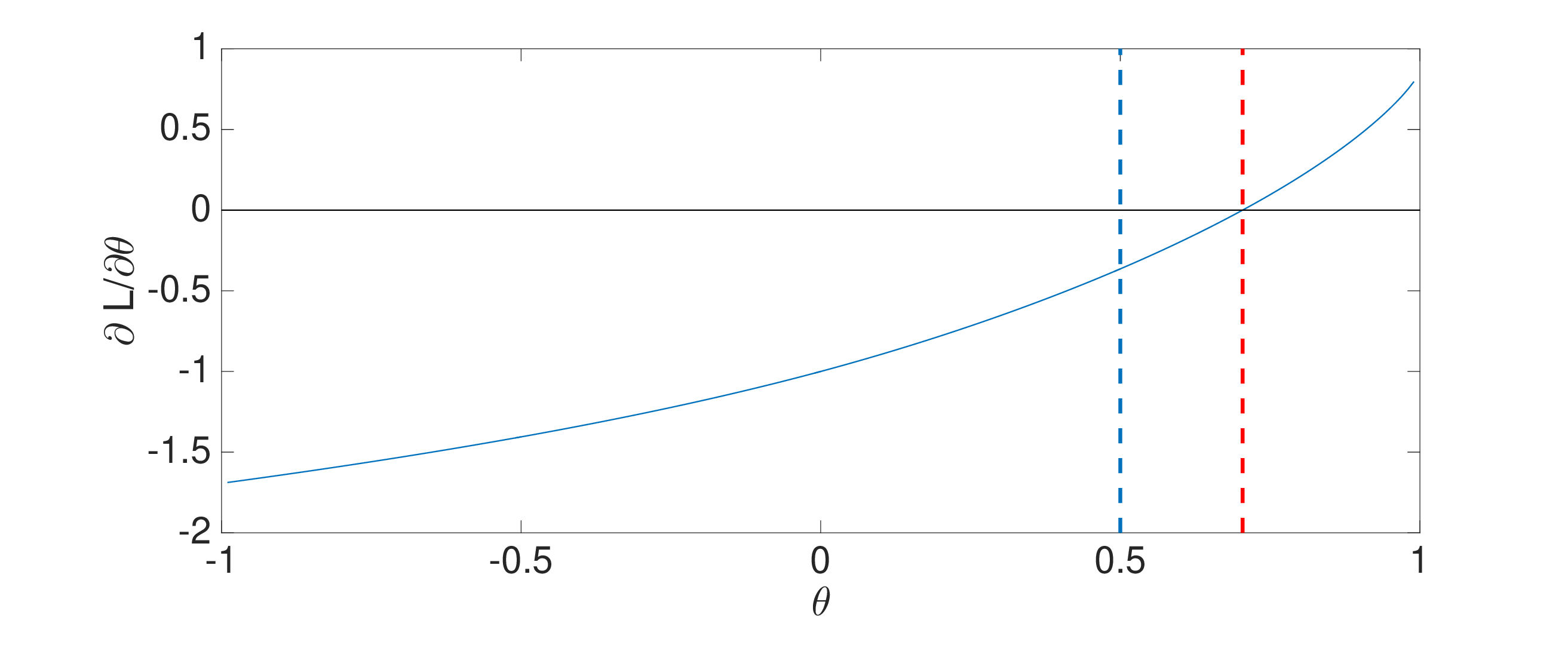}
\end{minipage}
\caption{The population level loss function $L_1(\theta_1)$ (left panel) and its derivative (right panel). The blue dashed line marks the true ES regression coefficient $\beta_1$, while the red one marks the minimizer of $L_1(\theta_1)$.}
\label{fig::analytic-loss}
\end{figure}

We further demonstrate the inconsistency of the superquantile regression approach by \cite{rockafellar2014superquantile} using a numerical experiment. We generate $200$ Monte Carlo datasets from Model (\ref{eq::toy-RRM}), and we consider sample sizes at $n=100$ or $n=1000$. 
Setting $\tau = 0.5$, Figure \ref{fig::numeric-loss} shows the histogram of the estimated slope term $\hat{\beta}_1$ among the $200$ Monte Carlo datasets, solved by the numerical integration method in Section 5.2 of \cite{rockafellar2014superquantile} with $100$ grid points. We can see the histograms are clearly concentrating toward $\tilde \beta_1 = 0.704$, instead of the true ES regression coefficient $\beta_1 = 0.5$.

\begin{figure}[hbt]
\centering
\begin{minipage}{0.49\textwidth}
\includegraphics[width = \textwidth]{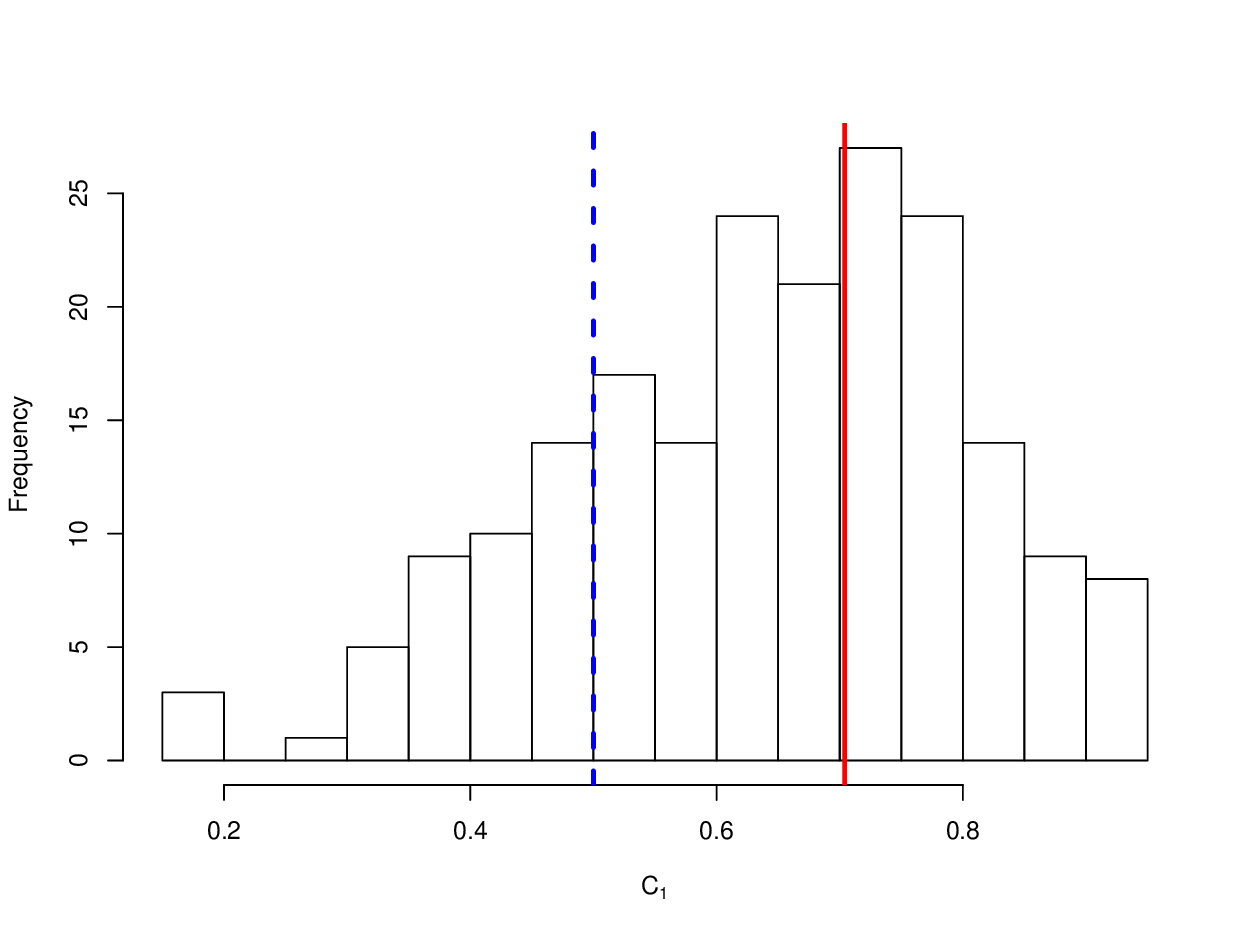}
\end{minipage}
\hfill
\begin{minipage}{0.49\textwidth}
\includegraphics[width = \textwidth]{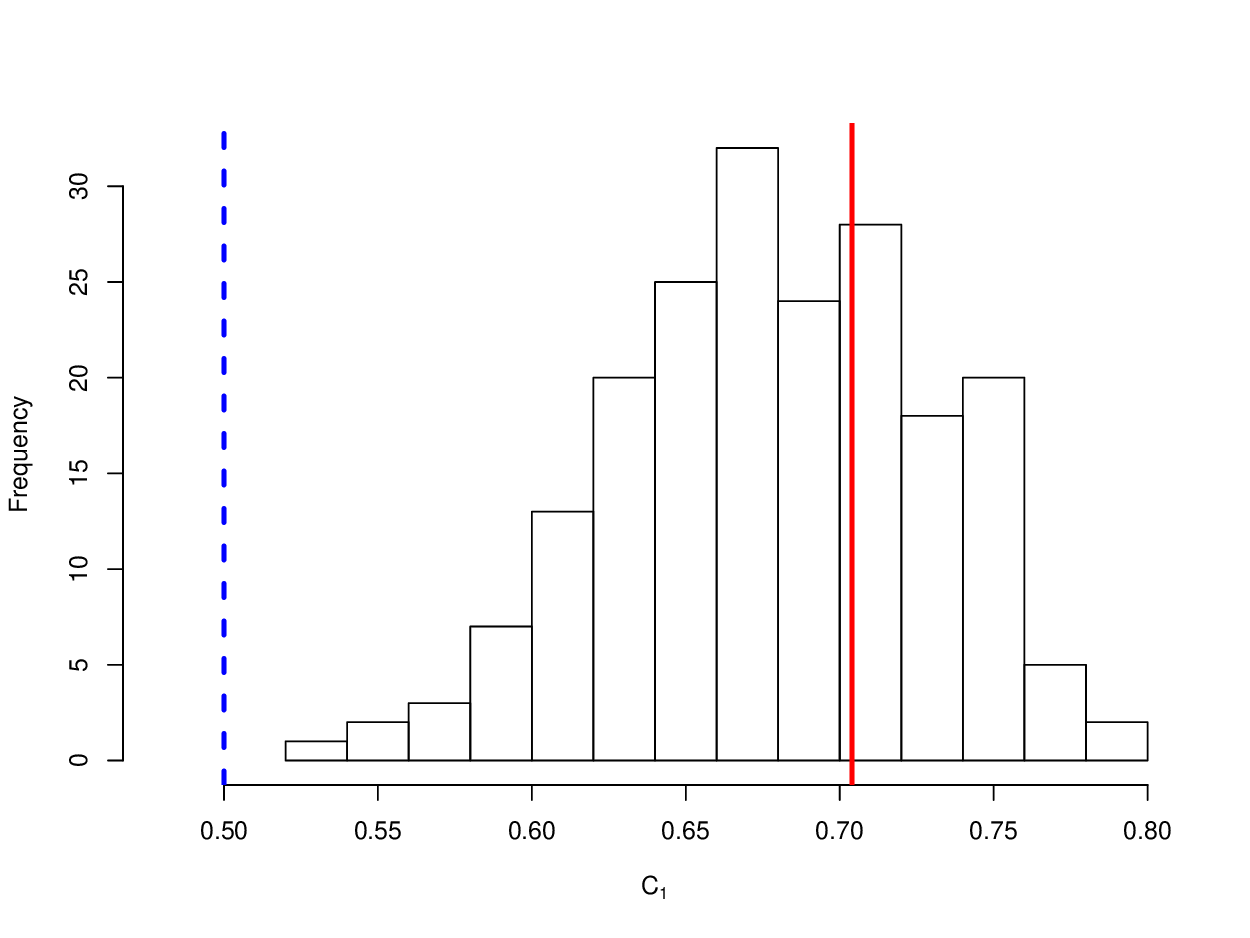}
\end{minipage}
\caption{The empirical distribution of the estimator $\hat{\beta}_1$ following \cite{rockafellar2014superquantile} at sample size $n=100$ (left) and $n=1000$ (right). The blue line marks the true ES regression coefficient $\beta_1 = 0.5$, while the red one marks the population-level solution to~\eqref{eq::SQR}, namely, $\tilde \beta_1 = 0.704$.}
\label{fig::numeric-loss}
\end{figure}

While the superquantile regression proposed in \cite{rockafellar2014superquantile} is not valid for the ES regression, it can still be valuable as a generalized regression technique.
As shown in  \cite{rockafellar2013fundamental}, the superquantile regression approach finds the best linear approximation to the response $Y$ using the covariates $X$, in the sense that the residuals minimize the superquantile loss function (\ref{eq::SQR}). 
Furthermore, \cite{rockafellar2018superquantile} shows that the superquantile regression coefficient in~\eqref{eq::SQR} is consistent for the ES regression in homoscedastic linear models; and \cite{golodnikov2019cvar} shows that the superquantile regression approach is equivalent to a composite quantile regression under certain scenarios.
Therefore, the superquantile regression approach can be useful for risk tuning and optimization that incorporates covariate information \citep{miranda2014superquantile}.
}

\section{Proof of theoretical results in Section 2}\label{append::sec2}
Recall that we omit the subscript $[Y\mid X]$ of $v$ and $\hat v$, $q$, and $\hat q$, and simply write as $v(s,x)$, $\hat v(s,x)$, $q(s,x)$, $\hat q(s,x)$, respectively, throughout the proof. 
\subsection{Proof of Theorem 2.1}\label{proof::thm::RRM-correct}
\begin{proof}
In order to show that the true ES coefficient $\beta$ is the unique minimizer of the loss function \begin{eqnarray}\label{eq::loss}
L(\theta) 
&=& \E_X \,\left[ \int_{0}^1 \,\rho_\tau\left(v_{[Y\mid X]}(\alpha,X) - X^T\theta\right) \,\mathrm{d}\alpha\right]\label{eq::RRM-correct},
\end{eqnarray}
where $\rho_\tau(u) = \{ \tau - \mathbbm{1} (u<0) \} u$ is the quantile loss function, we break the proof into the following three steps: (1) we show that the function $v(u,x)$ is increasing in $u \in (0,1)$ and strictly increasing in $u \in [\tau-\epsilon,\tau+\epsilon]$ for each possible value of $x$; (2) we verify the first-order optimality condition; (3) we check the positive definiteness of the Hessian matrix at the minimum. 
From the property of the quantile loss function, (see, e.g., \citet[Chapter 1.3]{koenker_2005}), the loss function in~\eqref{eq::loss} is \textbf{convex} in $\theta$.
Combined with (2) and (3), it follows that $\beta$ is the unique global minimizer of the loss function. 

\paragraph{Step 1}
For each $x$, we check the monotonicity of $v(u,x)$ over $u$. The conditional quantile of $Y \mid X=x$, i.e., $q(u,x)$, is increasing in $u\in(0,1)$ by definition. 
Then, by definition of $v(u,x)$, i.e., 
\begin{equation}
    v(u,x) = (1-u)^{-1} \int_{u}^1 q(s,x) \mathrm{d}s,
\end{equation}
we have that $v(u,x)$ is increasing in $u\in(0,1)$ since it is a tail-average of $q(u,x)$.
Since the cumulative distribution function of $Y \mid X=x$ is continuous and strictly increasing in a neighborhood of $q(\tau, x)$, $q(u, x)$ is 
continuous and 
strictly increasing in $u\in [\tau-\epsilon,\tau+\epsilon]$ for some $\epsilon>0$. 
For $u \in [\tau-\epsilon,\tau+\epsilon]$ and for all $x$,  
\begin{eqnarray*}
    v(u,x) = (1-u)^{-1} \Bigg\{(1-\tau) v(\tau,x) + \int_u^\tau q(s,x) \,\mathrm{d}s \Bigg\}
\end{eqnarray*}
is finite since $v(\tau,x)$ is finite for all $x$ and $q(s,x)$ are bounded on a closed interval due to continuity. 
Therefore, for $u \in [\tau-\epsilon,\tau+\epsilon]$,
\begin{eqnarray}
\dfrac{\partial v(u,x)}{\partial u} 
&=& \dfrac{\partial \{(1-u)^{-1} \int_{u}^1 q(s,x) \mathrm{d}s\}}{\partial u}\nonumber\\
&=& \dfrac{\int_{u}^1 q(s,x) \mathrm{d}s}{(1-u)^2} - \dfrac{q(u,x)}{1-u}\nonumber\\
&=& \dfrac{v(u,x) - q(u,x)}{1-u}\label{eq::derivative-SQ-proof}\\
 &=&\frac{1}{(1-u)^2}\int_{u}^1[q(s,x) - q(u,x)]\,\mathrm{d}s\nonumber > 0\nonumber.
\end{eqnarray} 
Hence $v(u,x)$ is strictly increasing in $u \in [\tau-\epsilon,\tau+\epsilon]$. 

\paragraph{Step 2}
Next, we check the first-order optimality condition of $\beta$ for the loss function $L(\theta)$ in~\eqref{eq::loss}. By defining $\xi\sim U(0,1)$, we equivalently write $L(
\theta)$ as
\begin{eqnarray}
L(\theta) 
&=& E_{(X,\,\xi)}\,\left[ \rho_\tau\left(v(\xi,X) - X^T\theta\right)\right].
\end{eqnarray} 
From the property of the check-loss function $\rho$, it follows that the function $L(\theta)$ is convex and differentiable in $\theta$. 
Then, 
\begin{align*}
\frac{\partial L(\theta)}{\partial \theta}\Bigg|_{\theta=\beta} \;
&= \;E_{(X,\xi)}\left(\left[\tau - \bm{1}\left\{v(\xi,X) < X^T\beta \right\}\right] \cdot (-X)\right)\\
&= \;E_{X}\left(\left[\tau - \Pr_{\xi|X}\left\{v(\xi,X) < X^T\beta \mid X\right\}\right] \cdot (-X)\right)\\
&=\;E_{X}\left(\left[\tau - \Pr_{\xi|X}\left\{v(\xi,X) < v(\tau,X) \mid X\right\}\right] \cdot (-X)\right)\\
&=\;E_{X}\left(\left[\tau - \Pr(\xi < \tau)\right] \cdot (-X)\right)\\
&=\;0,
\end{align*}
where the second equality follows by first conditioning on $X$, the third equality follows from $v(\tau,X) = X^T\beta$, the fourth equality follows from the monotonicity of $v(u,x)$ in step 1, and the last equality follows from $\xi \sim U(0,1)$ and independent of $X$. 
The first-order optimality condition, combined with the convexity of $L(\theta)$, implies that $\beta$ is the global minimum. 

\paragraph{Step 3}
Lastly, we show the minimizer of $L(\theta)$ is unique. Since $L(\theta)$ is convex, it suffices to show that the Hessian matrix at the minimum is positive definite. 
Since $v(\cdot,x)$ is strictly monotone, let $h(z,x)$ be the inverse of $v(\cdot,x)$, such that $v\circ(h(z,x),x) = z$.
By the conditions in Theorem 2.1, 
we have 
\begin{equation}\label{eq::derivative_v}
    \begin{aligned}
        \frac{\partial \Pr(v(\xi,x) \leq z)}{\partial z}\Bigg|_{z = v(\tau,x)}  &= \frac{\partial \Pr(\xi \leq h(z,x))}{\partial z}\Bigg|_{z = v(\tau,x)}\\
    & = \frac{\partial  h(z,x)}{\partial z}\Bigg|_{z = v(\tau,x)}\\
    & = \frac{\partial s}{\partial  v(s,x)}\Bigg|_{s = \tau}\\
    & = \left\{\frac{\partial  v(s,x)}{\partial s}\Bigg|_{s = \tau}\right\}^{-1}\\
    & = \dfrac{1-\tau}{v(\tau,x) - q(\tau,x)}.
    \end{aligned}
\end{equation}

Therefore, $L(\theta)$ is twice differentiable, and its second derivative at $\theta = \beta$ satisfies
\begin{align*}
    \frac{\partial^2 L}{\partial \theta\partial\theta^T}\Bigg|_{\theta = \beta}\; &=\; \E_{X}\left[ \frac{\partial \Pr_{\xi|X}(v(\xi,X) \leq z)}{\partial z}\Bigg|_{z = X^T\beta}\cdot XX^T \right] \\
    &= \E_{X}\left[ \frac{\partial \Pr_{\xi}(v(\xi,x) \leq z)}{\partial z}\Bigg|_{(z = X^T\beta,x=X)}\cdot XX^T \right] \\
    &= (1-\tau) \E_{X}\left[\frac{XX^T}{v(\tau,X) - q(\tau,X)} \right] \\
    & = (1-\tau) D_1\succ 0,
\end{align*}
where the second equality follows since $\xi$ is independent of $X$, and the third equality follows from~\eqref{eq::derivative_v} and $v(\tau) = x^T\beta$.
Therefore, the Hessian matrix of $L(\cdot)$ evaluated at $\beta$ 
is positive definite, establishing the uniqueness of the minimizer $\beta$.

\end{proof}

\subsection{Proof of Corollary 1}
\begin{proof}
    It follows closely from the proof to Theorem 2.1 in Appendix~\ref{proof::thm::RRM-correct} by simply replacing the distribution of $\xi$ to $\xi\sim U(\tau- \delta\tau,\tau + \delta(1-\tau))$ and noting that $\Pr(\xi < \tau) = \tau$. 
\end{proof}

\section{Proof and additional details of Section 3}\label{append::sec3}
\subsection{Proof of Theorem 3.1}\label{proof::thm::RRM-modified}

In order to show the consistency and the asymptotic normality of the i-Rock estimator with discrete covariates in Theorem 3.1, we 
start with one-sample case without covariates in~\ref{subsec::one-sample-SQ} and~\ref{subsec::proof-one-sample}, and generalize to the case with discrete covariates in~\ref{subsec::proof-mRock-discrete}. 

\subsubsection{Auxiliary results for the one-sample ES process}
\label{subsec::one-sample-SQ}

We first present asymptotic results in the one-sample case without any covariate. These results also apply to the empirical ES estimators at each covariate value in our regression setting.

We fix some notations for the discussion of the one-sample problem. Suppose the data $Y_1,\ldots,Y_n$ are $i.i.d.$ observations with a common distribution function $F(y)$. For any $0<s<1$, let $\hat{q}(s)$ be the sample quantile from the $n$ observations, we define the empirical ES estimator as:
\begin{equation}
\hat{v}(s) = \dfrac{\sumn Y_i\cdot\bm{1}\{ Y_i\geq \hat{q}(s)\}}{\sumn \bm{1}\{ Y_i\geq \hat{q}(s)\}}.
\label{eq::sampleSQ}
\end{equation}
While the parameter of interest is the $\tau$-level ES, here we consider the empirical ES process, which is the stochastic process given by $\{\hat{v}(s):s\in[\tau_L,\tau_U]\}$, where $0 < \tau_L  < \tau < \tau_U < 1$. Let $\ell^{\infty}[a,\,b]$ be the set of all uniformly bounded functions on the interval $[a,b]$. To further simplify notations, in the remainder of this subsection, we shall write $\hat{v}_s = \hat{v}(s)$ and $\hat{q}_s = \hat{q}(s)$, and  we define $q_L$ and $q_U$ as the $\tau_L$-th and $\tau_U$-th quantile, respectively. In the one-sample case, the notations here may be different than those in the regression setting.

We need the following technical condition, which is the one-sample counterpart for Conditions R-Y1 and R-Y2. 
\begin{customcond}{U}
The distribution function $F(y)$ is continuously differentiable on the interval $[q_L-\varepsilon, q_U + \varepsilon]$ for some $\varepsilon>0$; the density function $f(y)$ is bounded away from zero and above on the same interval. Furthermore, we have $\E[Y^2\cdot \bm{1}\{Y\geq 0\}]< + \infty$.
\label{cond::univariate}
\end{customcond}

Now we present the first main result in the one-sample case, which concerns the weak convergence of the empirical ES as a stochastic process indexed by the quantile level. Not only is the result an important technical tool for subsequent analysis, but it also is of interest on its own. 

\begin{theorem}
Suppose Condition \ref{cond::univariate} holds, then we have
\[
\hat{v}_s - v_s = \frac{1}{n}\sum_{i=1}^n \left\{\frac{(Y_i - q_s)\cdot \bm{1}(Y_i\geq q_s)}{1-s}
 - (v_s-q_s)
\right\} + o_P\left(n^{-1/2}\right),
\]
uniformly in $s\in[\tau_L,\tau_U]$. Furthermore, the centered empirical ES process converges weakly:
\[
\sqrt{n}\left[\hat{v}(\cdot) - v(\cdot)\right] \mathrel{\leadsto} \mathbb{G}(\cdot) \quad \text{in}\;\;\ell^{\infty}[\tau_L,\,\tau_U], 
\]
where $\mathbb{G}(\cdot)$ is a mean zero Gaussian Process.
\label{thm::univariate-process}
\end{theorem}

Theorem \ref{thm::univariate-process} gives the uniform (weak) Bahadur representation for the empirical ES process. To the best of our knowledge, the uniformity of the result is new. Restricting to a single quantile level $\tau$, \citet{chen2007nonparametric} and \citet{zwingmann2016asymptotics} study the asymptotic properties of the ES estimator $\hat{v}(\tau)$ under more general conditions; on the other hand we discuss process convergence.
Practically, Theorem \ref{thm::univariate-process} is a technical tool for simultaneous statistical inference for a range of expected shortfalls.


As a simple corollary of Theorem \ref{thm::univariate-process}, we can obtain the asymptotic distribution for the $\tau$-level empirical ES, which is also known from, e.g., \citet{chen2007nonparametric} and \citet{zwingmann2016asymptotics}. 
\begin{corollary}
\label{coro::SQ-fixed-tau}
Under Condition \ref{cond::univariate}, we have
\[
\sqrt{n}(\hat{v}_\tau - v_\tau) \converged \mathrm{N}(0, \sigma_\tau^2),
\]
with $(1-\tau)\sigma_\tau^2 = \var(Y\mid Y\geq q_\tau) + \tau(v_\tau - q_\tau)^2$.
\end{corollary}
\begin{proof}
    Combining the Central Limit Theorem with the Bahadur representation in Theorem \ref{thm::univariate-process} concludes the results with asymptotic variance 
    \begin{equation*}
        \begin{aligned}
            \sigma_\tau^2 & = \var\left\{\frac{(Y - q_\tau) \cdot \bm{1}(Y\geq q_\tau)}{1-\tau} - (v_\tau-q_\tau)\right\}\\
            & = \E\left\{\frac{(Y - q_\tau)^2 \cdot \bm{1}(Y\geq q_\tau)}{(1-\tau)^2} + (v_\tau-q_\tau)^2 -\frac{2(v_\tau - q_\tau)(Y-q_\tau)\cdot \bm{1}(Y\geq q_\tau)}{1-\tau}\right\}\\
            & = \frac{\E\{(Y - v_\tau)^2 + (v_\tau - q_\tau)^2 + 2(Y - v_\tau)(v_\tau - q_\tau) \mid Y\geq q_\tau\} P(Y\geq q_\tau)}{(1-\tau)^2} \\
            &~~+ (v_\tau-q_\tau)^2 -\frac{2(v_\tau - q_\tau)\E(Y-q_\tau \mid Y\geq q_\tau)P(Y\geq q_\tau)}{1-\tau}\\
            & = \frac{\E\{(Y - v_\tau)^2 \mid Y\geq q_\tau\}+(v_\tau - q_\tau)^2 }{1-\tau} - (v_\tau-q_\tau)^2\\
            & = \frac{1}{1-\tau}\{\var(Y \mid Y \geq q_\tau) + \tau (v_\tau-q_\tau)^2\}
        \end{aligned}
    \end{equation*}
\end{proof}

The asymptotic variance $\sigma_\tau^2$ consists of two parts. The first part is the variance in estimating $v_\tau$ when $q_\tau$ is known, whereas the second part is attributable to quantile estimation \citep{zwingmann2016asymptotics}. 

Next, we proceed to the study of the inverse ES function, which we define below:
\[
h(z) = \{s:v_s = z\}\quad \text{and}\quad \hat{h}(z) = \inf\{s\in[0,1]:\hat{v}_s\geq z\},
\]
for any $z\in[v_{\tau_L},v_{\tau_U}]$.
Note $v_s$ is strictly increasing in $s\in[\tau_L,\tau_U]$, and we show in Lemma \ref{lemma::SQ-monotonicity} that $\hat{v}(z)$ is also non-decreasing in $s\in[\tau_L,\tau_U]$; therefore the definitions above are well-defined.
The following Lemma shows that $\hat{h}(z)$, the empirical inverse ES, is also asymptotically Gaussian.

\begin{lemma}
Under Condition \ref{cond::univariate}, the inverse ES process satisfies:
\[
\hat{h}(z) - h(z) = -\frac{1}{n}\sum_{i=1}^n \left[\frac{(Y_i - q_{h(z)})\cdot \bm{1}\{Y_i\geq q_{h(z)}\}}{z - q_{h(z)}}
 - (1-h(z))
\right] + o_P\left(n^{-1/2}\right),
\]
uniformly in $z\in[v_\tau - \varepsilon',\,v_\tau + \varepsilon']$ for some $\varepsilon'>0$. In particular, we have:
\[
\sqrt{n}\left(\hat{h}(v_\tau) - \tau\right) \converged \mathrm{N}\left(0,\;\frac{(1-\tau)^2\sigma_\tau^2}{(v_\tau - q_\tau)^2}\right),
\]
with $\sigma_\tau^2$ defined in Corollary \ref{coro::SQ-fixed-tau}. Furthermore, the process $n^{1/2}[\hat{h}(z) - h(z)]$ is asymptotically equi-continuous over $z\in[v_\tau - \varepsilon',\,v_\tau + \varepsilon']$ with respect to the Euclidean distance.
\label{lemma::inverseSQ}
\end{lemma}

The asymptotic property of $\hat{h}(z)$ is essential for the analysis of the i-Rock regression approach. 
The proof of Lemma \ref{lemma::inverseSQ} builds upon the uniform representation in Theorem \ref{thm::univariate-process}, as well as the functional Delta method; See, e.g., Theorem 20.8 in \citet{van2000asymptotic}. Note, the asymptotic normality of $\hat{v}_\tau$ at a single level \citep{chen2007nonparametric,zwingmann2016asymptotics}, is not sufficient to establish the result in Lemma \ref{lemma::inverseSQ}.

\subsubsection{Proof for the one-sample case}
\label{subsec::proof-one-sample}
The proofs in this subsection rely on standard empirical process tools in, e.g., \citet{van1996weak}, and we adopt the same notations therein. Let $Y_1,\ldots,Y_n$ be $i.i.d.$ observations from the same population. 
For a class of function $y\mapsto f(y;\theta)$ indexed by $\theta \in \mathbb{R}^q$, let $\mathbb{E}_n[f(Y^* ;\theta)] = \sumn f(Y_i;\theta)/n$, $\E[f(Y^*;\theta)] = \E[f(Y_i;\theta)]$ and $\Gn[f(Y^*;\theta)] = {n}^{1/2}\{\mathbb{E}_n[f(Y^* ;\theta)] - \E[f(Y^* ;\theta)]\}$. We sometimes use the subscript and write $\mathbb{E}_n[f_\theta]$ instead of $\mathbb{E}_n[f(Y^* ;\theta)]$ for further simplicity. For a semi-metric space $\mathbb{T}$, we use $\ell^{\infty}(\mathbb{T})$ to denote the functional space that consists all bounded functions of $\mathbb{T}\mapsto \mathbb{R}$. To prove Theorem~\ref{thm::univariate-process}, we need the following technical lemmas. 

\begin{lemma}
\label{lemma::SQ-monotonicity}
Under Condition \ref{cond::univariate}, for any $s,t$ such that $\tau_L\leq s<t\leq \tau_U$, we have $\hat{v}_s \leq \hat{v}_t$ and $v_s < v_t$. Furthermore, as a function of $s$, $\hat{v}_s$ is left continuous whose right limit exists everywhere.
\end{lemma}

\begin{lemma}
\label{lemma::sto-diff-univariate}
Let $\psi(y;\theta,s) = (y - v_s)\bm{1}\{y\geq \theta\}$. Under Condition \ref{cond::univariate}, and suppose that $|\tilde{q}_s- q_s| = \op(1)$ uniformly over $s\in[\tau_L,\tau_U]$, then we have
\[
\sup_{s\in[\tau_L,\tau_U]}\left|\Gn[\psi_{(\tilde{q}_s,s)}] - \Gn[\psi_{(q_s,s)}]\right| = \op(1).
\]
Furthermore, the function class $\mathcal{F} = \{y\mapsto \psi(y,\theta,s) : \theta\in [q_L,q_U], s \in [\tau_L,\tau_U] \}$ is Donsker.
\end{lemma}

\begin{proof}[Proof of Theorem~\ref{thm::univariate-process}]
We first prove the Bahadur representation for a broader class of ES estimator. Consider any estimator $\tilde{v}_s$ that solves the following estimating equation:
\begin{equation}
0 = \sumn (Y_i - \tilde{v}_s)\bm{1}\{Y_i \geq \tilde{q}_s\},
\label{eq::score-univariate}
\end{equation}
where $\tilde{q}_s$ is any estimator for the $q_s$ that satisfies (i) $\sup_{s\in[\tau_L,\tau_U]}|\tilde{q}_s - q_s| = \Op(n^{-1/2})$; and (ii) $\tilde{q}_s\in\ell^{\infty}([\tau_L,\tau_U])$ as a stochastic process indexed by $s$. Choosing $\tilde{q}_s$ as the sample quantile in (\ref{eq::score-univariate}) recovers the empirical ES estimator.

Let $\psi(y;\theta,s) = (y - v_s)\bm{1}\{y\geq \theta\}$. Given the quantile estimators $\tilde{q}_s$ $(s\in[\tau_L,\tau_U])$, the estimating equation (\ref{eq::score-univariate}) for $\tilde{v}_s$ solves $\En[(Y^* - \tilde{v}_s)\bm{1}\{Y^*\geq \tilde{q}_s\}] =\En[\psi(Y^*;\tilde{q}_s,s)] + (v_s - \tilde{v}_s)\En[\bm{1}\{Y^*\geq \tilde{q}_s\}] = 0$. Hence the estimator $\tilde{v}_s$ satisfies:
\begin{eqnarray}
 &&\sqrt{n}(\tilde{v}_s - v_s )\En[\bm{1}\{Y^*\geq \tilde{q}_s\}]\nonumber\\
 &= & \sqrt{n}\En[\psi(Y^*;\tilde{q}_s,s)] \label{eq::stardardized1-univariate}\\
 &= & \sqrt{n}\left\{\E[\psi(Y^*;\tilde{q}_s,s)] - \E[\psi(Y^*;q_s,s)]\right\}\label{eq::intermediate1-univariate}\\
 &&\; + \;\underbrace{\Gn[\psi(Y^*;\tilde{q}_s,s)] - \Gn[\psi(Y^*;q_s,s)]}_{R_1(s)} \nonumber\\
 && \; + \;\Gn[\psi(Y^*;q_s,s)]\nonumber\\
 &=& \frac{\partial \E[\psi(Y^*;q_s,s)]}{\partial q_s}[\sqrt{n}\left( \tilde{q}_s - q_s \right)]+ \Gn[\psi(Y^*;q_s,s)] +  R_1(s)\nonumber\\
 &&\;+\;
 \underbrace{\sqrt{n}\left\{\E[\psi(Y^*;\tilde{q}_s,s)] - \E[\psi(Y^*;q_s,s)] -\frac{\partial \E[\psi(Y^*;q_s,s)]}{\partial q_s}\left( \tilde{q}_s - q_s \right)  \right\}}_{R_2(s)} \nonumber\\
 &=  & (v_s - q_s)f_Y(q_s)[\sqrt{n}\left( \tilde{q}_s - q_s \right)]+
 \Gn[\psi(Y^*;q_s,s)] +  R_1(s) + R_2(s),\nonumber
\end{eqnarray}
where (\ref{eq::intermediate1-univariate}) holds since $\E[\psi(Y^*;q_s,s)] = 0$ for all $s\in[\tau_L,\tau_U]$, and the last equality follows since $\partial \E[\psi(Y^*;\theta,s)]/\partial \theta =( v_s - \theta)f_Y(\theta)$.

Now we show that both $R_1(s)$ and $R_2(s)$ are negligible uniformly in $s$.
By Lemma \ref{lemma::sto-diff-univariate}, we immediately obtain $R_1(s) = \op(1)$ uniformly over $s\in[\tau_L,\tau_U]$. For $R_2$, we first re-write
\[
\E[\psi(Y^*;\theta,s)] = \int_{\theta}^{+\infty} yf_Y(y)\,\mathrm{d}y - v_s [1-F_Y(\theta)] \triangleq I_1(\theta) - v_s\times I_2(\theta),
\]
and hence by Taylor expansion with respect to $\theta$, and letting $\Delta_s =\tilde{q}_s - q_s $, we have 
\begin{eqnarray*}
&&\sup_{s\in[\tau_L,\tau_U]}|R_2(s)|\\ 
&\leq&\sup_{s\in[\tau_L,\tau_U]} |\sqrt{n} \Delta_s| \times \sup_{\substack{\theta\in[q_L,q_U]\\\\|\theta-\theta'|\leq |\Delta_s|}}\left| [I_1^{\prime}(\theta') - I_1^{\prime}(\theta)] - v_s\times [ I_2^{\prime}(\theta') - I_2^{\prime}(\theta)] \right| \\
&=&\sup_{s\in[\tau_L,\tau_U]} |\sqrt{n} \Delta_s| \times \sup_{\substack{\theta\in[q_L,q_U]\\\\|\theta-\theta'|\leq |\Delta_s|}}\left| \{I_1^{\prime}(\tilde \theta) - v_s I_2^{\prime}(\tilde \theta')\} (\theta'-\theta) \right| \\
&=&\op(1),
\end{eqnarray*}
where $\tilde \theta$ and $\tilde \theta'$ are between $\theta$ and $\theta'$, the second to last equality follows from mean value theorem and that both $I_1'(\theta) = -\theta f_Y(\theta)$ and $I_2'(\theta) = - f_Y(\theta)$ are continuous on $\theta\in[q_L ,q_U]$ under Condition \ref{cond::univariate}, and the last equality follows from (i) $|\Delta_s| = O_p(n^{-1/2})$, (ii) $v_s$ is uniformly bounded over $s\in[\tau_L,\tau_U]$, and (iii) $I_1'(\theta)$ and $I_2'(\theta)$ are uniformly bounded over $\theta\in[q_L ,q_U]$.

Combining the results for $R_1(s)$ and $R_2(s)$ with Equation (\ref{eq::stardardized1-univariate}), we have
\begin{equation}
\sqrt{n}(\tilde{v}_s - v_s)\En[\bm{1}\{Y^*\geq \tilde{q}_s\}] =   (v_s - q_s)f_Y(q_s)[\sqrt{n}\left( \tilde{q}_s - q_s \right)]+
 \Gn[\psi(Y^*;q_s,s)] + \op(1),
 \label{eq::pre-Bahadur-univariate}
\end{equation}
where the $\op(1)$ term is uniform in $s\in [\tau_L,\tau_U]$. From here, we can deduce the $n^{1/2}$-uniform consistency of $\tilde{v}_s$ as follows.
From the Lemma \ref{lemma::sto-diff-univariate}, the function class $\{y\mapsto \psi(y,\theta,s);s\in[\tau_L,\tau_U],\theta\in[q_L,q_U]\}$ is Donsker, therefore
\begin{equation}\label{eq::Gn}
    \sup_{s\in[\tau_L,\tau_U]}|\Gn[\psi(Y^*;q_s,s)]|
\leq
\sup_{\substack{s\in[\tau_L,\tau_U]\\\theta\in[q_L,q_U]}}|\Gn[\psi(Y^*;\theta,s)]| = \Op(1).
\end{equation}
Furthermore, the assumptions on $\tilde{q}_s$ at the beginning of the proof implies
\begin{equation}\label{eq::qs}
    \sup_{s\in[\tau_L,\tau_U]}\sqrt{n}|\tilde{q}_s - q_s| = \Op(1),\quad\text{and}\quad\En[\bm{1}\{Y^*\geq \tilde{q}_s\}] = 1-s +\op(1).
\end{equation}
Combining~\eqref{eq::pre-Bahadur-univariate},~\eqref{eq::Gn},~\eqref{eq::qs}, and that $v_s$ and $f_Y(q_s)$ are bounded, we have
\[
\sup_{s\in[\tau_L,\tau_U]}|\sqrt{n}(\tilde{v}_s - v_s)| = \Op(1).
\] 
From here, we can obtain the uniform Bahadur representation of $\tilde{v}_s - v_s$. Dividing both sides of (\ref{eq::pre-Bahadur-univariate}) by $(1-s)$, we obtain, since $\sqrt{n}(\tilde{v}_s - v_s)$ is asymptotically tight in $\ell^{\infty}([\tau_L,\tau_U])$, that
\begin{equation}
\sqrt{n}(\tilde{v}_s - v_s) = \frac{1}{1-s}\left\{\sqrt{n}\left( \tilde{q}_s - q_s \right)(v_s - q_s)f_Y(q_s)+
 \Gn[(Y^* - v_s)\bm{1}\{Y^*\geq q_s\}]\right\} + \op(1),
 \label{eq::Bahadur-general-univariate}
\end{equation}
uniformly over $s\in[\tau_L,\tau_U]$.

In particular, if we choose $\tilde{q}_s$ to be the sample quantile that satisfies $\En[\bm{1}\{Y^*\leq \tilde{q}_s\}] = s$, then for sufficiently large $n$, the estimator obtained from (\ref{eq::score-univariate}) 
is asymptotically equivalent to the empirical ES estimator $\hat{v}_s$ defined in (\ref{eq::sampleSQ}).
Since the sample quantile satisfies
\[
\sqrt{n}(\hat{q}_s - q_s) = \sumn \frac{s - \bm{1}\{Y_i \leq q_s\}}{f_Y(q_s)} + \op(1),
\]
uniformly in $s$ (see, e.g., Corollary 21.5 of \citet{van2000asymptotic}). Combining the above displayed equation with (\ref{eq::Bahadur-general-univariate}), we have
\[
\sqrt{n}(\hat{v}_s - v_s) = \frac{1}{1-s}\Gn[(Y^* - q_s)\bm{1}\{Y^*\geq q_s\}] + \op(1),
\]
uniformly over $s\in[\tau_L,\tau_U]$.

Finally, we show that the empirical ES process $\sqrt{n}(\hat{v}_s - v_s)$ converges towards a Gaussian Process in $\ell^{\infty}[\tau_L,\tau_U]$. In view of the uniform Bahadur representation, it suffices to consider the process $\Gn[(Y^*-q_s)\bm{1}\{Y^*\geq q_s\}]$.
Since $q_s$ is uniformly Lipschitz continuous in $s\in[\tau_L,\tau_U]$, it follows from Example 19.19 of \citet{van2000asymptotic} that the function class $\{y\mapsto (y-q_s)\bm{1}[y\geq q_s]: s\in[\tau_L,\tau_U]\}$ is Donsker. Therefore $\Gn[(Y^*-q_s)\bm{1}\{Y^*\geq q_s\}]\converged \mathbb{G}_{\infty}(s)$ as a function of $s$ on the space $\ell^{\infty}([\tau_L,\tau_U])$; here $\mathbb{G}_{\infty}(s)$ is a zero-mean Gaussian process with continuous sample path with respect to the semi-metric 
\[
\rho(s,t) = \left(\E\{(Y^*-q_s)\bm{1}[Y^*\geq q_s] - (Y^*-q_t)\bm{1}[Y^*\geq q_t]\}^2\right)^{1/2},\quad s,t\in[\tau_L,\tau_U].
\]
Since $\rho(s,t) \leq |q_s - q_t| \lesssim |s-t|$, the sample path of $\mathbb{G}_{\infty}(\cdot)$ is also continuous with respect to the Euclidean distance. This concludes the proof.
\end{proof}

\begin{proof}[Proof of Lemma~\ref{lemma::inverseSQ}]
Let $a = v_\tau - \varepsilon_0$ and $b = v_\tau + \varepsilon_0$ for some constant $\varepsilon_0$ such that $v_{\tau_U} - 2\varepsilon_0 \geq v_\tau \geq v_{\tau_L} +2\varepsilon_0$.
Define the function space $\mathbb{D}_1$ as the space of all non-decreasing, continuous function on $[\tau_L,\tau_U]$.
For any function $F\in \mathbb{D}_1$, we define the inverse map $\phi(\cdot): \mathbb{D}_1\mapsto \ell^{\infty}([a,b])$ such that $\phi(F)(z) = \inf\{s\in[\tau_L,\tau_U]:F(s)\geq z\}$ for $z\in[a,b]$ \footnote{We define $\phi(F)(z) = \tau_U$ if $\sup_{s\in[\tau_L,\tau_U]}F(s) < z$, so that $\phi(F)\in\ell^{\infty}([a,b])$.}. Note that $v_s$ as a function of $s\in[\tau_L,\tau_U]$ is continuously differentiable with $\partial v_s/\partial s = (v_s - q_s)/(1-s) > 0$. Following Lemma 21.4 in \citet{van2000asymptotic}, the map $\phi(\cdot)$ is Hadamard-differentiable at $v_s\in \mathbb{D}_1$, tangentially to the set of all continuous (with respect to the Euclidean distance) functions on $[\tau_L,\tau_U]$. The Hadamard-derivative of the inverse map $\phi$ at $v_s$ is $\phi'_{v}(h) = -h(v^{-1})/v'(v^{-1})$, for any continuous function $h$.

Next we apply the functional Delta method. Note $v_s,\hat{v}_s \in \mathbb{D}_1$, and $h(\cdot) = \phi\circ v_s(\cdot)\in \ell^{\infty}([a,b])$; since $\hat{v}_{\tau_L} \convergeip v_{\tau_L} < a$, $\hat{v}_{\tau_U} \convergeip v_{\tau_U} > b$, the inverse ES process $\hat{h}(\cdot) = \phi\circ \hat{v}_s(\cdot)$ with probability going to $1$ \footnote{Since the infimum in the definition of $\phi$ is taken only within $[\tau_L,\tau_U]$.}. 
Therefore, applying the functional Delta method (Theorem 20.8 in \citet{van2000asymptotic}) towards the inverse map $\phi$ gives
\begin{eqnarray*}
\sqrt{n}[\hat{h}(z) - h(z)] &=& -\left[\frac{\sqrt{n}(\hat{v}_s - v_s)}{v'(s)}\right]{\Bigg|}_{s = v^{-1}(z)} + \op(1)\\
&=& -\left[\frac{\Gn[(Y^* - q_s)\bm{1}\{Y^*\geq q_s\}]}{v_s - q_s}\right]{\Bigg|}_{s = v^{-1}(z)} + \op(1)
,
\end{eqnarray*}
in $\ell^{\infty}[a,b]$, which shows the first part of the Lemma. Since $n^{1/2}(\hat{v}_s - v_s)\converged \mathbb{G}_{\infty}(s)$, it follows that $n^{1/2}[\hat{h}(z) - h(z)]$ also converges towards a Gaussian process with continuous sample path (with respect to the Euclidean distance), since $v_s$ is continuously differentiable with respect to $s$. Asymptotic equi-continuity of $n^{1/2}[\hat{h}(z) - h(z)]$ is then a consequence of its convergence towards a continuous stochastic process.

For the second part of the Lemma, taking $z = v_\tau$ in the above displayed equation, and recalling $\sqrt{n}(\hat{v}_\tau - v_\tau)\converged \mathrm{N}(0,\sigma_\tau^2)$ from Theorem \ref{thm::univariate-process} concludes the proof.
\end{proof}

\subsubsection{Proof for the i-Rock approach with discrete covariates}
\label{subsec::proof-mRock-discrete}
With the results of one-sample ES process, we now establish the main results for the i-Rock estimator. 
We begin by giving some finite-sample properties of the empirical i-Rock loss function, defined as
\begin{equation}
L_n(\theta) = \frac{1}{M}\sum_{m=1}^M\left[\int_{0}^1\,\rho_\tau\left(\hat{v}_m(\alpha) - x_m^T\theta\right)\,\mathrm{d}\alpha\right].\label{eq::mRock-loss-discrete}
\end{equation} 
We suppose that $|L_n(\theta)| < +\infty$ for all $\theta$; otherwise we can restrict the domain of interest to $\{\theta: |L_n(\theta)| < +\infty\}$.
\begin{proposition} 
\label{prop::m-Rock-loss-discrete}
Under Condition R-X, the following holds true for the empirical i-Rock loss function $L_n(\theta)$:
\begin{enumerate}
  \item $L_n(\theta)$ is convex and Lipshitz continuous.
  \item The directional derivative of $L_n(\theta)$ exists at any $\theta$ and along any direction.
  \item Suppose there are no ties among the response at each covariate value, then any minimizer of $L_n(\theta)$, denoted by $\widehat{\beta}$, satisfies:
\begin{equation}
\label{eq::m-Rock-optimality}
\left\lVert M^{-1}\sum_{m=1}^M x_m \left[\tau - \hat{h}_m(x_m^T\widehat{\beta})\right]\right\rVert \leq \frac{C_1M}{n},
\end{equation}
for some universal constant $C_1 > 0$, where $\hat{h}_m(z)$ is the empirical inverse\footnote{The inverse is well-defined as $\hat{v}_m(\alpha)$ is monotonically increasing in $\alpha$; see Lemma \ref{lemma::SQ-monotonicity} in Section \ref{subsec::proof-one-sample}.} of the ES function:
\[
\hat{h}_m(z) = \inf\{s\in[0,1]:\hat{v}_m(s) \geq z\}.
\]
\end{enumerate}
\end{proposition}

Proposition \ref{prop::m-Rock-loss-discrete} shows that the function $L_n(\theta)$ enjoys some desirable properties. Therefore, theoretical and computational tools from convex optimization apply to the analysis of the i-Rock approach. With convexity, (\ref{eq::m-Rock-optimality}) gives the necessary first-order optimality condition for the i-Rock estimator. Though $L_n(\theta)$ is not everywhere differentiable, optimality requires all the directional derivatives of $L_n(\theta)$ to be non-negative at $\widehat{\beta}$, which leads to (\ref{eq::m-Rock-optimality}).

\begin{remark}
Proposition \ref{prop::m-Rock-loss-discrete} is a general result that does not depend on Conditions R-Y1 and R-Y2; nor does it depend on the choice of $\hat{v}_m(\alpha)$ in the loss function. The conclusions take hold for any estimator $\hat{v}_m(\alpha)$ that is (i) monotonic in $\alpha$, and (ii) not flat over $\alpha$ for any interval of length $M/n$. However, our subsequent asymptotic analysis may depend on the sampling properties of $\hat{v}_m(\alpha)$.
\end{remark}


\begin{proof}[Proof of Proposition~\ref{prop::m-Rock-loss-discrete}]
Part 1 of the Proposition follows from the properties of the check loss $\rho_\tau(\cdot)$ function. Note that
\begin{eqnarray}
L_n({u}_1) + L_n({u}_2) 
&=& \sumM\left[\int_{0}^1 \rho_\tau\left(\hat{v}_{m}(\alpha) - x_m^T{u}_1\right) + \rho_\tau\left(\hat{v}_{m}(\alpha) - x_m^T{u}_2\right) \,\mathrm{d}\alpha\right]\nonumber\\
& \geq &  \sumM\int_{0}^1 \rho_\tau\left(\hat{v}_{m}(\alpha) - x_m^T({u}_1+{u}_2)/2\right) \,\mathrm{d}\alpha\nonumber\\
& \geq &  L_n\left(\frac{{u}_1+{u}_2}{2}\right)\label{eq::RRM-convex-discrete}.
\end{eqnarray}
Moreover, 
\begin{eqnarray*}
| L_n({u}_1) - L_n({u}_2) |
&=& \left|\sumM\int_{0}^1 \rho_\tau\left(\hat{v}_{m}(\alpha) - x_m^T{u}_1\right) - \rho_\tau\left(\hat{v}_{m}(\alpha) - x_m^T{u}_2\right)\,\mathrm{d}\alpha\right|\\
& \leq &  \sumM\int_{0}^1 \left|\rho_\tau\left(\hat{v}_{m}(\alpha) - x_m^T{u}_1\right) - \rho_\tau\left(\hat{v}_{m}(\alpha) - x_m^T{u}_2\right)\right|\,\mathrm{d}\alpha \\
& \leq &  \sumM \int_{0}^1 \left| x_m^T({u}_1-{u}_2)\right|\,\mathrm{d}\alpha\\
& \leq & \sumM\left\lVert x_m\right\rVert \cdot \lVert {u}_1-{u}_2\rVert.
\end{eqnarray*}
Thus, the convexity and Lipschitz continuity of $L_n(\theta)$ follows.

For Part 2 of the Proposition, the previous Lipschitz continuity implies we can exchange the order of integration and differentiability.  Therefore
\begin{eqnarray}
\nabla_{{w}} L_n({u}) &=& \lim_{t\to 0+}\frac{L_n({u} + t{w}) - L_n({u})}{t}\nonumber\\
&=&\sumM \int_{0}^1 \nabla_{{w}} \rho_\tau(\hat{v}_m(\alpha) - x_m^T{u})\,\mathrm{d}\alpha\nonumber
\\
&=&- \sumM x_m^T{w}\int_{0}^1 \psi^*_\tau\left(\hat{v}_m(\alpha)  - x_m^T{u},-x_m^T{w}\right)  \,\mathrm{d}\alpha,\label{eq::derivative-discrete}
\end{eqnarray}
where $\psi_\tau^*$ originates from the gradient condition of the check loss function, as in \citet[page 33]{koenker_2005}:
\begin{eqnarray}
\psi_\tau^*(u,v) &=&
 \begin{cases}
\tau - \bm{1}\{u<0\},&\text{if }u\neq 0,\\
\tau - \bm{1}\{v<0\},&\text{if }u = 0.
\end{cases}
\nonumber\\
&=& \tau - \bm{1}\{u < 0\} - \bm{1}\{u =  0, v < 0\}\nonumber.
\end{eqnarray}

We now prove Part 3, i.e., the optimality condition for the i-Rock estimator. By the convexity of $L_n$, any minimizer $\hat{\beta}$ of $L_n$ must satisfy: $\nabla_{{w}}L_n(\hat{\beta})  = 0$,
for all ${w} \in \mathbb{R}^p$, $\lVert {w}\rVert = 1$. Using the expression in (\ref{eq::derivative-discrete}), we can re-write the optimality condition as
\begin{eqnarray}
0 &=&
\sumM x_m^T{w}\int_{0}^1 \psi^*_\tau\left(\hat{v}_m(\alpha) - x_m^T\hat{\beta},-x_m^T{w}\right)  \,\mathrm{d}\alpha \nonumber\\
& = &
\sumM x_m^T{w}\left(\tau -  \int_{0}^1 \bm{1}\{\hat{v}_m(\alpha) < x_m^T\hat{\beta}\} \,\mathrm{d}\alpha - \bm{1}\{x_m^T{w} > 0 \}\int_{0}^1 \bm{1}\{\hat{v}_m(\alpha) = x_m^T\hat{\beta}\} \,\mathrm{d}\alpha\right).\nonumber\\
\label{eq::optimality-discrete-1}
\end{eqnarray}
By the monotonicity of $\hat{v}_m(\alpha)$, each of the set $\{\alpha:\hat{v}_m(\alpha) < x_m^T\hat{\beta}\}$ is an interval on $[0,1]$. By relating the integration to Lebesgue measure, we have
\begin{eqnarray*}
\int_{0}^1 \bm{1}\{\hat{v}_m(\alpha) < x_m^T\hat{\beta}\} \,\mathrm{d}\alpha 
&=& 
1 - \int_{0}^1 \bm{1}\{\hat{v}_m(\alpha) \geq x_m^T\hat{\beta}\} \,\mathrm{d}\alpha \\
&=&
1 - Leb\left(\{\alpha\in(0,1):\hat{v}_m(\alpha) \geq x_m^T\hat{\beta}\}\right)\\
&=&\hat{h}_m(x_m^T\hat{\beta}),
\end{eqnarray*}
where $Leb(\cdot)$ is the Lebesgue measure on $\mathbb{R}$, and the last inequality follows from the definition of $\hat{h}_m(\cdot)$. 
Therefore, (\ref{eq::optimality-discrete-1}) implies that
\begin{eqnarray}
\left\lVert \sumM x_m\left[\tau - \hat{h}_m(x_m^T\hat{\beta})\right] \right\rVert_2
&=&
\sup_{\lVert {w}\rVert = 1}\left[\sumM x_m^T{w}\left(\tau - \hat{h}_m(x_m^T\hat{\beta})\right) \right]
\nonumber\\
&=&
\sup_{\lVert{w}\rVert = 1}\left[
\sumM x_m^T{w}\bm{1}\{x_m^T{w} > 0\}\cdot \int_{0}^1 \bm{1}\{\hat{v}_m(\alpha) = x_m^T\hat{\beta}\} \,\mathrm{d}\alpha\right]\nonumber\\
&\leq&
\sumM \lVert x_m\rVert \cdot Leb\{\alpha\in[0,1]: \hat{v}_m(\alpha) = x_m^T\hat{\beta}\} \nonumber\\
&\lesssim&
\frac{M^2}{n},\label{eq::lebesgue-discrete}
\end{eqnarray}
almost surely since the covariates are bounded; the last inequality follows since there are no ties among $Y_1,\ldots,Y_n$, and hence $Leb\{\alpha\in[0,1]: \hat{v}_m(\alpha) = x_m^T\hat{\beta}\} \leq  M/n$. The proof is now complete.
\end{proof}

Our subsequent theoretical analysis builds upon the generalized Z-estimation framework\footnote{We use the word `generalized' because the estimating equation (\ref{eq::m-Rock-optimality}) is not an empirical average over each data point.} from Proposition \ref{prop::m-Rock-loss-discrete}. In particular, Equation~\eqref{eq::m-Rock-optimality} suggests that the property of the i-Rock estimator is closely tied to that of $\hat{h}_m(\cdot)$, the inverse empirical ES. Lemma~\ref{lemma::inverseSQ-main} serves as an important technical tool to understand the i-Rock estimator via (\ref{eq::optimality-discrete-1}). 

\begin{lemma}
\label{lemma::inverseSQ-main}
Under a fixed discrete design in~\eqref{eq::fix_design} and Conditions R-Y1 and R-Y2, the inverse empirical ES satisfies
\[
\sqrt{\frac{n}{M}}\left[\hat{h}_m\{v_m(\tau)\} - \tau\right] \converged \mathrm{N}\left(0,\;\;\frac{(1-\tau)^2\sigma_m^2(\tau)}{\{v_m(\tau) - q_m(\tau)\}^2}\right),
\]
for each $m = 1,\ldots, M$, with $(1-\tau)\sigma_m^2(\tau) = \text{var}[Y_m\mid Y_m\geq q_m(\tau)] + \tau[v_m(\tau) - q_m(\tau)]^2$.
\end{lemma}

Lemma \ref{lemma::inverseSQ-main} is a simple corollary from Lemma \ref{lemma::inverseSQ} therein. It establishes a key asymptotic normality result of the inverse ES estimator $\hat{h}_m$, which serves as an important technical tool to understand the i-Rock estimator. 
While the asymptotic properties of the empirical ES estimator $\hat{v}_m(\tau)$ has been well studied in the literature \citep{chen2007nonparametric,nadarajah2014estimation,zwingmann2016asymptotics}, Lemma~\ref{lemma::inverseSQ-main} gives the first asymptotic analysis of the inverse empirical ES estimator $\hat{h}_m$. 
With the help of Proposition~\ref{prop::m-Rock-loss-discrete} and Lemma~\ref{lemma::inverseSQ-main}, we are now ready to prove the main result for the i-Rock estimator with discrete covariates.

\begin{proof}[Proof of Theorem 3.1]
First, we prove the consistency part. 
To this end, we start by showing that for any $\varepsilon_0 > 0$, $\lVert \widehat{\beta} - \beta \rVert\geq 2\varepsilon_0$ implies that $\inf_{w:\|w\|=1}\nabla_{{w}} L_n(u)|_{u=\beta + \varepsilon_0 {w}} \leq 0$. Let $\tilde w = \frac{\widehat{\beta} - \beta}{\lVert \widehat{\beta} - \beta \rVert}$, and let $f(t) = L_n(\beta + \varepsilon_0 \tilde w + t \tilde w)$. Due to the convexity of $L_n$, $f(t)$ is a convex function since the second derivative of f, namely, $f^{''}(t) = \tilde w ^T \nabla^2 L_n(\beta + \varepsilon_0 \tilde w + t \tilde w)\tilde w \geq 0$. Since $L_n$ is minimized at $\widehat{\beta}$, $f(t)$ is minimized at $t = \lVert \widehat{\beta} - \beta \rVert - \varepsilon_0$. If $\lVert \widehat{\beta} - \beta \rVert\geq 2\varepsilon_0$, then
\begin{equation*}
\begin{aligned}
     \inf_{w:\|w\|=1}\nabla_{{w}} L_n(u)|_{u=\beta + \varepsilon_0 {w}} &\leq  \nabla_{\tilde{w}} L_n(u)|_{u=\beta + \varepsilon_0 \tilde{w}} \\
     &= \lim_{t \to 0^+} \frac{ L_n(\beta + \varepsilon_0 \tilde w + t \tilde w) -  L_n(\beta + \varepsilon_0 \tilde w)}{t}\\
     &= f'(t)|_{t=0}\\
     &\leq f'(t)|_{t=\lVert \widehat{\beta} - \beta \rVert - \varepsilon_0}=0
\end{aligned}
\end{equation*}
where, for the last line, the inequality follows from the convexity of $f$ and the fact that $\lVert \widehat{\beta} - \beta \rVert - \varepsilon_0 \geq \varepsilon_0 >0$, and the equality holds since $t=\lVert \widehat{\beta} - \beta \rVert - \varepsilon_0$ is the minimizer to $f(t)$. 
Using (\ref{eq::derivative-discrete}) and (\ref{eq::lebesgue-discrete}) in the proof of Proposition \ref{prop::m-Rock-loss-discrete}, we have
\begin{equation} 
    \begin{aligned}
        0 &\geq \inf_{\lVert {w}\rVert = 1} \nabla_{{w}} L_n({\beta + \varepsilon_0{w}}) \\
   &=\inf_{\lVert {w}\rVert = 1}\Bigg[\sumM x_m^T{w}\left\{\hat{h}_m[x_m^T(\beta + \varepsilon_0{w})] - \tau\right\} \\
    &~~+ \sumM x_m^T{w}\bm{1}\{x_m^T{w} > 0 \}\int_{0}^1 \bm{1}\{\hat{v}_m(\alpha) = x_m^T(\beta + \varepsilon_0{w})\} \,\mathrm{d}\alpha \Bigg]\label{eq::direction_derivative}. 
    \end{aligned}
\end{equation}
The negativity of the directional derivative and the positivity of the last term in~\eqref{eq::direction_derivative} further implies that $\inf_{\lVert {w}\rVert = 1}\sumM x_m^T{w}\left\{\hat{h}_m[x_m^T(\beta + \varepsilon_0{w})] - \tau\right\} \leq 0 $.


For small enough $\varepsilon_0$, let 
\[
R_1(w) = \hat{h}_m[x_m^T(\beta + \varepsilon_0{w})] - h_m[x_m^T(\beta + \varepsilon_0{w})].
\]
Lemma \ref{lemma::inverseSQ} shows that $R_1(w) = \op(1)$ uniformly over $\lVert {w}\rVert = 1$.
Furthermore, since $h_m(z)$ is continuously differentiable and $h_m(x_m^T\beta) = \tau$, we have
\[
h_m(z) - \tau = (z-x_m^T\beta)h_m'(x_m^T\beta) + o(|z-x_m^T\beta|).
\]

Therefore, for sufficiently small $\varepsilon_0 > 0$, assuming $\sup_{m=1,\ldots,M}\|x_m\| \leq C$ since the covariates are bounded, we have 
\begin{eqnarray*}
&&\P\left(\lVert \widehat{\beta} - \beta \rVert \geq 2\varepsilon_0\right)\\
&\leq&
\P\left(
\inf_{\lVert {w}\rVert = 1}\left\{\sumM x_m^T{w}\left[\hat{h}_m(x_m^T\beta + \varepsilon_0 x_m^T{w}) - \tau\right]\right\} \leq 0
\right)\\
&\leq&
\P\Bigg(
\inf_{\lVert {w}\rVert = 1}\sumM x_m^T{w} R_1(w) + \varepsilon_0\inf_{\lVert {w}\rVert = 1} {w}^T\left[\sumM x_mx_m^Th_m'(x_m^T\beta)\right]{w} \\
&&~~+ \inf_{\lVert {w}\rVert = 1} o\bigg(\varepsilon_0 \sumM x_m^T{w} |x_m^T{w}|\bigg) \leq 0 \Bigg)\\
&\leq&
\P\Bigg(
\varepsilon_0\inf_{\lVert {w}\rVert = 1} {w}^T\left[\sumM x_mx_m^T\frac{1-\tau}{v_m(\tau) - q_m(\tau)}\right]{w} - o\bigg(\varepsilon_0\sumM \|x_m\|^2\bigg) \\
&&~~\leq \sumM \|x_m\| \cdot \sup_{\lVert {w}\rVert = 1} |R_1(w)|
\Bigg)\\
&=&
\P\left(
\varepsilon_0 (1-\tau)\lambda_{\min}(D_1) - o(\varepsilon_0 C^2) \leq C \sup_{\lVert {w}\rVert = 1} |R_1(w)|
\right)\\
&\to&0,
\end{eqnarray*}
since $h_m'(x_m^T\beta) = (v_m - q_m)^{-1}(1-\tau)$, and $D_1$ is positive definite, and $\sup_{\lVert {w}\rVert = 1} R_1(w) = \op(1)$; this concludes the consistency of $\widehat{\beta}$.


Next, we derive the Bahadur-type representation of $\widehat{\beta}$. From Proposition \ref{prop::m-Rock-loss-discrete}, we have
\begin{eqnarray}
\Op\left(\frac{1}{n}\right)
&=& 
\sumM x_m \left[\tau - \hat{h}_m(x_m^T\widehat{\beta})\right]\nonumber\\
&=&
\underbrace{\sumM x_m \left[\tau- \hat{h}_m(x_m^T\beta)\right]}_{R_1} + 
\underbrace{\sumM x_m \left[
h_m(x_m^T\beta) - h_m(x_m^T\widehat{\beta})\right]}_{R_2}\nonumber\\
&&\;+\;\underbrace{\sumM x_m \left\{
[\hat{h}_m(x_m^T\beta) - h_m(x_m^T\beta)] - [\hat{h}_m(x_m^T\widehat{\beta}) - h_m(x_m^T\widehat{\beta})]\right\}}_{R_3}.\label{eq::ASN-discrete}
\end{eqnarray}
Below we consider the three terms $R_1$ through $R_3$ separately.
For the term $R_2$, Taylor expansion of $h_m$ gives
\begin{eqnarray*}
\frac{R_2}{M} &=& -\frac{1}{M}\sumM x_m\left[ x_m^T(\widehat{\beta} - \beta)h_m'(x_m^T\beta) + \op(\lVert\widehat{\beta} - \beta\rVert )\right]\\
&=&-\left[(1-\tau)D_1 + \op(1)\right](\widehat{\beta} - \beta),
\end{eqnarray*}
since $h_m(\cdot)$ is continuously differentiable. 
For $R_3$, the asymptotic equi-continuity in Lemma \ref{lemma::inverseSQ} shows that
\[
\frac{R_3}{M} = \frac{1}{M}\sumM \left[x_m \op\left(\sqrt{\frac{M}{n}}\right)\right] = \op\left(\frac{1}{\sqrt{n}}\right),
\]
since we've shown that $\widehat{\beta}$ is consistent for $\beta$.
Therefore from (\ref{eq::ASN-discrete}) we have
\[
[(1-\tau)D_1 + \op(1)](\widehat{\beta} - \beta) = \frac{1}{M}R_1 + \op\left(\frac{1}{\sqrt{n}}\right),
\]
which proves the asymptotic representation of $\widehat{\beta}$.

Now for the term $R_1$, Lemma \ref{lemma::inverseSQ} shows $\hat{h}_m(z)$ is asymptotically Gaussian at $ z = v_m(\tau) = x_m^T\beta$ for each covariate value $x_m$, therefore:
\[
\sqrt{\frac{n}{M}}\left[\hat{h}_m(x_m^T\beta) - h_m(x_m^T\beta)\right] \converged \mathrm{N}\left(0,\; \frac{(1-\tau)^2\sigma_m^2}{(v_m - q_m)^2}\right).
\]
Summing the above equation over $m$ gives
\[
\frac{\sqrt{n}}{M}R_1 \converged \mathrm{N}\left[0,\;(1-\tau)^2\Omega_1\right],
\]
where $\Omega_1$ is defined in Theorem 3.1. 
Substituting $R_1$, $R_2$ and $R_3$ back into (\ref{eq::ASN-discrete}) gives
\[
\sqrt{n}\left[(1-\tau)D_1 + o(1)\right](\widehat{\beta} - \beta) \converged \mathrm{N}\left[0,\;(1-\tau)^2\Omega_1\right],
\]
which implies
\[
\sqrt{n}(\widehat{\beta} - \beta) \converged \mathrm{N}\left(0,\;D_1^{-1}\Omega_1D_1^{-1}\right). 
\]
The proof is now complete.

\end{proof}

\subsection{Other ES regression approaches in Section 3.2}\label{append::subsec::compete}
In this section, we give additional details to the competing approaches for the i-Rock approach in Section 3.2, but for general covariates. 
Note that our proposed i-Rock approach only assumes linear ES regression model in (3), while all of the other ES regression approaches except linearize (LN) approach assume linearity on both quantile and ES, namely, 
\begin{equation}\label{eq::joint_linear}
    q_{[Y\mid X]}(\tau) = X^T \eta,
    \quad v_{[Y\mid X]}(\tau) = X^T \beta.
\end{equation}
\subsubsection{Description of the approaches}
\begin{enumerate}
    \item 
The first approach to compare with is to linearize (LN) the initial ES estimators, namely, 
\begin{equation*}
    \hat \beta = \arg\min_{u} \sumn (\hat v_i(\tau) - x_i^T u)^2
\end{equation*}
where $\hat v_i(\tau)$ is the $\tau$-level initial ES estimator for $Y_i$ given $X_i$. 
 Motivated by the possible heteroscedasticity, the weighted least squares (WLS) method can be used as an alternative approach to the linearization method, and WLS solves
\begin{equation*}
\min_{{u}}\sumn w_i\left(\hat{v}_i(\tau) - x_i^T{u}\right)^2,
\end{equation*}
where $w_i$ is the weight attached to each covariate value.

\item Second, we consider the joint quantile and expected-shortfall approach proposed in \cite{dimitriadis2019joint,patton2019dynamic}. Given a non-decreasing function $G_1(\cdot)$ and a concave increasing function $G_2(\cdot)$, under \eqref{eq::joint_linear}, the joint approach minimizes a joint loss function
\[
 (\hat{\eta},\hat{\beta})  =  \min_{\eta,\beta} \sumn \ell_i(\eta,\beta; G_1,G_2),
  \]
where the loss function is
  \begin{equation}
  \ell_i(\eta,\beta; G_1,G_2) = \rho_{\tau}\left(G_1(Y_i) - G_1(x_i^T\eta)\right) + G'_2(x_i^T\beta)\left[Z_i(\eta) - x_i^T\beta\right] + G_2(x_i^T\beta),
  \label{eq::joint-loss}
  \end{equation} 
  and 
  \begin{equation}
      Z_i(\eta) = (1-\tau)^{-1}(Y_i - x_i^T \eta)\mathbbm{1}(Y_i\geq x_i^T \eta) + x_i^T \eta.
  \label{eq::def-Zi}
  \end{equation} 
  We set $G_1(u) = u$ and focus on the following two options for $G_2$ advocated by \cite{dimitriadis2019joint} and \cite{patton2019dynamic},
  \[
G_{2}(u) = \log(u),\quad\text{and}\quad G_{2}(u) = \sqrt{u}.
  \]
  which we name J1 and J2, respectively.

\item As an extension of the joint approach, we consider the third approach, Neyman-orthogonalized least squares (TSN), proposed in \cite{barendse2020efficiently}, where a two-step estimation procedure is proposed with a quantile regression followed by a least squares estimation, namely,
\begin{gather*}
\hat{\eta}  =  \min_\eta \sumn \rho_\tau(Y_i - x_i^T\eta),\\
\hat{\beta}  =  \min_\beta \sumn \left[Z_i(\hat{\eta}) - x_i^T\beta\right]^2,
\end{gather*}
where $Z_i(\eta)$ is defined in~\eqref{eq::def-Zi}.


\item The last approach to compare with involves a straightforward two-step least squares (TSLS) approach, which is natural yet scarce in the literature. Here, we first estimate the quantile and then average the responses above the estimated quantile. Specifically, under~\eqref{eq::joint_linear}, the two-step procedure is
\begin{equation}
\begin{aligned}
\hat{\eta} & =  \min_\eta \sumn \rho_\tau(Y_i - x_i^T\eta),\\
\widehat{\beta} &= \argmin_\theta \sumn \left[(Y_i - x_i^T\theta)^2\cdot \bm{1}\{Y_i\geq x_i^T\hat{\eta}\}\right].
\end{aligned}
\end{equation}
\end{enumerate}

\subsubsection{Asymptotic variance for the competing approaches}
Here, we compute the asymptotic variance-covariance matrices for the competing approaches under Model~\eqref{eq::joint_linear}. Let 
\[
m_1(x) = \var(Y\mid Y\geq q(\tau,x),X=x),\quad m_2(x) = [v(\tau,x) - q(\tau,x)].
\]

\begin{enumerate}
    \item If the initial estimator $\hat v(\tau,x)$ has conditional variance $m_1(x) + \tau m^2_2(x)$,
    then the WLS estimator has the following asymptotic variance-covariance matrix
    \[
\text{AVar}_{WLS} = \frac{1}{1-\tau}[\E\{w(X) XX^T\}]^{-1}\E\left[w^2(X) XX^T\{m_1(X) + \tau m^2_2(X)\} \right][\E\{w(X) XX^T\}]^{-1}.
\]
\item From Theorem 2.4 of \cite{dimitriadis2019joint}, the joint estimators have the following asymptotic variance-covariance matrix
\begin{align*}
\text{AVar}_{J1} &= \frac{1}{1-\tau}\left[\E\left\{\frac{XX^T}{v^2(\tau,X)}\right\}\right]^{-1}\E\left[XX^T\left\{\frac{m_1(X) + \tau m^2_2(X)}{v^4(\tau,X)}\right\} \right]\left[\E\left\{\frac{XX^T}{v^2(\tau,X)}\right\}\right]^{-1},\\
\text{AVar}_{J2} &= \frac{1}{1-\tau}\left\{\E\left\{\frac{XX^T}{v^{3/2}(\tau,X)}\right\}\right]^{-1}\E\left[XX^T\left\{\frac{m_1(X) + \tau m^2_2(X)}{v^3(\tau,X)}\right\} \right]\left[\E\left\{\frac{XX^T}{v^{3/2}(\tau,X)}\right\}\right]^{-1},
\end{align*}
for $G_2(u) = \log u$ and $G_2(u) = \sqrt{u}$, respectively.
\item From Theorem 1 of \cite{barendse2020efficiently},
the asymptotic variance-covariance matrix for TSN approach is
\[
\text{AVar}_{TSN} = \frac{1}{1-\tau}\{\E(XX^T)\}^{-1}\E\left[XX^T\{m_1(X) + \tau m^2_2(X)\} \right]\{\E(XX^T)\}^{-1}.
\]
\item Finally, recalling from Theorem \ref{thm::example-RRM-SQ} that the i-Rock estimator with bin-wise linear initial estimator has the following asymptotic variance-covariance matrix
\[
\text{AVar}_{i-Rock} = \frac{1}{1-\tau}\left[\E\left\{\frac{XX^T}{m_2(X)}\right\}\right]^{-1}\E\left[XX^T\left\{\frac{m_1(X)}{m_2^2(X)} + \tau\right\}\right]\left[\E\left\{\frac{XX^T}{m_2(X)}\right\}\right]^{-1}.
\]
\end{enumerate}

The matrices $\text{AVar}_{TSN}$, $\text{AVar}_{J1}$, $\text{AVar}_{J2}$, and $\text{AVar}_{i-Rock}$ all involve the weighted expectations of $m_1(X) + \tau m_2^2(X)$,
which are asymptotically equivalent to WLS in the form of (\ref{eq::bin-weighted-linearization}) where the response has conditional variance $m_1(x) + \tau m_2^2(x)$.
The key difference between these approaches lies in the weighting scheme. For $\text{AVar}_{TSN}$, $\text{AVar}_{J1}$, $\text{AVar}_{J2}$ and $\text{AVar}_{i-Rock}$, the weights are proportional to $1$, $[v(\tau,x)]^{-2}$, $[v(\tau,x)]^{-3/2}$, and $[m_2(x)]^{-1}$ respectively. 
Under the class of WLS formulation in~\eqref{eq::bin-weighted-linearization} when the initial estimators have conditional variance $m_1(x) + \tau m^2_2(x)$, the optimal weight is proportional to $\{m_1(x) + \tau m^2_2(x)\}^{-1}$, which is equivalent to the efficient weight in \cite{barendse2020efficiently}.

\subsubsection{Additional empirical results}
\begin{figure}[tb]
\centering
\begin{subfigure}[b]{0.43\textwidth}
         \centering
         \includegraphics[width=\textwidth]{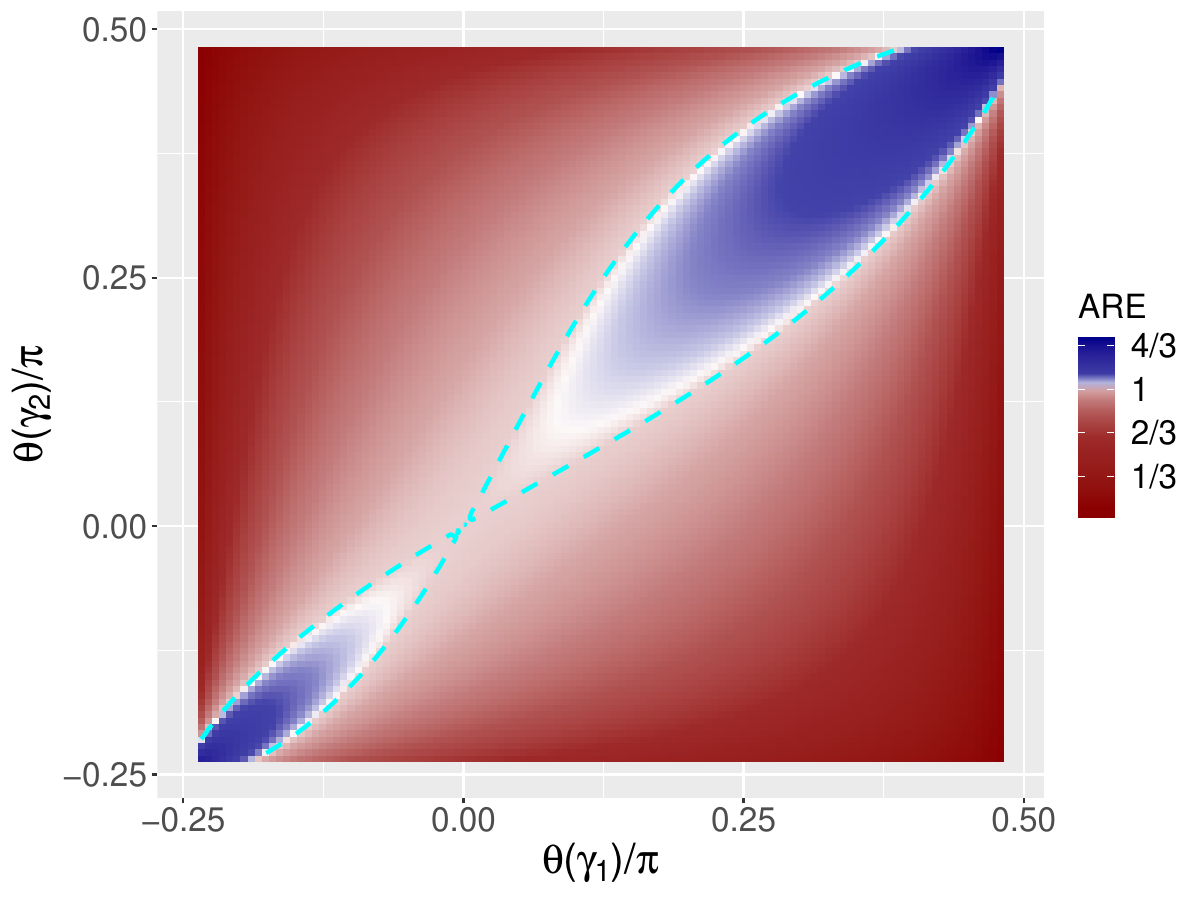}
         \caption{Method J1}
         \label{fig::J1-nid}
     \end{subfigure}
     \hfill
\begin{subfigure}[b]{0.43\textwidth}
         \centering
         \includegraphics[width=\textwidth]{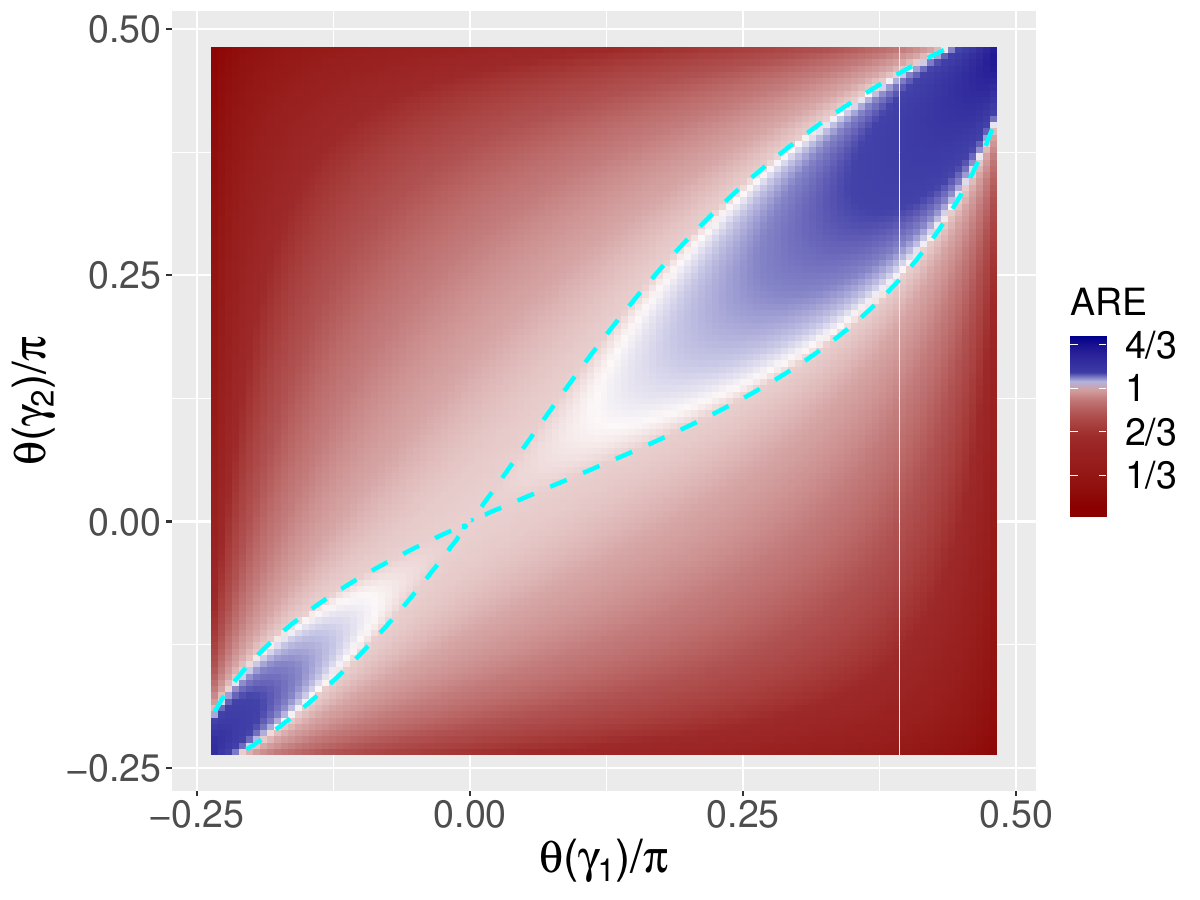}
         \caption{Method J2}
         \label{fig:J2-nid}
     \end{subfigure}
     \caption{
     The heatmap of ARE of the joint approaches relative to the i-Rock approach under Model~\eqref{eq::comparison-simple} for each value of $\gamma_1$ and $\gamma_2$ on the unit circle at $\tau = 0.9$.
     Here $\theta(\gamma)$ represents the angular coordinate of $\gamma$, and the $x$ and $y$ axes represent the angular coordinates of $\gamma_1$ and $\gamma_2$ varying between $-\pi/4$ and $\pi/2$ in the coordinate system, respectively. 
     The ARE is measured by the Frobenius norm of the asymptotic variance-covariance matrix.
     }
        \label{fig::asymp_effi_Joints}
\end{figure}

To better elucidate the circumstances in which the i-Rock approach is less or more efficient than the joint approach,
we present ARE of the joint approaches relative to the i-Rock approach under the linear location-scale model~\eqref{eq::comparison-simple} for various choices of $\gamma_1$ and $\gamma_2$ when $p=1$ in Figure~\ref{fig::asymp_effi_Joints}. 
Here, ARE less than one indicates that the i-Rock approach is more efficient.
In most cases, the i-Rock approach is more efficient than the joint approaches, with efficiency improvements reaching up to $30\%$.
Only when $\gamma_1$ is approximately parallel to $\gamma_2$, the joint approaches can be slightly more efficient than the i-Rock approach. 
This aligns with the theoretical interpretation that the joint approach is more efficient only when its asymptotic weight aligns more closely with the direction of the optimal weight; Specifically under Model~\eqref{eq::comparison-simple}, the optimal weight is $w_m \propto (x_m^T\gamma_2)^{-2}$, while the implicit weights for the two joint approaches are $w_m \propto (x_m^T \gamma_1)^{-t}$, where $t=2$ for J1 and $t=3/2$ for J2. 

\section{Proof and additional details of Section~\ref{sec::continuous}}\label{append::sec4}
\subsection{An initial ES estimator}\label{subsec::examples}
In this section, we present the exact formulation of the initial ES estimator in Section 4.2 of main manuscript and technical conditions for the asymptotic distribution of the resulting i-Rock estimator in Theorem 4.2.  


To fit a bin-wise linear ES regression, we first separate the intercept term from the covariates by writing $X_i = (1,\tX_i^T)^T$. Within each bin $A_m$, we center the covariates by subtracting off the bin center $\bar{x}_m = (1,\tilde{x}_m^T)^T$, and then fit the following two-step ES regression to obtain the initial estimator at quantile level $s\in(0,1)$:
\begin{eqnarray*}
(\hat{c}_0,\hat{c}_1) &=& \argmin_{\substack{c_0\in\mathbb{R}\\c_1\in\mathbb{R}^{p}}} \sum_{\substack{X_i\in A_m}} \left[\hat{Z}_i(s) - c_0 - c_1^T(\tilde{X}_i - \tilde{x}_m)\right]^2,\nonumber \\ 
\hat{v}(s,\bxm) &=& \hat{c}_0.
\end{eqnarray*}
where $\hat{Z}_i(s) = (1-s)^{-1} \{Y_i - \hat{q}(s,X_i)\}\bm{1}\{Y_i\geq \hat{q}(s,X_i)\} + \hat{q}(s,X_i)$ and $\hat{q}(s,x)$ can be any parametric or non-parametric estimator of the conditional quantile function of $Y\mid X=x$ satisfying Condition~\ref{cond::bin-Qt}. The solution to the bin-wise ES regression has the following closed-form
\begin{equation}\label{eq::ll-SQ}
\hat{v}(s,\bxm) 
= 
e_1^T\left[{\bm{X}}_m^T\bm{{W}}_m{\bm{X}}_m\right]^{-1}\;
\left[
\begin{array}{cc}
\sum_{i=1}^nw_{im}\hat{Z}_i(s)\vspace{0.8em}\\
\sum_{i=1}^n(\tilde{X}_i - \tilde{x}_m)w_{im}\hat{Z}_i(s)
\end{array}\right],
\end{equation}
where $e_1 = (1,0,\ldots,0)$ is a unit vector in $\mathbb{R}^{p+1}$,  
$
\bm{{X}}_m = \left[
\bm{1}_n,
\begin{array}{c}
(\tX_1 - \txm)^T\\
\vdots\\
(\tX_n - \txm)^T
\end{array}
\right] \in \mathbb{R}^{n\times (p+1)}$, and $w_{im} = \bm{1}(X_i\in A_m)$, and $
\bm{{W}}_m = \text{diag}(w_{1m},\ldots, w_{nm})\in\mathbb{R}^{n\times n}
$. 
Here, if ${\bm{X}}_m^T\bm{{W}}_m{\bm{X}}_m$ is not invertible, $({\bm{X}}_m^T\bm{{W}}_m{\bm{X}}_m)^{-1}$ is defined as the Moore–Penrose pseudo-inverse.
Furthermore, we use the following weights in the i-Rock estimating equation \eqref{eq::estimator-rock-binning},
\begin{equation}
\label{eq::ll-gweights}
\hgamma_m = (S_{0m} - S_{1m}^T\bm{S}_{2m}^{-1}S_{1m}),
\end{equation}
where the quantities $\{S_{jm}:j=0,1,2\}$ are defined through the partition
\begin{equation}\label{eq::ll-Sm}
    n^{-1} \left[{\bm{X}}_m^T\bm{{W}}_m{\bm{X}}_m\right] = 
\left[\begin{array}{cc}
S_{0m} & S_{1m}^T\\
S_{1m} & \bm{S}_{2m}
\end{array}
\right].
\end{equation}




To study the theoretical behavior of the i-Rock estimator in \eqref{eq::estimator-rock-binning} with weights~\eqref{eq::ll-gweights} and initial ES estimator~\eqref{eq::ll-SQ}, 
we further impose some regularity conditions on the data-generating process, the binning mechanism, the quantile estimator for the initial ES estimator, as well as a stronger version of Condition \ref{cond::R-bin-Y-density}.

\begin{customcond}{G-Y1'}
\label{cond::bin-Y-new}
All requirements in Condition \ref{cond::R-bin-Y-density} hold. In addition, we have
\begin{enumerate}[(i)]
\item For some constant $L_1 > 0$,
\[
\sup_{x\in\mathcal{X}}|f_{Y\mid X}(y_1; x) - f_{Y\mid X}(y_2; x)| \leq L_1|y_1 - y_2|.
\]
\item For some $\delta_0 > 0$,
\[
\sup_{x\in\mathcal{X}}\, \E\left[ (Y^{+})^{2+\delta_0}\Big| X=x\right] < + \infty,
\]
where $Y^{+} = \max\{Y,0\}$.
\end{enumerate}
\end{customcond}

For each bin $A_m$, let 
\[
\bh_m = \sup_{x\in A_m} \lVert x - \bxm \rVert_2,\quad \ubh_m = \inf_{x\notin A_m} \lVert x - \bxm\rVert_2,
\]
where $\bh_m$ is the radius of the bin, and $\ubh_m$ is the separation between bins. We further define $\bh = \max_m\{\bh_m\}$ and $\ubh = \min_m\{\ubh_m\}$. Both $\bh$ and $\ubh$ depend on the sample size $n$.
\begin{customcond}{G-A1}
\label{cond::bin-bandwidth}
There exists a constant $\varepsilon_1 > 0$, such that
\[
\bh \to 0, \qquad \ubh^p \gg \frac{\log(n)}{n^{\min\left\{1/2,\;1-2/(2+\delta_0) - \varepsilon_1\right\}}},
\]
where $1-2/(2+\delta_0) - \varepsilon_1>0$ and $\delta_0$ is in Condition \ref{cond::bin-Y-new}.
Furthermore, for some constants $0 < m_h < M_h < +\infty$,
\[
m_h \leq \liminf_{n\to\infty} (\bh^{-1}{\ubh} )\leq \limsup_{n\to\infty}(\bh^{-1}{\ubh} ) \leq M_h.
\]
\end{customcond}

\begin{customcond}{G-A2}
\label{cond::bin-Aux}
At least one of the following conditions holds.
\begin{enumerate}[(i)]
\item The covariate-dimension $p < 4$ and $\bh^{4-p} = o(\{\log(n)\}^{-1})$; furthermore, $v(s,x)$ is twice continuously differentiable with respect to $x$, and for some $\varepsilon_1 > 0$ and $L_2 > 0$, 
\[
\sup_{x\in\mathcal{X}}\left\lVert \frac{\partial^2 v(s,x)}{\partial x\partial x^T} - \frac{\partial^2 v(\tau,x)}{\partial x\partial x^T}\right\rVert_2 \leq L_2|s-\tau|,
\]
for all $|s-\tau| \leq \varepsilon_1$.
\item 
For sufficiently large $n$, there exists $\beta^{(m)}_{n}(s)$ such that
\[
v(s,x) = x^T\beta^{(m)}_{n}(s),\quad |s-\tau| \leq \varepsilon_2 n^{-1/4}\,;\,x\in A_m,
\]
for each $m = 1,\ldots,M$, where $\varepsilon_2 > 0$ is a universal constant.
\end{enumerate}
\end{customcond}
Condition~\ref{cond::bin-Y-new} is a standard assumption to ensure uniform consistency in the initial ES estimator (see, e.g., \cite{mack1982weak}).
Condition~\ref{cond::bin-bandwidth} ensures that each bin has an appropriate size, which is similar to the general bandwidth conditions for kernel ES estimation \citep{olma2021nonparametric}. Moreover, Condition~\ref{cond::bin-bandwidth} ensures that $\bh_m$ and $\ubh_m$ are of the same order, which holds if, for example, all the bins are hyperspheres or hypercubes and the density of $X$ is positive everywhere on a compact set. 
{\color{black} In Condition~\ref{cond::bin-bandwidth}, $\bar h^p$ relates to the volume of the bins. With a slight abuse of notation, the $p$ in this and the subsequent conditions refers to the number of continuous variables when both continuous and discrete variables are present.}
Condition \ref{cond::bin-Aux} requires either a low-dimensional model with smooth ES functions over $x\in\mathcal{X}$, or piece-wise linear ES functions near level $\tau$. We note the bandwidth condition in item (i) of Condition~\ref{cond::bin-Aux} is compatible with Condition \ref{cond::bin-bandwidth} when $p < 4$. 

In addition, we  require the following technical Condition \ref{cond::bin-Qt}  on the quantile regression estimator $\hat{q}(s,x)$ involved in our initial ES estimator (\ref{eq::ll-SQ}).
Condition \ref{cond::bin-Qt} is relatively weak in the sense that $\hat{q}(s,x)$ does not necessarily have to be based on either a parametric quantile regression model or the same binning mechanism as $\hat{v}(s,x)$. To provide a specific choice of such quantile estimator, we consider the local-linear quantile estimators 
\begin{eqnarray}\label{eq::ll_quantile}
     \hat q(s,x) &=& \hat \eta_m(s)^T x, \quad x\in A_m, \quad |s-\tau|\leq c_0 n^{-1/4},
\end{eqnarray}
where \begin{eqnarray*}
\hat \eta_m(s) &=& \argmin_{\substack{c\in\mathbb{R}^{p+1}}} \sum_{\substack{X_i\in A_m}} \rho_s(Y_i - c^TX_i).
\end{eqnarray*}

The quantile estimator in~\eqref{eq::ll_quantile} satisfies Condition G-Q below when the true quantile function is piecewise-linear, i.e., $q(s,x) = \eta_m(s)^Tx$ when $x\in A_m$. 
\\
\begin{customcond}{G-Q}
\label{cond::bin-Qt}
For some sequence 
$g_{1n} =o(n^{-1/4}) $,
the conditional quantile estimator $\hat{q}(s,x)$ satisfies
\begin{enumerate}[(i)]
\item 
\itemEq{
\sup_{\substack{x\in\mathcal{X}\\s:|s-\tau|\leq  n^{-1/4}  }} |\hat{q}(s,x) - q(s,x)| = O_p(g_{1n}).
}
\item 
For each $j = 0,1$, 
\small\[
 \sup_{\substack{m = 1,\ldots,M\\s:|s-\tau|\leq  n^{-1/4}}} \left\lvert \ddfrac{\sum_{\substack{i=1\\X_i\in A_m}}^n \left[\frac{\tilde{X}_i - \txm}{\bh_m}\right]^j[\hat{q}(s,X_i) - q(s,X_i)] [s - \bm{1}\{Y_i \leq q(s,X_i)\}]}{\sumn \bm{1}\{X_i\in A_m\}}\right\rvert =
 \op\left(n^{-1/2}\right).
\]
\end{enumerate}
\end{customcond}
With these technical conditions, we establish the asymptotic normality for the i-Rock estimator with the locally linear ES initial estimator in Theorem 4.2 in the main manuscript. 

\subsection{Some technical lemmas}
\label{subsec::lemmas-SQ-comt}
Here we collect some technical lemmas that are useful for our proofs later. These lemmas do not depend on the specific construction of an initial ES estimator, and hence are applicable for the results in both Sections \ref{subsec::asy_linear} and \ref{subsec::examples}. The proof of these lemmas can be found in Appendix~\ref{append::lemma}.

We fix some notations here. Recall $A_1,\ldots,A_M$ are the bins. For each bin, $\bxm$ is its geometric center, and $\hgamma_m$ is a weight (that only depends on the covariates) in the i-Rock estimation procedure (\ref{eq::estimator-rock-binning}). Let $\hat{v}(s,x)$ and $\hat{q}(s,x)$ be the initial binning ES and quantile estimators, respectively. The lemmas below do not depend on any particular choice of $\hgamma_m$, $\hat{v}$ or $\hat{q}$. The total number of bins $M$ is allowed to increase with the sample size. For the binning mechanism, let $w_{im} = \bm{1}\{X_i\in A_m\}$, and let
\[
\hat{\pi}_m = n^{-1}\sumn w_{im},
\]
be the proportion of data that falls into bin $A_m$. For each bin, we write $diam(A_m) = \bh_m = \sup_{x\in A_m}\lVert x - \bxm\rVert$ and $\ubh_m = \inf_{x\notin A_m}\lVert x - \bxm\rVert$.
We further define the inverse ES function as\footnote{Without loss of generality, we assume $\hat{v}(s,x)$ and $v(s,x)$ are (weakly) increasing in $s$.}:
\begin{equation}
\begin{gathered}
h(z,x)  :=  \int_0^1\bm{1}\{v(s,x) \leq z\}\,\d s = \sup\{s\in[0,1]: v(s,x) \leq z\},\\
\hat{h}(z,x) := \int_0^1\bm{1}\{\hat{v}(s,x) \leq z\}\,\d s = \sup\{s\in[0,1]: \hat{v}(s,x) \leq z\}.
\end{gathered}
\label{eq::def-h}
\end{equation}
In the Operations Research literature, these functions are called the `superdistribution' functions, in duality to the superquantile functions \citep{rockafellar2013superquantiles,rockafellar2013fundamental}.

We also use the following set of notations. 
For a vector $v$, let $\lVert v \rVert$ be its $\ell_2$ norm; for a matrix $A$, let $\lVert A \rVert$ be its operator norm. For two deterministic sequences $a_n$ and $b_n$, we write $a_n \ll b_n$ if $a_n=o(b_n)$ and $a_n \lesssim b_n$ if there exists a universal constant $C^* > 0$ such that $a_n \leq C^*b_n$; we define $a_n \asymp b_n$ if both $a_n = O(b_n)$ and $b_n = O(a_n)$ hold.
For stochastic sequences $A_n$ and $B_n$, we use the notations $A_n \ll_{\P} B_n$ and $A_n \lesssim_{\P}B_n$ to denote $A_n = \op(B_n)$ and $A_n  = \Op(B_n)$, respectively.

\begin{lemma}\label{lemma::bin-approx-moment-X}
Suppose the bins $A_m$ and the associated weights $\hgamma_m$ satisfy:
\[
\sup_{m=1,\ldots,M} \text{diam}(A_m)\convergeip 0,\quad \supM\left| \frac{\hgamma_m}{\hat{\pi}_m} - 1\right| \convergeip 0.
\]
Let $g(\cdot):\mathcal{X}\mapsto \mathbb{R}^m$ be a bounded and Lipschitz continuous function over $\mathcal{X}$, then we have
\[
\sumM\hgamma_mg(\bar{x}_m)\; \convergeip \;\E[g(X)].
\]
In addition, if $h(\cdot):\mathcal{X}\mapsto \mathbb{R}$ is a function such that $\E[|h(X)|] < \infty$, then we have
\[
\E\left[\sumM \bm{1}_{\{X\in A_m\}}g(\bar{x}_m)h(X)\right]\to \E[h(X)g(X)],
\]
as $n\to\infty$.
\end{lemma}

\begin{lemma}\label{lemma::bin-cond-implications}
Under Conditions \ref{cond::bin-covar}, \ref{cond::R-bin-Y-density} and \ref{cond::bin-SQ}, there is a constant $c_1 > 0$ such that the following results hold:
\begin{enumerate}
\item For some constants $0 < \underline{m}_1 < \overline{m}_1 < +\infty$, we have
\[
\underline{m}_1 \;\leq \inf_{\substack{x\in\mathcal{X}\\s: |s-\tau| \leq c_1}}|v(s,x) - q(s,x)| \leq \sup_{\substack{x\in\mathcal{X}\\s: |s-\tau| \leq c_1}}|v(s,x) - q(s,x)| \leq\; \overline{m}_1 .
\]
\item Both $q(s,x)$ and $v(s,x)$ are differentiable with respect to $s$ when $|s - \tau| \leq c_1$, and there exist constants $0< \underline{m}_2 < \overline{m}_2 < +\infty$ such that
\begin{eqnarray*}
\underline{m}_2\; \leq \inf_{\substack{x\in\mathcal{X}\\s: |s-\tau| \leq c_1}}\left|\frac{\partial q}{\partial s}(s,x)\right|  
&\leq&
\sup_{\substack{x\in\mathcal{X}\\s: |s-\tau| \leq c_1}}\left|\frac{\partial q}{\partial s}(s,x)\right| \leq
\;\overline{m}_2 ,\\
\underline{m}_2\;\leq \inf_{\substack{x\in\mathcal{X}\\s: |s-\tau| \leq c_1}}\left|\frac{\partial v}{\partial s}(s,x)\right| 
&\leq&
 \sup_{\substack{x\in\mathcal{X}\\s: |s-\tau| \leq c_1}}\left|\frac{\partial v}{\partial s}(s,x)\right| \leq \;\overline{m}_2 .
\end{eqnarray*}

\item There exists a constant $L > 0$ such that 
\begin{align*}
  \sup_{x\in\mathcal{X}}\left|\frac{\partial v}{\partial s}(s_1,x) - \frac{\partial v}{\partial s}(s_2,x)\right| &\leq L|s_1 - s_2|,\\
  \sup_{x\in\mathcal{X}}\left|\left[\frac{\partial v}{\partial s}(s_1,x)\right]^{-1} - \left[\frac{\partial v}{\partial s}(s_2,x)\right]^{-1}\right| &\leq L|s_1 - s_2|,
\end{align*}
for all $s_1,s_2\in[\tau - c_1,\tau + c_1]$.
\end{enumerate}
\end{lemma}

\begin{lemma}
\label{lemma::bin-technical-h}
Suppose the initial estimators $\hat{v}(s,\bar{x}_m)$ satisfy Condition \ref{cond::vhat1}, and the binning weights $\hgamma_m$ satisfy
\[
 \supM\left| \frac{\hgamma_m}{\hat{\pi}_m} - 1\right| \convergeip 0.
\] 
The following results hold, where $r_n$ is the same as in Condition \ref{cond::vhat1}.
\begin{enumerate}
\item \small \itemEq{
\sumgammaM \sup_{\substack{(z,z'): z = v(\tau,\bar{x}_m)\\|z'-z| \lesssim (r_n\vee n^{-1/2})}} \left|\{\hat{h}(z,\bar{x}_m) - h(z,\bar{x}_m)\} - \{\hat{h}(z',\bar{x}_m) - h(z',\bar{x}_m)\}\right| = \op\left(\frac{1}{\sqrt{n}}\right).
}
\item \itemEq{
    \sumgammaM \left[\hat{v}(\tau,\bar{x}_m) - v(\tau,\bar{x}_m)\right]^2 = \op\left(\frac{1}{\sqrt{n}}\right).
}
\item \itemEq{
    \sumgammaM \left|\tau - \hat{h}\circ(\hat{v}(\tau,\bar{x}_m),\bar{x}_m)\right| = \op\left(\frac{1}{\sqrt{n}}\right).
}
\end{enumerate}
\end{lemma}

\begin{lemma}\label{lemma::ll-coef}
Recall $S_{0m}$, $S_{1m}$, and $\bm{S}_{2m}$ given in (\ref{eq::ll-Sm}). Under Conditions \ref{cond::bin-bandwidth} and \ref{cond::bin-covar}, the following results hold:
\begin{enumerate}
\item \itemEq{\supM\left |S_{0m}^{-1}\,{(S_{0m} - S_{1m}^T\bm{S}_{2m}^{-1}S_{1m})} - 1\right| = \op(1).}
\item For any fixed $c_2 > 0$,
\[
\Pr\left(\sup_{m=1,\ldots,M}\left\lVert \bh_mS_{1m}^T\bm{S}_{2m}^{-1}\right\rVert  \geq c_2\right)\leq\frac{C_2}{n^3},
\]
where $C_2$ is a constant that may depend on $c_2$.
\item For some $\varepsilon_2 > 0$,
\[
\Pr\left(\inf_{m=1,\ldots,M}  |{\bh_m^{-p}}{S_{0m}}| \leq \varepsilon_2\right) \leq \frac{1}{n^3}
\]
\end{enumerate}
\end{lemma}

\begin{lemma}
\label{lemma::proof-U13}
Suppose Condition \ref{cond::bin-Y-new} holds. For any two sequences $a_n,b_n\to 0$, if the quantile estimator $\hat{q}(s,x)$ satisfies
\[
\sup_{\substack{x\in\mathcal{X}\\s:|s-\tau|\leq a_{n}}}|\hat{q}(s,x) - q(s,x)| = \Op(b_n),
\]
then we have
\begin{enumerate}
\item \itemEq{
\sup_{\substack{m=1,\ldots,M\\|s-\tau|\leq a_{n}}} \left|\ddfrac{\sumn w_{im}\kappa_{im}[Y_i - q(s,X_i)]\left[\bm{1}\{Y_i \geq \hat{q}(s,X_i)]\} - \bm{1}\{Y_i \geq q(s,X_i)\}\right]}{\sumn w_{im}} \right|\\ \qquad= \Op(a_{n}^2 + b_n^2);
}
\item \itemEq{
  \sup_{\substack{m=1,\ldots,M\\|s-\tau|\leq a_{n}}} \left|\ddfrac{\sumn w_{im}\kappa_{im}\left(q(s,X_i) - \hat{q}(s,X_i)\right) \left[\bm{1}\{Y_i \geq \hat{q}(s,X_i)\} - \bm{1}\{Y_i \geq q(s,X_i)\}\right]}{\sumn w_{im}} \right|\\ \qquad= \Op(a_{n}^2 + b_n^2);
}
\end{enumerate}
where $w_{im} = \bm{1}\{X_i\in A_m\}$ and $\kappa_{im} = [1 - S_{1m}^T\bm{S}_{2m}^{-1}(\tX_i - \txm)]$; $S_{1m}$, $\bm{S}_{2m}$ are given in \eqref{eq::ll-Sm}.
\end{lemma}

We comment on these lemmas.
Lemmas \ref{lemma::bin-approx-moment-X} and \ref{lemma::ll-coef} are about the covariates under the binning mechanism. 
Lemma \ref{lemma::bin-cond-implications} gives some technical implications on the data-generating process derived from the conditions in Section~\ref{subsec::asy_linear}. 
In addition, Lemma \ref{lemma::bin-technical-h} and \ref{lemma::proof-U13} are more technical. 
In particular, the results in \ref{lemma::bin-technical-h} are similar to, but stronger than the examples given by standard functional delta method \citep[Chapter 20]{van2000asymptotic}. 
For a fixed quantile level $\tau$, Lemma \ref{lemma::proof-U13} is the same as Lemma A.4 in \citet{olma2021nonparametric} and Lemma A.3 in \citet{kato2012weighted}. Our result is stronger in uniformity over $s$.

\subsection{Proof of Theorem \ref{thm::asymptotic-RRM-bin}}

\label{subsec::proof-main}
To simplify the notations, in the following proof we define $v_m(\alpha) = v(\alpha,\bxm )$, and $\hat{v}_m(\alpha) = \hat{v}(\alpha,\bxm )$; correspondingly we write ${h}_m(z) = h(z,\bxm )$ and $\hat{h}_m(z) = \hat{h}(z,\bxm )$ for the inverse ES function. When there is no confusion, we shall write $v_m = v(\tau,\bxm)$ without the index to refer to the targeting $\tau$th ES. 
In our proof, the number of bins $M = M_n$ increases with the sample size, though we often omit the subscript.

\begin{proof}[Proof of Theorem \ref{thm::asymptotic-RRM-bin}]
In this proof, we work with the following shifted i-Rock objective function:
\begin{equation}
 L_n(\delta) = \sumM \hgamma_m \int_{0}^1\; \left[\rho_\tau\left(\hat{v}_{m}(\alpha) - v_{m}(\tau) - \bxm \delta/\sqrt{n} \right) -  \rho_\tau\left(\hat{v}_{m}(\alpha ) - v_m(\tau )\right)\right]\,\d\alpha.
\label{eq::proof-rock-Ln}
\end{equation}
It follows that $\hat{\delta} = n^{1/2}(\widehat{\beta} - \beta)$ minimizes $L_n(\delta)$, where $\widehat{\beta}$ is the i-Rock estimator in (\ref{eq::estimator-rock-binning}). Therefore, it suffices to study the asymptotic properties of $\hat{\delta}$.
To this end, we first show that the function $L_n(\delta)$ in (\ref{eq::proof-rock-Ln}) converges (pointwise) in probability to a quadratic function of $\delta$. Then we apply the convexity argument in \citet{pollard1991asymptotics} to derive the asymptotic properties of $\hat{\delta}$. We define $\Delta_m(\delta) = \hat{v}_m(\tau) - v_m(\tau) - n^{-1/2}\bxm ^T\delta$.

By Knight's identity \citep{knight1998limiting}, 
\[
\rho_\tau(w-v) - \rho_\tau(w) = -v(\tau - \bm{1}\{w \leq 0\}) + \int_{0}^{v}(\bm{1}\{w \leq t\} - \bm{1}\{w \leq 0\}) \,\d t,
\]
for any $w$ and $v$, therefore
\begin{eqnarray*}
\rho_\tau(w-v_1)-\rho_\tau(w-v_2) 
&=&
[\rho_\tau(w-v_1) - \rho_\tau(w)] -
[\rho_\tau(w-v_2) - \rho_\tau(w)]
\\
&=& (v_2-v_1)(\tau - \bm{1}\{w \leq 0\}) + \int_{v_2}^{v_1}(\bm{1}\{w \leq t\} - \bm{1}\{w \leq 0\}) \,\d t.
\end{eqnarray*}
Taking $w = \hat{v}_m(\alpha)- \hat{v}_m(\tau)$, $v_1 = -\Delta_m(\delta)$, and $v_2 = v_m(\tau) - \hat{v}_m(\tau)$ in the above displayed equation, we obtain:
\begin{eqnarray*}
&&\int_{0}^1 \rho_\tau\left(\hat{v}_m(\alpha) - v_{m}(\tau) - \bxm ^T\delta/\sqrt{n}\right)\,\d\alpha -  \int_{0}^1\rho_\tau\left(\hat{v}_{m}(\alpha ) - v_{m}(\tau)\right)\,\d\alpha\\
&=&-n^{-1/2}\tilde{x}^T\delta\int_{0}^1 (\tau - \bm{1}\{\hat{v}_m(\alpha)\leq \hat{v}_m(\tau)\}) \d\alpha\\
&&\;+\; \int_{0}^1\int_{v_m - \hat{v}_m}^{-\Delta_m(\delta)}(\bm{1}\{\hat{v}_m(\alpha) \leq \hat{v}_m(\tau) + t\} - \bm{1}\{\hat{v}_m(\alpha) \leq \hat{v}_m(\tau)\})\,\d t\,\d\alpha\\
&=& -n^{-1/2}\bxm ^T\delta[\tau - \hat{h}_m(\hat{v}_m)] +\;\int_{v_m - \hat{v}_m}^{-\Delta_m(\delta)}[\hat{h}_m(\hat{v}_m + t) - \hat{h}_m(\hat{v}_m)]\,\d t,
\end{eqnarray*}
where the last equality follows from the definition of $\hat{h}$ in (\ref{eq::def-h}), and by exchanging the order of integration.
Therefore, summing over $m = 1,\ldots,M$ in the above equation gives the following decomposition for $L_n(\delta)$ (defined in (\ref{eq::proof-rock-Ln})):
\begin{eqnarray}
L_n(\delta)
&=&
\underbrace{-n^{-1/2}\sumM \hgamma_m\bxm ^T\delta[\tau - \hat{h}_m(\hat{v}_m)]}_{A_n(\delta)} + 
\sumM \hgamma_m\int_{v_m - \hat{v}_m}^{-\Delta_m(\delta)}[\hat{h}_m(\hat{v}_m + t) - \hat{h}_m(\hat{v}_m) ]\,\d t\nonumber\\
&=&
A_n(\delta)+
\underbrace{\sumM  \hgamma_m \int_{v_m - \hat{v}_m}^{-\Delta_m(\delta)}[h_m(\hat{v}_m + t) - h_m(\hat{v}_m)]\,\d t}_{B_n(\delta)}\nonumber\\
&&\;+\;
\underbrace{\sumM  \hgamma_m \int_{v_m - \hat{v}_m}^{-\Delta_m(\delta)}\left[\{\hat{h}_m(\hat{v}_m + t) - h_m(\hat{v}_m + t)\} - \{\hat{h}_m(\hat{v}_m)-h_m(\hat{v}_m)\}\right]\,\d t}_{C_n(\delta)}\nonumber\\
&\triangleq& A_n(\delta) + B_n(\delta) + C_n(\delta).
\label{eq::proof-rock-ABC}
\end{eqnarray}


For any fixed $\delta = O(1)$, we shall show that both $A_n(\delta)$ and $C_n(\delta)$ are $\op(n^{-1})$.
For $A_n(\delta)$, note $\bxm ^T\delta$ is uniformly bounded over $m$, hence
\[
|n\,A_n(\delta)| \lesssim \sqrt{n}\left(\sumgammaM |\tau - \hat{h}_m(\hat{v}_m)| \right) = \op(1),
\]
from Lemma \ref{lemma::bin-technical-h}.
For $C_n(\delta)$, we first define
\[
R_n = \supM\max\left\{\frac{|\bxm ^T\delta|}{\sqrt{n}},\;|\hat{v}_m(\tau) - v_m(\tau)|\right\},
\] 
and it follows from Condition \ref{cond::vhat1} that $R_n= \Op(r_n\vee n^{-1/2})$. By taking the supremum within each the integration in $C_n(\delta)$, we have
\begin{eqnarray*}
&&|n\,C_n(\delta)|\\
&\leq&\sqrt{n}\,\left(
\sumgammaM |\bxm ^T\delta| \times\;2\sup_{|s| \leq R_n} \left|\{\hat{h}_m(v_m + s) - h(v_m + s)\} - \{\hat{h}_m(v_m) - h_m(v_m)\}\right|\right)\\
&\lesssim_{\P}&
\sqrt{n}\left(\sumgammaM \sup_{|s| \lesssim r_n\vee n^{-1/2}} \left|\{\hat{h}_m(v_m + s) - h_m(v_m + s)\} - \{\hat{h}_m(v_m) - h_m(v_m)\}\right|\right)\\
&=&\op(1),
\end{eqnarray*}
where the last inequality follows from Lemma \ref{lemma::bin-technical-h}.

Next we turn to the convergence of $B_n(\delta)$, where we first give a linear approximation for $h_m(\hat{v}_m + t) - h_m(\hat{v}_m)$ in (\ref{eq::proof-rock-ABC}). Note the derivative for the inverse function $h_m(z) = h(z,\bxm)$ in (\ref{eq::def-h}) is:
\begin{equation}
h'_m(z) = \left[\frac{\partial{v_m(s)}}{\partial s}{\Bigg|_{s=h_m(z)}}\right]^{-1} = \frac{1-h_m(z)}{v_m(h_m(z)) - q_m(h_m(z))}.
\label{eq::prof-rock-hprime}
\end{equation}
By using the first order Taylor-expansion and the mean value theorem, there exists a $\xi_m$ between $\hat{v}_m$ and $\hat{v}_m + t$ such that
\begin{eqnarray}
|h_m(\hat{v}_m + t) - h_m(\hat{v}_m ) - th_m'(v_m)| &=& |t[h_m'(\xi_m)-h_m'(v_m)]|\nonumber\\
&\leq& |t|\times L\left( |\xi_m - \hat{v}_m| + |\hat{v}_m - v_m|\right)\nonumber\\
&\leq & L\; (t^2 + |\hat{v}_m - v_m| \times |t|),\label{eq::Taylor-h-proof-main1}
\end{eqnarray}
since $|\xi_m - \hat{v}_m| \leq |t|$, where $L$ is the Lipschitz constant in Lemma \ref{lemma::bin-cond-implications}.
Therefore, $B_n(\delta)$ can be approximated as follows:
\begin{eqnarray*}
&&\left|B_n(\delta)  - \frac{1}{2}\sumgammaM h_m'(v_m) \left\{\Delta_m^2(\delta) -[\hat{v}_m- v_m]^2\right\}\right| \\
&=&
\left| \sumgammaM \int_{v_m - \hat{v}_m}^{-\Delta_m(\delta)}[h_m(\hat{v}_m + t) - h_m(\hat{v}_m) - th'_m(v_m)]\,\d t\right|\\
&\lesssim&
\left| \sumgammaM \int_{v_m - \hat{v}_m}^{-\Delta_m(\delta)}(t^2 +  |\hat{v}_m - v_m|\times |t|) \,\d t\right|\\
 & \lesssim& \frac{1}{\sqrt{n}}\sumgammaM |\hat{v}_m - v_m|^2 + \frac{1}{n}
\sumgammaM |\hat{v}_m - v_m| +\op\left(\frac{1}{n}\right) \\
&=& 
 \op\left(\frac{1}{n}\right),
\end{eqnarray*}
where the second inequality follows from (\ref{eq::Taylor-h-proof-main1}), and the last equality holds from Lemma \ref{lemma::bin-technical-h}. Therefore, $B_n(\delta)$ can be approximated by a function of $\Delta_m^2(\delta)$.

We now show that the loss function $L_n$ is approximately a quadratic function in $\delta$. Let
\[
D_{1n} = \frac{1}{1-\tau}\left[\sumgammaM h_m'(v_m)\bxm \bxm ^T\right], \quad \text{and}\quad u_n = \frac{\sqrt{n}}{1-\tau}\sumgammaM h_m'(v_m)\bxm (\hat{v}_m - v_m).
\]
Collecting the results for $A_n(\delta)$, $B_n(\delta)$ and $C_n(\delta)$ into (\ref{eq::proof-rock-ABC}), we have shown that for any fixed $\delta\in\mathbb{R}^{p+1}$,
\begin{eqnarray}
n\cdot L_n(\delta) &=& \frac{n}{2}\sumgammaM h_m'(v_m)\{\Delta_m^2(\delta) - [\hat{v}_m - v_m]^2\} + \op(1)\nonumber\\
&=& 
\frac{1}{2}\delta^TD_{1n}\delta - \delta^Tu_n  + \op(1),
\label{eq::proof-rock-quadratic-conv}
\end{eqnarray}
where the last equality follows by expanding $\Delta_m(\delta) = \hat{v}_m(\tau) - v_m(\tau) - n^{-1/2}\bxm ^T\delta$.
Since the i-Rock loss function $L_n(\delta)$ is convex, standard convexity argument (see e.g., \citet{hjort2011asymptotics} and \citet{pollard1991asymptotics}) shows that the convergence in (\ref{eq::proof-rock-quadratic-conv}) is uniform in $\delta$ over any compact subset of $\mathbb{R}^{p+1}$. 
Furthermore, the calculation of $h'_m$ in (\ref{eq::prof-rock-hprime}) and Lemma \ref{lemma::bin-approx-moment-X} shows
\[
D_{1n} \convergeip D_1 = \E\left[\frac{XX^T}{v(\tau,X)- q(\tau,X)}\right],
\] 
since $v(\tau,x)- q(\tau,x)$ is bounded by Lemma \ref{lemma::bin-cond-implications}.
Therefore, (\ref{eq::proof-rock-quadratic-conv}) implies that for any compact set $\mathcal{B}\subset\mathbb{R}^{p+1}$,
\begin{equation}
\label{eq::proof-rock-quadratic-conv2}
n\cdot \sup_{\delta\in\mathcal{B}} \left|\,L_n(\delta) - Q_n(\delta)\right| =  \op(1),
\end{equation}
where $Q_n(\delta) = \frac{1}{n}(\frac{1}{2}\delta^TD_1\delta - \delta^Tu_n) $. This shows that $L_n(\delta)$ can be uniformly approximated by a quadratic function in $\delta$.



Finally, we show the convergence of $\hat{\delta}$, which establishes the asymptotic properties of the i-Rock estimator.
As a function of $\delta$, $Q_n(\cdot)$ in (\ref{eq::proof-rock-quadratic-conv2}) has a unique minimizer
\[
\tilde{\delta} = D_1^{-1}u_n,
\]
since $D_1$ is positive definite.
Given Condition \ref{cond::bin-linear} and (\ref{eq::proof-rock-quadratic-conv2}), we apply the Basic Corollary in \citet{hjort2011asymptotics} to conclude that the minimizers of $L_n(\delta)$ and $Q_n(\delta)$ are asymptotically equivalent, i.e.,
\[
\hat{\delta} = \tilde{\delta} +\op(1) = D_{1}^{-1}\left[\sqrt{n}\sumM \frac{\hgamma_m\bxm }{v_m(\tau) - q_m(\tau)}[\hat{v}_m(\tau) - v_m(\tau)]\right] + \op(1).
\]
The proof is now complete by noting that $\hat{\delta} = n^{1/2}(\widehat{\beta}-\beta)$.

\end{proof}

\subsection{Proof of Theorem \ref{thm::example-RRM-SQ}}

\label{subsec::proof-eg}

To prove Theorem \ref{thm::example-RRM-SQ}, it entails to show that all conditions of Theorem \ref{thm::asymptotic-RRM-bin} apply to our specific construction of the initial estimator in (\ref{eq::ll-SQ}). We break the main technical requirements  into three Propositions below, the proof of which can be found in Appendix~\ref{append::Prop345}. In our proofs here, $\hat{v}(s,\bxm)$ refers specifically to the estimator constructed in (\ref{eq::ll-SQ}), and the weight $\hgamma_m$ refers to the one in (\ref{eq::ll-gweights}). Furthermore, we fix the sequence $r_n$ to be the one defined in Proposition \ref{prop::condI-1} below; We shall verify later in the proof of Theorem \ref{thm::example-RRM-SQ} that $r_n$ indeed satisfies the requirements in Theorem \ref{thm::example-RRM-SQ}.

\begin{proposition}
\label{prop::condL}
Under the conditions of Theorem \ref{thm::example-RRM-SQ}, we have
\[
\sqrt{n}\,\sumM \left[\frac{\hgamma_m \bxm }{v(\tau,\bar{x}_m)-q(\tau,\bar{x}_m)}\{\hat{v}(\tau,\bar{x}_m) - v(\tau,\bar{x}_m)\}\right] \converged \mathrm{N}\left(0,\;\Omega_1\right),
\]
where $\Omega_1$ is defined in Theorem \ref{thm::example-RRM-SQ}. 
\end{proposition}

\begin{proposition}
\label{prop::condI-1}
Let
\[
r_n = \sqrt{\frac{\log n}{n\ubh^p}}.
\]
Under the condition of Theorem \ref{thm::example-RRM-SQ}, we have
\[
  \supM \left|\hat{v}(\tau,\bxm) - v(\tau,\bxm)\right| = \Op\left(r_n\right).
\]
\end{proposition}

\begin{proposition}
\label{prop::condI-2}
Under the condition of Theorem \ref{thm::example-RRM-SQ}, we have for any fixed $B > 0$, 
\[    
    \sup_{\substack{m=1,\ldots,M\\|t|\leq B\cdot (r_n + n^{-1/2})}} \left|[\hat{v}(\tau + t,\bar{x}_m) - v(\tau + t,\bar{x}_m)] - [\hat{v}(\tau,\bar{x}_m) - v(\tau,\bar{x}_m)]\right| = \op\left(n^{-1/2}\right),
    \]
    where $r_n$ is given in Proposition \ref{prop::condI-1}.
\end{proposition}

The proof of Theorem \ref{thm::example-RRM-SQ} is relatively straightforward with these Propositions, and we now give the details. 

\begin{proof}[Proof of Theorem \ref{thm::example-RRM-SQ}]
Under the conditions of Theorem \ref{thm::example-RRM-SQ}, the binning mechanism satisfies:
\[
\supM \text{diam}(A_m) \lesssim \bh = o(1),\quad\text{and}\quad\supM \left| \frac{\hgamma_m}{\hat{\pi}_m} - 1\right| = \op(1),
\]
which follows from Lemma \ref{lemma::ll-coef}. It then suffices to check Conditions \ref{cond::bin-linear} and \ref{cond::vhat1}. 

Proposition \ref{prop::condL} directly implies Condition \ref{cond::bin-linear}.
From Condition \ref{cond::bin-bandwidth}, we have 
\[
n^{-1/2} \ll r_n = \sqrt{\ddfrac{\log n}{n\ubh^p}} \ll n^{-1/4},
\]
therefore the sequence $r_n$ constructed in Proposition \ref{prop::condI-1} can be used in Condition \ref{cond::vhat1}. Next we check Condition \ref{cond::vhat1}. 

The second requirement in Condition \ref{cond::vhat1} follows from Proposition \ref{prop::condI-2}. Moreover, from Proposition \ref{prop::condI-1} and \ref{prop::condI-2} we have
\begin{eqnarray*}
&&\supMS |\hat{v}(s,\bxm) - v(s,\bxm)| \\
&\leq& \supM |\hat{v}(\tau,\bxm) - v(\tau,\bxm)| \\
&&\; +\;
\sup_{\substack{m=1,\ldots,M\\|t|\leq B\cdot r_n}} \left|[\hat{v}(\tau + t,\bar{x}_m) - v(\tau + t,\bar{x}_m)] - [\hat{v}(\tau,\bar{x}_m) - v(\tau,\bar{x}_m)]\right|\\
& = & \Op(r_n) + \op(n^{-1/2}).
\end{eqnarray*}
Hence the first requirement in Condition \ref{cond::vhat1} also holds. Since the monotonicity of $\hat{v}(s,x)$ (with respect to $s$) is assumed, we have checked all requirements of Theorem \ref{thm::asymptotic-RRM-bin}. The proof is now complete.

\end{proof}

\subsection{Proof of Propositions \ref{prop::condL}, \ref{prop::condI-1} and \ref{prop::condI-2}}\label{append::Prop345}

Here we prove the three Propositions used in the proof Theorem \ref{thm::example-RRM-SQ}. We fix some notations used in the proof. Recall $S_{0m}$, $S_{1m}$ and $\bm{S}_{2m}$ from (\ref{eq::ll-Sm}); and note $S_{0m} = n^{-1}\sumn w_{im} = \hat{\pi}_m$. 
For the weight of each bin in (\ref{eq::estimator-rock-binning}), we set
\[
\hgamma_m =S_{0m} - S_{1m}^T\bm{S}_{2m}^{-1}S_{1m},
\]
as in (\ref{eq::ll-gweights}).
Using the block matrix inverse, our estimator $\hat{v}(s,\bxm)$ in (\ref{eq::ll-SQ}) can be further simplified as:
\begin{equation}
\hat{v}(s,\bxm) = (n\hgamma_m)^{-1}\left[ \sumn w_{im}\hat{Z}_i(s) - S_{1m}^T\bm{S}_{2m}^{-1}\sumn (\tX_i - \txm)w_{im}\hat{Z}_i(s) \right],
\label{eq::ll-SQ-2}
\end{equation}
where $\hat{Z}_i(s)$ is defined in (\ref{eq::def-Zi}).
Furthermore, let $\tilde{v}(s,\bxm)$ be the oracle estimator where we know $q(s,x)$ and $Z_i(s)$, i.e.,
\begin{equation}
\tilde{v}(s,\bxm) = (n\hgamma_m)^{-1}\left[ \sumn w_{im}{Z}_i(s) - S_{1m}^T\bm{S}_{2m}^{-1}\sumn (\tX_i - \txm)w_{im}{Z}_i(s) \right].
\label{eq::ll-SQ-oracle}
\end{equation}

\subsubsection{Proof of Proposition \ref{prop::condL}}

\begin{proof}
We rely on the decomposition that 
\begin{eqnarray}
[\hat{v}(\tau,\bxm) - v(\tau,\bxm)] 
&=& [\hat{v}(\tau,\bxm) - \tilde{v}(\tau,\bxm)] + [\tilde{v}(\tau,\bxm) - v(\tau,\bxm)]\nonumber\\
&=&[\hat{v}(\tau,\bxm) - \tilde{v}(\tau,\bxm)] \nonumber \\
&&\;\;-\; (n\hgamma_m) ^{-1} \left[S_{1m}^T\bm{S}_{2m}^{-1}\sumn (\tX_i - \txm  )w_{im}[{Z}_i(\tau) - v(\tau,X_i)] \right]\nonumber\\
&&\;\; +\; (n\hgamma_m)^{-1}\left[ \sumn w_{im}[{Z}_i(\tau) - v(\tau,X_i)]  \right] \label{eq::condL-tot-CLT},
\end{eqnarray}
where the last equality follows from standard local-linear calculation \citep{fan2018local} since $v(\tau,x)$ is linear in $x$.

It suffices to consider the aggregation of the three terms in the decomposition above.
First, we give two claims below; and we verify them one by one at the end of this proof. In what follows, we define $\zeta_m = v(\tau,\bxm) - q(\tau,\bxm)$.

\begin{claim}\label{claim1}
\begin{equation}
\label{eq::claim1-condL}
\sqrt{n}\sumM \left\{\frac{\hgamma_m}{\zeta_m }\bxm  \left[ \sumn \frac{w_{im}}{n\hgamma_m}S_{1m}^T\bm{S}_{2m}^{-1}(\tX_i - \txm  ) [{Z}_i(\tau) - v(\tau,X_i)]  \right]\right\} = \op(1).
\end{equation}
\end{claim}

\begin{claim}\label{claim2}
\begin{equation}
\label{eq::claim2-condL}
\sqrt{n} \sumM \left\{ \frac{\hgamma_m}{\zeta_m}\bxm \left[\hat{v}(\tau,\bxm) -\tilde{v}(\tau,\bxm) \right]\right\}
= \op(1).
\end{equation}
\end{claim}

Claims \ref{claim1} and \ref{claim2} together show the first two terms in Equation (\ref{eq::condL-tot-CLT}) are asymptotically negligible when aggregated over the bins. In particular, they show that using our initial estimator is asymptotically equivalent to using the oracle estimator from~\eqref{eq::ll-SQ-oracle} in Proposition \ref{prop::condL}.
In what follows, the proof is given in three steps.  In the first step, we give our main argument, which establishes a Central Limit Theorem type result; This step shows the desired asymptotic normality in Proposition \ref{prop::condL}. In the next two steps, we verify Claims \ref{claim1} and \ref{claim2} separately.

\paragraph*{Step 1: A CLT-type result}
We give the asymptotic analysis for the aggregation of the last term in (\ref{eq::condL-tot-CLT}) over the bins, given by
\begin{eqnarray*}
&&\frac{1}{\sqrt{n}}\sum_{m=1}^M \left\{\frac{1}{\zeta_m }\bxm  \left[ \sumn w_{im}[{Z}_i(\tau) - v(\tau,X_i)] \right]\right\}\\
&=&
\frac{1}{\sqrt{n}}\sumn \underbrace{[{Z}_i(\tau) - v(\tau,X_i)] \left[\sumM \frac{w_{im}}{\zeta_m }\bxm    \right]}_{O_{in}},
\end{eqnarray*}
which holds by exchanging the order of summation. Our arguments are always made by conditional on the design $X_i$ first.

Since $\zeta_m $ and $\bxm  $ are deterministic, the term $\sumM (\zeta_m)^{-1}w_{im} \bxm $ in each $O_{in}$ only depends on the bin which $X_i$ falls into. For fixed $n$ and $M$, the random vectors $O_{in}$ are independent with mean $0$ across $i=1,\ldots,n$; and we apply the multivariate Lindeberg-Feller Central Limit Theorem for triangular arrays (E.g., Theorem 2.27 of \citet{van2000asymptotic}) in our proof below.

We check the first Lindeberg conditions here. In our setting it suffices to show:
\begin{equation}
\E\lVert O_{in}\rVert^2 \bm{1}\{\lVert O_{in}\rVert \geq \varepsilon \sqrt{n}\} \to 0,
\label{eq::Lindeberg-CLT}
\end{equation}
for all fixed $\varepsilon > 0$ as $n\to \infty$. Since $\bxm  $ is uniformly bounded, we have that
\[
\E\lVert O_{in}\rVert^{2+\delta_0} \lesssim \E|Z_i(\tau) - v(\tau,X_i)|^{2+\delta_0} \lesssim \E[|q(\tau,X)|^{2+\delta_0}] +  \E\left[ |Y^{+}|^{2+\delta_0}\right] <\infty,
\]
which follows from Condition \ref{cond::bin-Y-new}.
Furthermore note $|x|^2\bm{1}\{|x| \geq a\} \leq a^{-\delta}|x|^{2+\delta}$, the Lindeberg condition (\ref{eq::Lindeberg-CLT}) then follows from the Markov inequality.

Next we calculate the variance of each $O_{in}$. Parallel to the one-sample case in Corollary \ref{coro::SQ-fixed-tau} of Chapter 2, we have that
\begin{eqnarray}
\var[Z_i(\tau) - v(\tau,X_i) \mid X_i=x] 
&=& \frac{\var(Y\mid X = x, Y \geq q(\tau,x)) + \tau[v(\tau,x) - q(\tau,x)]^2}{1-\tau}\nonumber \\
&\triangleq& \sigma_\tau^2(x) \label{eq::proof-def-sigma}.
\end{eqnarray}
Therefore from $O_{in}$ in the beginning of Step 1, we have
\begin{eqnarray*}
\var(O_{in}) &=& \E_X\left\{\left[\sum_{m=1}^M\frac{w_{im}}{\zeta_m ^2}\bxm  \bxm  ^T\right]\E_{Y\mid X} \left[Z_i(\tau) - v(\tau,X_i)\right]^2 \right\}\\
&=& 
\E_X\left\{\sigma_\tau^2(X)\left[\sum_{m=1}^M\frac{w_{im}}{\zeta_m ^2}\bxm  \bxm  ^T\right]\right\}\\
&\to& 
\E_X\left\{\frac{\sigma_\tau^2(X)}{[v(\tau,X) - q(\tau,X)]^2}XX^T\right\},
\end{eqnarray*}
as $n\to\infty$, which follows from Lemma \ref{lemma::bin-approx-moment-X}. Hence, it follows that
\[
\frac{1}{\sqrt{n}}\sumn O_{in} \converged \mathrm{N}\left(0,\;\Omega_1\right),
\]
by the Lindeberg CLT, where $\Omega_1$ is given in Theorem \ref{thm::example-RRM-SQ}.

Together with Claims \ref{claim1} and \ref{claim2}, we have proved that 
\[
\frac{1}{\sqrt{n}}\sum_{m=1}^M \left\{\frac{\hgamma_m}{\zeta_m }\bxm  \left[ \hat{v}(\tau,\bxm) - v(\tau,\bxm)\right]\right\}\converged \mathrm{N}\left(0,\;\Omega_1\right),
\]
from the decomposition in (\ref{eq::condL-tot-CLT}). Therefore Proposition \ref{prop::condL} holds.

\paragraph*{Step 2: Verification of Claim \ref{claim1}}
 The left hand side of (\ref{eq::claim1-condL}) can be written as:
\begin{eqnarray}
\frac{1}{\sqrt{n}}\sumn \left\{ [{Z}_i(\tau) - v(\tau,X_i)] \underbrace{\left[ \sumM \frac{w_{im}}{\zeta_m }\bxm  S_{1m}^T\bm{S}_{2m}^{-1}(\tX_i - \txm  ) \right]}_{V_{in}}\right\}
\label{eq::goal-CLT-remainder1},
\end{eqnarray}
by re-arranging the summation.


We use the Markov inequality to bound (\ref{eq::goal-CLT-remainder1}). To this end, we calculate the variance for each term of (\ref{eq::goal-CLT-remainder1}). Note that $V_{in}$ depends on the covariates but not the response, by conditioning on $X$ first we have:
\[
\E\left\lVert\frac{1}{\sqrt{n}}\sumn [Z_i(\tau) - v(\tau,X_i)]V_{in}\right\rVert^2 
=\frac{1}{n}\E\left(\sumn \sigma_\tau^2(X_i)\lVert V_{in}\rVert^2\right) \lesssim \frac{1}{n}\sumn \E\lVert V_{in}\rVert^2\\,
\]
since $\sigma_\tau^2(x)$ in (\ref{eq::proof-def-sigma}) is bounded. Following the above displayed equation, we can further expand the variance of $V_{in}$ as:
\begin{eqnarray}
\frac{1}{n} \sumn\E\lVert V_{in}\rVert^2
&=& 
\frac{1}{n}\E\left[\sumn\sumM \frac{w_{im}\lVert\bxm  \rVert^2 }{\zeta_m ^2}\lVert S_{1m}^T\bm{S}_{2m}^{-1}(\tX_i - \txm  )\lVert^2\right]\nonumber\\
&\lesssim&
\E\left\{
\sumM S_{1m}^T\bm{S}_{2m}^{-1}\underbrace{\left[n^{-1}\sumn w_{im}(\tX_i - \txm  )(\tX_i - \txm  )^T\right]}_{\bm{S}_{2m}}\bm{S}_{2m}^{-1}S_{1m}
\right\}\nonumber\\
&=&\E\left[\sumM S_{1m}^T\bm{S}_{2m}^{-1}S_{1m}\right]\label{eq::DCT-remainder2}\\
&=&o(1)\nonumber,
\end{eqnarray}
where the definition of $\bm{S}_{2m}$ is in (\ref{eq::ll-Sm}), and the convergence to the $o(1)$ term in the end follows from the Dominated Convergence theorem as outlined below. First, the term inside the expectation of (\ref{eq::DCT-remainder2}) is bounded by
\[
\sumM S_{1m}^T\bm{S}_{2m}^{-1}S_{1m} \leq \sumM S_{0m}  = n^{-1}\sumM \sumn \bm{1}\{X_i\in A_m\}= 1,
\]
since $\hgamma_m  = S_{0m} - S_{1m}^T\bm{S}_{2m}^{-1}S_{1m} \geq 0$. Second, we have from Lemma \ref{lemma::ll-coef} that
\[
\left|\sumM S_{1m}^T\bm{S}_{2m}^{-1}S_{1m}\right| \leq \sumM S_{0m}\left|1-\frac{\hgamma_m }{S_{0m}}\right| = \op(1).
\]
The convergence in expectation of (\ref{eq::DCT-remainder2}) then follows.

Therefore, Claim \ref{claim1} holds by applying the Markov inequality to (\ref{eq::goal-CLT-remainder1}).

\paragraph*{Step 3: Verification of Claim \ref{claim2}}

By the construction of $\hat{v}$ and $\tilde{v}$ in (\ref{eq::ll-SQ-2}) and (\ref{eq::ll-SQ-oracle}), we have
\[
\hat{v}(\tau,\bxm) -\tilde{v}(\tau,\bxm) 
= 
(n\hgamma_m)^{-1}  
\left[ \sumn w_{im}[1-S_{1m}^T\bm{S}_{2m}^{-1}(\tX_i - \txm  )][\hat{Z}_i(\tau) - Z_i(\tau)] \right].
\]
Similar to the proof of Lemma 1 in \citet{olma2021nonparametric}, we consider the following decomposition of $\hat{Z}_i(\tau) - Z_i(\tau)$:
\begin{eqnarray}
(1-\tau)[\hat{Z}_i(\tau) - Z_i(\tau)] 
&=& 
[Y_i - q(\tau,X_i)]\left\{\bm{1}[Y_i \geq \hat{q}(\tau,X_i)] - \bm{1}[Y_i \geq q(\tau,X_i)]\right\}\nonumber\\
&&\;\; +\, (q(\tau,X_i) - \hat{q}(\tau,X_i))\cdot (\tau - \bm{1}[Y_i < q(\tau,X_i)])\nonumber\\
&&\;\; + \,(q(\tau,X_i) - \hat{q}(\tau,X_i))\cdot \left\{\bm{1}[Y_i \geq \hat{q}(\tau,X_i)] - \bm{1}[Y_i \geq q(\tau,X_i)]\right\}\nonumber\\
&\triangleq& u_{1i}(\tau) + u_{2i}(\tau) + u_{3i}(\tau),
\label{eq::NO-smallu}.
\end{eqnarray}
where we sometimes omit the index $\tau$ in this proof.
Using the above two displayed equations, and by re-arranging the order of summation in (\ref{eq::claim2-condL}) of Claim 2, we have
\begin{eqnarray}
&&\sqrt{n} \sumM \left\{ \frac{\hgamma_m}{\zeta_m}\bxm \left[\hat{v}(\tau,\bxm) -\tilde{v}(\tau,\bxm) \right]\right\}\nonumber\\
&=&
\frac{1}{\sqrt{n}}\sumn \left\{ [\hat{Z}_i(\tau) - Z_i(\tau)]\underbrace{\left[ \sumM \frac{w_{im}}{\zeta_m }[1 - S_{1m}^T\bm{S}_{2m}^{-1}(\tX_i - \txm  )] \bxm \right]}_{\kappa_i}\right\}\nonumber\\
&=&\frac{1}{\sqrt{n}(1-\tau)} \left(\sumn  u_{1i}\kappa_i   + \sumn  u_{2i}\kappa_i   + \sumn  u_{3i}\kappa_i  \right)
\nonumber\\
&\triangleq& U_{1n} + U_{2n} + U_{3n}, \label{eq::goal-Neyman-remainder}
\end{eqnarray}
where $U_{jn} = [\sqrt{n}(1-\tau)]^{-1}\sumn u_{ji}\kappa_i$ and $u_{ji}$ is defined in (\ref{eq::NO-smallu}). To check Claim 2, it suffices to consider the three terms in (\ref{eq::goal-Neyman-remainder}) separately. 

We consider $U_{2n}$ first. 
We consider $U_{2n}$ first. Separating $\kappa_i$ into two sums for the terms $1$ and $S_{1m}^T\bm{S}_{2m}^{-1}(X_i - \bxm)$, we have
\begin{eqnarray*}
&&\lVert(1-\tau)\sqrt{n} U_{2n} \rVert\\
&\lesssim&
\lVert \sumn u_{2i}\kappa_i\rVert 
\\
&\leq&
\left\lVert\sumM\sumn \left( \frac{ w_{im}}{\zeta_m}u_{2i}\bxm \right)\right\rVert + 
\left\lVert \sumM\sumn\left[ \frac{ w_{im}}{\zeta_m}  S_{1m}^T\bm{S}_{2m}^{-1}(X_i - \bxm) \cdot u_{2i}\bxm \right]\right\rVert
\\
&\lesssim&
 \supM \left|\frac{\sumn w_{im} u_{2i}}{n\hat{\pi}_m}\right|\cdot \sumM\frac{n\hat{\pi}_m\lVert \bxm\rVert}{|\zeta_m |} \\
 &&\;\; + \;
 \supM \left\lVert \ddfrac{\sumn w_{im} \left[\frac{X_i - \bxm  }{\bh_m }\right]u_{2i}}{n\hat{\pi}_m}\right\rVert 
 \cdot \sumM\frac{ n\hat{\pi}_m\cdot\lVert \bxm\rVert\cdot \lVert \bh_m \cdot S_{1m}^T\bm{S}_{2m}^{-1}\rVert }{|\zeta_m |} \\
 &\lesssim& 
 \op(n^{-1/2})\, \sumM(n\hat{\pi}_m) + \op(n^{-1/2})\,\sumM \left[ n\hat{\pi}_m\op(1)
 \right],
\end{eqnarray*}
where we have used the fact that $\lVert \bxm\rVert/|\zeta_m|$ is bounded; in the last inequality, the $\op(n^{-1/2})$ terms follow from Condition \ref{cond::bin-Qt}, and the $\op(1)$ term uses the bound of $\lVert \bh_m \cdot S_{1m}^T\bm{S}_{2m}^{-1}\rVert$ in Lemma \ref{lemma::ll-coef}. Noting that $\sumM\hat{\pi}_m = 1$, we conclude that $U_{2n} = \op(1)$.

Next we consider $U_{1n}$ and $U_{3n}$. Noting that $w_{im}| 1 - S_{1m}^T\bm{S}_{2m}^{-1} (\tX_i - \txm)| = 1+\op(1)$ uniformly as in Lemma \ref{lemma::ll-coef}, we have from (\ref{eq::goal-Neyman-remainder}) that
\begin{eqnarray*}
\lVert (1-\tau)\sqrt{n}U_{1n} \rVert &=& \left\lVert\sumn u_{1i}\kappa_i\right\rVert\\
&\leq& \sumn\sumM \frac{w_{im}|u_{1i}|\cdot \lVert \bxm\rVert}{|\zeta_m| } [1+\op(1)]\\
&\lesssim_{\P}&\supM\left[\frac{\sumn w_{im}|u_{1i}|}{n\hat{\pi}_m}\right]\cdot \sumM n\hat{\pi}_m\\
&=& n\cdot \supM\left[\frac{-\sumn w_{im}u_{1i}}{n\hat{\pi}_m}\right],
\end{eqnarray*}
since $\sumM \hat{\pi}_n = 1$, and $u_{1i} \leq 0, \forall i = 1,\ldots,n$. Similarly, since $u_{1i} \geq 0, \forall i = 1,\ldots,n$, 
\[
\lVert (1-\tau)\sqrt{n}U_{3n} \rVert
\;\;\lesssim_{\P}\;
n\cdot \supM\left[\frac{\sumn w_{im}u_{3i}}{n\hat{\pi}_m}\right]. 
\]
Therefore, it follows directly from Lemma \ref{lemma::proof-U13} that 
$U_{1n} = \Op(n^{1/2}g_{1n}^2)$ and $U_{3n} = \Op(n^{1/2}g_{1n}^2)$. 

Substituting $U_{1n}$, $U_{2n}$ and $U_{3n}$ into the decomposition in \eqref{eq::goal-Neyman-remainder}, we have 
\[
\sqrt{n} \sumM \left\{ \frac{\hgamma_m}{\zeta_m}\bxm \left[\hat{v}(\tau,\bxm) -\tilde{v}(\tau,\bxm) \right]\right\} = \Op(n^{1/2}g_{1n}^2) +  \op(n^{-1/2}) = \op(n^{-1/2}),
\]
from Condition \ref{cond::bin-Qt}. Hence Claim 2 holds. The proof of Proposition \ref{prop::condL} is now complete.

\end{proof}

\subsubsection{Proof of Proposition \ref{prop::condI-1}}
We define some additional notations. Let
\begin{equation}
\label{eq::def-kappaik}
\kappa_{im} = [1 - S_{1m}^T\bm{S}_{2m}^{-1}(\tX_i - \txm)],
\end{equation}
and let 
\begin{equation}
 A_{0m}= \sup_{i=1,\ldots,n}|w_{im}\kappa_{im}|,\quad A_{1m}  = \sumn w_{im}|\kappa_{im}|,\quad A_{2m} = \sumn w_{im}\kappa_{im}^2.
\label{eq::def-A1A2A3-kappa}
\end{equation}

\begin{proof}

Following the same calculation in (\ref{eq::condL-tot-CLT}),  for each bin $A_m$ we have:
\begin{eqnarray}
\left[\hat{v}(\tau,\bxm) - v(\tau,\bxm) \right]
&=& \left[\hat{v}(\tau,\bxm) - \tilde{v}(\tau,\bxm) \right] \nonumber\\
&&\;\;+\;\hgamma_m^{-1}S_{0m}\left[ \dfrac{\sumn w_{im}[{Z}_i(\tau) - v(\tau,X_i)]  [1 - S_{1m}^T\bm{S}_{2m}^{-1}(\tX_i - \tilde{x}_m)]}{\sumn w_{im}} \right]\nonumber\\
&\triangleq&\frac{S_{0m}}{\hgamma_m}\left[B_q(\tau,m)+ C(\tau,m)\right]
\label{eq::SQ-decomposition-s},
\end{eqnarray}
where $B_q(\tau,m) = S_{0m}^{-1}\hgamma_m[\hat{v}(\tau,\bxm) - \tilde{v}(\tau,\bxm) ]$ corresponds to the bias term that originates from using the estimated quantile in $\hat{Z}_i(\tau)$, and $C(\tau,m)$ has mean zero. 
Under the conditions of Theorem \ref{thm::example-RRM-SQ}, we give the following claims, which we verify at the end of the proof.

\begin{claim}\label{claim3}
\[
\supM \left| B_q(\tau,m)\right| =  \Op(g_{1n}^2) + \op\left(n^{-1/2}\right),
\]
where $g_{1n}$ is in Condition \ref{cond::bin-Qt}.
\end{claim}

\begin{claim}\label{claim4}
\[
\Pr\left(\supM A_{0m}\geq 2 \right)\lesssim\frac{1}{n^3},\quad  A_{1m}\leq nS_{0m},\;\;\text{and}\quad A_{2m} \leq nS_{0m},
\]
where the quantities $A_{jm}$ $(j=1,2,3)$ are defined in (\ref{eq::def-A1A2A3-kappa}).
\end{claim}

Claim \ref{claim3} shows that the bias in the initial ES estimator is asymptotically negligible, Claim \ref{claim4} is also useful but more technical.
Following Claims \ref{claim3} and \ref{claim4}, the proof proceeds in 5 steps. In steps 1 through 3, we establish the main argument that:
\[
\supM |C(\tau,m)| = \Op(r_n);
\]
In step 1, we give a truncation argument similar to \citet{mack1982weak}; In step 2, we derive exponential inequalities for the truncated process; Step 3 gives some auxiliary calculations that completes the proof. In the final two steps we verify Claims \ref{claim3} and \ref{claim4}.

\paragraph*{Step 1: Truncation}
For any sequence $ b_n > 0$ that satisfies: (i) $\sum_{n=1}^{\infty} b_n^{-2-\delta_0} < +\infty$ for $\delta_0$ in Condition \ref{cond::bin-Y-new}; (ii) $b_n$ is monotonically increasing; and (iii) $b_n\to+\infty,$ we define the truncated variable:
\[
Z_i^{(B)}(\tau) = \frac{[Y_i - q(\tau,X_i)]\bm{1}\{0 \leq Y_i - q(\tau,X_i) \leq b_n\}]}{1-\tau} + q(\tau,X_i),
\]
and the corresponding truncated process:
\[
C^{(B)}(\tau,m) = (nS_{0m})^{-1}\sumn \left\{w_{im}\left[Z_i^{(B)}(\tau) - v(\tau,X_i)\right][1 - S_{1m}^T\bm{S}_{2m}^{-1}(\tX_i - \txm)]\right\}.
\]
We shall give the precise choice of $b_n$ in step 2 below.
For sufficiently large $n$, in the following we show that $C^{(B)}(\tau,m)$ is equivalent to $C(\tau,m)$ with probability one. 

Comparing $C^{(B)}(\tau,m)$ with $C(\tau,m)$ in (\ref{eq::SQ-decomposition-s}), 
we see that $C^{(B)}(\tau,m)\neq C(\tau,m)$ only when $Y_i - q(\tau,X_i) \geq b_n$ for some $i=1,\ldots,n$; and we can calculate this probability:
\[
\Pr\left( [Y_i - q(\tau,X_i)] \geq b_n\right)\leq \frac{\E[Y_i - q(\tau ,X_i)]^{2+\delta_0}\bm{1}[Y_i \geq q(\tau,X_i)]}{b_n^{2+\delta_0}} \lesssim b_n^{-2-\delta_0},
\]
from Chebyshev's inequality and Condition \ref{cond::bin-Y-new}. 
Following Proposition 1 of \citet{mack1982weak}, under our choice of $b_n$ we have
\[
\Pr\left(\liminf_{n\to\infty}\left\{\supM |C^{(B)}(\tau,m) - C(\tau,m)| = 0\right\}\right) = 1,
\] 
where $\liminf_{n\to\infty}$ denotes the limit infimum for a sequence of events.


\paragraph*{Step 2: Exponential inequality}

Here we derive exponential tail bounds for the centered truncated sequence $C^{(B)}(\tau,m) - \E[C^{(B)}(\tau,m)]$. 
Note $C^{(B)}(\tau,m)$ does not have mean zero after truncation. With $\kappa_{im}$ in (\ref{eq::def-kappaik}), we write
\[
C^{(B)}(\tau,m) = (nS_{0m})^{-1}\sumn \left\{w_{im}\kappa_{im}\left[Z_i^{(B)}(\tau) - v(\tau,X_i)\right]\right\}.
\]
For a small enough $\varepsilon_3 > 0$, let the truncation threshold satisfy
\begin{equation} 
b_n \asymp n^{\frac{1}{2+\delta_0} + \varepsilon_3},
\label{eq::def-bn-truncation}
\end{equation}
where $\delta_0$ is in Condition \ref{cond::bin-Y-new}; it is easy to check that this choice of $b_n$ satisfies the requirements in Step 1 of the proof.

We apply Bernstein inequality on the truncated and centered process; to this end, we compute some key quantities below. Conditional on the covariates $X$, we have
\[
\sup_{x\in\mathcal{X}}\,\var\left[Z_i^{(B)}(\tau) - v(\tau,X_i)\;\Bigg|\; X_i = x\right]
\leq C_1 < +\infty,
\]
for some constant $C_1$, which follows from Condition \ref{cond::bin-Y-new}; hence
\[
\sumn \var\left[ w_{im}\kappa_{im}[Z_i^{(B)}(\tau) - v(\tau,X_i)]\;\Bigg|\; X\right] \leq C_1 A_{2m} \leq nC_1S_{0m}.
\]
Furthermore, each of summands in $C^{(B)}(\tau,m)$ can be bounded by
\[
\left|w_{im}\kappa_{im}[Z_i^{(B)}(\tau) - v(\tau,X_i)]\right| \lesssim A_{0m}b_n.
\]
Refer to Claim \ref{claim4} for the properties of $A_{0m}$ and $A_{2m}$.

Now, a direct application of the (conditional on $X$) Bernstein inequality (e.g., Theorem 2.8.4 of \citet{vershynin2018high}) and a union bound gives
\begin{eqnarray}
&&\Pr\left( \supM\left|C^{(B)}(\tau,m) - \E[C^{(B)}(\tau,m)]\right|  \geq M_1 r_n \;\Bigg|\; X\right)\nonumber \\
&\leq&
\sumM \Pr\left( \left|\sumn w_{im}\kappa_{im}\left[Z_i^{(B)}(\tau) - \E[Z_i^{(B)}(\tau)]\right]\right|  \geq M_1\cdot nS_{0m}r_n \;\Bigg|\; X\right) \nonumber\\
&\leq&
 2\exp\left\{\log n -\frac{(M_1^2 nr_n^2/2)\cdot\inf_{m}S_{0m}}{C_1 + (M_1 b_n r_n/3)\cdot\sup_{m}A_{0m}}\right\}\nonumber\\
 &\lesssim&
 2\exp\left\{\log n -\frac{(M_1^2 nr_n^2/2)\cdot\inf_{m}S_{0m}}{C_1(1+ \sup_{m}A_{0m})}\right\}\label{eq::Bernstein-2},
\end{eqnarray}
for sufficiently large $n$, where the $\log n$ factor comes from $M \lesssim \bh^{-p} \leq n$ under Condition \ref{cond::bin-bandwidth}; the last inequality follows since 
\[
b_nr_n = \sqrt{\ddfrac{\log n}{n^{1-2/(2-\delta_0) - 2\varepsilon_3}\,\bh^p}} \to 0,
\]
under Condition \ref{cond::bin-bandwidth}, with $b_n$ in (\ref{eq::def-bn-truncation}) and $r_n$ in Proposition \ref{prop::condI-1}

Here we give the unconditional tail bound from the conditional one in (\ref{eq::Bernstein-2}). Let $\Gamma$ denote the event that $\sup_m A_{0m} \leq 2$ and $\inf_{m} |\bh_m^{-p}S_{0m}|\geq \varepsilon_2$ for some $\varepsilon_2 > 0$; With Lemma \ref{lemma::ll-coef} and Claim \ref{claim4}, we have $\Pr(\Gamma^c) \lesssim n^{-3}$. 
With the law of total expectation applied to (\ref{eq::Bernstein-2}), the unconditional tail bound is:
\begin{eqnarray}
&&\Pr\left( \supM|C^B(\tau,m) - \E[C^B(\tau,m)]|  \geq M_1 r_n\right) \nonumber\\ 
&\leq& \E_X\left[2\exp\left\{\log n -\frac{(M_1^2 nr_n^2/2)\cdot\inf_{m}S_{0m}}{C_1(1+ \sup_{m}A_{0m})}\right\}\cdot \bm{1}\{\Gamma\}\right]+ \E\left[\exp\{\log n\} \cdot \bm{1}\{\Gamma^c)\right] \nonumber\\
&\lesssim& \E_X\left[2\exp\left\{\log n -\frac{\log n\cdot M_1^2\varepsilon_2/2}{ 3C_1}\right\}\right]+ n\Pr\left(\Gamma^c\right) \nonumber\\
&\lesssim&
 \frac{1}{n},
\label{eq::tail-centered-truncated}
\end{eqnarray}
for sufficiently large $M_1$ since $C_1$ and $\varepsilon_2$ are fixed.

\paragraph*{Step 3: Final calculations}

Noting that 
\begin{equation*}
\left| C(\tau,m) - C^{(B)}(\tau,m) \right| =  
(nS_{0m})^{-1}\sumn w_{im}|\kappa_{im}|\left\{ \frac{[Y_i - q(\tau,X_i)]\bm{1}\{b_n\leq Y_i - q(\tau,X_i)\}}{1-\tau}\right\},
\end{equation*}
the expectation of $C^{(B)}$ can be bounded by
\begin{eqnarray*}
&&\E\left[|C(\tau,m) - C^{(B)}(\tau,m)| \mid X = x\right] \\
&=& 
(nS_{0m})^{-1}\sumn \frac{w_{im}|\kappa_{im}|}{1-\tau} \E\left\{[Y_i - q(\tau,X_i)]\bm{1}[ Y_i - q(\tau,X_i) \geq b_n]\Bigm| X=x \right\}\\
&\leq& 
(nS_{0m})^{-1}\sumn \frac{w_{im}|\kappa_{im}|}{1-\tau} \E\left\{\frac{[Y_i - q(\tau,X_i)]^{2+\delta_0}\bm{1}[ Y_i\geq q(\tau,X_i)]}{b_n^{1+\delta_0}}\Bigm| X=x \right\}\\
&\lesssim& (nS_{0m})^{-1}A_{1m}\cdot b_n^{-1-\delta_0}\\
&\leq& b_n^{-1-\delta_0},
\end{eqnarray*}
where the second inequality follows from Chebyshev's inequality, the third inequality follows from the moment bound in Condition \ref{cond::bin-Y-new} and the last inequality from Claim \ref{claim4}.
Taking expectation again with expect to $X$ we obtain:
\[
\supM \E\left[|C^{(B)}(\tau,m) - C(\tau,m)|\right] \lesssim b_n^{-1-\delta_0} \lesssim r_n,
\]
from the choice of $b_n$ in (\ref{eq::def-bn-truncation}) and $r_n$ in Proposition \ref{prop::condI-1}.
 
Combining the above expectation bounds with the results of steps 1 and 2, we have
\[
\supM |C(\tau,m)| = \Op\left( r_n \right).
\]
Together with Claim \ref{claim3}, we've proved that
\[
\supM |B_q(\tau,m) + C(\tau,m)| = \Op(r_n),
\]
since $g_{1n}^2 \ll n^{-1/2}\ll r_n$ in Proposition \ref{prop::condI-1}.
Furthermore, note from Lemma \ref{lemma::ll-coef} we have $\hgamma_m^{-1}S_{0m} = 1 +\op(1) $ uniformly over $m=1,\ldots,M$. The conclusion of Proposition \ref{prop::condI-1} then holds from the decomposition (\ref{eq::SQ-decomposition-s}).

\paragraph*{Step 4: Verification of Claim \ref{claim3}}
We follow the same decomposition of $\hat{Z}_i(\tau) - Z_i(\tau)$ as in (\ref{eq::NO-smallu}) in the proof of Proposition \ref{prop::condL}. From the definition of $B_q$ in (\ref{eq::SQ-decomposition-s}) we have
\begin{eqnarray}
B_q(\tau,m) 
&=&
\frac{1}{(1-\tau)}\left\{(nS_{0m})^{-1}\sumn  [u_{1i}(\tau) + u_{2i}(\tau) +u_{3i}(\tau)] w_{im}\kappa_{im}\right\}\nonumber\\
&\triangleq&
\frac{1}{(1-\tau)}\left[U_{1n}(\tau,m) + U_{2n}(\tau,m) + U_{3n}(\tau,m)\right],\label{eq::decomposition-U123-tau}
\end{eqnarray}
where $u_{ji}(\tau)$ is defined in (\ref{eq::NO-smallu}). We consider the three terms separately.

We consider $U_{2n}(\tau,m)$ first. By separating $\kappa_{im}$ in (\ref{eq::def-kappaik}) into two parts, we have:
\begin{eqnarray*}
\supM |U_{2n}(\tau,m)|
&\leq&
\supM \left| \frac{\sumn w_{im}[S_{1m}^T\bm{S}_{2m}(\tX_i - \txm)]u_{2i}(\tau)}{nS_{0m}}\right|
\\
&&\;\;+\; 
\supM \left| \frac{\sumn w_{im}u_{2i}(\tau)}{nS_{0m}}\right| \\
&\leq&
\supM \left\lVert \ddfrac{\sumn w_{im} \left[\frac{\tX_i - \txm  }{\bh_m }\right]u_{2i}(\tau)}{nS_{0m}}\right\rVert \cdot \supM \lVert \bh_mS_{1m}^T\bm{S}_{2m}^{-1} \rVert \\
&&\;\;+\;\op(n^{-1/2}) 
\\
&=&
\op(n^{-1/2}),
\end{eqnarray*}
which follows from Condition \ref{cond::bin-Qt} and Lemma \ref{lemma::ll-coef}.

For $U_{1n}$, we have
\begin{eqnarray*}
\supM |U_{1n}(\tau,m)| 
&\leq& 
\supM \frac{\sumn w_{im}\kappa_{im}|u_{1i}(\tau)|}{nS_{0m}} \\
&\leq&
\supM A_{0m}\cdot \supM \frac{\sumn w_{im}|u_{1i}(\tau)|}{\sumn w_{im}} \\
 &=& \Op(g_{1n}^2),
\end{eqnarray*}
which follows from Claim \ref{claim4} and Lemma \ref{lemma::proof-U13}.
Similarly, 
\[
\supM |U_{3n}(\tau,m)| = \Op(g_{1n}^2).
\]
Combining the results with $U_{2n}(\tau,m)$, we have verified
\[
B_q(\tau,m) = \Op(g_{1n}^2 ) + \op(n^{-1/2}),
\] 
hence Claim \ref{claim3} holds. The proof is now complete.

\paragraph*{Step 5: Verification of Claim \ref{claim4}}
We check the conditions for $A_{0m}$, $A_{1m}$, and $A_{2m}$ separately.
For $A_{2m}$, standard algebra gives
\begin{align*}
\quad A_{2m} &= \sumn w_{im} + S_{1m}^T\bm{S}_{2m}^{-1}\left[\sumn w_{im}(\tX_i - \txm)(\tX_i - \txm)^T\right]\bm{S}_{2m}^{-1}S_{1m} \\
&\;\;\;\;-\; 2S_{1m}^T\bm{S}_{2m}^{-1}\sumn w_{im}(\tX_i - \txm)\\
&= nS_{0m}  - nS_{1m}^T\bm{S}_{2m}^{-1}S_{1m}\\
&\leq nS_{0m} ,
\end{align*}
similar to how we obtain (\ref{eq::DCT-remainder2}). For $A_{1m}$, Cauchy–Schwartz inequality gives
\[
A_{1m}\leq \sqrt{A_{2m} \sumn w_{im}} \leq nS_{0m}.
\]
For $A_{0m}$, note
\begin{eqnarray*}
\Pr\left(\supM A_{0m} \geq 2\right) 
&\leq& \Pr\left(\sup_{\substack{i=1,\ldots,n\\m=1,\ldots,M}} \left|w_{im}S_{1m}^T\bm{S}_{2m}^{-1}(\tX_i - \txm)\right|\geq 1\right)\\
 &\leq& \Pr\left(\supM \left\lVert \bh_m\cdot  S_{1m}^T\bm{S}_{2m}^{-1}\right\rVert\geq 1\right)\\
 &\lesssim& \frac{1}{n^3},
\end{eqnarray*}
 which follows from Lemma \ref{lemma::ll-coef}. We have verified Claim \ref{claim4}, and hence the proof of Proposition \ref{prop::condI-1} is complete.

\end{proof}

\subsubsection{Proof of Proposition \ref{prop::condI-2}}
\begin{proof}
For each $s \in (0,1)$, the decomposition in (\ref{eq::SQ-decomposition-s}) gives
\begin{eqnarray}
\left[\hat{v}(s,\bxm) - v(s,\bxm) \right]
&=&\frac{S_{0m}}{\hgamma_m}\left[B_q(s,m)+ C(s,m)\right]\nonumber\\
&&\;\;+\;\hgamma_m^{-1}S_{0m}\left[\dfrac{ \sumn w_{im}\kappa_{im}[v(s,X_i) - v(s,\bxm)]  }{\sumn w_{im}} \right]\nonumber\\
& = &\frac{S_{0m}}{\hgamma_m}\left[B_q(s,m)+ C(s,m) + B_{np}(s,m)\right]\label{eq::decomposition-BBC},
\end{eqnarray}
where the additional term $B_{np}(s,m)$ corresponds to the non-parametric binning bias; In (\ref{eq::SQ-decomposition-s}), this bias does not exist because $v(\tau,x)$ is linear in $x$. Following the above three-term decomposition, 
\begin{eqnarray*}
&&[\hat{v}(s ,\bxm) - v(s,\bxm)] - [\hat{v}(\tau,\bxm) - v(\tau,\bxm)]\\
&=& \frac{S_{0m}}{\hgamma_m}\left\{[B_q(s,m) - B_q(\tau,m)] + [B_{np}(s,m) - B_{np}(\tau,m)] + [C(s, m) - C(\tau,m)]\right\}.
\end{eqnarray*}

Recall that 
\[
r_n = \sqrt{\frac{\log n }{n\ubh^p}},
\]
from Proposition \ref{prop::condI-1}.
We give two claims below, the verification of which is at the end of the proof.
\begin{claim}\label{claim5}
    \[
\supMS \left| B_q(s,m)\right| = \op\left(n^{-1/2}\right).
\]
\end{claim}

\begin{claim}\label{claim6}
    \[
\supMS \left| B_{np}(s,m)\right| = \op\left(n^{-1/2}\right),
\]
\end{claim}
for any fixed $B > 0$, if any one of the requirements in Condition \ref{cond::bin-Aux} holds.

Claims \ref{claim5} and \ref{claim6} show that the bias terms are uniformly (over $s$) negligible; They are stronger than those in the proof of Proposition \ref{prop::condI-1}. 
Following Claims \ref{claim5} and \ref{claim6}, the proof proceeds in 5 steps. In steps 1 to 3, we establish the main argument:
\[
\supMS \left|C(s,m) - C(\tau,m)\right| = \op\left( n^{-1/2}\right). 
\]
In steps 4 and 5 we verify Claims \ref{claim5} and \ref{claim6}.

\paragraph*{Step 1: Decomposition}

We use the following decomposition of $Z_i(s)$ defined in (\ref{eq::def-Zi}):
\begin{eqnarray*}
&&(1-s)[Z_i(s) - v(s,X_i)] - (1-\tau)[Z_i(\tau) - v(\tau,X_i)]\\
&=&\underbrace{(q(s,X_i) - y_i)\bm{1}[q(\tau,X_i)\leq y_i \leq q(s,X_i)]}_{u_{4i}(s)} \\
&&\;\;-\; \underbrace{\left\{ (1-s)[v(s,X_i) - q(s,X_i)] - (1-\tau)[v(\tau,X_i) - q(s,X_i)]\right\}}_{\E_i[u_{4i}(s)]}\\
&&\;\;+\, \underbrace{[q(\tau,X_i) - q(s,X_i)]\{\tau-\bm{1}[y_i\leq q(\tau,X_i)]\}}_{u_{5i(s)}}\\
&\triangleq& u_{4i}(s) - \E_i[u_{4i}(s)] + u_{5i}(s),
\end{eqnarray*}
where we define $\E_i[\cdot]$ as the conditional expectation given $X = X_i$; and note $\E_i[u_{5i}(\tau,s)] = 0$. Therefore from the definition of $C(s,m)$ in (\ref{eq::SQ-decomposition-s}) and $\kappa_{im}$ in (\ref{eq::def-kappaik}), we have
\begin{equation}
\begin{aligned}
(1-s)C(s,m) - (1-\tau)C(\tau,m) &= 
\underbrace{(nS_{0m})^{-1}\sumn w_{im}\kappa_{im}\{u_{4i}(s) - \E_i[u_{4i}(s)]\}}_{U_{4n}(s,m)}\\
&\;\;+\; \underbrace{(nS_{0m})^{-1}\sumn w_{im}\kappa_{im}u_{5i}(s)}_{U_{5n}(s,m)}.
\label{eq::decomposition-U4U5}
\end{aligned}
\end{equation}

Furthermore, we separate $\kappa_{im}$ into:
\[
\kappa_{im} = \kappa_{im}\bm{1}[\kappa_{im} \geq 0] + \kappa_{im}\bm{1}[\kappa_{im} < 0] \triangleq \kappa^{(+)}_{im} - \kappa^{(-)}_{im},
\]
and correspondingly we define $U_{4n}(s,m) = U_{4n}^{(+)}(s,m)  + U_{4n}^{(-)}(s,m) $
\begin{equation*}
\begin{gathered}
U_{4n}^{(+)}(s,m) = (nS_{0m})^{-1}\sumn w_{im}\kappa^{(+)}_{im}\{u_{4i}(s) - \E_i[u_{4i}(s)]\}, \\
U_{4n}^{(-)}(s,m) = (nS_{0m})^{-1}\sumn w_{im}\kappa^{(-)}_{im}\{u_{4i}(s) - \E_i[u_{4i}(s)]\}.
\end{gathered}
\end{equation*}
In the following we consider $U_{4n}^{(+)}(s,m)$, $U_{4n}^{(-)}(s,m)$ and $U_{5n}(s,m)$ separately.

\paragraph*{Step 2: Bound for $U_{4n}$}

Let $s_+ = \tau + Br_n$, it suffices to consider the convergence of $U_{4n}^{(+)}(s,m)$ over the 
 over $s\in[\tau,s_+]$. The result for $s < \tau$ and/or $U_{4n}^{(-)}(s,m)$ follows analogously.

We use a monotonicity argument to show the uniformity over $s$. Since $u_{4i}(s)$ is monotonically increasing in $s$, we have the sandwich-type bound for $U_{4n}^{(+)}$:
\[
\ddfrac{\sumn w_{im}\kappa_{im}^{(+)}\left\{u_{4i}(\tau) - \E_i[u_{4i}(s_+)]\right\}}{nS_{0m}} \leq U_{4n}^{(+)}(s,m)\leq \ddfrac{\sumn w_{im}\kappa_{im}^{(+)}\left\{u_{4i}(s_+) - \E_i[u_{4i}(\tau)]\right\}}{nS_{0m}},
\]
which holds for all $s\in[\tau,s_+]$. Noting that $u_{4i}(\tau) = 0$, using the monotonicity argument in \citet[Theorem 19.1]{van2000asymptotic} gives
\begin{equation}
\sup_{\substack{m=1,\ldots,M\\s\in[\tau,s_+]}} |U_{4n}^{(+)}(s,m)| \leq \sup_{m=1,\ldots,M} |U_{4n}^{(+)}(s_+,m)| + \sup_{m=1,\ldots,M} \left|(nS_{0m})^{-1}\sumn w_{im}\kappa_{im}^{(+)}\E_i[u_{4i}(s_+)]\right|.
\label{eq::C-asymp-continuity-1}
\end{equation}
Next we bound the two terms separately.

For the first term in (\ref{eq::C-asymp-continuity-1}), note each of the summand in $U_{4n}^{(+)}(s_{+},m)$ is bounded from (\ref{eq::decomposition-U4U5}), and 
\[
0\leq \sumn |w_{im}\kappa_{im}^{(+)}u_{4i}(s_+)|^2 \leq (\underline{f}^{-1}|\tau - s_+|)^2\sumn w_{im}\kappa^2_{im} \leq  A_{2m}\underline{f}^{-2}B^2 r_n^2,
\] 
where $A_{2m}$ is in (\ref{eq::def-A1A2A3-kappa}) and $\underline{f}$ from Condition \ref{cond::R-bin-Y-density}.
 For any $\varepsilon_4 > 0$, application of the (conditional on $X$) Hoeffding's inequality and the union bound gives
\begin{eqnarray*}
\Pr\left(\sup_{m=1,\ldots,M}|U_{4n}^{(+)}(s_+,m)| \geq \varepsilon_4 n^{-1/2} \Bigg| X\right) 
&\leq&
\sum_{m=1}^M 2\exp\left\{-\frac{2n\varepsilon_4^2\cdot \inf_{m}S_{0m}}{\underline{f}^{-2}B^2r^2_n\cdot \sup_{m}A_{2m}}\right\}.
\end{eqnarray*}
Since $r_n^{-1} \gg \log(n)$ under Condition \ref{cond::bin-bandwidth}, we can obtain the unconditional tail bound, which implies
\[
\sup_{m=1,\ldots,M}|U_{4n}^{(+)}(s_+,m)|  = \op\left(n^{-1/2}\right),
\]
similar to how we obtain (\ref{eq::tail-centered-truncated}).

Next we consider the conditional expectation on the right hand side of (\ref{eq::C-asymp-continuity-1}). From (\ref{eq::decomposition-U4U5}), each $\E_i[u_{4i}(s_+)]$ is bounded as
\[
|\E_i[u_{4i}(s_+)]| \lesssim |q_i(s) - q_i(\tau)|^2 \leq \underline{f}^{-2} B^2r_n^2;
\]
Hence
\begin{eqnarray*}
\left|(nS_{0m})^{-1}\sumn w_{im}\kappa_{im}^{(+)}\E_i[u_{4i}(s_+)]\right| 
&\lesssim& (nS_{0m})^{-1}\sumn w_{im}|\kappa_{im}| B^2r_n^2\\
&\leq& (nS_{0m})^{-1}A_{1m} r_n^2\\
&= & \Op(r_n^2),
\end{eqnarray*}
where $A_{1m}$ and its property are in Claim \ref{claim4} of Proposition \ref{prop::condI-1}.

We now conclude from (\ref{eq::C-asymp-continuity-1}) that 
\[
\supMS |U_{4n}^{(+)}(s,m)| = \op\left(n^{-1/2}\right),
\]
since $r_n^2 = o(n^{-1/2})$ under Condition \ref{cond::bin-bandwidth}.

\paragraph*{Step 3: Bound for $U_{5n}$}

For any $s,s'\in(0,1)$, from the decomposition in (\ref{eq::decomposition-U4U5}) we have
\begin{eqnarray*}
|U_{5n}(s,m) - U_{5n}(s',m)| 
&=& (nS_{0m})^{-1} \max\{\tau,1-\tau\}\sumn w_{im}|\kappa_{im}| |q(s,X_i) - q_i(s',X_i)|\\
&\lesssim & \frac{A_{1m}}{nS_{0m}} |s - s'|\\
&\leq& |s-s'|,
\end{eqnarray*}
since $q(s,X_i)$ is Lipschitz continuous in $s$, and we use the bound for $A_{1m}$ in Claim \ref{claim4} of Proposition \ref{prop::condI-1}. 

We use a discretization argument to show the uniform convergence over $s$. Define 
\[
\tau - Br_n = s_0 < s_1< \ldots,s_J = \tau + Br_n,
\]
as an equally-spaced grid, such that $s_{j+1} - s_j \asymp n^{-1}$; therefore there are $J  \lesssim n$ sub-intervals. 
Similar to (\ref{eq::C-asymp-continuity-1}), we have
\begin{eqnarray*}
\supMS |U_{5n}(s,m)| 
&\leq& \sup_{\substack{m=1,\ldots,M\\j=0,\ldots,J}} |U_{5n}(s_j,m)| + \sup_{\substack{m=1,\ldots,M\\j=0,\ldots,J\\s\in I_j}} |U_{5n}(s,m) - U_{5n}(s_j,m)|\\
& \leq& \sup_{\substack{m=1,\ldots,M\\j=0,\ldots,J}} |U_{5n}(s_j,m)|+  \supM \left|\frac{A_{1m}}{n^2S_{0m}}\right|,
\end{eqnarray*}
where the last term is from the beginning of step 3.

Next we apply Bernstein inequality for the discretized $U_{5n}(s_j,m)$. For each $|s-\tau| \leq Br_n$, from (\ref{eq::decomposition-U4U5}) and the Lipschitz continuity of $q(s,X_i)$ (over $s$):
\begin{gather*}
\E[u_{5i}(s)\mid X] = 0,\qquad |w_{im}\kappa_{im}u_{5i}(s)| \leq C_{51}A_{0m}Br_n,\\
\sumn \var[w_{im}\kappa_{im}u_{5i}(s)\mid X] \leq C_{52}(Br_n)^2A_{2m},
\end{gather*}
where $C_{51}$ and $C_{52}$ are two constants, and $A_{0m}$ and $A_{2m}$ are in (\ref{eq::def-A1A2A3-kappa}).
For small enough $\varepsilon_5 > 0$, we apply the (conditional on $X$) Bernstein inequality as in (\ref{eq::Bernstein-2}), which shows that:
\begin{eqnarray*}
&&\Pr\left( \sup_{m,j}|U_{5n}(s_j,m)| \geq \varepsilon_5 n^{-1/2} \Bigg| X\right) \\
&\leq&
2\sum_{m=1}^M\sum_{j=0}^J \exp\left\{-\frac{n\varepsilon_5^2S^2_{0m}}{C_{52}B^2r_n^2A_{2m} + n^{1/2}\varepsilon S_{0m} C_{51}A_{0m}Br_n/3}\right\}\\
&\lesssim& 2\exp\left\{2 \log n - \frac{\varepsilon_5^2\cdot \inf_{m} S_{0m}}{B^2r_n^2 + n^{-1/2}Br_n\cdot \sup_{m}A_{0m}}\right\}.
\end{eqnarray*}
Similar to how we obtain (\ref{eq::tail-centered-truncated}), we can show that the corresponding unconditional probability is $o(1)$, which implies that:
\[
\supMS U_{5n}(s,m) = \op\left( n^{-1/2}\right).
\]

Therefore, with the decomposition in (\ref{eq::decomposition-U4U5}), we have established that 
\[
\supMS \left|C(s,m) - C(\tau,m)\right| = \op\left( n^{-1/2}\right),
\] 
from steps 1 through 3. Using Claims \ref{claim5}, \ref{claim6} and Equation (\ref{eq::decomposition-BBC}), we would complete the proof of Proposition \ref{prop::condI-2}.
{}
\paragraph*{Step 4: Verification of Claim \ref{claim5}}

With Condition \ref{cond::bin-Qt} and Lemma \ref{lemma::proof-U13}, the proof here is a simple extension of step 4 in the proof of Proposition \ref{prop::condI-1}. We only give an outline here. Using same decomposition used in (\ref{eq::decomposition-U123-tau}), we have
\begin{eqnarray*}
B_q(s,m) 
&=&
\frac{1}{(1-s)}[U_{1n}(s,m) + U_{2n}(s,m) + U_{3n}(s,m)].
\end{eqnarray*}

For $U_{2n}$, similar to step 4 in the proof of Proposition \ref{prop::condI-1}, we have:
\begin{eqnarray*}
\supMS |U_{2n}(s,m)|
&\leq&
\supMS \left\lVert \ddfrac{\sumn w_{im} \left[\frac{\tX_i - \txm  }{\bh_m }\right]u_{2i}(s)}{nS_{0m}}\right\rVert \cdot \supM \lVert \bh_mS_{1m}^T\bm{S}_{2m}^{-1} \rVert\\
&&\;\;+\;\supMS \left| \frac{\sumn w_{im}u_{2i}(s)}{nS_{0m}}\right| \\
&=&
\op(n^{-1/2}),
\end{eqnarray*}
which follows from Condition \ref{cond::bin-Qt}.

For $U_{1n}(s,m)$, it follows verbatim to part 4 of Proposition \ref{prop::condI-1} that:
\[
\supMS |U_{1n}(\tau,m)| 
 = \Op(g_{1n}^2),\quad \supMS |U_{3n}(\tau,m)| 
 = \Op(g_{1n}^2),
\]
from the uniform convergence (over $s$) in Lemma \ref{lemma::proof-U13}.

Noting that $g_{1n}^2 \ll n^{-1/2}$, we have verified Claim \ref{claim5}.

\end{proof}

\paragraph*{Step 5: Verification of Claim \ref{claim6}}

We check Claim \ref{claim6} under two scenarios separately. First, consider the case where the second requirement in Condition \ref{cond::bin-Aux} holds.
Then $v(s,x)$ is piece-wise linear in all the bins, for all $s \lesssim n^{-1/4} \ll r_n$. Therefore the same calculations in (\ref{eq::SQ-decomposition-s}) apply, and the non-parametric bias does not exist, i.e.,
\[
\supMS |B_{np}(s,m)| = 0,
\]
which follows from the nature of local-linear estimation \citep{fan2018local}.

In the following, we consider the case when only the first requirement in Condition \ref{cond::bin-Aux} holds. 
With a slight abuse of notation, we write $v(s,\tX_i) = v(s,X_i)$, $v(s,\txm) = v(s,\bxm)$. Let $v'_x$ denote the $p$-dimensional gradient vector with respect to covariates $x$ without the intercept term, and $\bm{v}''_{xx}$ be the $p$ by $p$ Hessian matrix. With Taylor expansion at each $\tilde{x}_m$,
\[
v(s,\tX_i) - v(s,\txm) = (\tX_i - \txm )^Tv'_x(s,\tilde{x}_m) +  \frac{1}{2}(\tX_i - \txm)^T \bm{v}''_{xx} (s,\hat{x}_{im}) (\tX_i - \txm),
\]
for some $\hat{x}_{im}$ in between $\tilde{x}_m$ and $\tX_i$. 

Note $B_{np}$ is the linear combination of $v(s,X_i) - v(s,\tilde{x}_m)$ as in (\ref{eq::decomposition-BBC}); We plug in the two terms in the above displayed equation into (\ref{eq::decomposition-BBC}) separately. First, the first-order terms sum up to exactly 0:
\[
\sumn w_{im}(\tX_i - \txm)^Tv'_x(s,\tilde{x}_m)[1-S_{1m}^T\bm{S}_{2m}^{-1}(\tX_i - \txm)] = 0,
\]
due to standard local-linear calculation \citep{fan2018local}.  

Second, note that $\bm{v}''_{xx}(\tau,x) = 0$ for all $x$ due to the linearity of $\tau$-th ES, therefore the first item in Condition \ref{cond::bin-Aux} implies $\lVert\bm{v}''_{xx}(s,x)\rVert_{2}\leq L_2|s-\tau|$ uniformly for all $x$. Hence
\begin{eqnarray*}
&&\frac{1}{2nS_{0m}}\sumn w_{im}(\tX_i - \txm)^T[\bm{v}''_{xx}(s,\hat{x}_{im})](\tX_i - \txm)[1-S_{1m}^T\bm{S}_{2m}^{-1}(\tX_i - \txm)]\\
&\leq&\left[\frac{L_2|s-\tau|}{2nS_{0m}}\sumn w_{im}\lVert\tX_i - \txm\rVert^2\right]\left|1-S_{1m}^T\bm{S}_{2m}^{-1}(\tX_i - \txm)\right|\\
&\lesssim& |s-\tau|\left| \frac{\sumn w_{im}\lVert \tX_i - \txm \rVert^2}{\sumn w_{im}}\right|(1+\op(1))\\
&\lesssim& |s-\tau|\bar{h}^2(1+\op(1)),
\end{eqnarray*}
where the first inequality owns to the operator norm bound for $\bm{v}''_{xx}$, and the $\op(1)$ terms are uniform in $m$ and independent of $s$ due to Lemma \ref{lemma::ll-coef}.

Combining the previous two displayed equations, we obtain
\[
\supMS|B_{np}(s,m)| = \Op\left(r_n\bar{h}^2\right)  = \op\left(n^{-1/2}\right),
\]
since 
\[
r_n\bh^2 = \sqrt{\frac{\bh^4\log n}{n \ubh^p}}  \ll \frac{1}{\sqrt{n}},
\]
under the first requirement of Condition \ref{cond::bin-Aux}.

Therefore, Claim \ref{claim6} holds under either requirement of Condition \ref{cond::bin-Aux}.

{\color{black}
\subsection{Proof of Theorem 4.3}
Similar to Theorem 4.1, we represent the weighted i-Rock estimator with bounded weights in the following Proposition. 
\begin{proposition}\label{prop::weighted_iRock_linearization}
Under the conditions in Theorem 4.1, the weighted i-Rock estimator in (25) with bounded weights $\omega_m$ satisfies
    \begin{align*}
        \left(\widehat{\beta} - \beta\right) &= \left(\sum_{m=1}^M \frac{\hgamma_m \omega_m \bar{x}_m \bar{x}_m^T}{v(\tau,\bar{x}_m)-q(\tau,\bar{x}_m)} \right)^{-1}\\
        &~~\times \sumM \left[\frac{\hgamma_m \omega_m \bxm }{v(\tau,\bar{x}_m)-q(\tau,\bar{x}_m)}\{\hat{v}(\tau,\bar{x}_m) - v(\tau,\bar{x}_m)\}\right]\\
&~~+ \op(n^{-1/2}). 
    \end{align*}
\end{proposition}
Proposition~\ref{prop::weighted_iRock_linearization} implies that the weighted i-Rock estimator is
asymptotically equivalent to the weighted least squares in (19) with weights $\omega_m/\{v(\tau,\bxm) - q(\tau,\bxm)\}$. The proof of Proposition~\ref{prop::weighted_iRock_linearization} follows closely the proof of Theorem 4.1. 
We define the weighted shifted i-Rock loss function as 
\begin{equation}
 L_n^{(\omega)}(\delta) = \sumM \hgamma_m \omega_m \int_{0}^1\; \left[\rho_\tau\left(\hat{v}_{m}(\alpha) - v_{m}(\tau) - \bxm \delta/\sqrt{n} \right) -  \rho_\tau\left(\hat{v}_{m}(\alpha ) - v_m(\tau )\right)\right]\,\d\alpha.
\label{eq::proof-rock-Ln-weighted}
\end{equation}
Since $\omega_m, m=1,\ldots,M$ are deterministic and bounded, we have 
\begin{equation*}
    n L_n^{(\omega)}(\delta) = \frac{1}{2} \delta^T \tilde D_{1n} \delta - \delta^T \tilde u_n + o_p(1),
\end{equation*}
where 
\begin{align*}
    \tilde D_{1n} &= \sum_{m=1}^M \frac{\hgamma_m \omega_m \bar{x}_m \bar{x}_m^T}{v(\tau,\bar{x}_m)-q(\tau,\bar{x}_m)}, \\
    \tilde u_n &= \sumM \left[\frac{\hgamma_m \omega_m \bxm }{v(\tau,\bar{x}_m)-q(\tau,\bar{x}_m)}\{\hat{v}(\tau,\bar{x}_m) - v(\tau,\bar{x}_m)\}\right].
\end{align*}

\begin{corollary}\label{corollary2}
    Under the conditions in Theorem 4.1, the weighted i-Rock estimator in (25) with optimal weights $\omega_m$ in (26) satisfies
    \begin{align*}
        \left(\widehat{\beta} - \beta\right) &= D_2^{-1} \sumM \left[\frac{\hgamma_m  \bxm }{\sigma^2_\tau(\bxm)}\{\hat{v}(\tau,\bar{x}_m) - v(\tau,\bar{x}_m)\}\right]\\
&~~+ \op(n^{-1/2}),
    \end{align*}
    where $D_2 = E\{XX^T/\sigma^2_\tau(X)\}$. 
\end{corollary}
Recognizing that $\omega_m, m=1,\ldots,M$ are deterministic, we have the following proposition as a modified version of Proposition~\ref{prop::condL}. 
\begin{proposition}\label{prop::condL_weighted}
    Under the conditions in Theorem 4.2, we have
    \begin{equation*}
        \sqrt{n} \sumM \left[\frac{\hgamma_m  \bxm }{\sigma^2_\tau(\bxm)}\{\hat{v}(\tau,\bar{x}_m) - v(\tau,\bar{x}_m)\}\right] \overset{d}{\to} N(0, D_2),
    \end{equation*}
    where $D_2$ is defined in Corollary~\ref{corollary2}. 
\end{proposition}
Propositions~\ref{prop::condI-1}, ~\ref{prop::condI-2}, and~\ref{prop::condL_weighted} lead to the desired  result in Theorem 4.3. 

\subsection{Weighted ES regression approaches}
\subsubsection{Efficiency comparison}
The main theoretical conclusion is that
the following estimators are asymptotically equivalent: (1) the optimal M-estimator of the joint approach in \cite{dimitriadis2022efficiency}, (2) the optimally weighted two-step approach in \cite{barendse2020efficiently}, and (3) the optimally weighted i-Rock approach proposed in (21) and (22) of our main manuscript. 
While all three weighted approaches achieve the same efficiency, the weighted i-Rock approach does not require modeling of the quantile functions, while the other two approaches rely on the linear quantile assumption. 
As shown in  \cite{dimitriadis2022efficiency}, these estimators achieve semi-parametric efficiency for a class of models specified in Theorem 5.2 of \cite{dimitriadis2022efficiency}.

Here, we provide additional details for the comparison of these approaches for the ES regression estimation.

\begin{enumerate}
    \item[(1)] \textbf{Joint approach:}\\
We start with the joint quantile and ES regression model proposed in \cite{dimitriadis2019joint}. The joint approach assumes linear quantile and ES function, i.e., 
\begin{equation}\label{eq::joint_linear}
    q_{[Y\mid X]}(\tau) = X^T \eta^0,
    \quad v_{[Y\mid X]}(\tau) = X^T \beta^0,
\end{equation}
and consider the following joint quantile and ES estimator
\[
 (\hat{\eta},\hat{\beta})  =  \min_{\eta,\beta} \sumn \ell_i(\eta,\beta; G_1,G_2),
  \]
where the joint loss function is
  \begin{equation}
  \ell_i(\eta,\beta; G_1,G_2) = \rho_{\tau}\left(G_1(Y_i) - G_1(x_i^T\eta)\right) + G'_2(x_i^T\beta)\left[Z_i(\eta) - x_i^T\beta\right] + G_2(x_i^T\beta),
  \label{eq::joint-loss}
  \end{equation} 
  and 
  \begin{equation}
      Z_i(\eta) = (1-\tau)^{-1}(Y_i - x_i^T \eta)\mathbbm{1}(Y_i\geq x_i^T \eta) + x_i^T \eta.
  \label{eq::def-Zi}
  \end{equation}
Under this framework, \cite{dimitriadis2022efficiency} provides a specific choice of $G_1$ and $G_2$ functions, namely,
\begin{equation}\label{eq::G2}
\begin{aligned}
    G_1(s)& = 0,\\
    G_2(s)& = \frac{c_1(X^T \eta^0 - s)}{\sqrt{\tau m_1(X)}} \arctan\left(\frac{\sqrt{\tau}(X^T \eta^0 - s)}{\sqrt{m_1(X)}}\right) \\
    &~~+ s\frac{\pi c_1(1+c_2)}{2\sqrt{\tau m_1(X)}} - \frac{c_1}{2\tau}\log\{m_1(X) + \tau m_2^2(X)\},
\end{aligned}
\end{equation}
so that the M-estimator attains the efficiency bound
\begin{equation}\label{eq::asy-var-joint}
\text{AVar}_{\text{M-joint}} = \left[E\left\{\frac{(1-\tau)XX^T}{m_1(X) + \tau m_2^2(X)}\right\}\right]^{-1},
\end{equation}
where $m_1(X) = \var(Y\mid Y \geq q(\tau, X), X)$, $m_2(X) = v(\tau, X) - q(\tau, X)$. 

Let $f$ denote the conditional density function of $Y$ given $X$.
For a more restricted class of conditional distributions, namely, 
\begin{equation}\label{eq::cond1}
    m_1(X) \propto m_2^2(X), \quad f(q(\tau,X)) \propto m_2^{-1}(X),
\end{equation}
the M-estimator with $G_1, G_2$ specified in~\eqref{eq::G2} attains optimal semi-parametric efficiency, the same as the Z-estimation. 
A representative example that satisfies condition~\eqref{eq::cond1} is the class of linear location–scale models. 
When some of the conditions (e.g., \eqref{eq::cond1}) are violated, the Z-estimator can beis more efficient than the M-estimator, which creates an efficiency gap. 
The asymptotic variance attainable by of the Z-estimator is
\begin{equation}\label{eq::asy-var-joint-Z}
\text{AVar}_{\text{Z-joint}} \propto \left[E\left\{\frac{XX^T}{\frac{(1-\tau)\tau m_1(X) m_2^2(X)}{m_1(X) + \tau^3 m_2^2(X)}}\right\}\right]^{-1}.
\end{equation}

\item[(2)] \textbf{Two-step approach:}\\
We consider  the weighted two-step approach proposed by \cite{barendse2020efficiently} under the linear assumption in~\ref{eq::joint_linear}, namely, 
\begin{align}
\hat{\eta}  &=  \min_\eta \sumn \rho_\tau(Y_i - x_i^T\eta),\label{eq::quantile_regression}\\
\hat{\beta}  &=  \min_\beta \sumn w_i \left[Z_i(\hat{\eta}) - x_i^T\beta\right]^2\label{eq::weighted_two_step}. 
\end{align}
The optimal weights proposed in~\cite{barendse2020efficiently} are 
\begin{equation}\label{eq::weight_two_step}
    w_i = \frac{1}{m_1(x_i) + \tau m_2^2(x_i)}, 
\end{equation}
and  the weighted two-step estimator is asymptotically normal
with asymptotic variance 
\begin{equation}\label{eq::asy-var-two-step}
\text{AVar}_{\text{two-step}} = \left[E\left\{\frac{(1-\tau)XX^T}{m_1(X) + \tau m_2^2(X)}\right\}\right]^{-1},
\end{equation}
which is the same as \eqref{eq::asy-var-joint}.
\item[(3)] \textbf{i-Rock approach:}\\
From Theorem 4.3 of the main manuscript, we know that the optimally weighted i-Rock estimator in (21) and (22) has the same asymptotic variance as \eqref{eq::asy-var-joint}.
\end{enumerate}

\subsubsection{Empirical results}

\textbf{Numerical algorithms:}\\
In practice, all three weighted approaches need an estimation of the weights. 
\begin{enumerate}
    \item[(1)] For the joint approach, \cite{dimitriadis2022efficiency} proposes to estimate $\hat{\eta}$ and plug it into the weights $G_2$. However, the joint loss function is non-convex and non-differentiable, thus solving the optimization problem can be computationally challenging. 

\item[(2)] For the weighted two-step estimation, \cite{barendse2020efficiently} does not offer an explicit numerical approach. To estimate the weights in~\eqref{eq::weight_two_step}, we require both the quantile and ES estimators to start with, and follow a three-step implementation: (1) obtain linear quantile regression estimator $\hat \eta$ from~\eqref{eq::quantile_regression}, (2) fit unweighted least-squares with pseudo-responses, i.e.,~\eqref{eq::weighted_two_step} with $w_i = 1$, denoted as $\tilde \beta$, and (3) obtain final ES estimator $\hat \beta$ by fitting weighted least-squares with pseudo-responses in~\eqref{eq::weighted_two_step} using weights 
\begin{equation}
    w_i = \frac{1}{(x_i^T \tilde \beta - x_i^T\hat \eta)^2}.
\end{equation}

\item[(3)] For the weighted i-Rock approach, we could directly use a plug-in estimator for the optimal weights in (22) of main manuscript from the initial quantile and ES estimators that are not based on linearity assumptions. 
\end{enumerate}
\begin{figure}[tbh!]
\centering 
\begin{subfigure}[b]{0.4\textwidth}
         \centering
         \includegraphics[width = \textwidth]{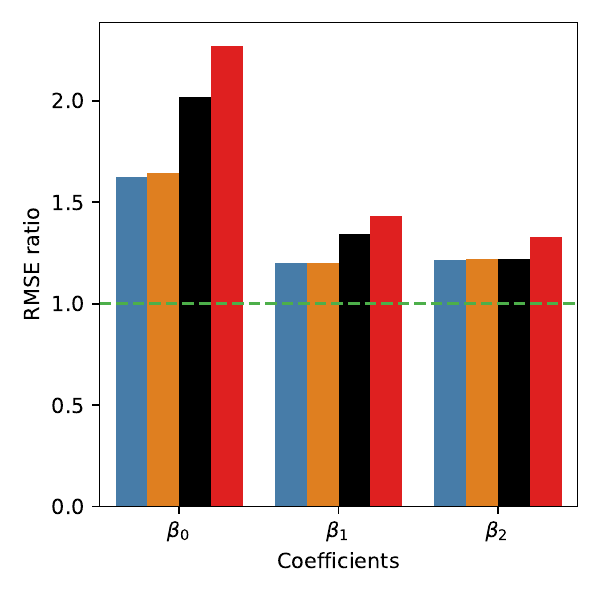}
         \caption{$n = 5000, \tau = 0.8$}
     \end{subfigure}
     \begin{subfigure}[b]{0.4\textwidth}
         \centering
         \includegraphics[width = \textwidth]{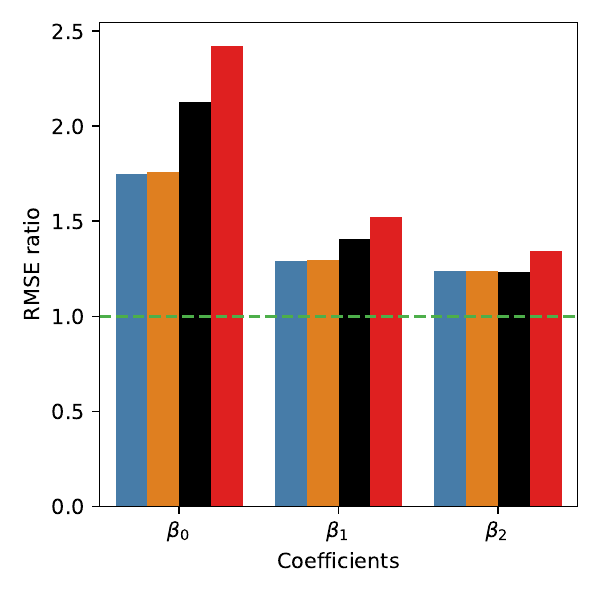}
         \caption{$n = 10000, \tau = 0.8$}
     \end{subfigure}
     \begin{subfigure}[b]{0.4\textwidth}
         \centering
         \includegraphics[width = \textwidth]{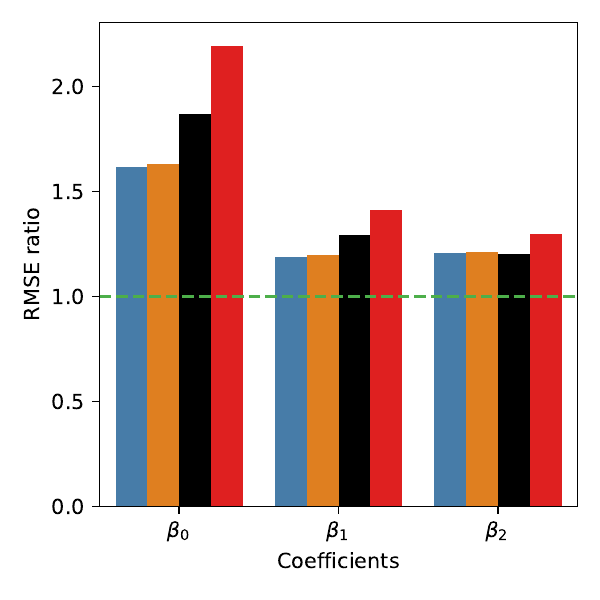}
         \caption{$n = 5000, \tau = 0.9$}
     \end{subfigure}
     \begin{subfigure}[b]{0.4\textwidth}
         \centering
        \includegraphics[width = \textwidth]{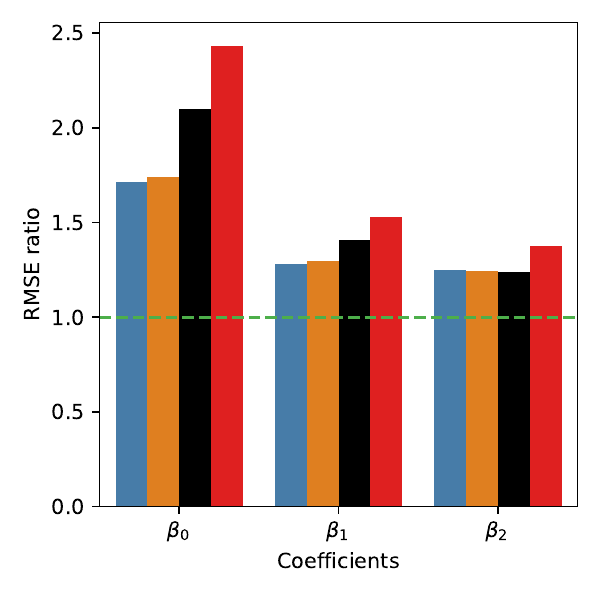}
         \caption{$n = 10000, \tau = 0.9$}
     \end{subfigure}
     \begin{subfigure}[b]{\textwidth}
         \centering
         \includegraphics[width = \textwidth]{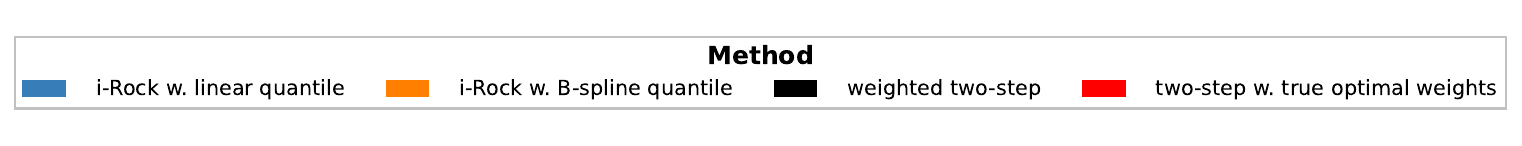}
     \end{subfigure}
     \caption{
     RMSE ratio of the (unweighted) two-step estimator over the (unweighted) 
     i-Rock estimator (with linear or B-spline quantile function estimation) or weighted two-step estimator (with estimated weights or true weights) under the linear heteroscedastic model (Case 5.1) at various quantile levels and sample sizes. The RMSE ratio greater than one indicates better efficiency for the estimator under consideration than the unweighted two-step estimator.}
     \label{fig::weighted_case1}
\end{figure}
\begin{figure}[tbh!]
    \centering
    \begin{subfigure}[b]{\textwidth}
    \centering
\includegraphics[width=0.9\linewidth]{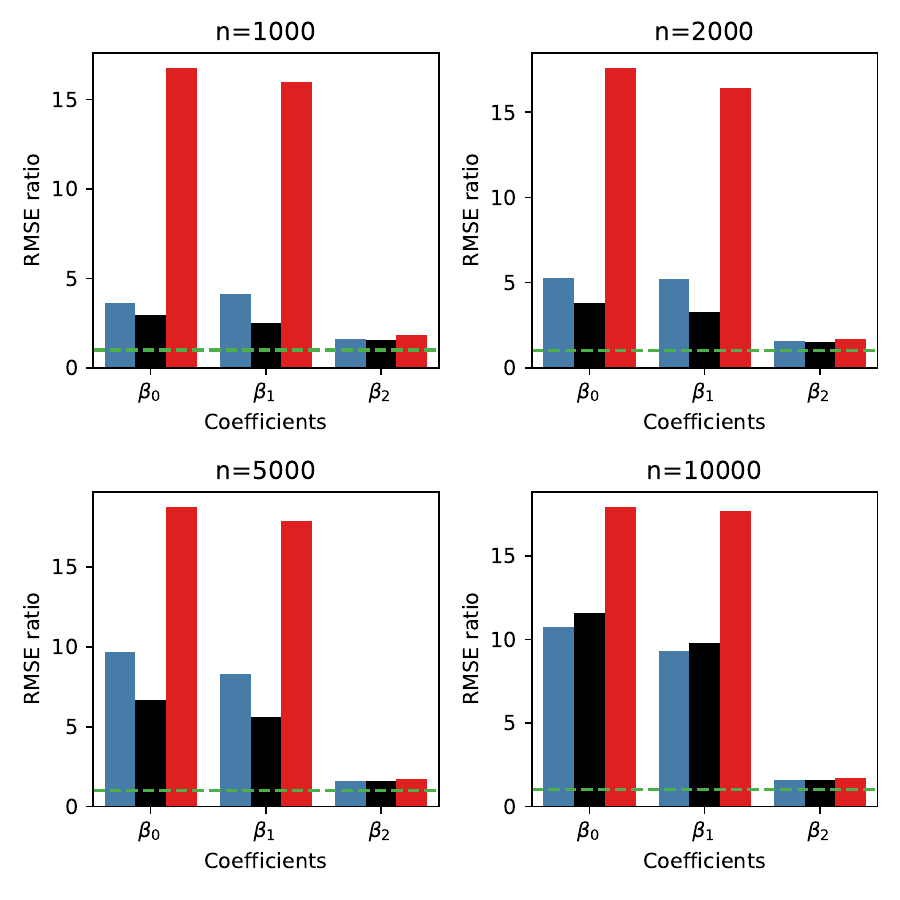}
    \end{subfigure}
    \begin{subfigure}[b]{\textwidth}
         \centering
    \includegraphics[width = \textwidth]{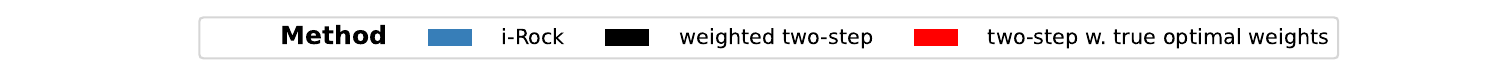}
     \end{subfigure}
     \caption{RMSE ratio of the (unweighted) two-step estimator over the (unweighted) 
     i-Rock estimator (with linear or B-spline quantile function estimation) or weighted two-step estimator (with estimated weights or true weights) under Case 5.3 at various quantile levels and sample sizes. The RMSE ratio greater than one indicates better efficiency for the estimator under consideration than the unweighted two-step estimator.}
     \label{fig::weighted_case3}
\end{figure}

\textbf{Numerical comparisons:}\\
While the three approaches can all achieve optimal efficiency, the weights need to be estimated from data. 
The i-Rock approach, without further weighting, incorporates implicit weights that are usually well correlated with the optimal weights. 
To demonstrate the finite-sample efficiency of the (unweighted) i-Rock approach, we compare it with the weighted two-step (with true and estimated weights) and unweighted two-step approaches. 
We consider the  two data-generating mechanisms given as Cases 5.1 and 5.3 in the main paper.
 
The results are summarized in Figures~\ref{fig::weighted_case1} and \ref{fig::weighted_case3}. In general, the i-Rock and the weighted two-step estimators exhibit superior efficiency relative to the unweighted two-step method. The optimal weighted two-step estimator with true weights attains the highest efficiency, consistent with theoretical expectations, though it is infeasible in practice. On the other hand, the weighted two-step procedure based on estimated weights can achieve greater efficiency than the i-Rock approach when the sample size is sufficiently large ($n=10,000$ in Figure \ref{fig::weighted_case3})  to ensure reliable weight estimation. However, its performance deteriorates when the available data are limited, as substituting the estimated weights $\hat{w}_m$ may induce instability in the weighted two-step estimator.
}

\section{Proof of technical lemmas}\label{append::lemma}
\subsection{Proof of Lemma~\ref{lemma::SQ-monotonicity}}
\begin{proof}
Under Condition \ref{cond::univariate}, the ES process $v_s = \E[Y\mid Y\geq q_s]$ is continuous in $s$, and in particular
\[
\frac{\partial v_s}{\partial s} = \frac{v_s - q_s}{1-s} > 0,\quad s\in[\tau_L,\tau_U],
\]
 which indicates $v_s$ is strictly increasing in $s$.

Next we show that the sample ES $\hat{v}_s\leq \hat{v}_t$ for $\tau_L \leq s < t\leq \tau_U$.
Without loss of generality, we can assume $\hat{q}_s < \hat{q}_t$, where $\hat{q}$ is the sample quantile; otherwise $\hat{v}_s = \hat{v}_t$. Let  $m_1 = \sum_{i=1}^n \bm{1}\{Y_i\geq \hat{q}_t\}$ and $m_2 = \sum_{i=1}^n \bm{1}\{Y_i\geq \hat{q}_s\}$; by the choice of sample quantiles $\hat{q}_s$, we have $m_2 \geq m_1 > 0$. Hence
\begin{eqnarray*}
\hat{v}_t - \hat{v}_s &=& \frac{\sum_{i=1}^nY_i \cdot \bm{1}\{Y_i\geq \hat{q}_t\}}{m_1} - \frac{\sum_{i=1}^nY_i \cdot \bm{1}\{Y_i\geq \hat{q}_s\}}{m_2}
\\
&=& \frac{(m_2-m_1)\sum_{i= 1}^nY_i \cdot \bm{1}\{Y_i\geq \hat{q}_t\} - m_1\sum_{i= 1}^nY_i \cdot \bm{1}\{\hat{q}_t > Y_i\geq \hat{q}_s\}  }{m_1m_2}\\
&\geq & \frac{\hat{q}_t(m_2-m_1) \sum_{i=1}^n \bm{1}\{Y_i\geq \hat{q}_t\} - m_1\hat{q}_t  \sum_{i=1}^n \bm{1}\{ \hat{q}_t > Y_i\geq \hat{q}_s\}}{m_1m_2}
\\
&\geq & 0,
\end{eqnarray*} 
where the equality in the penultimate inequality holds if and only if $m_1 = m_2$.
Therefore, $\hat{v}_s$ is non-decreasing with respect to $s$.

From its monotonicity, the one-sided limit of $\hat{v}_s$ from either the left or right exists. To show the continuity from the left, note that the quantile function $\hat{q}_s$ is left-continuous over $s\in(0,1)$, thus for any $s\in(\tau_L,\tau_U)$,
\begin{gather*}
\lim_{\varepsilon\to0+ }\sum_{i=1}^n\bm{1}\{Y_i\geq \hat{q}_{s - \varepsilon}\} = \lim_{\varepsilon\to0+ }\sum_{i=1}^n\bm{1}\{Y_i\geq \hat{q}_{s} - \varepsilon\} = \sum_{i=1}^n\bm{1}\{Y_i\geq \hat{q}_{s}\} > 0,\\
\lim_{\varepsilon\to0+ }\sum_{i=1}^n Y_i \bm{1}\{Y_i\geq \hat{q}_{s - \varepsilon}\} = \sum_{i=1}^nY_i\bm{1}\{Y_i\geq \hat{q}_{s}\}.
\end{gather*}
Since $\hat{v}_s$ is the ratio of the above displayed equations, we conclude that $\hat{v}_s$ is also continuous from the left.
\end{proof}

\subsection{Proof of Lemma~\ref{lemma::sto-diff-univariate}}

\begin{proof}

Define a class of functions $\mathcal{F} = \{y\mapsto \psi(y,\theta,s) : \theta\in [q_L,q_U], s \in [\tau_L,\tau_U] \}$. We shall show that $\mathcal{F}$ is a Donsker class of functions.
First, note that $\mathcal{F} = \mathcal{F}_1\times\mathcal{F}_2 \triangleq \{fg: f\in\mathcal{F}_1,g\in\mathcal{F}_2\}$, where
\[
\mathcal{F}_1 = \{z\mapsto z- v_s: s \in [\tau_L,\tau_U]\},\quad \mathcal{F}_2 = \{z\to \bm{1}[z\geq \theta]: \theta \in [q_L,q_U]\}.
\]
Since $\mathcal{F}_1$ contains only linear functions and $\mathcal{F}_2$ contains only indicator functions of half lines, it is clear that both $\mathcal{F}_1$ and $\mathcal{F}_2$ are VC classes of functions, and therefore $\mathcal{F}$ also satisfy the uniform entropy condition. (See e.g. Example 19.19 of \citet{van2000asymptotic}.) Next, let $F(y) = [|y| + |v_{\tau_U}| + |v_{\tau_L}|]\bm{1}\{y\geq q_{L}\}$, we can easily verify that $\sup_{f\in\mathcal{F}} |f(z)| \leq F(z)$ and $\E[F^2] < +\infty$ under Condition \ref{cond::univariate}, i.e., $F$ is a suqare-integrable envelope function for $\mathcal{F}$. Therefore, we conclude that $\mathcal{F}$ is Donsker, which follows from Lemma 19.14 of \citet{van2000asymptotic}.

Let $\mathbb{T} = [q_L,q_U]\times [\tau_L,\tau_U]$ be the product space equipped with the semimetric $\rho((\theta,s),(\theta',s')) = \{\E[\psi(Y^*;\theta,s) - \psi(Y^*;\theta',s')]^2\}^{1/2}$. As a consequence of Donskerness, the stochastic process $\Gn[\psi(Y^*;\theta,s)]$ indexed by $(\theta,s)$ is stochastically equi-continuous on $(\mathbb{T},\rho)$, and that $(\mathbb{T},\rho)$ is totally bounded.

Similar to Lemma 19.24 in \citet{van2000asymptotic}, define the map
\[
\begin{array}{cccccc}
g: &\ell^{\infty}(\mathbb{T})&\times &\ell^{\infty}([\tau_L,\tau_U]) &\mapsto &\mathbb{R}\\
&z(\cdot,\cdot)&\times &v(\cdot)&\mapsto&
\sup_{s\in[\tau_L,\tau_U]}|z(v(s),s) - z(q_s,s)|.
\end{array}
\]
First, it is easy to verify that $g(\cdot,\cdot)$ is continuous (with respect to the product metric on $\ell^{\infty}(\mathbb{T})\times\ell^{\infty}([\tau_L,\tau_U])$) at $(z_0,v_0)$, as long as $z_0(\cdot,\cdot)$ is uniformly continuous over $(\mathbb{T},\rho)$.
Second, by its Donskerness, $\Gn(\psi(Y^*;\theta,s))\converged \mathbb{G}_{\infty}(\theta,s)$ in $\ell^{\infty}(\mathbb{T})$, where almost all sample paths of the limit $\mathbb{G}_{\infty}(\theta,s)$ is uniformly continuous on $(\mathbb{T},\rho)$. Third, by assumption we have $\tilde{q}_s \convergeip q_s$ on $\ell^{\infty}([\tau_U,\tau_L])$, which implies that the bivariate process $[\Gn(\psi(Y^*;\theta,s)),\tilde{q}_s]$ also converges weakly. Hence by the continuous mapping theorem,
\begin{eqnarray*}
\sup_{s\in[\tau_L,\tau_U]}\left|\Gn[\psi_{(\tilde{q}_s,s)}] - \Gn[\psi_{(q_s,s)}]\right|
&=&g\circ\left\{\Gn[\psi(Y^*;\theta,s)],\tilde{q}_s\right\} \\
&=&g\circ\left\{\Gn[\psi(Y^*;\theta,s)],\tilde{q}_s\right\}- g\circ\left\{\mathbb{G}_{\infty}[\psi(\theta,s)],q_s\right\} \\
 &\converged& 0,
\end{eqnarray*}
since $g\left\{\mathbb{G}_{\infty}[\psi(\theta,s)],q_s\right\} = 0$. Weak convergence to a constant then implies convergence in probability, which concludes the proof.
\end{proof}



\subsection{Proof of Lemma~\ref{lemma::bin-approx-moment-X}}
\begin{proof}

First we prove the second statement. By the Lipschitz continuity of $g(\cdot)$, we have
\begin{eqnarray*}
&&\left\lVert \E\left[ \sumM \bm{1}\{X\in A_m\} g(\bar{x}_m )h(X) \right] - \E[g(X)h(X)] \right\rVert \\
&\leq& \E\left[ \sumM \bm{1}\{X\in A_m\} \cdot \lVert g(\bxm ) - g(X)\rVert \cdot |h(X)|\right] \\
&\lesssim& \supM\text{diam}(A_m)\cdot \E\left[ |h(X)|\right] \\
&=& \op(1),
\end{eqnarray*}
where the last equality follows from the binning conditions in  Lemma \ref{lemma::bin-approx-moment-X} as well as the absolute integrability of $h$.

For the first claim, we first show the convergence when $\hgamma_m $ is replaced by $\hat{\pi}_m$, where $\hat{\pi}_m$ is given in Lemma \ref{lemma::bin-approx-moment-X}. By re-arranging the summation we have
 \[
\sumM \hat{\pi}_m  g(\bxm ) = \frac{1}{n}\sumM\sumn \bm{1}\{X_i\in A_m \} g(\bxm ) = \frac{1}{n}\sumn\underbrace{\sumM  \bm{1}\{X_i\in A_m \} g(\bxm )}_{U_{i}^{(n)}},
 \]
where $U_i^{(n)}$ depends on the sample size through binning. For each fixed $n$, $U_{i}^{(n)}$ are $i.i.d.$ across $i=1,\ldots,n$, and $\E[|U_i^{(n)}|] < +\infty$  since $g(\cdot)$ is bounded. The Law of Large Numbers gives
\[
\sumM \hat{\pi}_m  g(\bxm ) -\E[U_i^{(n)}] \convergeip 0.
\] 
Furthermore, using the second statement of the lemma, we have $\E[U_{i}^{(n)}] \to \E[g(X)]$, thus
\begin{equation}\label{eq::integration-bins-1}
\sumM \hat{\pi}_m  g(\bxm )\convergeip \E[g(X)].
\end{equation}

Next, we show the first statement of the lemma holds with $\hgamma_m$. Under the conditions on $\hgamma_m$ in Lemma \ref{lemma::bin-approx-moment-X}, we have
\[
\left| \sumM \hgamma_m g(\bxm ) - \sumM \hat{\pi}_m  g(\bxm ) \right| \leq  \sumM \left|\frac{\hgamma_m  - \hat{\pi}_m }{\hat{\pi}_m }\right|\hat{\pi}_m |g(\bxm )| = \op(1),
\] 
since $|g(\cdot)|$ is bounded.
The proof is now complete by combining the above displayed equation with (\ref{eq::integration-bins-1}). 

\end{proof}

\subsection{Proof of Lemma~\ref{lemma::bin-cond-implications}}
\begin{proof}
We prove the three items separately. 

For the first statement, note that $|v(s,x) - q(s,x)|$ is a continuous function in $s$ and $x$ under Condition \ref{cond::bin-SQ}; therefore the upper bound holds since continuous functions are always bounded on compact intervals. We consider the lower bound. By the uniform continuity of $q(s,x)$ on compact intervals, there is a constant $c_1 > 0$ such that $|s-\tau| \leq 2c_1$ implies $|q(s,x) - q(\tau,x)| \leq \varepsilon_0$ for all $x$, where $\varepsilon_0$ is defined in Condition \ref{cond::R-bin-Y-density}. Without loss of generality we assume $c_1 < \varepsilon_0/2$. For each $|s-\tau| < c_1,$ we have
\begin{eqnarray*}
v(s,x) - q(s,x) &=& \dfrac{\int_{s}^1 q(\alpha,x) - q(s,x)\,\d\alpha}{1-s} \\
&\geq& \ddfrac{ \inf_{x,|\alpha - \tau|\leq 2c_1}\left[\frac{\partial q(\alpha,x)}{\partial s}\right]\cdot \int_{\tau + c_1}^{\tau + 2c_1} |\alpha - s|\,\d\alpha}{1-s}\\
&\geq& \ddfrac{1}{\sup_{x,|y-q(\tau,x)|\leq \varepsilon_0}f_{Y\mid X}(y;x)}\cdot \frac{c_1^2}{2(1-s)}.
\end{eqnarray*}
The lower bound in the first statement hence follows from Condition \ref{cond::R-bin-Y-density}.

For the second statement, the derivative of $q(s,x)$ is bounded because $f_{Y\mid X}(y;x)$ is bounded in Condition \ref{cond::R-bin-Y-density}. For the derivative of $v(s,x)$, note
\begin{equation}
\frac{\partial v(s,x)}{\partial s} = \frac{v(s,x) - q(s,x)}{1-s};
\label{eq::SQ-derivative-continuous}
\end{equation}
the boundedness of the above derivative then follows from item 1 of Lemma \ref{lemma::bin-cond-implications}.

Finally we prove the third statement. Since the second statement of Lemma \ref{lemma::bin-cond-implications} implies both $v(s,x)$ and $q(s,x)$ are uniformly (in $x$) Lipschitz continuous in $s\in[\tau-c_1,\tau+c_1]$. Therefore, the derivative $\partial v(s,x)/\partial s$ is also uniformly (over both $x$ and $s\in[\tau-c_1,\tau+c_1]$) Lipschitz continuous from (\ref{eq::SQ-derivative-continuous}). Furthermore, the Lipschitz continuity of $[\partial v(s,x)/\partial s]^{-1}$ follows since $\partial v(s,x)/\partial s$ is uniformly bounded away from $0$ and $+\infty$.
\end{proof}

\subsection{Proof of Lemma \ref{lemma::bin-technical-h}}
\begin{proof}
We need to check items 1, 2 and 3 in the lemma separately. As a preliminary result, note that
\[
\left|\sumM (\hgamma_m -\hat{\pi}_m)\right| \leq \sumM  \hat{\pi}_m \left|\frac{\hgamma_m -\hat{\pi}_m}{\pi_m}\right| = \op(1),
\]
under the conditions of the lemma. Therefore $\sumgammaM = \Op(1)$.

To check item 2, it follows that
\[
\sqrt{n}\sumgammaM \left[\hat{v}_m(\tau) - v_m(\tau)\right]^2 \leq \Op(\sqrt{n}\,r_n^2),
\]
which is a direct consequence of Condition \ref{cond::vhat1}. 

Next we check item 3 in Lemma \ref{lemma::bin-technical-h}. From the monotonicity and (left-)continuity of $\hat{v}(s,\bxm )$ we have $\tau \leq  \hat{h}_m[\hat{v}_m(\tau)] < \tau + g_n$, for any $g_n > 0$ satisfying $\hat{v}_m(\tau + g_n) > \hat{v}_m(\tau)$. Therefore it suffices to show that there exists a sequence $0 < g_n \ll n^{-1/2}$, such that
\[
\inf_{m=1,\ldots,M}[\hat{v}_m(\tau + g_n) - \hat{v}_m(\tau)] > 0,
\]
with high probability; the above displayed inequality means the functions $\hat{v}_m(\cdot)$ are not flat near $\tau$. 
Note for any $ 0 < g_n \ll n^{-1/2}$, we shall have
\begin{eqnarray*}
\inf_{m=1,\ldots,M}[\hat{v}_m(\tau + g_n) - \hat{v}_m(\tau)] 
&\geq& 
\inf_{m}[v_m(\tau + g_n) - v_m(\tau)] \\
&&\;- 
\underbrace{\sup_{\substack{m=1,\ldots,M\\s:|s-\tau| \lesssim n^{-1/2}}}\left|[\hat{v}_m(s) - v_m(s)] - [\hat{v}_m(\tau) - v_m(\tau)]\right|}_{\Op(G_n)}\\
& \geq & 
g_n \cdot \left[\inf_{\substack{m=1,\ldots,M\\|s-\tau| \leq g_n}} v_m'(s)\right] - \Op(G_n),
\end{eqnarray*}
where $G_n \ll n^{-1/2}$ as in the second requirement of Condition \ref{cond::vhat1}. By Lemma \ref{lemma::bin-cond-implications}, $v_m'(s)$ is uniformly bounded from below; therefore by choosing any $g_n$ such that $G_n \ll g_n \ll n^{-1/2}$, the last displayed inequality is positive with probability tending to 1. Item 3 in Lemma \ref{lemma::bin-technical-h} thus takes hold.

\def\hxi{\hat{\xi}}
Finally we check item 1 in Lemma \ref{lemma::bin-technical-h}. Our proof follows the classical treatment in \citet{bahadur1966note}. We first show that $\hat{h}_m(z)$ converge uniformly at a rate of $r_n$, which is given in Lemma \ref{lemma::bin-technical-h}. For each $s$ in a shrinking neighbourhood of $\tau$, and for any fixed $C_1 > 0$, it follows from the definition of $\hat{h}$ in (\ref{eq::def-h}) that
\begin{gather*}
\hat{h}_m[v_m(s)] < s - C_1r_n \;\Rightarrow\; v_m(s) \leq \hat{v}_m(s - C_1r_n),\\
\hat{h}_m[v_m(s)] > s + C_1r_n \;\Rightarrow\; v_m(s) \geq \hat{v}_m(s + C_1r_n),
\end{gather*}
which shows that
\begin{gather*}
\sup_{m=1,\ldots,M}|\hat{h}_m[v_m(s)] - h_m[v_m(s)]| > C_1r_n \\
\Downarrow \\
\sup_{\substack{m=1,\ldots,M\\|u| < C_1r_n}}|\hat{v}_m(s + u) - v_m(s + u)| \geq C_1 r_n\,\inf_{|u| \leq C_1r_n} v_m'(s+u).
\end{gather*}
Note $v_m'(\cdot)$ is uniformly bounded in Lemma \ref{lemma::bin-cond-implications}, hence for sufficiently large $C_1$, the probability of the right hand side of the above displayed equation converges to zero by Condition \ref{cond::vhat1}.
Adding uniformity with respect to $s$, we have that
\begin{equation}
\sup_{\substack{m=1,\ldots,M\\|z - v_m(\tau)|\leq C_2\, (r_n+n^{-1/2}) }}|\hat{h}_m(z) - h_m(z)| = \Op(r_n) = \op(1),
\label{eq::bin-rate-inverse}
\end{equation}
for some $C_2 > 0$; we can use the range $|z - v_m(\tau)|\leq C_2\, (r_n+n^{-1/2}) $ since $h_m$ is uniformly (over $m$) Lipschitz continuous by Lemma \ref{lemma::bin-cond-implications}.

Next we consider the asymptotic equi-continuity of $\hat{h}_m$. Let $z_m = v_m(\tau)$, and fix a $z_m'$ such that $|z_m' - z_m| \leq C_2(n^{-1/2} + r _n)$. Define $\hxi_m = \hat{h}_m(z_m)$, $\hxi_m' = \hat{h}_m(z_m')$. 
Fixing $n$, from the monotonicity and (left-)continuity of $\hat{v}_m$ we have:
\[
\hat{v}_m(\hxi_m) \leq z_m \leq \hat{v}_m(\hxi_m + \varepsilon_n) ,\quad \hat{v}_m(\hxi_m') \leq z_m' \leq \hat{v}_m(\hxi_m' + \varepsilon_n),
\]
for any $\varepsilon_n > 0$; See \citet[Chapter 19]{van2000asymptotic}. Letting $\Delta_m(\cdot) = \hat{v}_m(\cdot) - v_m(\cdot)$, the first set of inequalities above on the left implies
\[
\Delta_m(\hxi_m) \leq [z_m - v_m
(\hxi_m)] \leq \Delta_m(\hxi_m + \varepsilon_n) + v_m(\hxi_m + \varepsilon_n) - v_m(\hxi_m).
\]
Re-arranging the above displayed inequalities gives
\begin{equation}
\begin{aligned}
\Delta_m(\hxi_m) - \Delta_m(\hxi_m' + \varepsilon_n) - \eta_{k}(z_m') \leq [z_m - z_m'] &- [v_m(\hxi_m) - v_m(\hxi_m')] \\
&\leq \Delta_m(\hxi_m + \varepsilon_n) - \Delta_m(\hxi_m') +  \eta_{k}(z_m'),
\end{aligned}
\label{eq::bound-hk-asymp-equi-conti}
\end{equation}
where 
\[
\eta_{k}(z_m') = \max\left\{|v_m(\hxi_m +\varepsilon_n) - v_m(\hxi_m)|,|v_m(\hxi_m' +\varepsilon_n) - v_m(\hxi_m')|\right\}.
\]

We derive the desired asymptotic equi-continuity of $\hat{h}_m$ from (\ref{eq::bound-hk-asymp-equi-conti}). To this end, we bound its left and right hand sides separately. An application of the results in (\ref{eq::bin-rate-inverse}) shows that both $\hxi_m$ and $\hxi_m'$ converges in probability towards $\tau$ uniformly over $m$. It then follows from the Lipschitz continuity of $v_m$ in Lemma \ref{lemma::bin-cond-implications} that
\[
\sup_{\substack{m=1,\ldots,M\\|z_m' - v_m(\tau)| \leq C_2(n^{-1/2} + r_n)}}  |\eta_m(z_m')| = \Op(\varepsilon_n).
\] 
In addition, from (\ref{eq::bin-rate-inverse}) we have
\[
\sup_{m=1,\ldots,M}|\hxi_m - \tau| = \Op(r_n), \quad \sup_{\substack{m=1,\ldots,M\\|z'_m - v_m(\tau)|\leq C_2(n^{-1/2} + r_n)}}|\hxi_m' -\tau | = \Op(r_n + n^{-1/2}),
\]
Then, by choosing $\varepsilon_n =o(n^{-1/2} \wedge r_n)$ we have 
\[
\sup_{\substack{m=1,\ldots,M\\|z_m' - v_m(\tau)| \leq C_2(n^{-1/2} + r_n)}} |\Delta_m(\hxi_m) - \Delta_m(\hxi_m' + \varepsilon_n)| = \op\left(n^{-1/2}\right),
\]
 from the second statement of Condition \ref{cond::vhat1}. The right hand side of Equation (\ref{eq::bound-hk-asymp-equi-conti}) can be bounded with the same argument, hence we have
 \begin{equation}
\sup_{\substack{m=1,\ldots,M\\|z_m' - v_m(\tau)| \leq C_1(n^{-1/2} + r_n)}}\left |[z_m - z_m'] - [v_m(\hxi_m) - v_m(\hxi_m')]\right| = \op\left(n^{-1/2}\right),
\label{eq::bound-hk-asymp-equi-conti2}
\end{equation}
by our choice of $\varepsilon_n$.

Finally, we connect Equation (\ref{eq::bound-hk-asymp-equi-conti2}) with the desired asymptotic equi-continuity of $\hat{h}_m$. Let $\xi_m = h_m(z_m)$, $\xi_m' = h_m(z_m')$, and therefore $z_m - z_m' = v_m(\xi_m) - v_m(\xi_m')$. By the first-order Taylor expansion of $v_m(\cdot)$ we have
\begin{eqnarray*}
\left|[z_m - z_m'] - [v_m(\hxi_m) - v_m(\hxi_m')]\right| 
&=& 
\left|v_m'(\tilde{s}_1)[\xi_m - \hxi_m] - v_m'(\tilde{s}_2)[\xi_m' - \hxi_m']\right| \\
&\geq& 
v_m'(\tilde{s}_1) \left|[\xi_m - \hxi_m] - [\xi_m' - \hxi_m']\right|\\
&&\;\; - \underbrace{\,\supM\left|[v_m'(\tilde{s}_1) - v_m'(\tilde{s}_2)]\cdot[\hxi_m' - \xi_m'] \right|}_{\Op(H_n)}\\
&\geq& c_1\,\left|[h_m(z_m) - \hat{h}_m(z_m)] - [h_m(z_m') - \hat{h}_m(z_m')]\right| \\
&&\;\;-\, H_n,
\end{eqnarray*}
for some $\tilde{s}_1$ in between $\xi_m$ and $\hxi_m$ and some  $\tilde{s}_2$ in between $\xi_m'$ and $\hxi_m'$; the last inequality follows by expanding $\xi_m$ and $\hat{x}_m$, and that $v_m'$ is bounded in Lemma \ref{lemma::bin-cond-implications}. Since both $h_m$ and $v_m'$ is Lipschitz continuous, it follows from (\ref{eq::bin-rate-inverse}) that we can take $H_n = (r_n + n^{-1/2})^2$.
We conclude from the above displayed equation and (\ref{eq::bound-hk-asymp-equi-conti2}) that
\begin{align*}
&\sup_{\substack{m=1,\ldots,M\\z_m = v_m(\tau)\\|z_m' - v_m(\tau)| \leq C_1\cdot (r_n + n^{1/2})}} \left|[h_m(z_m) - \hat{h}_m(z_m)] - [h_m(z_m') - \hat{h}_m(z_m')]\right|\\ 
&= \op\left(n^{-1/2}\right) + \Op\left((r_n + n^{-1/2})^2\right),
\end{align*}
which proves item 1 of Lemma \ref{lemma::bin-technical-h}. The proof is now complete.

\end{proof}




\subsection{Proof of Lemma \ref{lemma::ll-coef}}
\begin{proof}
We first prove the first statement; By definition
\[
\left\lVert S_{0m}^{-1}S_{1m}\right\rVert \leq \frac{\sumn \lVert X_i-\tilde{x}_m\rVert \bm{1}\{X_i\in A_m\}}{\sumn \bm{1}\{X_i\in A_m\}}\leq \bh_m,
\]
for all $m = 1,\ldots,M$.
Therefore
\[
\left| 1- \frac{\hgamma_m}{S_{0m}}\right| = \left|S_{0m}^{-1} S_{1m}^T\bm{S}_{2m}^{-1}S_{1m}\right| \leq  \lVert \bh_m\bm{S}_{2m}^{-1}S_{1m} \rVert,
\]
and hence the first statement follows from the second statement. It suffices to show
\[
\sup_m\lVert \bh_m\bm{S}_{2m}^{-1}S_{1m} \rVert = \op(1).
\]

For the second statement, we first give a uniform probability order bound for $\lVert S_{1m}\rVert_2$, where
\[
n\cdot S_{1m} = \sumn (X_i - \tilde{x}_m)\bm{1}\{X_i\in A_m\} \in \mathbb{R}^p.
\] 
We apply the covering argument to show the convergence of the $\ell_2$ norm.
For any $\alpha\in\mathbb{R}^p$, with $\lVert \alpha\rVert  = 1$, we have
\begin{eqnarray}
\left| \E\left( \alpha^TS_{1m}\right)\right |
&=& 
\left| \alpha^T \int_{z\in A_m}(z - \tilde{x}_m ) f_X(z) \d z\right| \nonumber \\
&\leq& f_X(\tilde{x}_m)\cdot \left|\alpha^T \int_{z\in A_m} (z-\tilde{x}_m)\d z\right|\nonumber\\
&&\;\;+\,
\alpha^T \int_{z\in A_m} |f_X(z) - f_X(\tilde{x}_m)|\cdot\left\lVert z-\tilde{x}_m\right\rVert \d z\nonumber\\
&=& 0 + O\left(\bar{h}_m^{p+2}\right) \label{eq::kernel-expect},
\end{eqnarray}
uniformly over $m$, where the last inequality owns to $\tilde{x}_m$ being the geometric center of $A_m$, as well as the Lipschitz continuity of $f_X$. 
Similarly, we have the following uniform bound for variance
\begin{eqnarray*}
\var\left( \alpha^TS_{1m}\right)
&\leq&
\E[(\alpha^T (X_i - \tilde{x}_m)\bm{1}\{X_i\in A_m\})^2]\\
&=& O\left(\bar{h}_m^{p+2}\right).
\end{eqnarray*}
Furthermore, note the boundedness of $\lVert X_i - \tilde{x}_m\rVert \leq \bar{h}_m$ when $X_i \in A_m$. Application of the Bernstein's inequality \citep[Theorem 2.8.4]{vershynin2018high} gives for any $\varepsilon > 0$,
\begin{eqnarray*}
&&\Pr\left(n\cdot |\alpha^TS_{1m} | \geq  n\bar{h}_m^{p+1}\varepsilon\right) \\
&\leq&
\Pr\left( |\E(\alpha^TS_{1m}) | \geq \bar{h}_m^{p+1}\varepsilon\right)  + \Pr\left( n\cdot|\alpha^TS_{1m} - \E(\alpha^TS_{1m})| \geq n\bar{h}_m^{p+1}\varepsilon\right) \\
&\leq&
 2\exp\left\{-\frac{n^2\bar{h}_m^{2p+2}\varepsilon^2/2}{nO(\bar{h}_m^{p+2}) + n\bar{h}_m^{p+2}\varepsilon/3}\right\} \\
 &=& 2\exp\left\{-C_1n\bar{h}_m^{p}\varepsilon^2\right\},
\end{eqnarray*}
for some constant $C_1 > 0$ whenever $n$ is sufficiently large. 
With the standard covering argument, see e.g., \citet[Chapter 4]{vershynin2018high}, we have
\[
\lVert S_{1m}\rVert = \sup_{\alpha} \alpha^TS_{1m} \leq 2\sup_{j=1,\ldots,J} \alpha_t^TS_{1m},
\]
where $\{\alpha_j\}$ forms a $1/2$-net in the unit $p$-dimensional sphere and the covering number $J \leq 2^p$. Using a union bound over $m$ and $t$ gives
\begin{eqnarray*}
\Pr\left(\sup_{m=1,\ldots,M}\left\lVert \frac{n\cdot S_{1m}}{ n\bar{h}_m^{p+1}}\right\rVert \geq 2\varepsilon \right) 
&=& 2 \sum_{m=1}^M\sum_{j=1}^J \exp\{-C_1n\bar{h}_m^{p}\varepsilon^2\}\\
&\leq& 2\exp\left\{\log M + \log J - C_1 n\varepsilon^2\cdot \inf_{k}\bar{h}_m^p\right\}\\
&\lesssim& \frac{1}{n^3},
\end{eqnarray*}
for sufficiently large $n$ under Condition \ref{cond::bin-bandwidth}, which implies
\[
\sup_{m=1,\ldots,M}\left\lVert \frac{S_{1m}}{\bh_m^{p+1}}\right\rVert = \op\left(1\right).
\]

Next, we prove an analogous result for the operator norm of $\bm{S}_{2m}^{-1}$. Basic matrix algebra gives
\begin{equation}
\lVert \bm{S}_{2m}^{-1}\rVert_{op} =\left(\min_{\alpha\in\mathbb{R}^p} \alpha^T \bm{S}_{2m} \alpha\right)^{-1},
\label{eq::S2-inverse}
\end{equation}
and hence it suffices to bound the right hand side.
For any fixed $\alpha$ in the $p$-dimensional unit sphere, we have, using similar to the derivation in (\ref{eq::kernel-expect}):
\begin{gather*}
\E\left( n\cdot\alpha^T\bm{S}_{2m}\alpha\right) \geq m_1\cdot n\ubh_m^{p+2},\\
\var\left( n\cdot\alpha^T\bm{S}_{2m}\alpha\right) = O\left(n\bh_m^{p+4}\right),
\end{gather*}
for some constant $C_2$; the expectation is lower bounded since the $A_m$ covers a ball with radius $\ubh_m$. With Bernstein's inequality, similar to that used for $S_{1m}$, we have for any $\varepsilon > 0$: 
\[
\Pr\left( n\cdot|\alpha^T\bm{S}_{2m}\alpha - \E(\alpha^T\bm{S}_{2m}\alpha) | \geq n\bh_m^{p+2}\varepsilon \right) 
\leq  2\exp\left\{-C_2 n\bar{h}_m^{p}\varepsilon^2\right\},
\]
and hence for any sufficiently large $M_2 > 0$,
\begin{eqnarray*}
&&\Pr\left(\frac{ n\bh_m^{p+2}}{n\cdot\alpha^T\bm{S}_{2m}\alpha }\geq M_2 \right) \\
&=&
\Pr\left(n\cdot\alpha^T\bm{S}_{2m}\alpha \leq \frac{1}{M_2}  n\bh_m^{p+2} \right)  \\
&\leq&  
\Pr\left(n\cdot |\alpha^T\bm{S}_{2m}\alpha - \E(\alpha^T\bm{S}_{2m}\alpha) | \geq n\cdot\E[\alpha^T\bm{S}_{2m}\alpha]  - \frac{1}{M_2}n\bh_m^{p+2} \right) \\
&\leq &\Pr\left(n\cdot |\alpha^T\bm{S}_{2m}\alpha - \E(\alpha^T\bm{S}_{2m}\alpha) | \geq m_1  n\bh_m^{p+2}/2 \right) \\
&\leq& 2\exp\{-C_2n\bh_m^{p}m_1^2/4\},
\end{eqnarray*}
since $\bh_m/\ubh_m$ is uniformly bounded and the expectation is bounded from below.

Next, applying the same covering argument again, and note the relationship from (\ref{eq::S2-inverse}), we have that 
\begin{eqnarray*}
\Pr\left(\sup_{m=1,\ldots,M}\left\lVert \bh_m^{p+2} \bm{S}_{2m}^{-1}\right\rVert_{op} \geq M_2 \right) 
&\lesssim& \exp\{\log M + p\log 2 -C_2n\ubh^{p}m_1^2/4\} \lesssim \frac{1}{n^3},
\end{eqnarray*}
implying
\[
\sup_{m=1,\ldots,M}\left\lVert \bh_m^{p+2} \bm{S}_{2m}^{-1}\right\rVert_{op} =  \Op(1). 
\]
Therefore the second statement follows by combining  the norm bounds for $S_{1m}$ and $\bm{S}_{2m}^{-1}$. In particular,
\begin{eqnarray*}
&&\Pr\left(\sup_m \lVert \bh_m \bm{S}_{2m}^{-1}S_{1m}\rVert \geq \frac{1}{2}\right)\\
&\leq& \Pr\left(\sup_m \left\lVert (\bh_m^{p+1})^{-1}S_{1m}\right\rVert \geq \frac{1}{2M_2}\right) + \Pr\left(\sup_m \left\lVert \bh_m^{p+2} \bm{S}_{2m}^{-1}\right\rVert_{op} \geq M_2\right)\\
&\lesssim& \frac{1}{n^3}.
\end{eqnarray*}

For the third statement, we give a bound for
\[
S_{0m} = n^{-1}\sumn\bm{1}[X_i\in A_m],
\]
similar to what we did in the second statement. Note that
\[
\E[S_{0m}] = \Pr(X\in A_m) \gtrsim \ubh_m^p,\quad \var[S_{0m}] \leq \Pr(X\in A_m) \lesssim \bh_m^p,
\]
since the density of $X$ is bounded. Therefore for small enough $\varepsilon_0 > 0$, an application of Bernstein's inequality gives
\begin{eqnarray*}
\Pr\left( S_{0m} \leq \varepsilon_0 \ubh_m^{p} \right) 
&\leq&
\Pr\left( S_{0m} - \E[S_{0m}] \leq -\varepsilon_0 \ubh_m^{p}/2 \right) \\
&\leq&
 \exp\left\{-C_3 n\bh_m^p \varepsilon_0^2\right\},
\end{eqnarray*}
for some constant $C_3 > 0$. Taking a union bound with all $m = 1,\ldots,M$ shows
\[
\Pr\left(\inf_m \frac{S_{0m}}{\bh_m^p} \leq \varepsilon_0 \right) \leq \sumM \exp\left\{-C_3n\bh_m^p\varepsilon_0^2\right\} \leq \frac{1}{n^3},
\] 
for sufficiently large $n$ with the bandwidth in Condition \ref{cond::bin-bandwidth}. The proof is now complete.

\end{proof}

\subsection{Proof of Lemma \ref{lemma::proof-U13}}
\begin{proof}
We only give detailed proof for item 1; the conclusion for item 2 holds similarly and we only give an outline. We prove the conclusion specifically for $a_n = r_n$ in Proposition \ref{prop::condI-1} and $b_n = g_{1n}$ as in Condition \ref{cond::bin-Qt}.
We use the same notations in (\ref{eq::NO-smallu}) and define
\begin{gather*}
u_{1i}(s) = [Y_i - q(s,X_i)]\left[\bm{1}\{Y_i \geq \hat{q}(s,X_i)] - \bm{1}\{Y_i \geq q(s,X_i)\}\right],\\
u_{3i}(s) = \left(q(s,X_i) - \hat{q}(s,X_i)\right)\cdot \left[\bm{1}\{Y_i \geq \hat{q}(s,X_i)\} - \bm{1}\{Y_i \geq q(s,X_i)\}\right];
\end{gather*}
Correspondingly the left-hand sides in Lemma \ref{lemma::proof-U13} can be written as
\[
U_{1n}(s,m) = \frac{\sumn w_{im}\kappa_{im}u_{1i}(s)}{\sumn w_{im}},\quad U_{3n}(s,m)=\frac{\sumn w_{im}\kappa_{im}u_{3i}(s)}{\sumn w_{im}}.
\]
Moreover, let
\[
R_q = \supMS |\hat{q}(s,X_i) - q(s,X_i)| = \Op(g_{1n}),
\] 
as in Condition \ref{cond::bin-Qt}.

We consider the decomposition:
\begin{eqnarray*}
|U_{1n}(s,m)| &\lesssim_{\P}& (nS_{0m})^{-1} \sumn w_{im}[y_i - q(s,X_i)]\bm{1}\{q(s,X_i)\leq y_i \leq \hat{q}(s,X_i)\}\\
&&\;\; + (nS_{0m})^{-1}\sumn w_{im}[q(s,X_i) - y_i]\bm{1}\{\hat{q}(s,X_i)\leq y_i \leq q(s,X_i)\}\\
&\leq& (nS_{0m})^{-1}\sumn w_{im} [y_i - q(s,X_i)]\bm{1}\{q(s,X_i)\leq y_i\leq q(s,X_i) + R_q\}\\
&&\;\; + (nS_{0m})^{-1}\sumn w_{im}[q(s,X_i) - y_i]\bm{1}\{q(s,X_i) - R_q \leq y_i \leq q(s,X_i)\}\\
&\triangleq& U^{(+)}_{1n}(s,m) + U^{(-)}_{1n}(s,m),
\end{eqnarray*}
where we use a constant to upper bound $|\kappa_{im}|$ (see Claim \ref{claim6} in the proof of Proposition \ref{prop::condI-1}); and the second inequality holds by monotonicity of the indicator functions. By symmetry, hereafter we focus on the term $U^{(+)}_{1n}(s,m)$.

Let $s_{-} = \tau - Br_n$, $s_{+} = \tau + Br_n$. For any $s\in[s_{-},s_{+}]$ we have
\begin{eqnarray*}
0\leq U^{(+)}_{1n}(s,m) 
&\leq&(nS_{0m})^{-1}\sumn w_{im}[y_i - q(s_{-},X_i)]\bm{1}\{q(s_{-},X_i)\leq y_i \leq q(s_{+},X_i) + R_q\}\\
&\triangleq& \bar{U}_{1n}^{(+)}(k),
\end{eqnarray*}
Therefore we can drop the supremum over $s$ by relying on $\bar{U}_{1n}^{(+)}(k)$, we bound its expectation and centered process separately. In what follows we bound the centered empirical process and the expectation of $\bar{U}_{1n}^{(+)}(k)$ separately.

We give a tail bound for $\E[\bar{U}_{1n}^{(+)}(k)] - \bar{U}_{1n}^{(+)}(k)$ using Hoeffding's inequality (conditional on $X$) and a union bound. 
By Condition \ref{cond::R-bin-Y-density}, each summand in $\bar{U}_{1n}^{(+)}(k)$ is bounded by $[y_i - q(s_{-},x)]\bm{1}\{q(s_{-},x)\leq y_i\leq q(s_{+},x) + R_q\} \lesssim ( s_{+} - s_{-}) + R_q $ for all $x$. 
For any $\delta_1 > 0$, there exists a large enough $M_1 > 0$ that
\begin{eqnarray*}
&&\Pr\left( \sup_{m=1,\ldots,M} \left|\bar{U}_{1n}^{(+)}(k) - \E[\bar{U}_{1n}^{(+)}(k)] \right|\geq M_1(g_{1n} + r_n)\Bigg| X\right)\\
&\leq& 2\exp\left\{ \log n - 2nM_1^2\cdot \inf_mS_{0m}\right\} + \delta_1,
\end{eqnarray*}
where the $\delta_1$ comes from the probability that $R_q \geq M_1g_{1n}$. Similar to how we obtain (\ref{eq::tail-centered-truncated}), the following unconditional  tail bound holds from the above displayed conditional bound:
\begin{eqnarray*}
\Pr\left( \sup_{m=1,\ldots,M} \left|\bar{U}_{1n}^{(+)}(k) - \E[\bar{U}_{1n}^{(+)}(k)] \right|\geq M_1(g_{1n} + r_n)\right)
&\lesssim&
2\delta_1.
\end{eqnarray*}

Here we bound the expectation $\E[\bar{U}_{1n}^{(+)}(k)]$. By Condition \ref{cond::bin-Y-new} we have
\begin{eqnarray*}
&&\E[Y_i - q(s_{-},X_i)]\bm{1}\{q(s_{-},X_i)
\leq Y_i\leq q(s_{+},X_i) + R_q\} \\
& =  &
\E[Y_i - q(s_{-},X_i)]\bm{1}\left\{0 \leq [Y_i - q(s_{-},X_i)]\leq R_q + \underline{f}^{-1}|s_{+} - s_{-}|\right\} \\
&=&
O\left((g_{1n} + r_n)^2\right),
\end{eqnarray*}
since $s_+ - s_- \lesssim r_n$ and $R_q = \Op(g_{1n})$.
Therefore
\[
\supM \E[\bar{U}^{(+)}_{1n}(k)] = \Op\left((g_{1n} + r_n)^2\right).
\]

Combining the bounds for the expectation and the centered empirical process, we arrive at
\[
\supMS U_{1n}^{(+)}(s,m) \lesssim_{\P} \supM \bar{U}^{(+)}_{1n}(k)  = \Op\left((g_{1n} + r_n)^2\right).
\]
Repeating the same procedure for $U_{1n}^{(-)}(s,m)$ would complete the proof for the first item of the Lemma.

For $U_{3n}(s,m)$, note 
\[
|U_{3n}(s,m)| \lesssim_{\P} R_q(nS_{0m})^{-1}\sumn w_{im}\bm{1}[q(s,X_i) - R_q \leq y_i \leq q(s,X_i) + R_q];
\]
therefore we can follow the same line of reasoning and establish
\[
\supMS |U_{3n}(s,m)| = \Op\left((g_{1n} + r_n)^2\right).
\]
The proof is now complete.

\end{proof}

\section{Discussions on initial ES estimator for the i-Rock approach}\label{append::init_ES}
\subsection{On the non-linearity of the initial ES estimator}\label{append::subsec::init_est}
To obtain i-Rock estimator, we need initial estimators of the ES process $\{v(s,X), s\in(\tau-\delta \tau, \tau + \delta(1-\tau))\}$. However, if the initial ES estimator at quantile level $\tau$ is linear in $X$ and $\hat{v}(s,x)$ is monotone in $s$, the i-Rock estimator will 
coincide with the initial ES estimator at quantile level $\tau$, as detailed in Proposition~\ref{prop::linear_ES}. With this result, we opt for non-parametric initial ES estimators for the i-Rock approach. 
\begin{proposition}\label{prop::linear_ES}
    Let $\hat v(\alpha,\cdot)$ denotes an initial estimator of $v(\alpha,\cdot)$, for all $\alpha \in [0,1]$.
    If $\hat v(\tau,x) = x^T \xi$ and $\hat v(s,x)$ is increasing in $s$\footnote{When $\hat v(s,x)$ is increasing in $s$, the solution to~\eqref{eq::linear_init} is not unique but the true coefficient $\xi$ is one of the solutions. If $\hat v(s,x)$ is strictly increasing in $s$, then the solution to ~\eqref{eq::linear_init} is unique and equals to $\xi$.}, then
    \begin{equation}\label{eq::linear_init}
         \xi \in \argmin_{u} \sum_{i=1}^n \hat \gamma_i \int_{0}^1 \rho_{\tau} (\hat v(\alpha,x_i) - x_i^T u) d \alpha,
    \end{equation}
    for any user-specified weights $\hat \gamma_i >0,i=1,\ldots,n$.
\end{proposition}

\begin{proof}[Proof of Proposition~\ref{prop::linear_ES}]
    Define a random variable $S_i = \hat v(U,x_i)$ where $U$ follows a Uniform distribution between 0 and 1. Suppose $(x_i,y_i)$ are fixed, i.e., the only randomness of $S_i$ comes from $U$, we have
    \begin{equation*}
        \begin{aligned}
            \argmin_{q} \int_{0}^1 \rho_{\tau} (\hat v(\alpha,x_i) - q) d \alpha &= \argmin_{q} E_U(\rho_{\tau} (S_i - q))\\
            &= q_{[S_i]}(\tau)\\
            &= \hat v(\tau,x_i)\\
            &=x_i^T \xi,
        \end{aligned}
    \end{equation*}
    where the first equality follows from the distribution of $S_i$, the second equality follows from that the quantile is the minimizer of the expected quantile loss function, and the third equality follows from the monotonicity of $\hat v$. 
    The conclusion holds since $\xi$ serves as the minimizer of every term in \eqref{eq::linear_init}. 
\end{proof}
\subsection{On the monotonicity of the initial ES estimator}

\label{subsec::monotonicity}

In Condition \ref{cond::vhat1}, we require the initial ES estimator $\hat{v}(s,\bxm)$ to be monotonically increasing in $s\in(0,1)$. 
The monotonicity can be achieved by re-arrangement of a given estimator \citep{chernozhukov2009improving,chernozhukov2010quantile}. Here we provide a technical perspective that suggests monotonicity may not be necessary for Theorem \ref{thm::asymptotic-RRM-bin}.

Even when the initial ES estimators are not monotone, the i-Rock approach still has a clear interpretation as finding the $\tau$th `re-arranged' ES \citep{chernozhukov2009improving}.
To see this, consider the univariate case with no covariate as an illustrative example. Following (\ref{eq::mRock-est-discrete}), the i-Rock approach solves:
\begin{equation}
  \min_{C}\int_0^1\rho_\tau\left(\hat{v}(s) - C\right)\d s \approx \frac{1}{J}\min_C\sum_{j=1}^J\rho_\tau\left(\hat{v}(s_j) - C\right),
  \label{eq::Rock-monotone-illustrate}
\end{equation}
where we discretize the integral above as a grid $s_1,\ldots,s_J\in(0,1)$. The solution to (\ref{eq::Rock-monotone-illustrate}) is approximately the $\tau$-th quantile of the (possible unordered) set $\{\hat{v}(s_1),\ldots,\hat{v}(s_J)\}$.
Operationally, the i-Rock approach gives exactly the $\tau$-th \textit{monotonically re-arranged} superquantile in \citet{chernozhukov2009improving}.

Following this insight, we now demonstrate that the proof of Theorem \ref{thm::asymptotic-RRM-bin} may adapt to situations where $\hat{v}(s,\bxm)$ is not  monotonic.
As we've demonstrated in the proof of Theorem \ref{thm::asymptotic-RRM-bin}, central to the main result is the asymptotic properties of $\hat{h}(\cdot,\tilde{x}_m)$; Without monotonicity, $\hat{h}(\cdot,\tilde{x}_m)$ is defined by
\[
\hat{h}(z,x) := \int_0^1\bm{1}\{\hat{v}(s,x) \leq z\}\,\d s = \sup\{s\in[0,1]: \hat{v}(s,x) \leq z\}.
\]
The functional $\hat{h}(\cdot,\tilde{x}_m)$ is the \textit{monotonized inverse} operator in \citet{chernozhukov2010quantile}; note when $\hat{v}(\cdot,\tilde{x}_m)$ is indeed monotonic, $\hat{h}(\cdot,\tilde{x}_m)$ reduces to the classic inverse operator defined in (\ref{eq::def-h}). 
Corollary 3 of \citet{chernozhukov2010quantile} establishes the Hadamard differentiability of $\hat{h}(\cdot,\tilde{x}_m)$, and shows that its asymptotic property does not rely on the (finite-sample) monotonicity of $\hat{v}(\cdot,\tilde{x}_m)$.
With some technical modification, we expect that their proof can be adapted to our setting, therefore our main result, i.e., Lemma \ref{lemma::bin-technical-h} and Theorem \ref{thm::asymptotic-RRM-bin}, can be established without the monotonicity requirement in Condition \ref{cond::vhat1}.

\begin{remark}
As yet another technical solution, one may pursue the following strategy: first find $J$ equally-spaced grid points $0<\tau_1<\ldots<\tau_J < 1$ that spans the interval $[0,1]$. Then we can estimate the initial ES on the grid with linear interpolations in between the grid points. The monotonicity follows with probability going to $1$ provided that $r_n \ll (\tau_{j+1} - \tau_j) \ll n^{-1/4}$, where $r_n$ is in Condition \ref{cond::vhat1}.
\end{remark}

\subsection{On the bias in the initial ES estimator}
\label{subsec::bias}
Here we illustrate the importance to control the bias in the initial ES estimator. Consider the following example with fixed design in a unit cube $[0,1]^p$ (excluding intercept); and we define the bins as hypercubes with edge length $\bh$ and therefore we have $M \leq \lceil \bh^{-p}\rceil$ total bins. Suppose we use a standard Nadaraya-Watson type estimator for the initial ES in each bin $m$, and the $\tau$-th ES estimator can be represented as $\hat{v}_m(\tau) - v_m(\tau) = B_m + U_m$, where $\E[U_m] = 0$ and $B_m$ is the bias.

We consider the plausibility of Condition \ref{cond::bin-linear}. From the results in \citep{kato2012weighted}, $B_m = \Op(\bh^2)$ and $U_m = \Op((n\bh^p)^{-1/2})$. Because each $\hat{v}_m$ are based on local observations in disjoint bins, $(B_m,U_m)$ is independent across $m = 1,\ldots,M$. Therefore, the aggregation for $U_m$ gives
\[
\sqrt{n}\sum_{m=1}^M U_m = \sqrt{nM}\cdot\Op\left(\frac{1}{\sqrt{nh^{p}}}\right) = \Op(1),
\]
by the Central Limit Theorem, provided that $\sqrt{n}\bh^p \to 0$.
On the other hand, for the bias terms we have
\[
\sqrt{n}\sum_{m=1}^M B_m  = \sqrt{n}M\cdot\Op(\bh^2) = \Op\left(\sqrt{n} h^{2-p}\ \right);
\]
the order of the above sum goes to infinity if $p > 1$, hence the bias dominates in the aggregation of $\hat{v}_m$ in Condition \ref{cond::bin-linear}. Therefore, it is critical to reduce the bias when constructing the initial ES estimator.

\section{Numerical investigations}
\subsection{Implementation}\label{subsec::imple}
\begin{algorithm}[tb]
\caption{
The i-Rock estimation of the $\tau$-th ES regression with discrete covariates.
}
\label{alg::mRock_dis}
\begin{algorithmic}[1]
\setlength{\abovedisplayskip}{3pt} 
\setlength{\belowdisplayskip}{0pt} 
\setlength{\abovedisplayshortskip}{0pt} 
\setlength{\belowdisplayshortskip}{0pt} 

\STATE Input:
$\{(x_m,Y_{mj}):\; j = 1,\ldots, n_m;  \,m = 1,\ldots,M\}$, $\tau$, $\delta = 0.5$ (default). 
\STATE Form an equally-spaced grid over the interval $[\tau - \delta\tau,\,\tau + \delta(1-\tau)]$ as
\[
\tau - \delta\tau = s_0< s_1 <\ldots < s_J  = \tau + \delta(1-\tau).
\]
\FOR{$m=1$ to $M$}
\FOR{$j=0$ to $J$}
\STATE 
\label{alg::step-discrete} 
Obtain the initial ES estimator at level $s_j$, 
\[
\hat{v}_{mj} \gets \ddfrac{\sum_{j=1}^{n_m} Y_{mj}\bm{1}\{Y_{mj} \geq \hat{q}(s,x_m)\}}{(1-s)n/M}
\]
\ENDFOR
\ENDFOR
\STATE Solve the (approximate) optimization problem via quantile regression
\[
\begin{aligned}
\widehat{\theta} \;\;&\gets \;\;\min_{\theta\in\mathbb{R}^{p+1}} \sum_{m=1}^M n_m \int_{\tau - \delta\tau}^{\tau + \delta(1-\tau)}\rho_\tau\left(\hat{v}\left(s, x_m\right) - x_m^T\theta\right)\d s.\\
&\approx  \;\; \min_{\theta\in\mathbb{R}^{p+1}}  \frac{1}{1+J}\sum_{m=1}^M \sum_{j=0}^J 
n_m \rho_\tau\left(\hat{v}_{ij} - x_m^T\theta\right).
\end{aligned}
\]
\end{algorithmic}
\end{algorithm}

\begin{algorithm}[!h]
\caption{
The i-Rock estimation of the $\tau$-th ES regression with continuous or mixed covariates.
}
\setlength{\abovedisplayskip}{3pt} 
\setlength{\belowdisplayskip}{0pt} 
\setlength{\abovedisplayshortskip}{0pt} 
\setlength{\belowdisplayshortskip}{0pt} 
\label{alg::mRock1}
\begin{algorithmic}[1]
\STATE Input: $\{(X_i, Y_i): i=1,\ldots,n\}$, $\tau$, $\delta = 0.5$ (default). 
\STATE \label{alg::step-interval} Form an equally-spaced grid over the interval $[\tau - \delta\tau,\,\tau + \delta(1-\tau)]$ as
\begin{equation*}
    \tau - \delta\tau = s_0< s_1 <\ldots < s_J  = \tau + \delta(1-\tau).
\end{equation*}
\vspace{-12pt}
\STATE Partition each continuous covariate by equally-spaced quantiles with $k = \lceil 0.5 \sqrt{p} \times n^{1/(2p)}\rceil$ bins, and each discrete covariate by distinct values. The disjoint bins $B_m,m=1,\ldots,M$ are created from all combinations of the partitions of the covariates.
\FOR{$m=1$ to $M$}
\STATE Calculate the geometric center $\bar x_i$ of bin $B_m$ by combining the mean of the two ends for continuous covariates and the unique value of the discrete covariates within bin $B_m$. 
\STATE Find the closest observation to the geometric center of bin $B_m$:
\begin{equation*}
    X_{(m)} = \argmin_{x \in B_m} \left\| x - \bar x_m \right\|_2
\end{equation*}
\STATE Let $j_0$ be such that $s_{j_0} \leq \tau - 0.5\delta\tau < s_{j_0 + 1}$.
\FOR{$j=j_0$ to $J$}
\STATE \label{alg::step-qt} Obtain the conditional quantile estimator at level $s_j$,
\[
 \hat{q}^{(m)}(s_j,x),\quad x\in B_m.
\]
\vspace{-12pt}
\STATE \label{alg::step-sq} Obtain the initial ES estimator at level $s_j$ from~\eqref{eq::ll-SQ} where $\bm{{X}}_m$ takes only continuous covariates, with quantile estimates $\hat{q}^{(m)}(s_j,x)$, 
\[
\hat{v}_{mj} \gets \hat{v}\left(s_j,X_{(m)}\right).
\]
\vspace{-12pt}
\ENDFOR
\ENDFOR
\STATE \label{alg::step-winsorization} Fill $\{\hat{v}(s_j,x):s_j \in [\tau - \delta \tau,\tau - 0.5\delta \tau)\}$ with the constant $\hat{v}(\lceil \tau - 0.5\delta \tau\rceil,x)$.
\STATE \label{alg::step-i-rock} Solve the (approximate) optimization problem via quantile regression
\[
\begin{aligned}
\widehat{\theta} \;\;&\gets \;\;\min_{\theta\in\mathbb{R}^{p+1}} \sum_{m=1}^M \hat{\gamma}_m \int_{\tau - \delta\tau}^{\tau + \delta(1-\tau)}\rho_\tau\left\{\hat{v}\left(s, X_{(m)}\right) - X_{(m)}^T\theta\right\}\d s.\\
&\approx  \;\; \min_{\theta\in\mathbb{R}^{p+1}}  \frac{1}{1+J}\sum_{m=1}^M \sum_{j=0}^J\hat{\gamma}_m\rho_\tau\left(\hat{v}_{mj} - X_{(m)}^T\theta\right),
\end{aligned}
\]
where $\hat{\gamma}_m$ is in (D.2) of the online Supplementary Material. 
\end{algorithmic}
\end{algorithm}

For discrete covariates, we provide an implementation of the i-Rock approach summarized in Algorithm~\ref{alg::mRock_dis} based on Section~\ref{sec::discrete} and Corollary~\ref{coro::truncation}. 
For continuous or mixed covariates, we adopt a variant of the formulation in~\eqref{eq::estimator-rock-binning} summarized in Algorithm~\ref{alg::mRock1}, which uses the bin-wise local-linear initial ES estimator introduced in Section~\ref{subsec::examples}.
Below we comment on several aspects of our implementation. 
We partition the covariate space by binning each covariate. Discrete covariates are naturally partitioned according to their distinct values, while continuous covariates are divided using breakpoints at equally spaced quantiles.
For the subsequent experiments, we set the number of bins for each continuous covariate as $k = \lceil 1.6 \sqrt{p} \times \{\sqrt{n}/\log(n)\}^{1/p}\rceil$ in our simulation study. 

In Algorithm \ref{alg::mRock1}, we obtain initial ES estimators only at quantile levels in $[\tau - 0.5\delta \tau, \tau + \delta (1 - \tau)]$, due to Corollary~\ref{coro::truncation} and the use of winsorization \citep{wilcox2005trimming}, which extrapolates the initial ES estimators with a constant into the lower tails.
In Step~\ref{alg::step-winsorization} of Algorithm~\ref{alg::mRock1}, we use $50\%$ left-winsorization, i.e., setting $\{\hat{v}(s,x):s \in [\tau - \delta \tau,\tau - 0.5\delta \tau)\}$ to a constant $\hat{v}(\tau - 0.5\delta \tau,x)$. 
The shrinkage of the estimating interval for initial ES estimators not only significantly reduces the computational costs, but improves statistical accuracy by avoiding the ES estimation at extreme quantile levels that may have high variability. 
In our experiments, we fix $\delta=0.5$ and the number of intervals 
$J = \lceil \sqrt{70n\log(n)}\rceil$.

\subsection{Asymptotic normality check for Section 5}
We provide additional numerical results to Section 5 of the main manuscript. 
Under Model (25) of the main manuscript, we verify that the sampling distribution of the i-Rock estimator $(\hat \beta_0,\hat \beta_1,\hat \beta_2)$ follows a normal distribution very closely.  We test whether $\hat \beta_0/\sigma_0^*$, $\hat \beta_1/\sigma_1^*$, and  $\hat \beta_2/\sigma_2^*$ follow the standard normal distribution, respectively, where $(\sigma_0^*)^2$, $(\sigma_1^*)^2$, and $(\sigma_2^*)^2$ are the theoretical asymptotic variances calculated from (24) of the main manuscript for $\beta_0$, $\beta_1$, and $\beta_2$, respectively. To this end, we 
perform the Kolmogorov–Smirnov (KS) test comparing against the standard normal distribution. 
For the i-Rock estimator with the B-spline quantile function, the p values from the KS test across different quantile levels and sample sizes are summarized in Table~\ref{tab::2dim_KS}.
To verify that these p values are uniformly distributed on $(0,1)$, another KS test was performed for the p values comparing against $U(0,1)$, resulting in a p-value of 0.994. This result empirically validates the asymptotic normality of the i-Rock estimator and its asymptotic variance in (24) of the main manuscript.
\begin{table}[!htp]
    \centering
    \begin{tabular}{cc|ccc}
$\tau$ & $n$ &     $\hat \beta_0/\sigma_0$ &  $\hat \beta_1/\sigma_1$ &  $\hat \beta_2/\sigma_2$ \\
\hline
0.9 & 10000 &   0.018&
 0.661&
 0.167\\
0.9 & 5000 & 0.738&
 0.556&
 0.988 \\
0.8 & 10000 & 0.162&
 0.319&
 0.724\\
0.8 & 5000 & 0.456&
 0.238&
 0.804\\
\end{tabular}
    \caption{The p values from the KS test against the standard normal distribution for the i-Rock estimators divided by their theoretical asymptotic standard error under Model (25).}
    \label{tab::2dim_KS}
\end{table}

Under Model (26), we present the Q-Q plots in Figure~\eqref{fig::2d_qq} to verify the asymptotic normality for the i-Rock estimators with B-spline quantile function. 
 \begin{figure}[t]
    \centering
\includegraphics[scale=0.4]{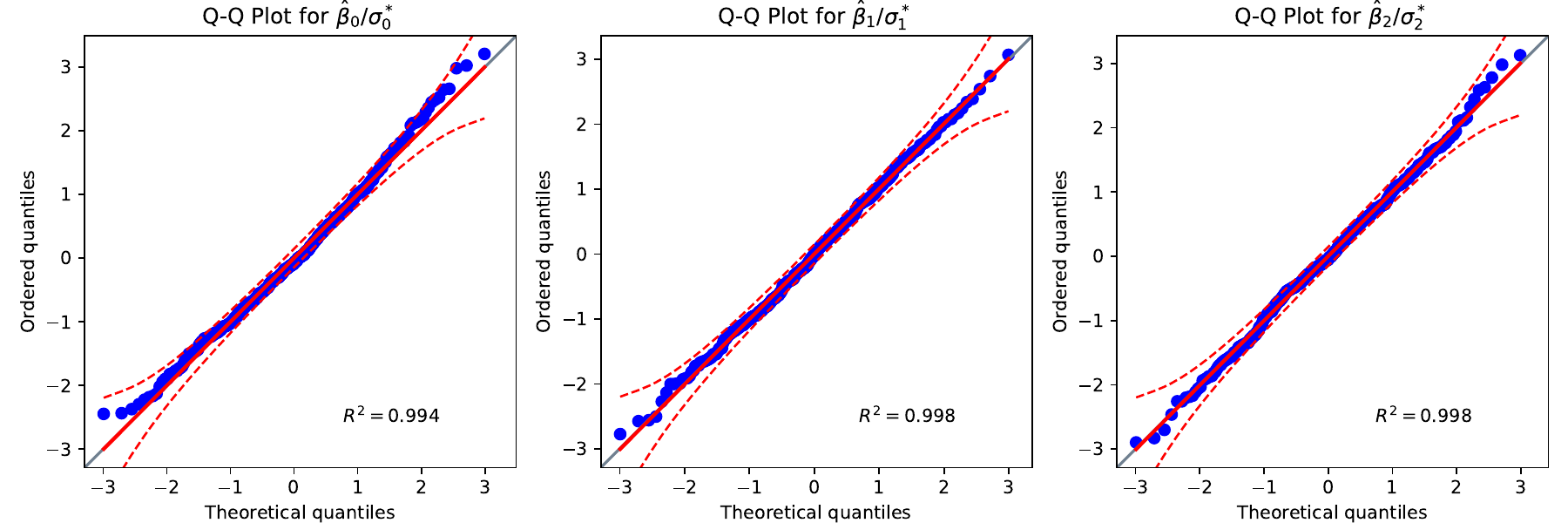}
    \caption{Q-Q plot with 99\% confidence interval (in red dashed curves) under Model (26) for the i-Rock estimators with the B-spline quantile model, normalized by the theoretical asymptotic variances, namely, $\hat \beta_0/\sigma_0^*$ and $\hat \beta_1/\sigma_1^*$, when $\tau=0.9$ and $n=10000$.}
    \label{fig::2d_qq}
\end{figure}

\subsection{Comparisons to quantile average approach}
We designed a numerical comparison with the ``quantile average approach", specifically averaging over a sequence of linear quantile estimators from quantile levels $\tau$ to $1$ with step size 0.002. Here, we compare our proposed estimator with the quantile average approach for the three cases considered in the main manuscript. 

In case 1, we consider a linear heteroscedastic model
\begin{equation}\label{eq::2d_cont_disc}
    Y_i = \{1 + U \} + (2 + 2 U) X_{i,1} + \{3 + 3 U\} X_{i,2}, \quad i=1,\ldots,n,
\end{equation}
where 
$U$ follows $U(0,1)$, $X_{i,1}$ follows $U(0,4)$, 
and $X_{i,2}$ are independently distributed from $\{0,1\}$ with equal probability. Here, both quantile and ES are linear in the covariates. Figure~\ref{fig::2d_cont_disc} shows the relative bias and the RMSE ratio comparisons for $\tau \in \{0.8,0.9\}$ and $n\in\{5000,10000\}$. 
The i-Rock estimators with both the linear and the B-spline quantile regression estimation outperform the two-step estimator in bias and RMSE in all settings. The i-Rock approach performs similarly to the quantile average approach, while the latter exhibits a larger bias.

\begin{figure}[htb]
\centering 
\begin{subfigure}[b]{0.48\textwidth}
         \centering
         \includegraphics[width = \textwidth]{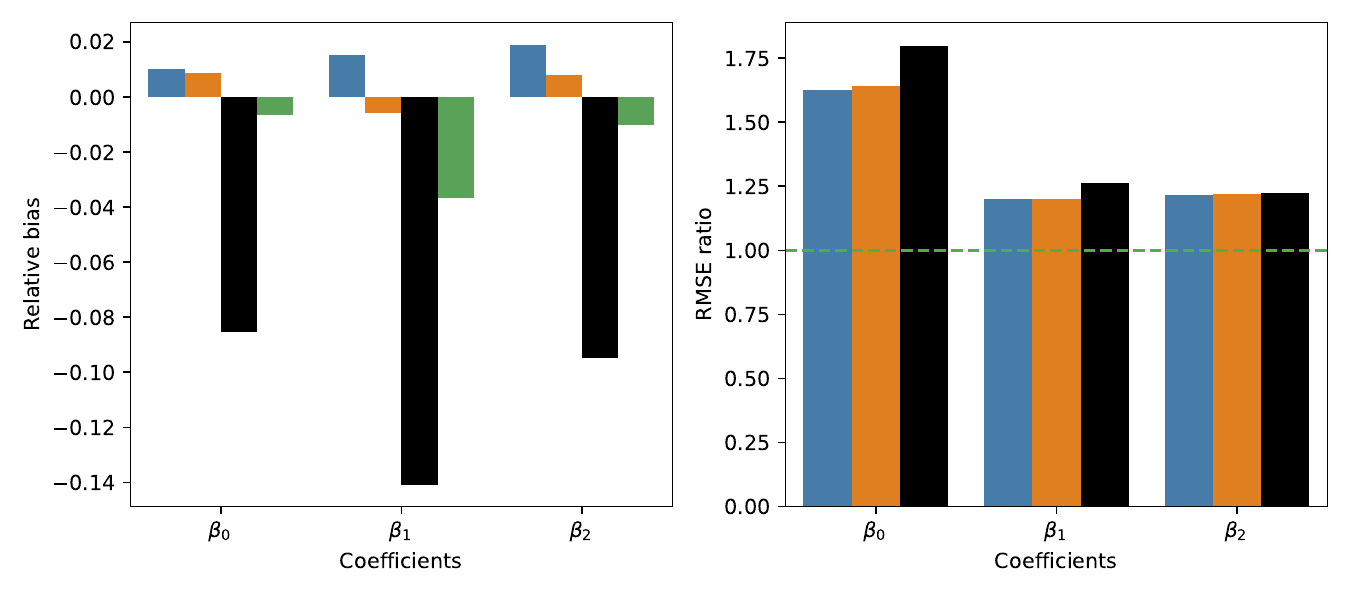}
         \caption{$n = 5000, \tau = 0.8$}
     \end{subfigure}
     \begin{subfigure}[b]{0.48\textwidth}
         \centering
         \includegraphics[width = \textwidth]{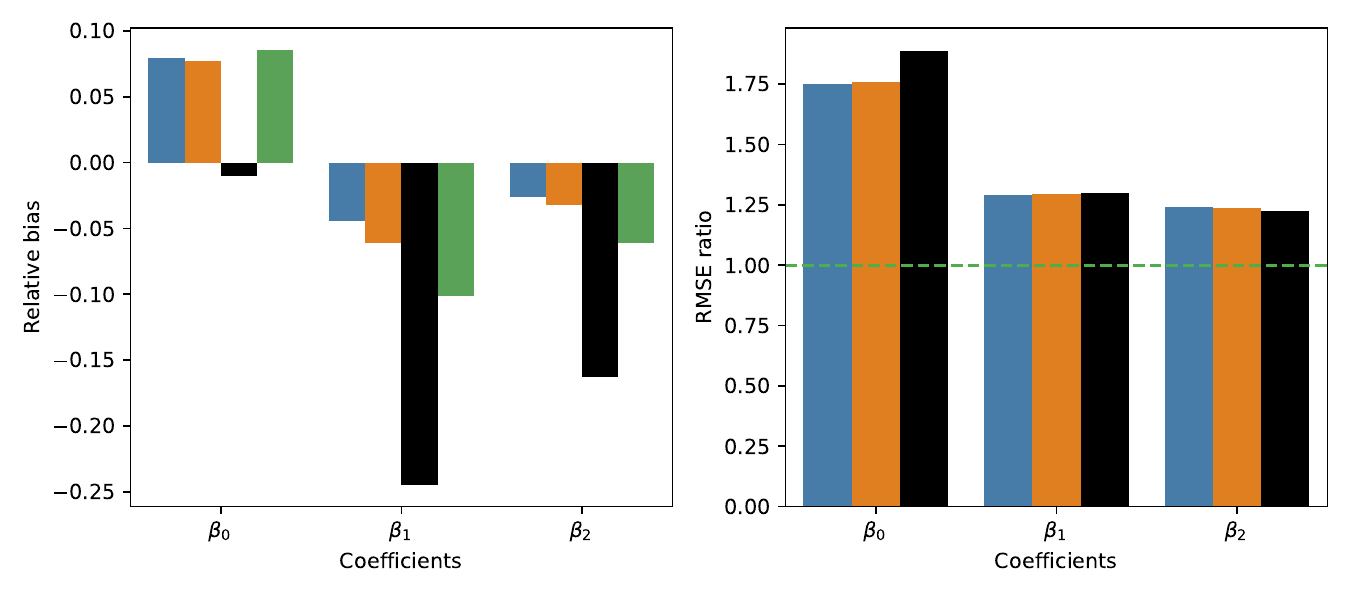}
         \caption{$n = 10000, \tau = 0.8$}
     \end{subfigure}
     \begin{subfigure}[b]{0.48\textwidth}
         \centering
         \includegraphics[width = \textwidth]{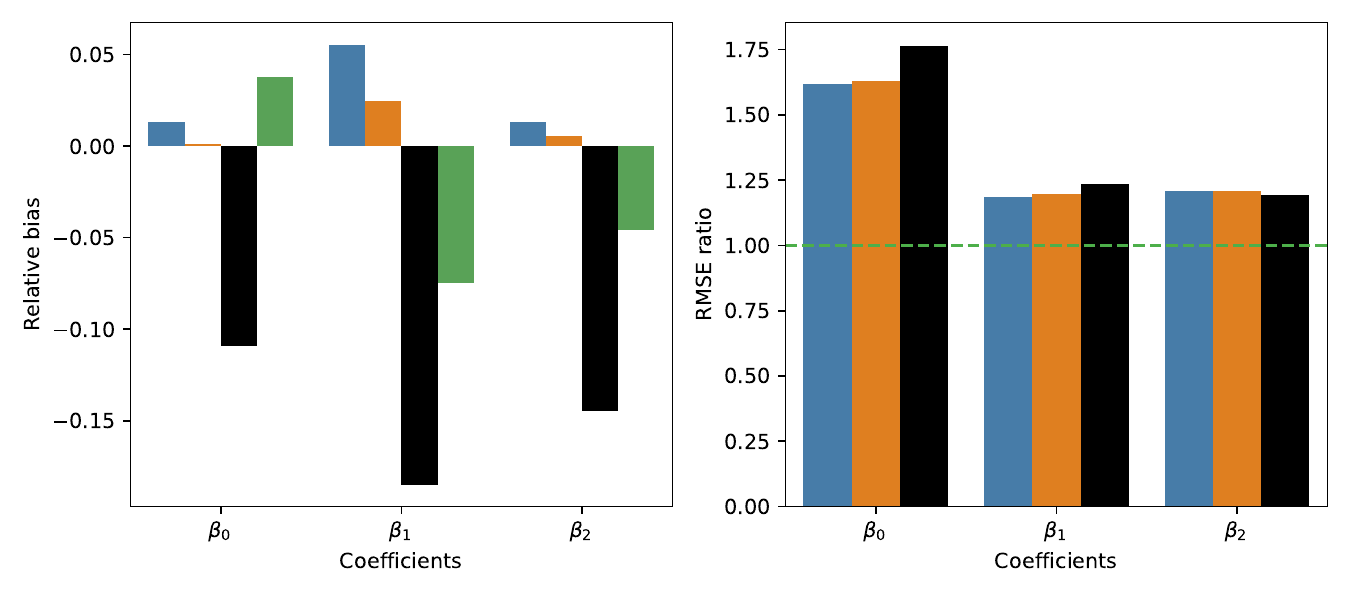}
         \caption{$n = 5000, \tau = 0.9$}
     \end{subfigure}
     \begin{subfigure}[b]{0.48\textwidth}
         \centering
         \includegraphics[width = \textwidth]{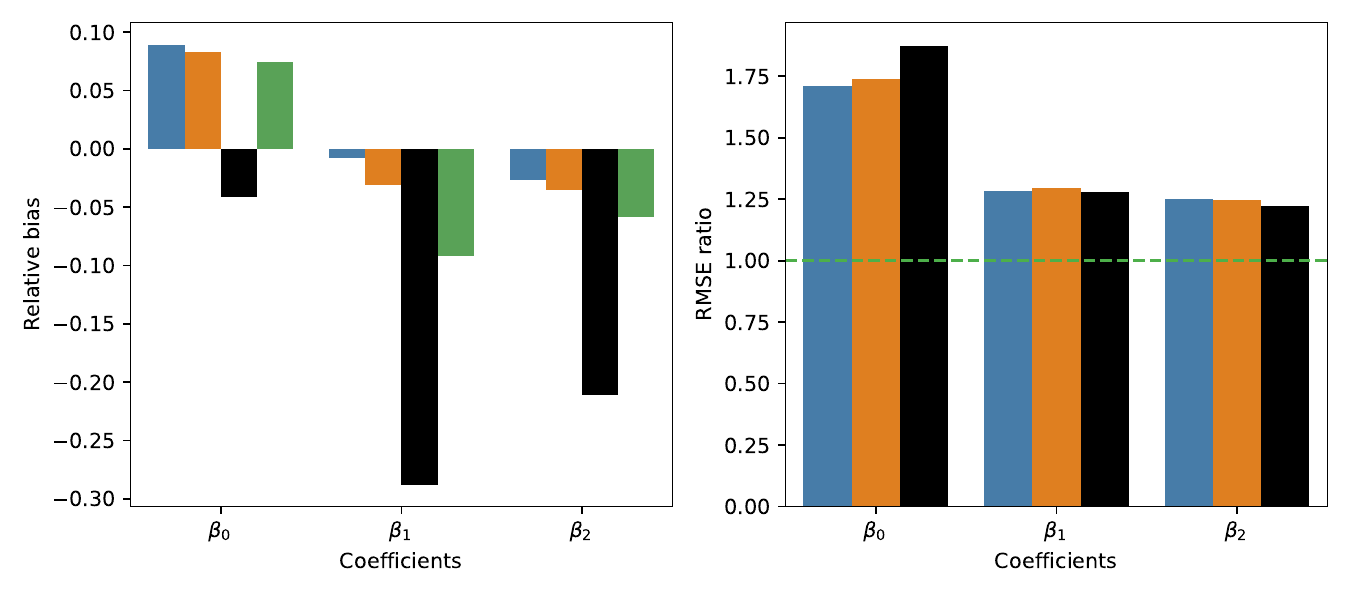}
         \caption{$n = 10000, \tau = 0.9$}
     \end{subfigure}
     \begin{subfigure}[b]{\textwidth}
         \centering
         \includegraphics[width = \textwidth]{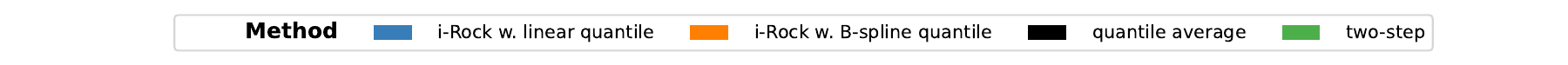}
     \end{subfigure}
\caption{Numerical comparisons of the i-Rock approach (with linear or B-spline quantile function estimation), quantile average approach, and two-step approach under linear heteroscedastic model~\eqref{eq::2d_cont_disc} at various quantile levels and sample sizes.}
\label{fig::2d_cont_disc}
\end{figure}

In case 2, we consider
\begin{equation}\label{eq::sim_2d_nonlinear}
Y_i = -1 + 2 X_{i,1} -3 X_{i,2} + \left(24 X_{i,1}^2 + 12 X_{i,2}^2 + 5\right) (\epsilon_i-\nu_0), \quad i=1,\ldots,n,
\end{equation}
where $(X_{i,1}, X_{i,2})$ is uniformly distributed in a two-dimensional square $[-1,2]^2$,
$\epsilon_i$ follows the skewed-$t_5$ distribution with skewness 2 \citep{theodossiou1998financial} that is independent of the covariates, and $\nu_0$ is the $0.9$-th ES of the distribution of $\epsilon_i$. In this case, the quantiles are non-linear but the $\tau$th ES is linear. 
Figure~\ref{fig:2d_nl} shows the results for $\tau=0.9$, $n = 10000$. The i-Rock with linear quantile function estimation, the quantile average estimator, and the two-step estimator suffer from the misspecification of the quantile function, while the i-Rock with the B-spline quantile function estimation is significantly less biased and more efficient than the rest of the approaches.  

\begin{figure}[tb]
    \centering
     \includegraphics[scale=0.5]{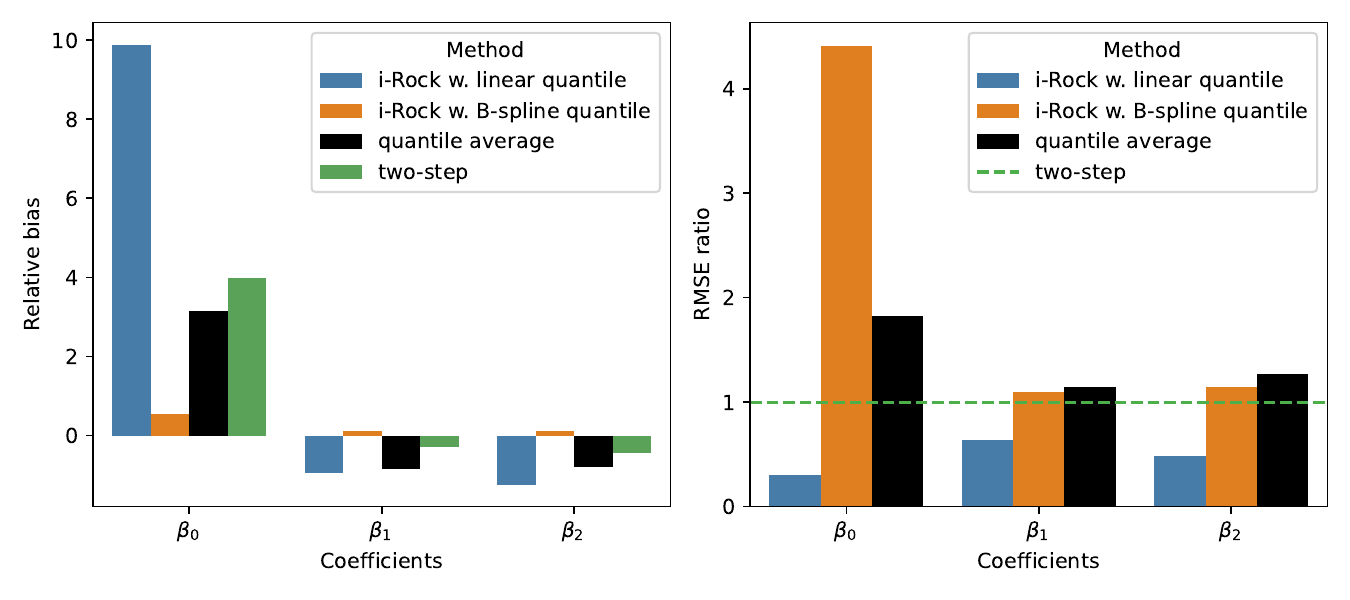}
    \caption{Numerical comparisons of the i-Rock approach (with linear or  B-spline quantile function estimation), quantile average approach, and two-step approach under Model~\eqref{eq::sim_2d_nonlinear} at $\tau=0.9$, $n = 10000$.
    }
\label{fig:2d_nl}
\end{figure}

In case 3, we consider a highly heterogeneous data-generating process
\begin{equation}\label{eq:2d}
    Y_i = \{1 - \log(1 - U)\} + (2 + 2 U) X_{i,1} + \{3 - 30 \log(1-U)\} X_{i,2}, \quad i=1,\ldots,n,
\end{equation}
where $U$ is uniformly distributed on $(0,1)$, and $(X_{i,1},X_{i,2})$ are independently distributed from $\text{binomial}(2, 0.5)$. 
The relative bias and RMSE ratios are summarized in Figure~\ref{fig::2d_disc} at several sample sizes. In this case, the i-Rock estimator 
is significantly more efficient than the quantile average and the two-step estimator, due to the automatic effective weighting schemes of the former. In addition, the quantile average estimator can be significantly more biased than the other two approaches. 

\begin{figure}[htb]
\centering 
     \begin{subfigure}[b]{0.48\textwidth}
         \centering
         \includegraphics[width = \textwidth]{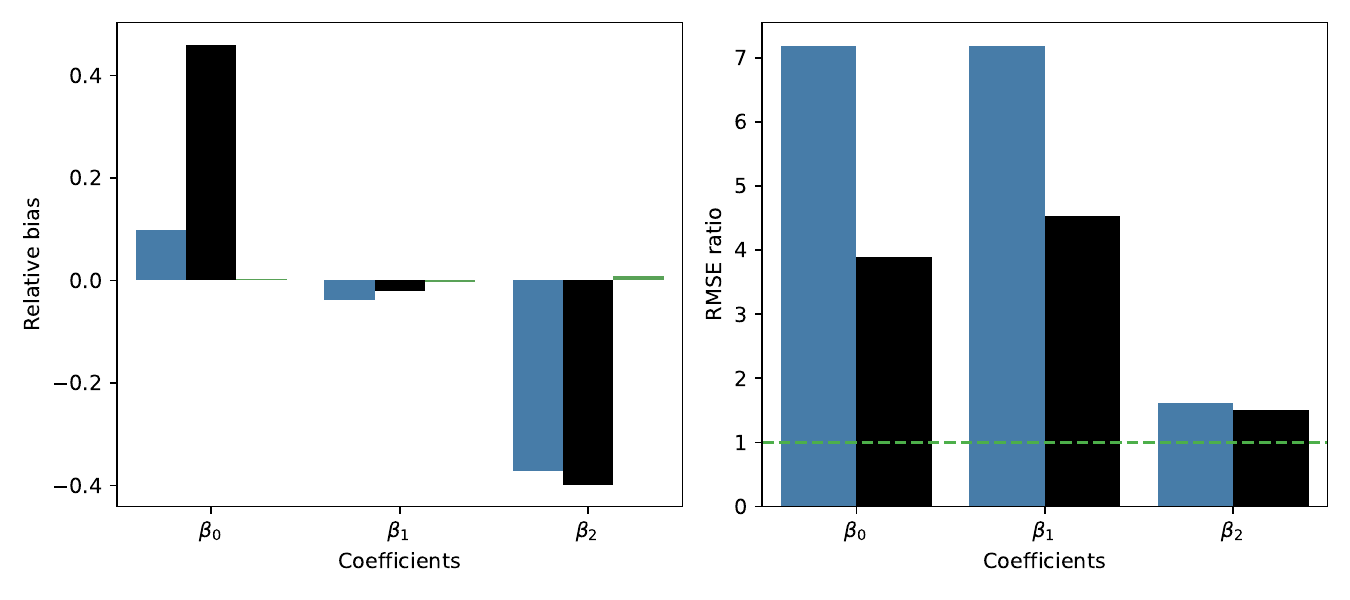}
         \caption{$n = 1000, \tau = 0.9$}
     \end{subfigure}
     \begin{subfigure}[b]{0.48\textwidth}
         \centering
         \includegraphics[width = \textwidth]{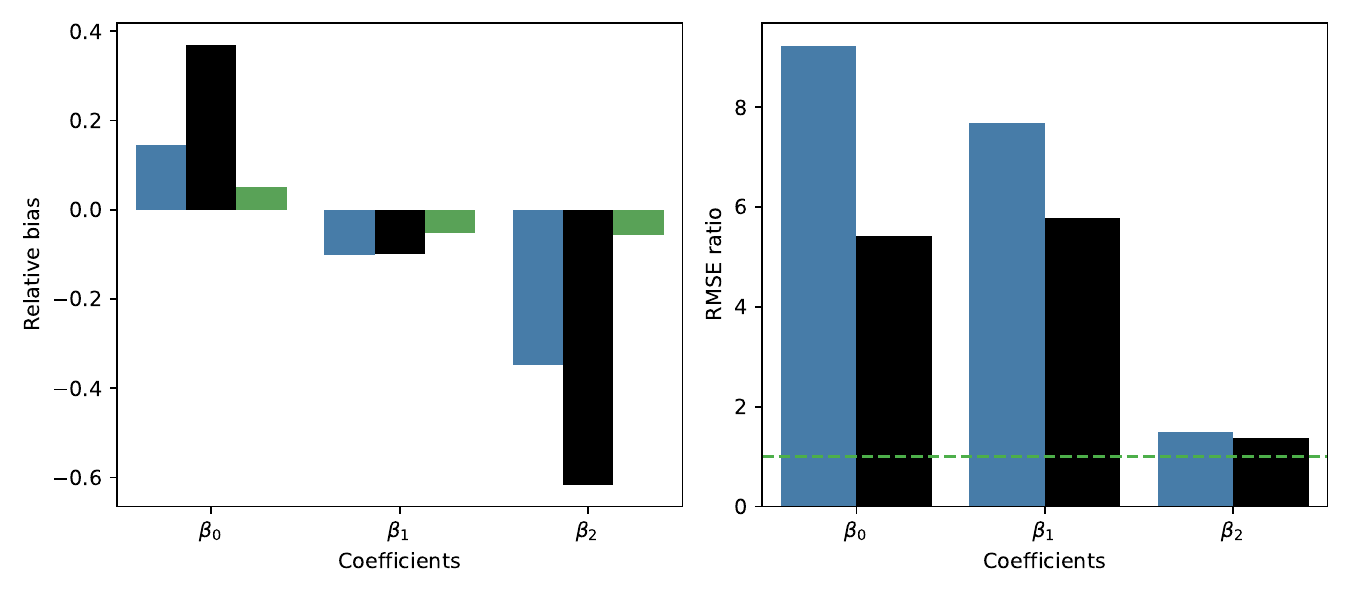}
         \caption{$n = 2000, \tau = 0.9$}
     \end{subfigure}
     \begin{subfigure}[b]{0.48\textwidth}
         \centering
         \includegraphics[width = \textwidth]{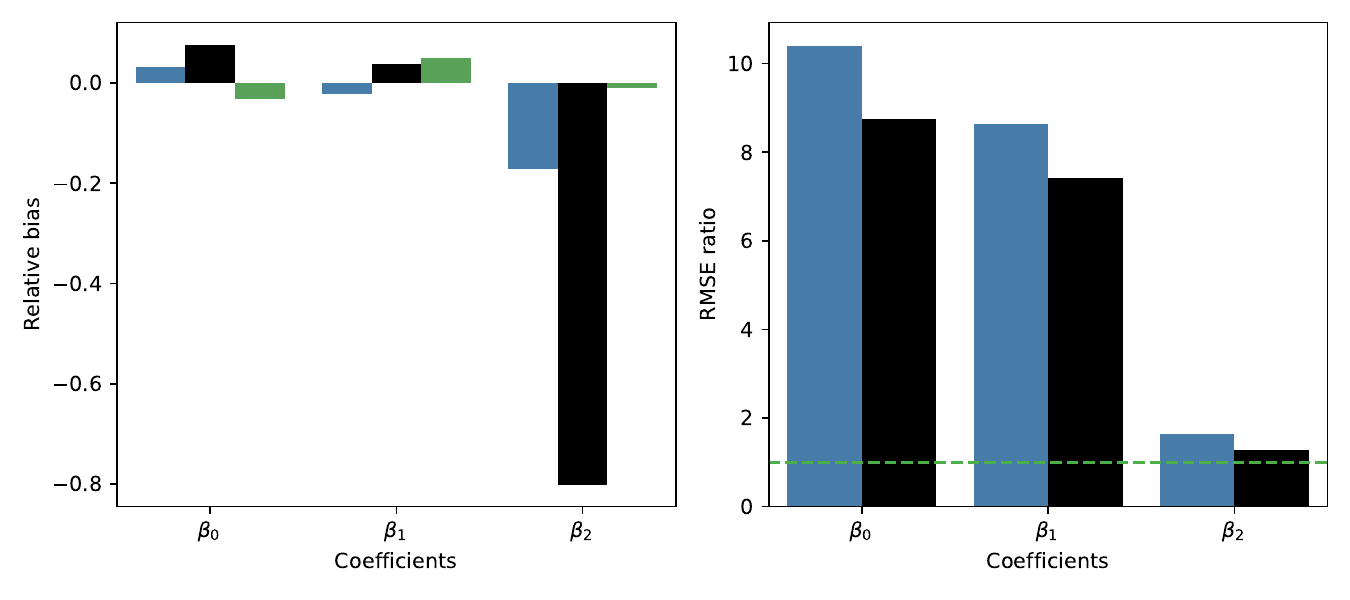}
         \caption{$n = 5000, \tau = 0.9$}
     \end{subfigure}
     \begin{subfigure}[b]{\textwidth}
         \centering
         \includegraphics[width = \textwidth]{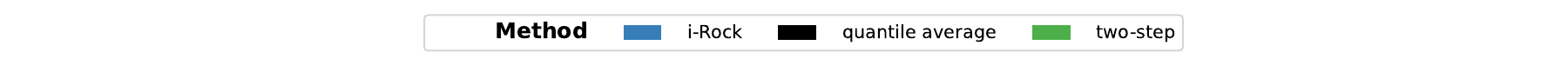}
     \end{subfigure}
\caption{Numerical comparisons of the i-Rock, the quantile average, and the two-step approach under model~\eqref{eq:2d} at various sample sizes and $\tau = 0.9$.}
\label{fig::2d_disc}
\end{figure}

In conclusion, when the quantile is linear in the covariates and the data exhibit little heterogeneity, the quantile average approach is nearly as efficient as the proposed method, albeit with slightly more bias. However, if the quantile is non-linear or the data are highly heterogeneous, the quantile average approach becomes less efficient than the proposed i-Rock approach and may display significant bias.

\subsection{Additional simulation results}
We consider the following case with unbounded covariates,
\begin{equation}\label{eq::gaussian}
    Y_i = \{1 - \log(1 - U)\} +  \{3 - 30 \log(1-U)\} X_i, \quad i=1,\ldots,n,
\end{equation}
where $X_i= |Z_i|$ and $Z_i \sim N(0,1)$. Figure~\ref{fig::gaussian} indicates that the proposed i-Rock approach is still significantly more efficient than the two-step approach with comparable bias, even with unbounded covariates. 


\begin{figure}[htb]
\centering 
         \includegraphics[width = \textwidth]{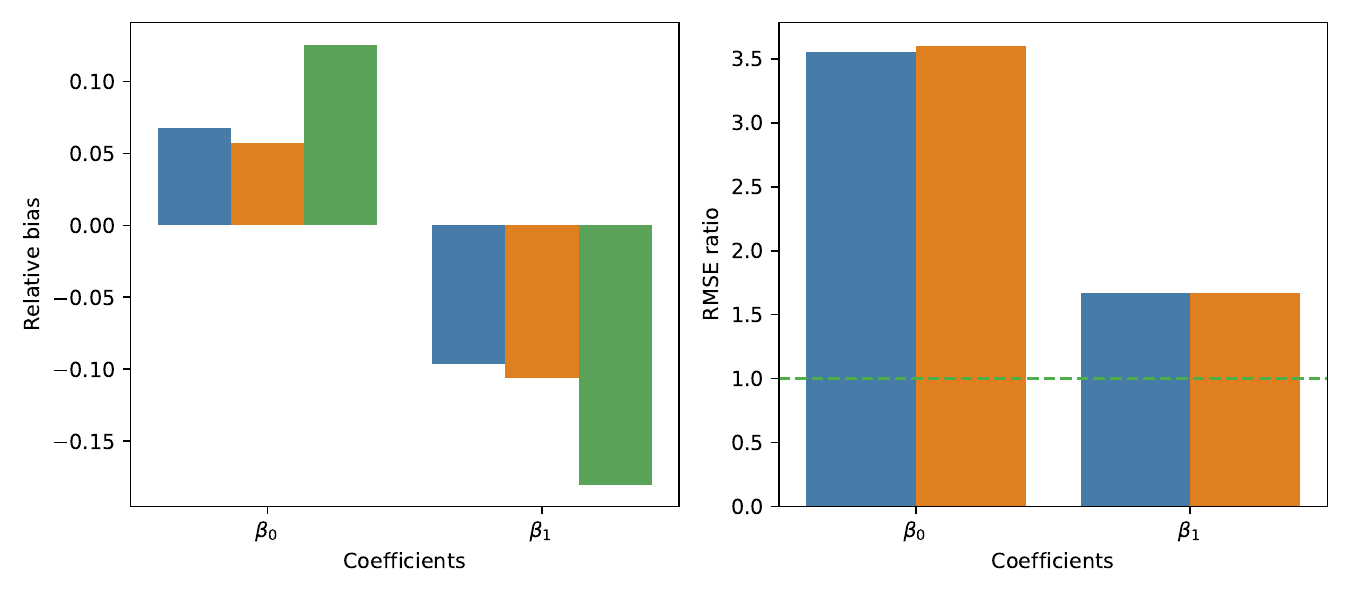}
     \begin{subfigure}[b]{\textwidth}
         \centering
        \includegraphics[width = \textwidth]{fig/legend.pdf}
     \end{subfigure}
\caption{Numerical comparisons of the i-Rock approach (with linear or B-spline quantile function estimation) and two-step estimator under linear heteroscedastic model~\eqref{eq::gaussian} when $n = 10000$ and $\tau = 0.9$.}
\label{fig::gaussian}
\end{figure}
\begin{figure}[tb]
\centering 
\begin{subfigure}[b]{0.48\textwidth}
         \centering
         \includegraphics[width = \textwidth]{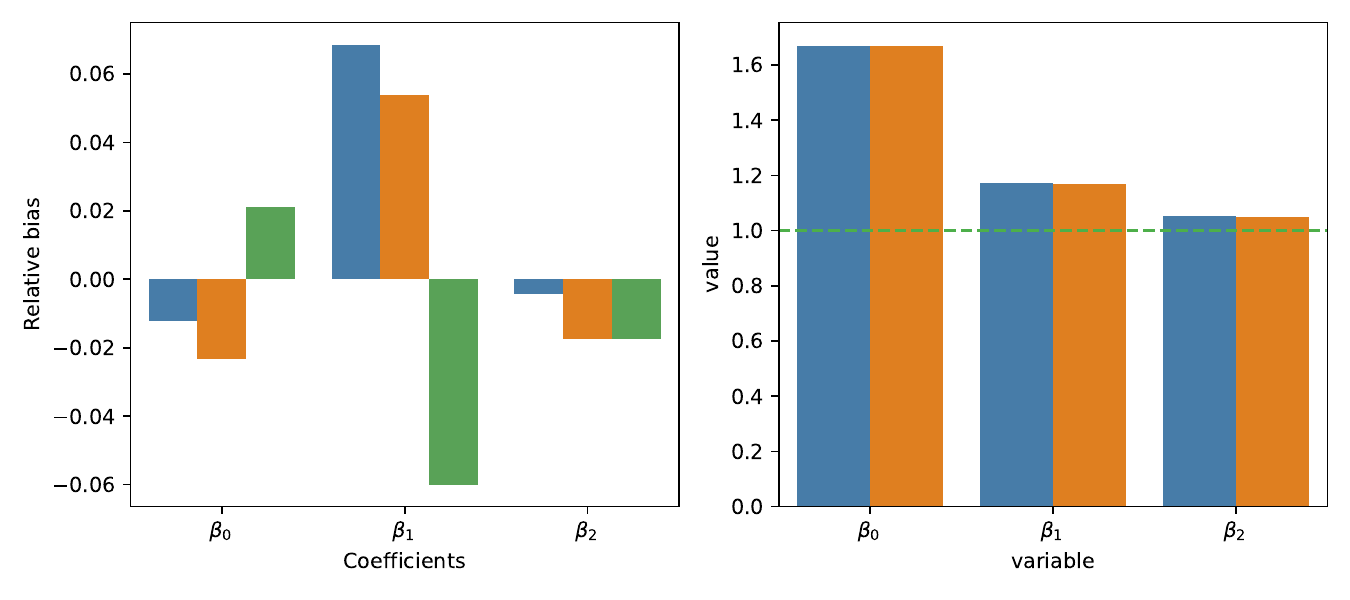}
         \caption{$n = 5000, \tau = 0.8$}
     \end{subfigure}
     \begin{subfigure}[b]{0.48\textwidth}
         \centering
         \includegraphics[width = \textwidth]{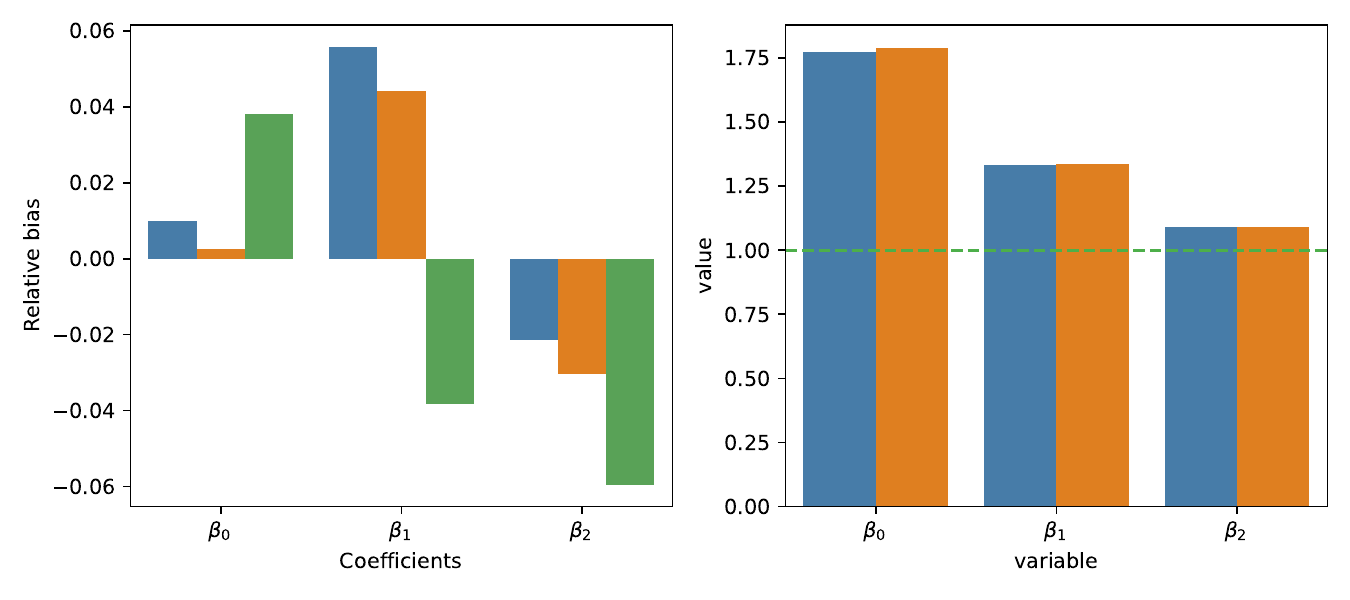}
         \caption{$n = 10000, \tau = 0.8$}
     \end{subfigure}
     \begin{subfigure}[b]{0.48\textwidth}
         \centering
         \includegraphics[width = \textwidth]{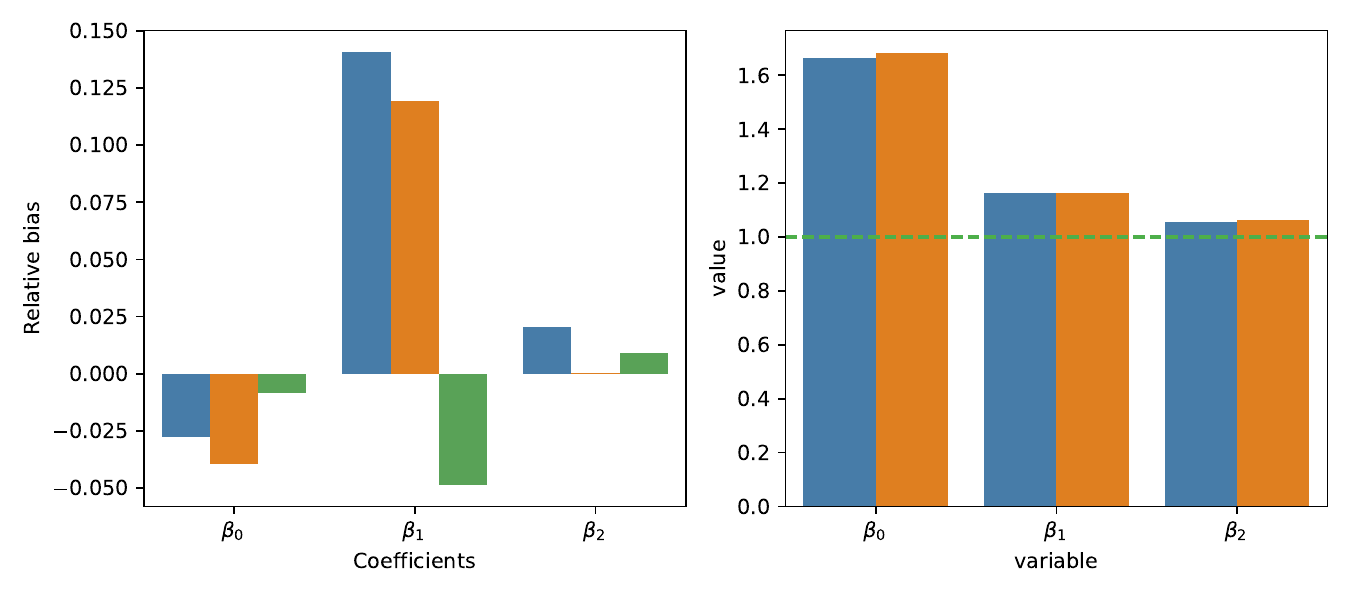}
         \caption{$n = 5000, \tau = 0.9$}
     \end{subfigure}
     \begin{subfigure}[b]{0.48\textwidth}
         \centering
         \includegraphics[width = \textwidth]{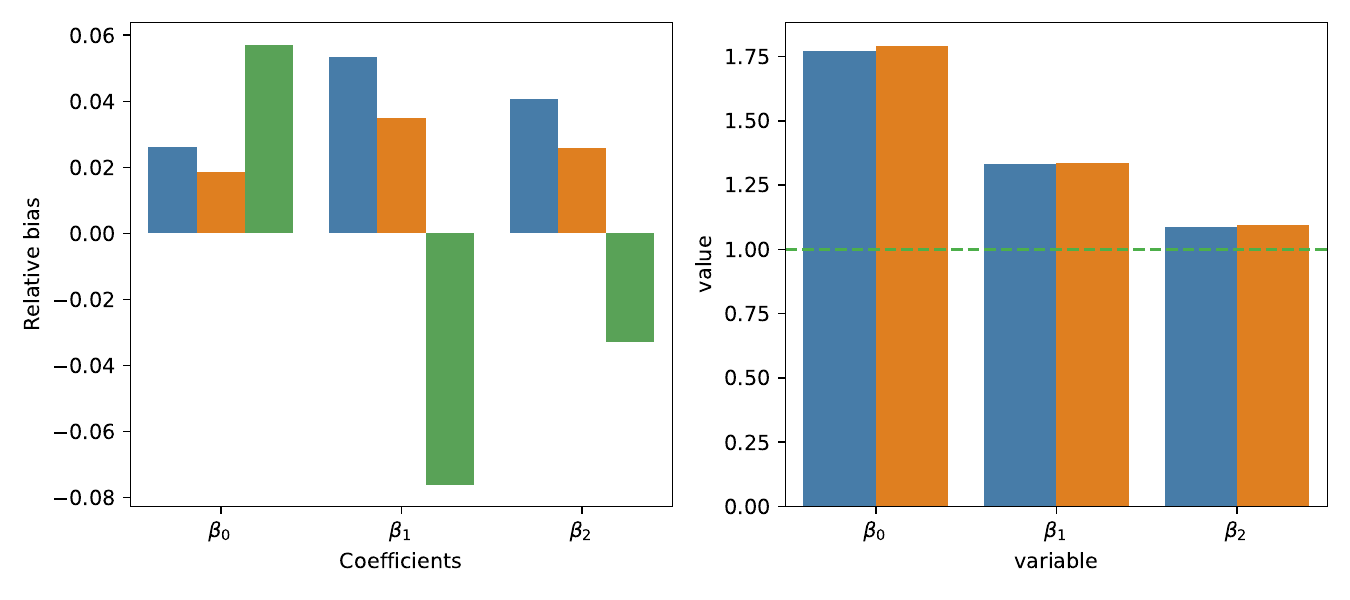}
         \caption{$n = 10000, \tau = 0.9$}
     \end{subfigure}
     \begin{subfigure}[b]{\textwidth}
         \centering
         \includegraphics[width = \textwidth]{fig/legend.pdf}
     \end{subfigure}
\caption{Numerical comparisons of the i-Rock approach (with linear or B-spline quantile function estimation) and two-step estimator under linear heteroscedastic model~\eqref{eq::2d_cont_disc} at various quantile levels and sample sizes.}
\label{fig::2d_cont_disc}
\end{figure}

To show how the proposed approach works with correlated covariates, we generate data as a random sample from a linear heteroscedastic model with two-dimensional covariates, namely,
\begin{equation}\label{eq::2d_cont_disc}
    Y_i = \{1 + U \} + (2 + 2 U) X_{i,1} + \{3 + 3 U\} X_{i,2}, \quad i=1,\ldots,n,
\end{equation}
where 
$U$ follows $U(0,1)$, $X_{i,1}$ given $X_{i,2}$  follows $U(0,3+X_{i,2})$,  
and $X_{i,2}$ are independently distributed from $\{0,1\}$ with equal probability. Figure~\ref{fig::2d_cont_disc} shows the relative bias and the RMSE ratio comparisons for $\tau \in \{0.8,0.9\}$ and $n\in\{5000,10000\}$. The i-Rock estimators with both the linear and the B-spline quantile regression estimation outperform the two-step estimator in bias and RMSE in all settings. The finding is very similar to Case 5.1 in the main manuscript with independent covariates. This suggests that our proposed approach is robust for moderately correlated covariates.

\subsection{Additional results to Section 6}
\begin{figure}[htbp]
    \centering
    \begin{subfigure}[b]{\textwidth}
         \centering
\includegraphics[width=\textwidth]{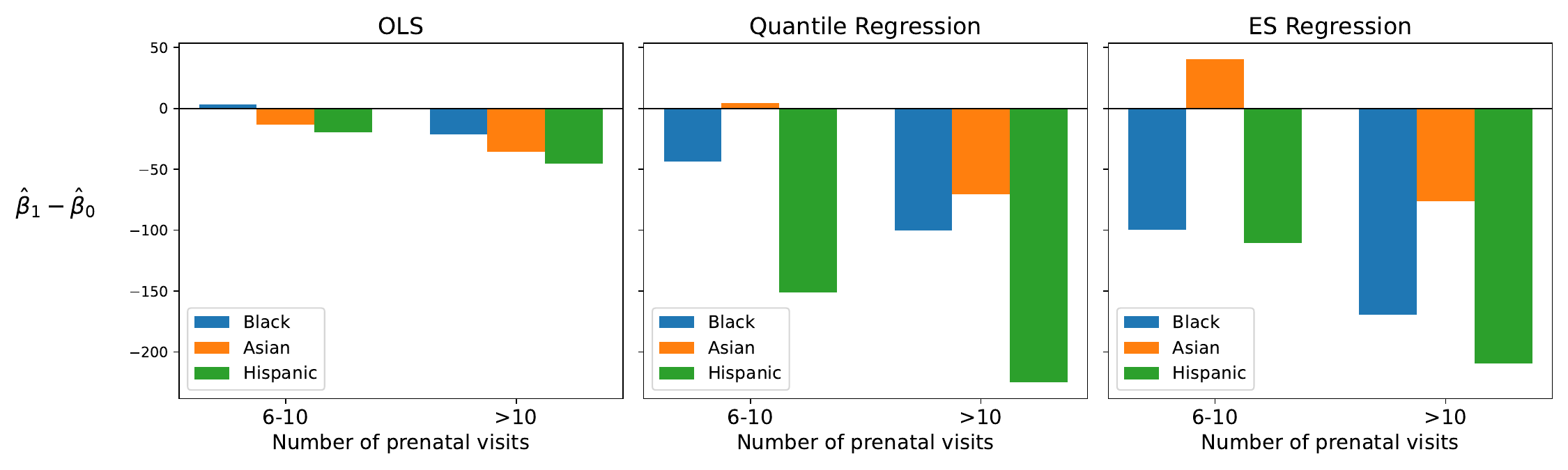}
    \end{subfigure}
    \caption{The quantities shown from part of $\beta_1 - \beta_0$ associated with Equation (27) of main manuscript are the birth weight disparities  of the disadvantaged groups for subgroups defined by the number of prenatal visits: $[6, 10]$ and $>10$, with the subgroup of  $ \leq 5$ prenatal  visits serving as the baseline. We report the results from three regressions, i.e., OLS, 0.05-th quantile regression, 0.05-th ES regression. }
    \label{fig:prenatal_visit_difference_in_difference_three}  
\end{figure}



\begin{table}[tb!]
\centering
\caption{ 
Estimated coefficients for (1) the lower ES regression of birth weight at the quantile level $\tau = 0.05$ using the i-Rock approach, (2) the 0.05-th quantile regression, and (3) the ordinary least squares (OLS). 
The numbers in the parenthesis show the standard errors.}
\label{tab::comparison_OLS}
\renewcommand{\arraystretch}{0.9} 
\resizebox{\textwidth}{!}{
\begin{tabular}{lccc|lccc}
  \toprule
  Covariates & 0.05 ES & 0.05 quantile & OLS &Covariates & 0.05 ES & 0.05 quantile & OLS\\
  \hline
\multirow{2}{*}{\shortstack[l]{\textbf{ (Intercept) }}} &
\multirow{2}{*}{\shortstack{1627.34\\{\tiny( 5.37 )}}}  & \multirow{2}{*}{\shortstack{2086.50\\{\tiny( 9.30 )}}}  & 
\multirow{2}{*}{\shortstack{3268.27\\{\tiny(  1.79 )}}} & \multirow{2}{*}{\shortstack[l]{\textbf{Gestational diabetes}\\{\tiny(baseline: no gestational diabetes)}} } & \multirow{2}{*}{\shortstack{-34.02\\{\tiny( 6.20 )}}} &\multirow{2}{*}{\shortstack{-33.50\\{\tiny( 3.81 )}}}  & \multirow{2}{*}{\shortstack{9.28\\{\tiny( 1.59 )}}}\\
\\ \hline 
\multirow{2}{*}{\shortstack[l]{\textbf{Race $=$ black}\\{\tiny(baseline: white)}} } & \multirow{2}{*}{\shortstack{-249.97\\{\tiny( 7.98 )}}} & \multirow{2}{*}{\shortstack{-231.25\\{\tiny( 3.94)}}}  & \multirow{2}{*}{\shortstack{-186.17\\{\tiny( 1.32 )}}} & \multirow{2}{*}{\shortstack[l]{\textbf{Gestational hypertension}\\{\tiny(baseline: no gestational hypertension)}} } & \multirow{2}{*}{\shortstack{-447.88\\{\tiny( 4.57 )}}} &\multirow{2}{*}{\shortstack{-415.00\\{\tiny( 3.91 )}}}  &  \multirow{2}{*}{\shortstack{-217.27\\{\tiny( 1.46 )}}} 
\\
\\ \hline 
\multirow{2}{*}{\shortstack[l]{\textbf{Race $=$ asian}\\{\tiny(baseline: white)}}} & \multirow{2}{*}{\shortstack{-193.55\\{\tiny( 6.27 )}}} &\multirow{2}{*}{\shortstack{-206.50\\{\tiny( 6.10 )}}}  &  \multirow{2}{*}{\shortstack{-227.26\\{\tiny( 1.98 )}}} & \multirow{2}{*}{\shortstack[l]{\textbf{Cigarettes at 3rd trimester}\\{\tiny(baseline: no cigarette use at 3rd trimester)}} } & \multirow{2}{*}{\shortstack{-174.59\\{\tiny( 5.68 )}}}&\multirow{2}{*}{\shortstack{-201.75\\{\tiny( 8.21 )}}}  &  \multirow{2}{*}{\shortstack{-203.55\\{\tiny( 2.54 )}}} \\
\\ \hline 
\multirow{2}{*}{\shortstack[l]{\textbf{Race $=$ hispanic}\\{\tiny(baseline: white)}} }  & \multirow{2}{*}{\shortstack{-44.65\\{\tiny( 2.77 )}}}& \multirow{2}{*}{\shortstack{-40.25\\{\tiny( 2.41 )}}}  & \multirow{2}{*}{\shortstack{-61.77\\{\tiny( 1.12 )}}} & \multirow{2}{*}{\shortstack[l]{\textbf{ Mother's age $<20$ }\\{\tiny(baseline: age $[20,34]$)}} } & \multirow{2}{*}{\shortstack{-19.98\\{\tiny( 14.25 )}}}&\multirow{2}{*}{\shortstack{-38.50\\{\tiny( 7.90 )}}}  &  \multirow{2}{*}{\shortstack{-78.88\\{\tiny( 2.27 )}}} \\
\\ \hline
\multirow{2}{*}{\shortstack[l]{\textbf{Prenatal visits $\in [6,10]$ }\\{\tiny(baseline: $[0,5]$)}} } & \multirow{2}{*}{\shortstack{437.05\\{\tiny( 6.24 )}}} &\multirow{2}{*}{\shortstack{378.25\\{\tiny( 8.59 )}}}  &  \multirow{2}{*}{\shortstack{102.65\\{\tiny( 1.64 )}}} 
& \multirow{2}{*}{\shortstack[l]{\textbf{ Mother's age $>34$ }\\{\tiny(baseline: age $[20,34]$)}} } & \multirow{2}{*}{\shortstack{-94.45\\{\tiny( 4.31 )}}}&\multirow{2}{*}{\shortstack{-66.75\\{\tiny( 3.13 )}}} & \multirow{2}{*}{\shortstack{-8.38\\{\tiny( 1.11 )}}} \\
\\ \hline
\multirow{2}{*}{\shortstack[l]{\textbf{Prenatal visits $> 10$ }\\{\tiny(baseline: $[0,5]$)}} } & \multirow{2}{*}{\shortstack{832.55\\{\tiny( 5.85 )}}}&\multirow{2}{*}{\shortstack{663.50\\{\tiny( 8.45 )}}} & \multirow{2}{*}{\shortstack{228.33\\{\tiny( 1.60 )}}} & \multirow{2}{*}{\shortstack[l]{\textbf{Receipt of WIC}\\{\tiny(baseline: not receipt of WIC)}}} & \multirow{2}{*}{\shortstack{8.25\\{\tiny( 3.11 )}}}&\multirow{2}{*}{\shortstack{-1.75\\{\tiny( 2.93 )}}} & \multirow{2}{*}{\shortstack{4.85\\{\tiny( 1.06 )}}}
\\
\\ \hline
\multirow{2}{*}{\shortstack[l]{\textbf{Education: $\geq$ college}\\{\tiny(baseline: high school and below)}}}
&
\multirow{2}{*}{\shortstack{63.21\\{\tiny( 4.36 )}}} &\multirow{2}{*}{\shortstack{51.75\\{\tiny( 2.90 )}}} 
&
\multirow{2}{*}{\shortstack{32.08\\{\tiny( 1.21 )}}} 
& \multirow{2}{*}{\shortstack[l]{\textbf{Unmarried}\\{\tiny(baseline: married)}} } & \multirow{2}{*}{\shortstack{-77.56\\{\tiny( 5.99 )}}}&\multirow{2}{*}{\shortstack{-71.50\\{\tiny( 3.78 )}}} & \multirow{2}{*}{\shortstack{-56.91\\{\tiny( 1.05 )}}} \\
\\ \hline
\multirow{2}{*}{\shortstack[l]{\textbf{Education: some college}\\{\tiny(baseline: high school and below)}}}
& \multirow{2}{*}{\shortstack{-4.25\\{\tiny( 4.72 )}}} &\multirow{2}{*}{\shortstack{1.50\\{\tiny( 2.90 )}}} & \multirow{2}{*}{\shortstack{14.08\\{\tiny( 1.14 )}}}  \\ \\
\bottomrule
\end{tabular}
}
\end{table}

\begin{table}[tb]
\centering
\caption{{ 
Estimated coefficients for the lower ES regression of birth weight at different quantile levels $\tau \in 
\{0.01, 0.025, 0.05\}$ using the i-Rock approach and two-step approach. The numbers in the parenthesis show the standard errors.
}}
\label{tab::BW-mRock}
\renewcommand{\arraystretch}{0.9} 
\resizebox{\textwidth}{!}{
\begin{tabular}{l|cccc|cccc}
  \toprule
  Covariates & \multicolumn{8}{c}{Coefficients}\\
  & \multicolumn{4}{c}{I-Rock} & \multicolumn{4}{c}{Two-step}\\
  \multicolumn{1}{r|}{$\tau$}& 0.01 & 0.025 & 0.05 &0.5 & 0.01 & 0.025 & 0.05& 0.5\\
  \hline
\multirow{2}{*}{\shortstack[l]{\textbf{ (Intercept) }}} &
\multirow{2}{*}{\shortstack{1081.29\\{\tiny(14.04)}}}&
\multirow{2}{*}{\shortstack{1344.10\\{\tiny(10.57)}}}&
\multirow{2}{*}{\shortstack{1627.34\\{\tiny(6.33)}}}& 
\multirow{2}{*}{\shortstack{2812.83\\{\tiny(2.78)}}}&
\multirow{2}{*}{\shortstack{1087.19\\{\tiny(13.20)}}}&
\multirow{2}{*}{\shortstack{1340.72\\{\tiny(8.83)}}}&
\multirow{2}{*}{\shortstack{1618.95\\{\tiny(6.19)}}}&
\multirow{2}{*}{\shortstack{2805.71\\{\tiny(2.63)}}}\\
\\ \hline 
\multirow{2}{*}{\shortstack[l]{\textbf{Race $=$ black}\\{\tiny(baseline: white)}} } & 
\multirow{2}{*}{\shortstack{-258.29\\{\tiny(10.91)}}} & 
\multirow{2}{*}{\shortstack{-260.51\\{\tiny(8.56)}}} & 
\multirow{2}{*}{\shortstack{-249.97\\{\tiny(5.52)}}}& 
\multirow{2}{*}{\shortstack{-198.10\\{\tiny(1.54)}}}& 
\multirow{2}{*}{\shortstack{-265.41\\{\tiny(11.21)}}} & 
\multirow{2}{*}{\shortstack{-256.47\\{\tiny(7.45)}}} & 
\multirow{2}{*}{\shortstack{-246.44\\{\tiny(4.81)}}} & 
\multirow{2}{*}{\shortstack{-196.31\\{\tiny(1.46)}}} 
\\
\\ \hline 
\multirow{2}{*}{\shortstack[l]{\textbf{Race $=$ asian}\\{\tiny(baseline: white)}}} & 
\multirow{2}{*}{\shortstack{-161.53\\{\tiny(15.53)}}}& 
\multirow{2}{*}{\shortstack{-180.43\\{\tiny(10.65)}}}& 
\multirow{2}{*}{\shortstack{-193.55\\{\tiny(7.74)}}}& 
\multirow{2}{*}{\shortstack{-214.91\\{\tiny(2.26)}}}& 
\multirow{2}{*}{\shortstack{-153.69\\{\tiny(13.91)}}}& 
\multirow{2}{*}{\shortstack{-168.56\\{\tiny(9.70)}}}& 
\multirow{2}{*}{\shortstack{-182.07\\{\tiny(6.92)}}}& 
\multirow{2}{*}{\shortstack{-214.73\\{\tiny(2.02)}}}\\
\\ \hline 
\multirow{2}{*}{\shortstack[l]{\textbf{Race $=$ hispanic}\\{\tiny(baseline: white)}} } & 
\multirow{2}{*}{\shortstack{-43.93\\{\tiny(9.11)}}}& 
\multirow{2}{*}{\shortstack{-41.54\\{\tiny(5.42)}}}& 
\multirow{2}{*}{\shortstack{-61.77\\{\tiny(3.48)}}}& 
\multirow{2}{*}{\shortstack{-53.46\\{\tiny(1.18)}}}& 
\multirow{2}{*}{\shortstack{-46.06\\{\tiny(8.84)}}}& 
\multirow{2}{*}{\shortstack{-37.22\\{\tiny(5.89)}}}& 
\multirow{2}{*}{\shortstack{-34.49\\{\tiny(3.77)}}}& 
\multirow{2}{*}{\shortstack{-53.46\\{\tiny(1.14)}}}\\
\\ \hline
\multirow{2}{*}{\shortstack[l]{\textbf{Prenatal visits $\in [6,10]$ }\\{\tiny(baseline: $[0,5]$)}} } & 
\multirow{2}{*}{\shortstack{350.53\\{\tiny(11.22)}}}& 
\multirow{2}{*}{\shortstack{435.78\\{\tiny(7.97)}}}& 
\multirow{2}{*}{\shortstack{437.05\\{\tiny(6.40)}}}& 
\multirow{2}{*}{\shortstack{154.48\\{\tiny(3.06)}}}& 
\multirow{2}{*}{\shortstack{344.26\\{\tiny(9.07)}}}& 
\multirow{2}{*}{\shortstack{438.75\\{\tiny(6.79)}}}& 
\multirow{2}{*}{\shortstack{444.79\\{\tiny(6.44)}}}& 
\multirow{2}{*}{\shortstack{161.06\\{\tiny(2.87)}}}
 \\
\\ \hline
\multirow{2}{*}{\shortstack[l]{\textbf{Prenatal visits $> 10$ }\\{\tiny(baseline: $[0,5]$)}} } &
\multirow{2}{*}{\shortstack{854.48\\{\tiny(9.56)}}}&
\multirow{2}{*}{\shortstack{900.25\\{\tiny(8.73)}}}&
\multirow{2}{*}{\shortstack{832.55\\{\tiny(7.11)}}}& 
\multirow{2}{*}{\shortstack{320.53\\{\tiny(2.81)}}}&
\multirow{2}{*}{\shortstack{847.21\\{\tiny(7.68)}}}&
\multirow{2}{*}{\shortstack{902.18\\{\tiny(7.26)}}}&
\multirow{2}{*}{\shortstack{840.83\\{\tiny(6.92)}}}&
\multirow{2}{*}{\shortstack{329.99\\{\tiny(2.93)}}}
\\
\\ \hline
 \multirow{2}{*}{\shortstack[l]{\textbf{Gestational diabetes}\\{\tiny(baseline: no gestational diabetes)}} } & 
 \multirow{2}{*}{\shortstack{-13.16\\{\tiny(18.17)}}}& 
 \multirow{2}{*}{\shortstack{-30.80\\{\tiny(10.13)}}}& 
 \multirow{2}{*}{\shortstack{-34.02\\{\tiny(8.17)}}}& 
 \multirow{2}{*}{\shortstack{-17.07\\{\tiny(2.46)}}}& 
 \multirow{2}{*}{\shortstack{-4.08\\{\tiny(12.95)}}}& 
 \multirow{2}{*}{\shortstack{-16.72\\{\tiny(9.99)}}}& 
 \multirow{2}{*}{\shortstack{-24.61\\{\tiny(8.03)}}}& 
 \multirow{2}{*}{\shortstack{-16.04\\{\tiny(2.57)}}}
 \\ \\ \hline
  \multirow{2}{*}{\shortstack[l]{\textbf{Gestational hypertension}\\{\tiny(baseline: no gestational hypertension)}} } & 
  \multirow{2}{*}{\shortstack{-422.36\\{\tiny(10.67)}}}&
  \multirow{2}{*}{\shortstack{-454.03\\{\tiny(6.38)}}}&
  \multirow{2}{*}{\shortstack{-447.88\\{\tiny(5.38)}}}&
  \multirow{2}{*}{\shortstack{-270.77\\{\tiny(2.35)}}}&
  \multirow{2}{*}{\shortstack{-418.85\\{\tiny(9.26)}}}&
  \multirow{2}{*}{\shortstack{-452.79\\{\tiny(6.67)}}}&
  \multirow{2}{*}{\shortstack{-458.08\\{\tiny(4.69)}}}&
  \multirow{2}{*}{\shortstack{-282.37\\{\tiny(2.15)}}}
  \\ \\ \hline
  \multirow{2}{*}{\shortstack[l]{\textbf{Cigarettes at 3rd trimester}\\{\tiny(baseline: no cigarette use at 3rd trimester)}} } & 
  \multirow{2}{*}{\shortstack{-105.82\\{\tiny(23.73)}}}& 
  \multirow{2}{*}{\shortstack{-153.01\\{\tiny(18.76)}}}& 
  \multirow{2}{*}{\shortstack{-174.59\\{\tiny(12.17)}}}& 
  \multirow{2}{*}{\shortstack{-209.12\\{\tiny(3.77)}}}& 
  \multirow{2}{*}{\shortstack{-106.01\\{\tiny(21.27)}}}& 
  \multirow{2}{*}{\shortstack{-133.00\\{\tiny(14.80)}}}& 
  \multirow{2}{*}{\shortstack{-150.98\\{\tiny(10.68)}}}& 
  \multirow{2}{*}{\shortstack{-210.29\\{\tiny(3.20)}}} \\\\ \hline
  \multirow{2}{*}{\shortstack[l]{\textbf{ Mother's age $<20$ }\\{\tiny(baseline: age $[20,34]$)}} } & 
  \multirow{2}{*}{\shortstack{14.81\\{\tiny(11.99)}}}& 
  \multirow{2}{*}{\shortstack{-8.61\\{\tiny(12.29)}}}& 
  \multirow{2}{*}{\shortstack{-19.98\\{\tiny(8.84)}}}& 
  \multirow{2}{*}{\shortstack{-58.57\\{\tiny(2.59)}}}& 
  \multirow{2}{*}{\shortstack{24.17\\{\tiny(12.74)}}}& 
  \multirow{2}{*}{\shortstack{-3.57\\{\tiny(9.52)}}}& 
  \multirow{2}{*}{\shortstack{-19.85\\{\tiny(7.33)}}}& 
  \multirow{2}{*}{\shortstack{-58.62\\{\tiny(2.69)}}}\\ \\ \hline
  \multirow{2}{*}{\shortstack[l]{\textbf{ Mother's age $>34$ }\\{\tiny(baseline: age $[20,34]$)}} } & 
  \multirow{2}{*}{\shortstack{-109.76\\{\tiny(6.59)}}}& 
  \multirow{2}{*}{\shortstack{-109.05\\{\tiny(4.76)}}}& 
  \multirow{2}{*}{\shortstack{-94.45\\{\tiny(3.74)}}}& 
  \multirow{2}{*}{\shortstack{-29.00\\{\tiny(1.58)}}}& 
  \multirow{2}{*}{\shortstack{-105.82\\{\tiny(6.53)}}}& 
  \multirow{2}{*}{\shortstack{-105.36\\{\tiny(4.13)}}}& 
  \multirow{2}{*}{\shortstack{-93.95\\{\tiny(3.40)}}}& 
  \multirow{2}{*}{\shortstack{-31.48\\{\tiny(1.67)}}} \\ \\ \hline
  \multirow{2}{*}{\shortstack[l]{\textbf{Receipt of WIC}\\{\tiny(baseline: not receipt of WIC)}}} & 
  \multirow{2}{*}{\shortstack{28.80\\{\tiny(10.21)}}}& 
  \multirow{2}{*}{\shortstack{17.59\\{\tiny(7.11)}}}& 
  \multirow{2}{*}{\shortstack{8.25\\{\tiny(4.41)}}}& 
  \multirow{2}{*}{\shortstack{-5.53\\{\tiny(1.13)}}}& 
  \multirow{2}{*}{\shortstack{31.41\\{\tiny(9.87)}}}& 
  \multirow{2}{*}{\shortstack{20.40\\{\tiny(6.78)}}}& 
  \multirow{2}{*}{\shortstack{12.61\\{\tiny(4.59)}}}& 
  \multirow{2}{*}{\shortstack{-4.58\\{\tiny(1.10)}}}\\ \\ \hline
  \multirow{2}{*}{\shortstack[l]{\textbf{Unmarried}\\{\tiny(baseline: married)}} } & 
  \multirow{2}{*}{\shortstack{-72.60\\{\tiny(9.44)}}} & 
  \multirow{2}{*}{\shortstack{-76.23\\{\tiny(6.22)}}} & 
  \multirow{2}{*}{\shortstack{-77.56\\{\tiny(4.00)}}} & 
  \multirow{2}{*}{\shortstack{-56.89\\{\tiny(1.54)}}}& 
  \multirow{2}{*}{\shortstack{-73.41\\{\tiny(9.04)}}} & 
  \multirow{2}{*}{\shortstack{-77.99\\{\tiny(5.97)}}} & 
  \multirow{2}{*}{\shortstack{-78.35\\{\tiny(3.84)}}}& 
  \multirow{2}{*}{\shortstack{-59.81\\{\tiny(1.51)}}} \\\\\hline
  \multirow{2}{*}{\shortstack[l]{\textbf{Education: $\geq$ college}\\{\tiny(baseline: high school and below)}}}
&
\multirow{2}{*}{\shortstack{75.03\\{\tiny(11.98)}}} &
\multirow{2}{*}{\shortstack{74.28\\{\tiny(7.49)}}} &
\multirow{2}{*}{\shortstack{63.21\\{\tiny(4.86)}}} &
\multirow{2}{*}{\shortstack{42.64\\{\tiny(1.21)}}} &
\multirow{2}{*}{\shortstack{77.99\\{\tiny(11.58)}}} &
\multirow{2}{*}{\shortstack{65.96\\{\tiny(7.56)}}} &
\multirow{2}{*}{\shortstack{54.69\\{\tiny(5.23)}}}&
\multirow{2}{*}{\shortstack{40.25\\{\tiny(1.23)}}} \\
\\ \hline
\multirow{2}{*}{\shortstack[l]{\textbf{Education: some college}\\{\tiny(baseline: high school and below)}}}
& \multirow{2}{*}{\shortstack{-11.17\\{\tiny(12.41)}}}
& \multirow{2}{*}{\shortstack{-5.51\\{\tiny(8.64)}}}
& \multirow{2}{*}{\shortstack{-4.25\\{\tiny(5.48)}}}
& \multirow{2}{*}{\shortstack{12.67\\{\tiny(1.39)}}}
& \multirow{2}{*}{\shortstack{-6.68\\{\tiny(10.83)}}}
& \multirow{2}{*}{\shortstack{-11.71\\{\tiny(8.15)}}}
& \multirow{2}{*}{\shortstack{-11.80\\{\tiny(5.93)}}}& 
\multirow{2}{*}{\shortstack{9.68\\{\tiny(1.44)}}}
\\ \\
\bottomrule
\end{tabular}
}
\end{table}   

\begin{figure}[tb]
\includegraphics[scale=0.4]{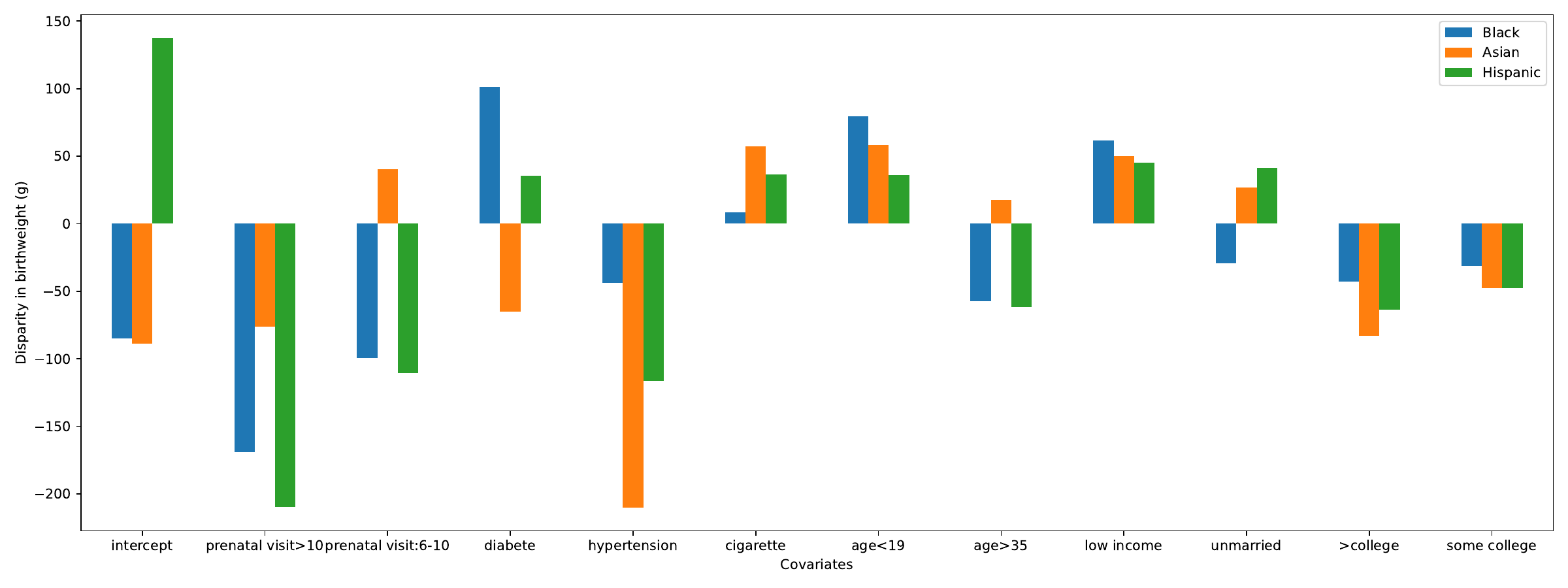}
    \caption{The estimator for $\beta_1 - \beta_0$ across covariates, where $\beta_1$, $\beta_0$ represents the coefficient for lower $0.05$ ES regression of birthweight for the disadvantaged group (i.e., black, asian, and hispanic) and advantaged group (i.e., white), respectively. }
\label{fig::disparity_driver}
\end{figure}

As a complement to the results in Section 6 of the main manuscript, we present a comparison of the ES regression, the quantile regression, and the ordinary least squares regression (OLS),  in Table~\ref{tab::comparison_OLS}. Let us highlight just two additional findings from this exercise. 

First, prenatal visits contribute to birth weight in general, but much more so at the lower tail.  After adjusting for other factors used in the model, frequent prenatal visits ($> 10$ times) improves the average birth weight by 228 grams over the baseline group of infrequent prenatal visits of 0-5 times. If we consider the birth weight in the lowest 5 percentiles, frequent prenatal visits improves the average birth weight by 832 grams. Note that 832 grams for low birth weight babies (around 1625 grams) is a far more significant improvement than a 228-gram change for a typical baby (over 3000 grams of birth weight).  The practical implication of the study results from the ES regression speaks for itself. 

Second, a comparison of the effect sizes over racial subpopulations between the 0.05-th ES, the 0.05-th  quantile, and the mean (through the OLS) yields useful findings. Take a look at the black vs. white populations, the average birth weight difference is 186 grams. The differences are 231 grams at the 0.05-th quantile, and 250 grams at the 0.05-th expected shortfall. They suggest that the birth weight of the black populations (relative to the white population) widens gradually in the lower tail. However, it is a different story for the Asian population vs. the white population. Still, the Asian babies weigh less, but the difference (after adjusting for other factors in the model) is 194 grams at the 0.05-th expected shortfall, 207 grams at the 0.05-th quantile, and 227 grams at the mean, which suggests that the difference narrows gradually in the lower tail.

In Figure~\ref{fig:prenatal_visit_difference_in_difference_three}, we also study the differences in the disparity between each disadvantaged racial groups and
white individuals when the number of prenatal visits is in $[6, 10]$ or $> 10$, as compared to the subgroup of no more than 5 prenatal visits, using OLS, quantile regression, and ES regression. Based on quantile and expected shortfall regressions, we find that disparities between Black and Hispanic individuals relative to White individuals are most pronounced when the number of prenatal visits is high, whereas such disparities are not captured by OLS, particularly for the Black population.

In addition, we present the estimates for $\beta_1-\beta_0$ in Equation (27) element-wise across different covariates in Figure~\ref{fig::disparity_driver}. From the magnitude of the estimates, we observe that the number of prenatal visits and gestational hypertension are two noteworthy factors to the health disparity. 

Similar findings are observed if we fit the ES regression at multiple quantile levels, as shown in Table~\ref{tab::BW-mRock}. In addition, we obtained nearly identical results whether we use the i-Rock estimator or the two-step estimator. Certainly the linear quantile assumption in this example was not found to be violated.

\bibliographystyle{agsm}
\bibliography{SuperQuant}

\end{document}